\documentclass[a4paper,10pt]{article}

\RequirePackage{natbib}
\RequirePackage[colorlinks,citecolor=blue,urlcolor=blue]{hyperref}
\usepackage{bm}
\usepackage{graphicx,xcolor}
\usepackage{placeins}
\usepackage{amsmath,amsthm,amssymb}
\usepackage{booktabs}
\usepackage{tikz}

\usepackage{setspace}

\newtheorem{lem}{Lemma}

\newtheorem{theo}{Theorem}
\newtheorem{theoA}{Theorem B}
\newtheorem{assum}{Assumption}
\newtheorem{rem}{Remark}
\input{ee.sty}

\newcommand{\pim}{\underset{T\rightarrow\infty}{\text{plim}}}
\newcommand{\lm}{\underset{T\rightarrow\infty}{\text{lim}}}
\newcommand{\inp}{\stackrel{p}{\rightarrow}}
\newcommand{\ins}{\stackrel{p^b}{\rightarrow}}

\newcommand{\op}{o_p}
\newcommand{\Op}{O_p}

\usepackage[normalem]{ulem}

\newcommand{\vlam}{\bm \lambda}
\renewcommand{\vec}{\mbox{\bf vect}}
\newcommand{\vecc}{\mbox{vec}}

\renewcommand{\diag}{\mbox{\bf diag}}
\newcommand{\vbetai}{\vbeta_{(i)}}

\renewcommand{\vepsi}{\bm \epsilon}

\newcommand{\black}[1]{{\color{black}#1}}

\newcommand{\weaks}{\stackrel{d_p^b}{\Rightarrow}}
\newcommand{\mat}{\mbox{\bf mat}}

\newtheorem*{namedthm}{\namedthmname}
\newcounter{namedthm}
\newenvironment{named}[1]
  {\def\namedthmname{#1}%
   \refstepcounter{namedthm}%
   \namedthm\def\@currentlabel{#1}}
  {\endnamedthm}
\makeatother

\renewenvironment{equation*}{\gather\textstyle}{\nonumber\endgather}
\renewcommand{\var}{\mbox{Var}}
\renewcommand{\cov}{\mbox{Cov}}

\begin{document}

\title{Bootstrapping Structural Change Tests}
\author{Otilia Boldea, \small{Tilburg University}\\Adriana Cornea-Madeira, \small{University of York}\\Alastair R. Hall, \small{University of Manchester}}
\maketitle \thispagestyle{empty}
\begin{abstract}
This paper analyses the use of bootstrap methods to test for parameter change in linear models estimated via Two Stage Least Squares (2SLS). Two types of test are considered:  one where the null hypothesis is of no change and the alternative hypothesis involves discrete change at $k$ unknown break-points in the sample; and a second test where the null hypothesis is that there is discrete parameter change at $l$ break-points in the sample against an alternative in which the parameters change at $l+1$ break-points. In both cases, we consider inferences based on a $\sup$-$Wald$-type statistic using either the wild recursive bootstrap or the wild fixed bootstrap. We establish the asymptotic validity of these bootstrap tests under a set of general conditions that allow the errors to exhibit conditional and/or unconditional heteroskedasticity, and report results from a simulation study that indicate the tests yield reliable inferences in the sample sizes often encountered in macroeconomics. The analysis covers the cases where the first-stage estimation of 2SLS involves a model whose parameters are either constant or themselves subject to discrete parameter change. If the errors exhibit unconditional heteroscedasticity and/or the reduced form is unstable then the bootstrap methods are particularly attractive because the limiting distributions of the test statistics are not pivotal.

\vspace*{0.2in}
 \vspace*{0.2in}
\par
\noindent
----------------------------------------------------------------------------------------------------------------------------------------------

\noindent
\tiny{The first author acknowledges the support of the VENI Grant 451-11-001.  The second author acknowledges the support of the British Academy's PDF/2009/370. The third author acknowledges the support of ESRC grant RES-062-23-1351. We thank the following for very valuable comments on this work: Yacine A\"{i}t-Sahalia, Donald Andrews, Giuseppe Cavaliere, Xu Cheng, Russell Davidson, Frank Diebold, S\'ilvia Gon\c{c}alves, Bruce Hansen, Jan Magnus, Ulrich M\"uller, Anders Rahbek, Nikolaus Schweizer, Terence Tao, three anonymous referees, and the participants at University of Pennsylvania Econometrics Seminar, 2011; the Workshop on Recent Developments in Econometric Analysis, University of Liverpool, 2012; Tilburg University seminar, 2012; University of York seminar 2012; the European Meeting of the Econometric Society, Malaga 2012 and Toulouse 2014; CFE-ERCIM Conference, Pisa 2014; Durham University Business School seminar, 2015; Inference Issues in Econometrics, Amsterdam 2017; IAAE Conference, Montreal, 2018.}
\end{abstract}

\pagebreak
\section{Introduction}

Linear models with endogenous regressors are commonly employed in time series econometric analysis.\footnote{For example, \citet{Brady:2008} examines consumption smoothing using by regressing consumption growth on consumer credit, the latter being endogenous  because it depends on liquidity constraints. \citet{Zhang/Osborn/Kim:2008}, \citet{Kleibergen/Mavroeidis:2009}, \citet{Hall/Han/Boldea:2012} and \citet{Kim/Manopimoke/Nelson:2014} investigate the New Keynesian Phillips curve, where inflation is driven by expected inflation and marginal costs, both endogenous since they are correlated with inflation surprises. \citet{Bunzel/Enders:2010} and \citet{Qian/Su:2014} estimate the forward-looking Taylor rule, a model where the Federal fund rate is set based on expected inflation and output, both endogenous as they depend either on forecast errors or on current macroeconomic shocks. All these studies test for structural change in their estimated equations as part of their analysis.} In many cases, the parameters of these models are assumed constant throughout the sample. However, given the span of many economic time series data sets, this assumption may be questionable and a more appropriate specification may involve parameters that change value during the sample period. Such parameter changes could reflect legislative, institutional or technological changes, shifts in governmental and economic policy, political conflicts, or could be due to large macroeconomic shocks such as the oil shocks experienced over the past decades and the productivity slowdown. It is therefore important to test for parameter - or structural - change.  Various tests for structural change have been proposed with one difference between them being in the type of structural change against which the tests are designed to have power. In this paper, we focus on the scenario in which the potential structural change consists  of discrete changes in the parameter values at unknown points in the sample, known as break - (or change-) points. Within this framework, two types of hypotheses tests are of natural interest: tests of no parameter change against an alternative of change at a fixed number of break-points, and tests of whether the parameters change at $\ell$ break-points against an alternative that they change at $\ell+1$ points. These hypotheses tests are of interest in their own right, and also because they can form the basis of a sequential testing strategy for estimating the number of parameter break-points, see \citet{Bai/Perron:1998}.

\citet{Hall/Han/Boldea:2012} (HHB hereafter) propose various statistics for testing these hypotheses in linear models with endogenous regressors based on Two Stage Least Squares (2SLS).\footnote{\citet{Perron/Yamamoto:2015} propose an alternative approach based on OLS.} Their tests are the natural extensions of the analogous tests for linear models with exogenous regressors estimated via OLS that are introduced in the seminal paper by \citet{Bai/Perron:1998}. A critical issue in the implementation of these tests in a 2SLS setting is whether or not the reduced form for the endogenous regressors is stable. If it is then, under certain conditions, HHB''s test statistics converge in distribution to the same distributions as their OLS counterparts and are pivotal, see HHB and \citet{Perron/Yamamoto:2014}. However, if the reduced form itself  is unstable and/or there is unconditional heteroskedasticity, then these limiting distributions no longer apply (HHB), and are, in fact, no longer pivotal (\citeauthor{Perron/Yamamoto:2014}, \citeyear{Perron/Yamamoto:2014}). This is a severe drawback as in most cases of interest the reduced form is likely to be unstable. This problem has been circumvented in two ways. HHB suggest a testing strategy based on dividing the sample into sub-samples over which the RF is stable but this is inefficient compared to inferences based on the whole sample, and can be infeasible if the sub-samples are small. \citet{Perron/Yamamoto:2015} propose using a variant of \citet{HansenBE:2000}'s  fixed regressor bootstrap to calculate the critical values of the test. Their simulation evidence suggests the use of this bootstrap improves the reliability of inferences but they do not establish the asymptotic validity of the method.\footnote{An alternative approach is to estimate the number and location of the breaks via an information criteria, see  \citet{Hall/Osborn/Sakkas:2015}. However, this approach has the drawback that inferences can be sensitive to the choice of penalty function.}

In this paper, we explore the use of bootstrap versions of 2SLS-based tests for parameter change in far greater detail than previous studies. We consider inferences based on two different types of bootstrap versions of the structural change tests, provide formal proofs of their asymptotic validity and report simulations results that demonstrate that the bootstrap tests provide reliable inferences in the finite sample sizes encountered in practice. More specifically, we consider the case where the right-hand side variables of the equation of interest contains endogenous regressors, contemporaneously exogenous variables, lagged values of both and lagged values of the dependent variable.  This equation of interest is part of a system of equations that is completed by the reduced form for the endogenous regressors and equations for the contemporaneously exogenous variables. This system of equations is assumed to follow a Structural Vector Autoregressive (SVAR) model in which the parameters of the mean are subject to discrete shifts at a finite number of break-points in the sample. Both the number and location of the break-points are unknown to the researcher. These break-points define regimes over which the parameters are constant, and it is assumed that the implied reduced form VAR is stable within each such regime. The errors of the VAR are assumed to follow a vector martingale difference sequence that potentially exhibits both conditional and unconditional heteroskedasticity. Given this error structure, we explore methods for inference based on the wild bootstrap proposed by \citet{Liu:1988} because it has been found to replicate the conditional and unconditional heteroskedasticity of the errors in other contexts. In particular, we consider two versions of the wild bootstrap: the wild recursive bootstrap (which generates recursively the bootstrap observations) and the wild fixed-regressor bootstrap (which adds the wild bootstrap residuals to the estimated conditional mean, thus keeping all lagged regressors fixed). These bootstraps have been proposed by  \citet{Goncalvez/Kilian:2004} to test the significance of parameters in autoregressions with (stationary) conditional heteroskedastic errors.   Our primary focus is on bootstrap versions of $\sup$-$Wald$ - type statistics to test for structural changes in the parameters of  the equation of interest (with endogenous variables) estimated by 2SLS, but our validity arguments also extend straightforwardly to analogous $\sup$-$F$-type statistics.

While our primary focus is on models where the reduced form for the endogenous regressors is unstable, our results also cover the case where this reduced form is stable. In the latter case, the test statistics have a pivotal limiting distribution under conditions covered by our framework,  specialized to errors that are unconditionally homoskedastic. For these situations, the bootstrap methods we propose are expected to provide a superior approximation to finite sample behaviour compared to the limiting distribution because the bootstrap, by its nature, incorporates sample information. Thus bootstrap versions of the tests are attractive in this setting as well.

In the case where there are no endogenous regressors in the equation of interest, our framework reduces to a linear model estimated by Ordinary Least Squares (OLS). For this set-up, \citet{HansenBE:2000} proposes the wild fixed-design bootstrap to test for structural changes using a $\sup$-$F$ statistic. Very recently, \\ \citet{Georgievetal:2018} (GHLT, hereafter) consider  \citet{HansenBE:2000}'s bootstrap for versions of $\sup$-$F$ type tests for parameter variation in predictive regressions with exogenous regressors.  Both \citet{HansenBE:2000} and GHLT establish the asymptotic validity of this bootstrap within the settings they consider.\footnote{In fact, GHLT demonstrate that \citet{HansenBE:2000}'s proof of the asymptotic validity of the bootstrap needs an amendment when the predictive regressors are (near-) unit root processes.} There are some similarities and important differences between our framework (specialized to the no endogenous regressor case) and those in \citet{HansenBE:2000} and GHLT. We adopt similar assumptions about the error process to GHLT and like both \citet{HansenBE:2000} and GHLT consider fixed regressor bootstrap tests of a null of constant parameters versus an alternative of parameter change. Important differences include: GHLT allow for strongly persistent variables whereas our framework assumes the system is stable within (suitably defined) regimes;  our analysis covers tests for additional breaks in the model, the use of the recursive bootstrap and also inferences based on $\sup$-$Wald$ tests. Thus our results for this case complement those of \citet{HansenBE:2000} and GHLT.\footnote{The wild fixed-regressor bootstrap is also included in the recent simulation study exploring the finite sample properties of inference methods about the location of the break-point in models estimated via OLS reported in \citet{Chang/Perron:2018}.}

Although the frameworks are different, \citet{HansenBE:2000}, GHLT and our own study all find their bootstrap versions of the structural change tests work well in finite samples. Interestingly, \citet{Chang/Perron:2018} find that bootstrap-based inferences about the location of breaks have similar advantages in finite samples.\footnote{\citet{Chang/Perron:2018} report results from a comprehensive simulation study that investigates the finite sample methods of various methods for constructing confidence intervals for the break fractions in linear regression models with exogenous regressors. They consider variants of the intervals based on i.i.d., wild and sieve bootstraps.} Collectively, our paper and these other recent studies suggest the use of the bootstrap can yield reliable inferences in linear models with multiple break-points in the sample sizes encountered in practice.

An outline of the paper is as follows. Section 2 lays out the model, test statistics and their bootstrap versions. Section 3 details the assumptions and contains theoretical results establishing the asymptotic validity of the bootstrap methods. Section 4 contains simulation results that provide evidence on the finite sample performance of the bootstrap tests. Section 5 concludes. Appendix A contains all the tables for Section 4, with additional simulations relegated to a Supplementary Appendix. Appendix B contains the proofs, with some background results relegated to the same Supplementary Appendix. 

\textit{Notation: } Matrices and vectors are denoted with bold symbols, and scalars are not. Define for a scalar $N$, the generalized vec operator $\vec_{s=1:N}(\mA_s) = (\mA_1', \ldots, \mA_N')'$, stacking in order the matrices $\mA_s, (s=1,\ldots, N)$, which have the same number of columns. Let $\diag_{s=1:N}(\bm{A}_s) = \diag(\mA_1, \ldots, \mA_{N})$ be the matrix that puts the blocks $\mA_1, \ldots, \mA_N$ on the diagonal. If it is clear over which set $\vec$ and $\diag$ operations are taken, then the subscript $s=1:N$ is dropped on these operators. If $N$ is the number of breaks in a quantity,  $T_{1},\ldots, T_{N}$ are the ordered candidate change-points and $T$ the number of time series observations, where $\tau_{0}=0$, $\tau_{N+1}=1$, and where $\vtau_N=(\tau_{0},\vec_{s=1:N}(\tau_s)',\tau_{N+1})$ is a  partition of the time interval $[1,T]$ divided by $T$, such that $[T \tau_{s}] = T_{s,\vtau_N}$ for $s=1,\ldots N+1$. Define the regimes where parameters are assumed constant as $ I_{s,\vtau_N} = [T_{s-1}+1, T_{s}]$ for $s=1, \ldots, N$.  Below the breaks in the structural equation are denoted by $\vtau_N=\vlam_m$,  and those in the reduced form by $\vtau_N=\vpi_h$, where $m$ an $h$ are the number of breaks in each equation. A superscript zero on any quantity refers to the true quantity, which is a fixed number, vector or matrix.   For any random vector or matrix  $\bm Z$, denote by $||\mZ||$ the Euclidean norm for vectors, or the Frobenius norm for matrices. Finally, $\vzeros_a$ and $\vzeros_{a\times a}$ denote, respectively, an $a\times 1$ vector and a $a\times a$ matrix of zeros, and $1_A$ denotes an indicator function that takes the value one if event $A$ occurs.

\section{The model and test statistics with their bootstrap versions}
\label{sect2}
This section is divided into three sub-sections. Section 2.1 outlines the model. Section 2.2 outlines the hypotheses of interest and the test statistics. Section 2.3 presents the bootstrap versions of the test statistics.
\subsection{The model} \label{sect2.1}
Consider the case where the equation of interest takes the form
\begin{equation}
\label{e1}
y_t\;=\; \underbrace{\begin{array}{c}\vw_t^\prime\end{array}}_{1 \times (p_1+q_{1})} \underbrace{\begin{array}{c}\vbeta_{(i)}^0 \end{array}}_{(p_1+q_1) \times 1}\,+\,\ u_t,\;\;\; i=1,\ldots,m+1,\;\;\; t\in I_{i,\vlam_m^0},
\end{equation}
where $\vw_t=\vec(\vx_t, \vz_{1,t})$, $\vz_{1,t}$ includes the intercept, $\vr_t$ and lagged values of $y_t$, $\vx_t$, and $\vr_t$, and $\vbeta_{(i)}^0$ are parameters in regime $i$. The key difference between $\vx_t$ and $\vr_t$ is that $\vx_t$ represents the set of explanatory variables which are correlated with $u_t$, and $\vr_t$ represents the set of explanatory variables that are uncorrelated with $u_t$. We therefore refer to $\vx_t$ as the {\it endogenous} regressors and $\vr_t$ as the {\it contemporaneously exogenous} regressors.\footnote{This terminology is taken from \citet{Wooldridge:2006}[p.349] and reflects that fact $\vr_t$ may be correlated with $u_n$ for $t\neq n$.} Equation (\ref{e1}) can be re-written as:
\begin{align*}
y_t \;=\; \vx_t^\prime \vbeta_{\vx, t}^0\,+\, \vz_{1,t}^\prime \vbeta_{\vz, t}^0\, +\, u_t\; =\; \vw_t^\prime\vbeta_t^0\, +\,u_t,
\end{align*}
where  $\vbeta_t^0 = \vbeta_{(i)}^0$ if $t \in  I_{i,\vlam_m^0}, i=1,\ldots,m$ and similar notation holds for $\vbeta_{\vx, t}$ and $\vbeta_{\vz, t}$. For simplicity, we refer to (\ref{e1}) as the ``structural equation'' (SE).

The SE is assumed to be part of a system that is completed by the following equations for $\vx_t$ and $\vr_t$.  The reduced form (RF) equation for the endogenous regressors $\vx_t$ is a regression model with $h$ breaks ($h+1$ regimes), that is:
\begin{equation}
\label{e2}
\underbrace{\begin{array}{c}\vx_t'\end{array}}_{1 \times p_1}\;=\;\underbrace{\begin{array}{c}\vz_t'\end{array}}_{1 \times q} \underbrace{\begin{array}{c}\mDelta_{(i)}^0\end{array}}_{q \times p_1}\,+\,\underbrace{\begin{array}{c}\vv_t'\end{array}}_{1 \times p_1},\qquad j=1,\ldots, h+1,\qquad t\in I_{i,\vpi_h^0}.
\end{equation}
The vector $\vz_t$ includes the constant, $\vr_t$ and lagged values of $y_{t}$, $\vx_t$ and $\vr_t$. It is assumed that the variables in $\vz_{1,t}$ are a strict subset of those in $\vz_t$ and therefore we write $\vz_t=\vec(\vz_{1,t},\vz_{2,t})$. Equation (\ref{e2}) can also be rewritten as:
\begin{align*}
\vx_t'\; =\; \vz_t^\prime \mDelta_t^0\, +\,\vv_t^\prime,
\end{align*}
where $\mDelta_t^0 = \mDelta_{(i)}^0$ if $t\in I_{i,\vpi_h^0}$, $i=1,\ldots, h+1$. The contemporaneously exogenous variables  $\vr_t$ are assumed to be generated as follows,
\begin{equation}
\label{er}
\underbrace{\begin{array}{c}\vr_t^\prime\end{array}}_{1 \times p_2}\;=\;\vz_{3,t}^\prime \mPhi_{(i)}^0\,+\,\underbrace{\begin{array}{c}\vzeta_t^\prime\end{array}}_{1 \times p_2}\qquad i=1,\ldots,d+1,\;\;\; t\in I_{i,\vomega_d^0},
\end{equation}
where $\vz_{3,t}$ includes the constant and lagged values of $\vr_t$, $y_t$ and $\vx_t$.

Equations (\ref{e1}), (\ref{e2}) and (\ref{er}) imply $\tilde{\vz}_t=\vec(y_t, \vx_t, \vr_{t})$ evolves over time via a SVAR process whose parameters are subject to discrete shifts at unknown points in the sample. To present the reduced form VAR version of the model, define $n=\dim(\tilde{\vz}_t)$ and let $\vtau_N$ denote the partition of the sample such that all three equations have constant parameters within the associated regimes.\footnote{For example, suppose $m=1$, $h=2$ and $d=1$ with $\vlam_m^0=[0,0.5,1]^\prime$,$\vpi_h^0=[0,0.3,0.5,1]^\prime$ and $\vomega_d^0=[0,0.7,1]^\prime$, then $N=3$ and $\vtau_N=[0,0.3,0.5,0.7,1]^\prime$.} We can then write equations (\ref{e1}), (\ref{e2}) and (\ref{er}) as:
 \begin{equation}
\label{ah1}
\tilde\vz_t \;=\;\vc_{\tilde{z},s}\,+\, \sum_{i=1}^{p} \mC_{i,s} \tilde\vz_{t-i}+ \ve_t,\qquad [\tau_{s-1}T]+1\,\leq t\,\leq [\tau_s T],\; s=1,2,\ldots, N+1,
\end{equation}
where $\ve_t=\mA_s^{-1}\vepsi_t$,
\begin{eqnarray}
\mA_s\;=\;\left[\,\begin{array}{ccc} 1 & -\vbeta_{\vx,s}^{0\prime} & -\vbeta_{\vr, s}^{0\prime}\\ \vzeros_{p_1} & \mI_{p_1} & \mDelta_{\vr,s}^{0'}\\ \vzeros_{p_2} & \vzeros_{p_2\times p_1} & \mI_{p_2}\end{array}\,\right],\label{AS}
\end{eqnarray}
$\vbeta_{\vr, s}^{0\prime}$ denotes the sub-vector of $\vbeta_{s}^{0\prime}$ that contain the coefficients on $\vr_t$ in (\ref{e1}) ($\vbeta_{\vr, s}^{0\prime}$ and $\vbeta_{s}^{0\prime}$ are the values of $\vbeta_{\vr, t}^{0\prime}$ and $\vbeta_{t}^{0\prime}$ for $[\tau_{s-1}T]+1\,\leq t\,\leq [\tau_s T]$);  $\mDelta_{\vr,s}^{0'}$ denotes the sub-matrix of $\mDelta_{s}^{0'}$ that contains the coefficients on $\vr_t$ in (\ref{e2}) ($\mDelta_{\vr, s}^{0\prime}$ and $\mDelta_{s}^{0\prime}$ are the  values of $\mDelta_{\vr, t}^{0\prime}$ and $\mDelta_{t}^{0\prime}$ for $[\tau_{s-1}T]+1\,\leq t\,\leq [\tau_s T]$), and $\vepsi_t=\vec(u_t,\vv_t,\vzeta_t)$. For ease of notation, we assume the order of the VAR is the same in each regime, but our results easily extend to the case where the order varies by regime.

\subsection{Testing parameter variation}\label{sect2.2}

As stated in the introduction, this paper focuses on the issue of testing for structural change in the SE. Within the model described above, there are two types of test that are of particular interest. The first tests the null hypothesis of no parameter change against the alternative of a fixed number of parameter changes in the sample that is, a test of $H_0:\, m=0$ versus $H_1:\,m=k$. The second tests the null of a fixed number of parameter changes against the alternative that there is one more, that is, it tests $H_0:\, m=\ell$ versus $H_1:\,m=\ell+1$. We consider appropriate test statistics for each of these scenarios in turn below.

As the tests are based on the Wald principle, calculation of the test statistics  here requires 2SLS estimation of the SE under $H_1$. On the first stage, the RF is estimated via least squares methods. If the number and location of the breaks in the RF are known then this estimation is straightforward. However, in general, neither the number or location of the breaks is known and so they must be estimated. For our purposes here, it is important that both $h$ and $\vpi_h^0$ are consistently estimated and that $\hat \vpi_h$, the estimator of $\vpi_h^0$, converges sufficiently fast (see Lemma \ref{lem7} in the Appendix B). These properties can be achieved by estimating the RF either as a system or equation by equation, and using a sequential testing strategy to estimate $h$; see, respectively \citet{Qu/Perron:2007} and \citet{Bai/Perron:1998}. Provided the significance levels of the tests shrink to zero slowly enough, $\hat h$ approaches $h$ with probability one as the sample size $T$ grows;  {\it e.g.} see \citet{Bai/Perron:1998}[Proposition 8]. The same consistency result holds if we estimate $h$ via the information criteria; {\it e.g.} see \citet{Hall/Osborn/Sakkas:2013}. For this reason, in the rest of the theoretical analysis, we treat $h$ as known. However, we explore the potential sensitivity of the finite sample performance of the tests for structural change in the SE to estimation of $h$ in our simulation study. Let $\hat \mDelta_{(j)}$ be the estimator of $\mDelta_{(j)}^0$, $\hat \mDelta_t=\sum_{j=1}^{h+1}\hat \mDelta_{(j)}1_{t\in \hat I_j^*}$, where $\hat{I}_j^*=\left\{[\hat{\pi}_{j-1}T]+1, [\hat{\pi}_{j-1}T]+2, \ldots,[\hat{\pi}_{j}T]\right\}$, and $\hat{\vx}_t=\hat \mDelta_t\vz_t^\prime$ that is, $\hat{\vx}_t$ is the predicted value for $\vx_t$ from the estimated RF.\\[0.1in]

\noindent
{\it Case (i): $H_0:\, m=0$ versus $H_1:\,m=k$}\\
\noindent
Under $H_1$, the second stage estimation involves estimation via OLS of the model,
\begin{equation}
\label{2sls_step2_mod}
y_t\;=\;\hat{\vw}_t^\prime\vbeta_{(i)}\,+\,\text{error},\quad i=1,...,k+1,\quad  t\in I_{i,\vlambda_k},
\end{equation}
for all possible $k$-partitions $\vlambda_k$. Let $\hat \vbeta_{(i)}$ denote the OLS estimator of $\vbeta_{(i)}$ in (\ref{2sls_step2_mod}), $\hat \vbeta_{\vlam_k}$ denote the OLS estimator of $\vec_{i=1:k+1}(\hat\vbeta _{(i)})$ in (\ref{2sls_step2_mod}) (that is,  $\hat \vbeta_{\vlam_k}$ is the OLS estimator of $\vec_{i=1:k+1}(\vbetai)$ based on partition  $\vlambda_k$).\footnote{Strictly, $\hat \vbeta_{(i)}$ depends on $\vlambda_k$ but we have suppressed this to avoid excessive notation.} To present the $\sup$-$Wald$ test, we define $\mR_k=\tilde{\mR}_k\otimes \mI_{p}$ where $\tilde{\mR}_k$ is the $k\times(k+1)$ matrix whose $(i,j)^{th}$ element, $\tilde{R}_k(i,j)$, is given by: $\tilde{R}_k(i,i)=1$, $\tilde{R}_k(i,i+1)=-1$, $\tilde{R}_k(i,j)=0$ for $i=1,2,\ldots, k$, and $j\neq i$, $j \neq i+1$. Also let $\mLambda_{\epsilon,k}=\{\vlambda_k: |\lambda_{i+1}-\lambda_i|\ge\epsilon,\lambda_1\ge\epsilon,\lambda_k\le 1-\epsilon\}$. With this notation, the test statistic is:
\begin{align}
\label{supwald}
&\sup\text{-}Wald_T\;=\;\underset{\vlambda_k \in\mLambda_{\epsilon,k}}{\sup}  \ Wald_{T\vlam_k}, \\\label{supwald2}
& Wald_{T\vlam_k}\; =\;  T \, \hat \vbeta_{\vlam_k}^\prime \,
\mR_k' \,\left(\mR_k \hat \mV_{\vlam_k} \mR_k^\prime\right)^{-1}\mR_k \,\hat \vbeta_{\vlam_k},
\end{align}
where:
\begin{align}
&\hat{\mV}_{\vlam_k}=\diag_{i=1:k+1}(\hat \mV_{(i)}) , \qquad \hat \mV_{(i)}\;=\;\hat \mQ_{(i)}^{-1}  \ \hat \mM_{(i)} \ \hat \mQ_{(i)}^{-1}\, , \qquad \hat \mQ_{(i)} = T^{-1} \sum_{t \in  I_{i,\vlam_k}} \hat \vw_t \hat \vw_t'\, , \label{hatV} \\
&\hat{\mM}_{(i)} \inp \lm \var\left(T^{-1/2}\sum_{t \in I_{i,\vlam_k}} \mUpsilon_t^{0'} \vz_t\left(u_t+\vv_t'\vbeta_{\vx}^0\right)\right),
\label{hatMi}
\end{align}
and $\vbeta_{\vx}^0$ is the true value of $\vbeta_{\vx,(i)}^0$ for $i=1,2,\ldots,m+1$ under $H_0$.

As mentioned in the introduction, our framework assumes the errors are a m.d.s. that potentially exhibits heteroskedasticity, and so the natural choice of $\hat{\mM}_{(i)}$ is the Eicker-White estimator, see \citet{Eicker:1967} and \citet{White:1980}. This can be constructed using the estimator of $\vbeta_{x,(i)}$ in \eqref{2sls_step2_mod} under either $H_0$ or $H_1$, where $\vbeta_{x,(i)}$ are the elements of $\vbeta_{(i)}$ containing the coefficients on $\hat \vx_t$. For the purposes of the theory presented below, it does not matter which is used because the null hypothesis is assumed to be true. However, the power properties may be sensitive to this choice. In our simulation study reported below, we use the Eicker-White estimator based on $\hat \vbeta_{\vx,(i)}$, the estimator of  $\vbeta_{x,(i)}$ under $H_1$, that is,
$$
\hat{\mM}_{(i)}\;=\;\widehat{EW}\left[\,\hat\mUpsilon_t^\prime \vz_t(\hat u_t+\hat\vv_t^\prime \hat \vbeta_{\vx,(i)});\,I_{i,\vlam_k}\,\right],
$$
where $\hat{u}_t=y_t-\vw_t^\prime\hat\vbeta_{(i)}$ for $t\in I_{i,\vlam_k}$, $\hat\vv_t=\vx_t-\hat \mDelta_t \vz_t^\prime$, $\hat\mUpsilon_t=[\hat \mDelta_t,\mPi]$ and, for any vector $\va_t$ and  $I\subseteq\{1,2,\ldots,T\}$, $\widehat{EW}\left[\va_t;\,I\,\right]=T^{-1}\sum_{t\in I}\va_t\va_t^\prime$.\\[0.1in]

\noindent
{\it Case (ii): $H_0:\, m=\ell$ versus $H_1:\,m=\ell+1$}\\
\noindent
Following the same approach used by \citet{Bai/Perron:1998} for OLS based inferences, suitable tests statistics can be constructed as follows. The model with $\ell$ breaks is estimated via a global minimization of the sum of squared residuals associated with the second stage of the 2SLS estimation of the SE. For each of the $\ell+1$ regimes of this estimated model, the $\sup$-$Wald$ statistic for testing no breaks versus one break is calculated. Inference about $H_0:m=\ell$ versus $H_1:\,m=\ell+1$ is based on the largest of these $\ell+1$ $\sup\text{-}Wald$ statistics.

More formally, let the estimated SE break fractions for the $\ell$-break model be $\hat{\vlambda}_\ell$ and the associated break points be denoted $\{\hat{T}_i\}_{i=1}^\ell$ where $\hat{T}_i=[T\hat{\lambda}_i]$. Let $\hat{I}_i=I_{i,\hat{\vlambda}_\ell}$, the set of observations in the $i^{th}$ regime of the $\ell$-break model and partition this set as $\hat{I}_i=\hat{I}_i^{(1)}(\varpi_i)\cup\hat{I}_i^{(2)}(\varpi_i)$ where  $\hat{I}_i^{(1)}(\varpi_i)=\{t:\,[\hat{\lambda}_{i-1}T]+1,[\hat{\lambda}_{i-1}T]+2,\ldots,[\varpi_iT]\}$ and  $\hat{I}_i^{(2)}(\varpi_i)=\{t:\,[\varpi_iT]+1,[\varpi_iT]+2,\ldots,[\hat{\lambda}_iT]\}$.  Consider estimation of the model
\begin{equation}
\label{2sls_step2_mod_caseii}
y_t\;=\;\hat{\vw}_t^\prime\vbeta_{(j)}\,+\,\text{error},\quad j=1,2\quad  t\in \hat{I}_{i}^{(j)},
\end{equation}
for all possible choices of $\varpi_i$ (where for notational brevity we suppress the dependence of $\vbeta_{(j)}$ on $i$). Let $\hat{\vbeta}(\varpi_i)=\vec\left(\hat{\vbeta}_{(1)}(\varpi_i), \hat{\vbeta}_{(2)}(\varpi_i)\right)$ be the OLS estimators of $\vec(\vbeta_{(1)},\vbeta_{(2)})$ from (\ref{2sls_step2_mod_caseii}). Also let $\mathcal{N}(\hat{\vlambda}_\ell)=[\hat{\lambda}_{i-1}+\epsilon,\hat{\lambda}_{i}-\epsilon]$.
The $\sup\text{-}Wald$ statistic for testing $H_0:\,m=\ell$ versus $H_1:\,m=\ell+1$ is given by
\begin{equation}\label{supWald_ell}
\sup\text{-}Wald_T(\ell+1\,|\,\ell)\;=\;\max_{i=1,2,\ldots\ell+1}\,\left\{\,\sup_{\varpi_i\in \mathcal{N}(\hat{\vlambda}_\ell)}T\hat{\vbeta}(\varpi_i)^\prime\mR_1^\prime(\mR_1\hat{\mV}(\varpi_i)\mR_1^\prime]^{-1}\mR_1\hat{\vbeta}(\varpi_i)\,\right\}
\end{equation}
where\footnote{The comment above about the calculation of $\hat{\mM}_{(i)}$ apply equally to $\hat{\mM}_{j}(\varpi_i)$.}
\begin{align*}
&\hat{\mV}(\varpi_i)\;=\;\diag\left(\,\hat \mV_{1}(\varpi_i),\hat \mV_{2}(\varpi_i)\,\right) , \qquad \hat \mV_{j}(\varpi_i)\;=\;\{\hat \mQ_{j}(\varpi_i)\}^{-1}  \ \hat \mM_{j}(\varpi_i) \ \{\hat \mQ_{j}(\varpi_i)\}^{-1},\\
& \hat \mQ_{j}(\varpi_i) = T^{-1} \sum_{t \in  \hat{I}_{i}^{(j)}} \hat \vw_t \hat \vw_t^\prime,\qquad
\hat{\mM}_{j}(\varpi_i)=\widehat{EW}\left[\,\hat\mUpsilon_t^{\prime} \vz_t(\hat u_t+\hat\vv_t^\prime\hat \vbeta_{\vx,(i)});\,\hat{I}_{i}^{(j)}\,\right].
\end{align*}

\subsection{Bootstrap versions of the test statistics}\label{sect2.3}
In this section, we introduce the bootstrap analogues of the test statistics presented in the previous section. As noted above, our framework assumes the error vector $\vepsi_t$ to be a m.d.s that potentially exhibits conditional and unconditional heteroskedasticity, and so we use the wild bootstrap proposed by \citet{Liu:1988} because it has been found to replicate the conditional and unconditional heteroskedasticity of the errors in other contexts.\footnote{\citet{Liu:1988} developed the wild bootstrap has been developed in \citet{Liu:1988} following suggestions in \citet{Wu:1986} and \citet{Beran:1986} in the context of static linear regression models with (unconditionally) heteroskedastic errors.} We consider both the wild recursive (WR) bootstrap and the wild fixed (WF) bootstrap. These procedures differ in their treatment of the right-hand side variables in the bootstrap samples as described below.\\[0.1in]

\noindent
{\it Generation of the bootstrap samples:}\\
Let $\tilde{\vz}^b_t=\vec(y_t^{b}, \vx_t^{b}, \vr_{t})$ where $y_t^{b}$ and $\vx_t^{b}$ denote the bootstrap values of $y_t$ and $\vx_t$. Note that because $\vr_t$ is contemporaneously exogenous its sample value is used in the bootstrap samples. In all cases below, the bootstrap residuals are obtained as $u_t^b=\hat u_t\nu_t$ and $\vv_t^b=\hat\vv_t\nu_t$, where $\hat u_t$ and $\hat \vv_t$ are the (non-centered) residuals under the null hypothesis and $\nu_t$ is a random variable that is discussed further in Section 3.

For the WR bootstrap, $\{y_t^{b\prime}\}$ and $\{\vx_t^{b\prime}\}$ are generated recursively as follows:
\begin{eqnarray}
\label{wr1}
\vx_t^{b\prime}\; &=&\; \vz_t^{b\prime} \hat\mDelta_t\, +\,\vv_t^{b\prime},\\
y_t^b \;&=&\; \vx_t^{b\prime} \hat\vbeta_{\vx, t}\,+\, \vz_{1,t}^{b\prime} \hat\vbeta_{\vz, t}\, +\, u_t^b,\label{wr2}
\end{eqnarray}
where the vector $\vz_t^b$ contains a constant, $\vr_t$ and lags of $y_t^b$, $\vx_t^b$ and $\vr_{t}$; $\hat \vbeta_{\vx,t}$ and  $\hat \vbeta_{\vz,t}$ are the sample estimates of  $\vbeta_{\vx,t}^0$ and  $ \vbeta_{\vz,t}^0$ under $H_0$ of the test in question.

For the WF bootstrap, $\vz_t$ is kept fixed and, following  \citet{Goncalvez/Kilian:2004}, the bootstrap samples are generated as follows:
\begin{eqnarray}
\label{wf1}
\vx_t^{b\prime}\; &=&\; \vz_t^\prime \hat\mDelta_t\, +\,\vv_t^{b\prime},\\
y_t^b \;&=&\; \vx_t^{b\prime} \hat\vbeta_{\vx, t}\,+\, \vz_{1,t}^\prime \hat\vbeta_{\vz, t}\, +\, u_t^b,\label{wf2}
\end{eqnarray}
where again $\hat \vbeta_{\vx,t}$ and  $\hat \vbeta_{\vz,t}$ are the sample estimates of  $\vbeta_{\vx,t}^0$ and  $ \vbeta_{\vz,t}^0$ under $H_0$ of the test in question.


\vspace*{0.1in}
\noindent
{\it Case (i): $H_0:\,m=0$ vs $H_1:\,m=k$}\\
First consider the WR bootstrap.  2SLS estimation is implemented in the bootstrap samples as follows. On the first stage, the following model is estimated via OLS
$$
\vx_t^b\; =\; \vz_t^{b\prime} \mDelta_j\, +\,\mbox{error},  \qquad t\in\hat{I}_j^*,\qquad j=1,2,\ldots,h+1,
$$
to obtain $\hat{\mDelta}_j^b=\left\{\sum_{t\in \hat I_j^*}\vz_t^b\vz_t^{b\prime}\right\}^{-1}\sum_{t\in \hat I_j^*} \vz_t^b \vx_t^{b\prime}$. Define $\hat \mDelta_t^b  = \sum_{j=1}^{\hat h+1}1_{t\in \hat I_j^*}\hat{\mDelta}_j^b$, $\hat \vx_t^{b\prime}= \vz_t^{b\prime}\hat\mDelta^b_t$, and $\hat \vw_t^b=\vec(\hat \vx_t^{b}, \vz_{1,t}^{b})$. For a given $k$-partition $\vlambda_k$, the second stage of the 2SLS in the bootstrap samples involves OLS estimation of
\begin{equation}
\label{2sls_step2_mod_boot}
y_t^b\;=\;\hat{\vw}_t^{b\prime}\vbeta_{(i)}\,+\,\text{error},\quad i=1,...,k+1,\quad  t\in I_{i,\vlambda_k},
\end{equation}
and let $\hat \vbeta_{\vlam_k}^b$ be the resulting OLS estimator of $\vec_{i=1:k+1}(\vbetai)$. The WR bootstrap version of the $\sup$-$Wald_T$ statistic is:
 \begin{align}
\label{supwaldb}
\sup\text{-}Wald_T^b&=\;\underset{\vlambda_k \in\mLambda_{\epsilon,k}}{\sup}  \ Wald_{T\vlam_k}^b, \\ \label{supwald2b}
 Wald_{T\vlam_k}^b&=\;  T \, \hat \vbeta_{\vlam_k}^{b\prime} \,
\mR_k^\prime \,\left(\mR_k \hat \mV_{\vlam_k}^b \mR_k^\prime\right)^{-1}\mR_k \,\hat \vbeta_{\vlam_k}^b,
\end{align}
where:
\begin{align}
&\hat{\mV}^b_{\vlam_k}=\diag_{i=1:k+1}(\hat \mV_{(i)}^b) , \qquad \hat \mV_{(i)}^b\;=\;(\hat \mQ_{(i)}^b)^{-1}  \ \hat \mM_{(i)}^b \ (\hat \mQ_{(i)}^b)^{-1}\, , \qquad \hat \mQ_{(i)}^b = T^{-1} \sum_{t \in  I_{i,\vlam_k}} \hat \vw_t^b \hat \vw_t^{b\prime}\, , \label{hatVb} \\
&\hat{\mM}_{(i)}^b =\widehat{EW}\left[\, \hat\mUpsilon_t^{b\prime} \vz_t^b\left(u_t^b+\vv_t^{b\prime}\hat\vbeta_{\vx,(i)}^b\right);\,I_{i,\vlam_k}\right], \qquad \hat\mUpsilon^{b}_{t} = (\hat\mDelta^{b}_t, \mPi).
\label{hatMib}
\end{align}

Now consider the WF bootstrap, for which $y_t^b$ and $\vx_t^b$ are generated via \eqref{wf1}-\eqref{wf2}. The first stage of the 2SLS involves LS estimation of
$$
\vx_t^b\; =\; \vz_t^{\prime} \mDelta_j\, +\,\mbox{error},  \qquad t\in\hat{I}_j^*,\qquad j=1,2,\ldots,h+1,
$$
to obtain $\hat{\mDelta}_j^b=\left\{\sum_{t\in \hat I_j^*}\vz_t\vz_t^{\prime}\right\}^{-1}\sum_{t\in \hat I_j^*} \vz_t \vx_t^{b\prime}$. Now re-define $\hat \mDelta_t^b  = \sum_{j=1}^{\hat h+1}1_{t\in \hat I_j^*}\hat{\mDelta}_j^b$, $\hat \vx_t^{b\prime}= \vz_t^{\prime}\hat\mDelta^b_t$, and $\hat \vw_t^b=\vec( \hat \vx_t^{b}, \vz_{1,t})$. For a given $k$-partitions $\vlambda_k$, the second stage of the 2SLS in the bootstrap samples involves OLS estimation of \eqref{2sls_step2_mod_boot} and let $\hat \vbeta_{\vlam_k}^b$ be the resulting OLS estimator of $\vec_{i=1:k+1}(\vbetai)$. The WF bootstrap $\sup\text{-}Wald$ statistic is defined as in \eqref{supwaldb} with $Wald_{T\vlam_k}^b$ defined as in (\ref{supwald2b}) only with $\hat{\vw}_t$ and $\hat{\mDelta}_t^b$ redefined in the way described in this paragraph, and $\hat{\mM}_{(i)}^b$ in \eqref{hatMib} replaced by $\hat{\mM}_{(i)}^b=\widehat{EW}\left[ \hat\mUpsilon_t^{b\prime} \vz_t(u_t^b+\vv_t^{b\prime}\hat\vbeta_{\vx,(i)}^b);\,I_{i,\vlam_k}\,\right]$.\\[0.1in]

\noindent
{\it Case (ii): $H_0:\, m=\ell$ versus $H_1:\,m=\ell+1$}\\
\noindent
For each bootstrap the first stage of the 2SLS estimation and the construction of $\hat{\vw}_t$ is the same as described under {\it Case (i)} above.  Let $\hat{I}_{i}^{(j)}$ be defined as in the discussion of {\it Case (ii)} in Section \ref{sect2.2}, and consider
\begin{equation}
\label{2sls_step2_mod_caseii_boot}
y_t^b\;=\;\hat{\vw}_t^{b\prime}\vbeta_{(j)}\,+\,\text{error},\quad j=1,2\quad  t\in \hat{I}_{i}^{(j)},
\end{equation}
for all possible choices of $\varpi_i$ (where, once again, we suppress the dependence of $\vbeta_{(j)}$ on $i$). Let $\hat{\vbeta}^b(\varpi_i)=\vec\left(\hat{\vbeta}_{(1)}^{b}(\varpi_i), \hat{\vbeta}_{(2)}^{b}(\varpi_i)\right)$ be the OLS estimators of $\vec(\vbeta_{(1)},\vbeta_{(2)})$ from (\ref{2sls_step2_mod_caseii_boot}). The bootstrap version of\linebreak $\sup\text{-}Wald_T(\ell+1\,|\,\ell)$ is given by
\begin{equation}\label{supWald_ellb}
\sup\text{-}Wald_T^b(\ell+1\,|\,\ell)\;=\;\max_{i=1,2,\ldots\ell+1}\,\left\{\,\sup_{\varpi_i\in \mathcal{N}(\hat{\vlambda}_\ell)}T\hat{\vbeta}^b(\varpi_i)^\prime\mR_1^\prime(\mR_1\hat{\mV}^b(\varpi_i)\mR_1^\prime]^{-1}\mR_1\hat{\vbeta}^b(\varpi_i)\,\right\}
\end{equation}
where
\begin{align*}
&\hat{\mV}^b(\varpi_i)\;=\;\diag\left(\,\hat \mV_{1}^b(\varpi_i),\hat \mV_{2}^b(\varpi_i)\,\right) , \qquad \hat \mV_{j}^b(\varpi_i)\;=\;\{\hat \mQ_{j}^b(\varpi_i)\}^{-1}  \ \hat \mM_{j}^b(\varpi_i) \ \{\hat \mQ_{j}^b(\varpi_i)\}^{-1},\\
& \hat \mQ_{j}^b(\varpi_i) = T^{-1} \sum_{t \in  \hat{I}_{i}^{(j)}} \hat \vw_t^b \hat \vw_t^{b\prime},
\end{align*}
and $\hat{\mM}_{j}^b(\varpi_i)\, =\,\widehat{EW}\left[\,\hat\mUpsilon_t^{b\prime} \vz_t^b(u_t^b+\vv_t^{b\prime}\hat\vbeta_{\vx,(i)}^b); \hat{I}_{i}^{(j)}\,
\right]$ for  WR, and $\hat{\mM}_{j}^b(\varpi_i)\, =\,\widehat{EW}\left[\,\hat\mUpsilon_t^{b\prime} \vz_t(u_t^b+\vv_t^{b\prime}\hat\vbeta_{\vx,(i)}^b); \hat{I}_{i}^{(j)}\,\right]$ for  WF.

\section{The asymptotic validity of the bootstrap tests}\label{asymval}
In this section, we establish the asymptotic validity of the bootstrap versions of the test statistics described above. To this end we impose the following conditions.

\begin{assum}
\label{a1}
If $m>0$, $T_{i}^{0}=[T\lambda_i^0]$, where $0<\lambda_1^0<\ldots<\lambda_m^0<1$.
\end{assum}

\begin{assum}
\label{a2}
If $m>0$, $\vbeta_{(i+1)}^0-\vbeta^0_{(i)}\neq \vzeros_{p_1+q_1}$ is a vector of constants for $i=1,\ldots, m$.
\end{assum}

\begin{assum}
\label{a3}
\black{If $h>0$, }then $T_{i}^*=[T\pi_i^0]$, where $0<\pi_1^0<\ldots<\pi_h^0<1$.
\end{assum}

\begin{assum}
\label{a4}
If $h>0$, $\mDelta^0_{(j+1)}-\mDelta^0_{(j)} \neq \vzeros_{q\times p_1}$ is a matrix of constants for $j=1,\ldots, h$.
\end{assum}

\begin{assum}
\label{a5}
If $k>0$, then $0<\omega_1^0<\ldots<\omega_k^0<1$ and $\mPhi_{(i+1)}^0- \mPhi_{(i)}^0 \neq \vzeros_{q\times p_2}$ is a matrix of constants for $i=1,\ldots, k$.
\end{assum}

\begin{assum}
\label{a6}
The first and second stage estimations in 2SLS are over respectively all partitions of $\vpi$ and $\vlambda$ such that $T_i-T_{i-1}>\max\left(q-1,\epsilon T\right)$ for some $\epsilon >0$ and $\epsilon< \min_i(\lambda_{i+1}^0-\lambda_i^0)$ and $\epsilon< \min_j(\pi_{j+1}^0-\pi_j^0)$.
\end{assum}

\begin{assum}
\label{a7}
(i) $p <\infty$; (ii) $\left|\mI_{n}\,-\,\mC_{1,s}a\, -\, \mC_{2,s} a^2\, -\, \cdots\, -\, \mC_{p,s} a^{ p}\right|\neq 0$, for all $s=1,\ldots, N+1$, and all $\left|a\right|\leq 1$.
\end{assum}

\begin{assum}
\noindent
\label{a9}
 $\rk\left(\mUpsilon_{t}^0\right)=p_1+q_1$ where $\mUpsilon_t^0 = (\mDelta_t^0, \mPi)$ and $\mPi^\prime=\left(\mI_{q_1},\,\vzeros_{q_1\times(q-q_1)}\right)$.
\end{assum}

\begin{assum} \label{a8}
\noindent The innovations can be written as $\vepsi_t=\mS \mD_t \vl_t$, where:\\
(i) $\mS$ is a $ n \times  n$ lower triangular non-stochastic matrix with real-valued diagonal elements $s_{ii}=1$ and elements below the diagonal equal to $s_{ij}$ (which are also zero for $i>p_1+1,\, j <p_1+1$), such that $\mS \mS'$ is positive definite;  $\mD_t = \diag_{i=1:n}(d_{it})$, a non-stochastic matrix where  $d_{it} = d_i(t/T): [0,1] \rightarrow D^{n}[0,1]$, the space of cadlag strictly positive real-valued functions equipped with the Skorokhod topology;\\
(ii) $\vl_t=\vec(l_{u,t},\vl_{\vv,t},\vl_{\vzeta,t})$ is a $n \times 1$ vector m.d.s.  w.r.t to $\mathcal F_t = \{ \vl_t, \vl_{t-1}, \ldots\}$ to which it is adapted, with conditional covariance matrix $\mSigma_{t\mid t-1} = \E(\vl_t \vl_t^\prime\mid\mathcal{F}_{t-1})=\diag\left(\mSigma_{t|t-1}^{(1)},\mSigma_{t|t-1}^{(2)}\right)$  and  unconditional variance $\E(\vl_t\vl_t^\prime)=\mI_n$.\\
(iii) $\sup_t \E||\vl_t||^{4+\delta}<\infty$ for some $\delta>0$. \\
(iv) $\E\left((\vl_t \vl_t^\prime) \otimes \vl_{t-i} \right)=\,\vrho_i$ for all $i\geq 0$, with $\sup_{i\geq 0} \|\vrho_i\| <\infty$.\\
(v)  $\E\left((\vl_t \vl_t^\prime) \otimes (\vl_{t-i} \vl_{t-j}^\prime)\right)=\,\vrho_{i,j}$, for all $i,j \geq 0$  with $\sup_{i,j\geq 0} \|\vrho_{i,j}\| <\infty$.
\end{assum}
\begin{named}{Assumption 9$^\prime$}\label{A8prime} Let $\vn_t = \vec(l_{u,t},\vl_{\vv,t})$. Then:\\
\noindent (i) Assumption \ref{a8}(iv) holds with $\E[(\vn_t \vn_t') \otimes \vn_{t-i}] = \vzeros$ for all $i\geq 1$.\\
\noindent (ii) Assumption \ref{a8}(v) holds with $\E[(\vn_t \vn_t') \otimes (\vn_{t-i} \vn_{t-j}')] = \vzeros$ for all $i,j \geq 1$ and $i\neq j$.\\
\noindent (iii) Assumption \ref{a8}(v) holds with $\E[(\vn_t \vn_t') \otimes (\vn_{t-i} \vl_{\vzeta,t-j}')] = \vzeros$ for all $i\geq 1$ and $j\geq 0$.
\end{named}

\begin{assum}
\label{aboot}(i) $\nu_{t}\stackrel{IID}{\sim}(0,1)$ independent of the original data generated by \eqref{e1}, \eqref{e2} and \eqref{er}; (ii) $\E^b\left|\nu_t\right|^{4+\delta^*}=\bar c<\infty$, for some $\delta^*>0$, for all $t$, where $E^b$ denotes the expectation under the bootstrap measure.
\end{assum}

Before presenting our main theoretical results, we discuss certain aspects of the assumptions.

\vspace*{0.1in}
\begin{rem}
\label{rem0}
Assumptions \ref{a1}-\ref{a5} indicate that the breaks are ``fixed'' in the sense that the size of the associated  shifts in the parameters between regimes is constant and does not change with the sample size.
\end{rem}
\begin{rem}
\label{rem1}
It follows from Assumption \ref{a7} that $\tilde{\vz}_t$ follows a finite order VAR in \eqref{ah1} that is stable within each regime.
\end{rem}
\begin{rem}
\label{rem6}
Assumption \ref{a9} is the identification condition for estimation of the structural equation parameters; see HHB for further discussion.
\end{rem}

 \begin{rem}
                \label{rem3}
                From Assumption \ref{a8} it follows that $\vepsi_t$ is a vector m.d.s. relative to $\mathcal{F}_{t-1}$ with time varying conditional and unconditional variance given by $\E(\vepsi_{t}\vepsi_{t}'|\mathcal{F}_{t-1})=\mS \mD_t\mSigma_{t\mid t-1}\mD_t^\prime\mS^\prime$ and $\E(\vepsi_{t}\vepsi_{t}')=\mS \mD_t\mD_t^\prime\mS^\prime$ respectively. The m.d.s. property implies that all the dynamic structure in the SE and RF for $\vx_t$ is accounted for by the variables in $\vz_{1,t}$ and $\vz_t$ respectively.
                As noted by \citet{Boswijketal:2016} and GHLT, Assumption \ref{a8} allows for $\vepsi_t$ to exhibit conditional and unconditional heteroskedasticity of unknown and general form that can include single or multiple variance shifts, variances that follow a broken trend or follow a smooth transition model.  When $\mD_t=\mD$, the unconditional variance is constant but we may have conditional heteroskedasticity. When $\mSigma_{t\mid t-1}=\mI_n$, the unconditional variance may still be time-varying. Note that Assumption \ref{a8} (i)-(ii) imply that $\vx_t$ is endogenous and $\vr_t$ is contemporaneously exogenous in the SE.  Assumption \ref{a8}(iv) allows for leverage effects (the correlation between the conditional variance and $\vl_{t-i}$ is nonzero, when $i\geq 1$). Assumption \ref{a8}(v)
allows for (asymmetric) volatility clustering (the conditional variance is correlated with cross-products $\vl_{t-i}\vl_{t-j}$, for $i,j\geq 1$).\footnote{The clustering is asymmetric if $\rho_{i,j}\neq 0$ when $i\neq j$.}
   \end{rem}

\begin{rem}

\ref{A8prime} is only imposed in the case of the WR bootstrap.  \ref{A8prime}(i)-(iii) is needed because the WR bootstrap sets to zero certain covariance terms in the distribution of the bootstrapped parameter estimates given the data. This happens because these moments depend on products of bootstrap errors at different lags and these terms have zero expectation under the bootstrap measure due to the fact that $\nu_t$ is mean zero and i.i.d.  \ref{A8prime}(i) is a restriction on the leverage effects and \ref{A8prime}(ii) is a restriction of the asymmetric effects allowed in volatility clustering. Note that \ref{A8prime}(i) is only needed when we have an intercept in \eqref{ah1}.

\ref{A8prime}(iii) arises because the WR design  bootstraps the lags of $y_{t}$ and $\vx_t$ in \eqref{ah1}, but it does not bootstrap $\vr_t$ and its lags. Therefore, certain fourth cross-moments involving both types of quantities are set to zero by the WR bootstrap,  leading to the restriction on clustering effects in \ref{A8prime}(iii) (where $i=j$ is imposed for replicating certain variances, and $i\neq j$ is imposed for replicating certain covariances in the asymptotic distribution of the parameter estimates).
\end{rem}

\begin{rem}
                \label{rem_varsigma}
                There are several choices for the distribution of $\nu_t$, the random variable used in construction of the bootstrap errors: \citet{Goncalvez/Kilian:2004} use the standard normal distribution, while \citet{Mammen:1993} and \citet{Liu:1988} suggested a two-point distribution. In this paper, we report simulation results for \citeauthor{Liu:1988} (1988)'s two-point distribution, which we found performed the best compared to the other distributions in simulations not reported here. This conclusion is similar to \citet{Davidson/Flachaire:2008} and \citet{Davidson/MacKinnon:2010}.
        \end{rem}

The following theorems establish the asymptotic validity of the bootstrap versions of the $sup\text{-}Wald$ tests.

\begin{theo}
\label{theo1boot}
If the WF bootstrap is used let Assumptions \ref{a1}-\ref{aboot} hold and if the WR bootstrap is used let Assumptions \ref{a1}-\ref{aboot} and \ref{A8prime} hold. If $y_t$, $\vx_t$ and $\vr_t$ are generated by (\ref{e1}), (\ref{e2}) and (\ref{er}) and $m=0$ then it follows that
 $$\sup_{c\in\mathbb{R}}\left| P^{b} \left(\sup\text{-}Wald^b_T\leq c\right)-P(\sup\text{-}Wald_T\leq c)\right|\inp 0$$ as $T\rightarrow\infty,$ where $P^b$ denotes the probability measure induced by the bootstrap.
\end{theo}

\begin{theo}
\label{theo2boot}
If the WF bootstrap is used let Assumptions \ref{a1}-\ref{aboot} hold and if the WR bootstrap is used let Assumptions \ref{a1}-\ref{aboot} and \ref{A8prime} hold. If $y_t$, $\vx_t$ and $\vr_t$ are generated by (\ref{e1}), (\ref{e2}) and (\ref{er}) and $m=\ell$ then it follows that: $$\sup_{c\in\mathbb{R}}\left| P^{b} \left(\sup\text{-}Wald^b_T(\ell+1\,|\,\ell)\leq c\right)-P(\sup\text{-}Wald_T(\ell+1\,|\,\ell)\leq c)\right|\inp 0$$ as $T\rightarrow\infty,$ where $P^b$ denotes the probability measure induced
by the bootstrap.
\end{theo}

\vspace*{0.1in}
\begin{rem}
\label{remlimitd}
The proof rests on showing the sample and bootstrap statistics have the same limiting distribution. Although this distribution is known to be non-pivotal if the RF is unstable (see \citeauthor{Perron/Yamamoto:2014}, 2014), to our knowledge this distribution has not previously been presented in the literature. A formal characterization of this distribution is provided in the Supplementary Appendix.

\end{rem}
\begin{rem}
\label{remboot1}
Theorems \ref{theo1boot}-\ref{theo2boot} cover the case where the reduced form is stable and the errors are unconditionally homoskedastic. In this case, the $\sup$-$Wald$ tests are asymptotically pivotal and so the bootstrap is expected to provide a superior approximation to finite sample behaviour compared to the limiting distribution because the bootstrap, by its nature, incorporates sample information. However, a formal proof is left to future research.
\end{rem}

\begin{rem}
\label{remboot2}
HHB also propose testing the hypotheses described above using $\sup$-$F$ tests. While $F$-tests are designed for use in regression models with homoskedastic errors,\footnote{If the reduced form is stable then the limiting distribution of the $sup$-$F$ statistics are only pivotal if the errors are homoskedastic.} wild bootstrap versions of the tests can be used as a basis for inference when the errors exhibit heteroskedasticity. In the Supplementary Appendix, we present WR bootstrap and WF bootstrap versions of appropriate $\sup$-$F$ statistics for testing both $H_0:\,m=0$ versus $H_1:\,m=k$ and $H_0:m=\ell$ versus $H_1:\,m=\ell+1$, and show that these bootstrap versions of the $\sup$-$F$ tests are asymptotically valid under the same conditions as their $\sup$-$Wald$ counterparts. Simulation evidence indicated no systematic difference in the finite sample behaviour of the $\sup$-$Wald$ and $\sup$-$F$ tests for a given null and bootstrap method, and so further details about this approach are relegated to the Supplementary Appendix.
\end{rem}

\begin{rem}
\label{remOLS}
In the special case where there are no endogenous regressors in the equation of interest then our framework reduces to one in which a linear regression model is estimated via OLS. For this set-up, the asymptotic validity of wild fixed bootstrap versions of $\sup$-$F$ test for parameter variation (our Case(i) above) has been established under different sets of conditions by \citet{HansenBE:2000} and GHLT.  \citet{HansenBE:2000} considers the case where the marginal distribution of the exogenous regressors changes during the sample.  GHLT consider  \citet{HansenBE:2000}'s bootstrap in the context of predictive regressions with strongly persistent exogenous regressors. Our results complement these earlier studies because we provide results for the wild recursive bootstrap and a theoretical justification for tests of $\ell$ breaks against $\ell+1$ based on bootstrap methods.
\end{rem}

\section{Simulation results}

In this section, we investigate the finite sample performance of the bootstrap versions of the $\sup$-$Wald$ and $\sup$-$F$ statistics.
We consider a number of designs that involve stability or instability in the SE and/or the RF. In all the designs the variable $x_t$ is endogenous and the SE is estimated by 2SLS. Recalling from above that $h$ and $m$ denote the true number of breaks in the RF and SE respectively, the four scenarios we consider are as follows.
\begin{itemize}
\item {\it Scenario: (h,m)=(0,0)}\\
The DGP is as follows:
\begin{eqnarray}
\label{xtsim2}
x_t &=& \alpha_x + \vr_t'\vdelta^0_{r} + \delta^0_{x_1} x_{t-1} + \delta^0_{y_1} y_{t-1} + v_t, \;\;\;\;\; \mbox{for}\;\; t=1,\dots,T,\\
\label{ytsim2}
y_t &=& \alpha_y + x_t\beta^0_x + \beta^0_{r_1}r_{1,t}+ \beta^0_{y_1} y_{t-1} +u_t, \;\;\;\;\; \mbox{for}\;\; t=1,\dots,T,
\end{eqnarray}
where the parameters of the SE - see equation \eqref{ytsim2} - are $\alpha_y=0.5$, $\beta^0_x=0.5$, $\beta^0_{r_1}=0.5$, $\beta^0_{y_1}=0.8$; the parameters of the RF in equation \eqref{xtsim2} are $\alpha_x=0.5$, $\vdelta_r^0=(1.5,1.5,1.5,1.5)'$ a $4\times 1$ parameter vector, $\delta^0_{x_1}=0.5$, $\delta^0_{y_1}=0.2$; $\vr_t=(r_{1,t},\vr_{2,t}')'$.

\item {\it Scenario: (h,m)=(1,0)}\\
The DGP is as follows:
\begin{eqnarray}
\label{xtsim3i}
x_t &=& \alpha_{x,(1)} + \vr_t'\vdelta^0_{r,(1)} + \delta^0_{x_1,(1)} x_{t-1} + \delta^0_{y_1,(1)} y_{t-1} + v_t, \;\;\;\;\;\;\; \mbox{for}\;\; t=1,\dots,[T/4],
\\
\label{xtsim3ii}
&=&\alpha_{x,(2)} + \vr_t'\vdelta^0_{r,(2)} + \delta^0_{x_1,(2)} x_{t-1} + \delta^0_{y_1,(2)} y_{t-1} + v_t, \;\;\;\;\;\;\ \mbox{for}\;\;  t=[T/4]+1,\dots,T,    \\
\label{ytsim3}
y_t &=& \alpha_y + x_t\beta^0_x + \beta^0_{r_1}r_{1,t}+ \beta^0_{y_1} y_{t-1} +u_t, \;\;\;\;\;\;\;\;\;\;\;\;\;\;\;\;\;\;\;\;\;\;\;\;\;\; \mbox{for}\;\; t=1,\dots,T,
\end{eqnarray}
where the parameters of the SE - equation (\ref{ytsim3}) - are the same as in scenario $(h,m)=(0,0)$, and the RF parameters - equations \eqref{xtsim3i}-\eqref{xtsim3ii} - are: $\alpha_{x,(1)}=0.1$, $\alpha_{x,(2)}=0.5$, $\vdelta_{r,(1)}^0=(0.1,0.1,0.1,0.1)'$, $\vdelta_{r,(2)}^0=(1.5,1.5,1.5,1.5)'$, $\delta^0_{x_1,(1)}=0.1$, $\delta^0_{x_1,(2)}=0.5$, $\delta^0_{y_1,(1)}=0.1$, and $\delta^0_{y_1,(2)}=0.2$. In our simulation study, prior to testing the null hypothesis of zero breaks in the SE parameters from \eqref{ytsim3}, we test sequentially for breaks in the RF parameters (assuming for a maximum of 2 breaks) by applying our bootstrap $\sup$-$Wald$ test.

\item {\it Scenario: (h,m)=(0,1)}\\
The DGP is as follows:
\begin{eqnarray}
\label{xtsim4}
x_t &=& \alpha_x + \vr_t'\vdelta^0_{r} + \delta^0_{x_1} x_{t-1} + \delta^0_{y_1} y_{t-1} + v_t, \;\;\;\;\;\;\;\;\;\;\;\;\;\;\;\;\;\;\;\;\;\;\;\;\; \mbox{for}\;\; t=1,\dots,T,\\
\label{ytsim4i}
y_t &=& \alpha_{y,(1)} + x_t\beta^0_{x,(1)} + \beta^0_{r_1,(1)}r_{1,t}+\beta^0_{y_1,(1)} y_{t-1} +u_t, \;\;\;\;\;\;\; \mbox{for}\;\; t=1,\dots,[3T/4],\\
\label{ytsim4ii}
&=&\alpha_{y,(2)} + x_t\beta^0_{x,(2)} + \beta^0_{r_1,(2)}r_{1,t}+ \beta^0_{y_1,(2)} y_{t-1} +u_t,\;\;\;\;\;\;\; \mbox{for}\;\; t=[3T/4]+1,\dots,T,
\end{eqnarray}
where the parameter values for the RF - equation (\ref{xtsim4}) - are as in scenario $(h,m)=(0,0)$, and the parameters on the SE  - equations  \eqref{ytsim4i}-\eqref{ytsim4ii} - are: $\alpha_{y,(1)} =0.5$, $\alpha_{y,(2)} =-0.5$; $\beta^0_{x,(1)}=0.5$, $\beta^0_{x,(2)}=-0.5$; $\beta^0_{r_1,(1)}=0.5$, $\beta^0_{r_1,(2)}=-0.5$, $\beta^0_{y_1,(1)}=0.8$, and $\beta^0_{y_1,(2)}=0.1$.

\item {\it Scenario: (h,m)=(1,1)}\\
The DGP is as follows:
\begin{eqnarray}
x_t &=& \alpha_{x,(1)} + \vr_t'\vdelta^0_{r,(1)} + \delta^0_{x_1,(1)} x_{t-1}
+ \delta^0_{y_1,(1)} y_{t-1} + v_t, \;\;\;\;\;\;\;\;\; \mbox{for}\;\; t=1,\dots,[T/4],\label{xtsim5i}\\
&=&\alpha_{x,(2)} + \vr_t'\vdelta^0_{r,(2)} + \delta^0_{x_1,(2)} x_{t-1}
+ \delta^0_{y_1,(2)} y_{t-1} + v_t, \;\;\;\;\;\;\;\;\; \mbox{for}\;\;  t=[T/4]+1,\dots,T, \label{xtsim5ii}   \\
y_t &=& \alpha_{y,(1)} + x_t\beta^0_{x,(1)} + \beta^0_{r_1,(1)}r_{1,t}+\beta^0_{y_1,(1)}
y_{t-1} +u_t, \;\;\;\;\;\;\;\;\; \mbox{for}\;\; t=1,\dots,[3T/4],\label{ytsim5i}\\
&=&\alpha_{y,(2)} + x_t\beta^0_{x,(2)} + \beta^0_{r_1,(2)}r_{1,t}+ \beta^0_{y_1,(2)}
y_{t-1} +u_t,\;\;\;\;\;\;\;\;\; \mbox{for}\;\; t=[3T/4]+1,\dots,T,    \label{ytsim5ii}
\end{eqnarray}
where the parameters of the RF - equations \eqref{xtsim5i}-\eqref{xtsim5ii} - are as in scenario $(h,m)=(1,0)$ and the parameters in the SE - equations \eqref{ytsim5i}-\eqref{ytsim5ii} - are as in $(h,m)=(0,1)$. In our simulation study, prior to testing the null hypothesis of zero breaks in the SE parameters from \eqref{ytsim3}, we test sequentially for breaks in the RF parameters (assuming for a maximum of 2 breaks) by applying our bootstrap $\sup$-$Wald$ test.

\end{itemize}

For the four scenarios above we consider the following choices for $u_t$, $\vv_t$ and $\vr_t$:
\begin{enumerate}
\item[{\it Case A:}] $u_t$ and $v_t\stackrel{IID}{\sim}N(0,1)$, $\cov(u_t,v_t)=0.5$, $t=1,\ldots,T$, $\vr_t\stackrel{IID}{\sim}N(\vzeros_{4 \times 1},\mI_{4})$.

\item[{\it Case B:}] $u_t$ and $v_t$ are GARCH(1,1) processes \textit{i.e.} $u_t=\tilde u_t/\sqrt{\var(\tilde u_t)}$ and $v_t=\tilde v_t/\sqrt{\var(\tilde v_t)}$ with $\tilde u_t=\sigma_{\tilde u,t}\vartheta_{\tilde u,t}$ and $\tilde v_t=\sigma_{\tilde v,t}\vartheta_{\tilde v,t}$, $\vartheta_{\tilde u,t}$ and $\vartheta_{\tilde v,t}\stackrel{IID}{\sim}N(0,1)$, $\cov(\vartheta_{\tilde u,t},\vartheta_{\tilde v,t})=0.5$,  $\sigma_{\tilde u,t}^2=\gamma_0+\gamma_1\tilde u_{t-1}^2+\gamma_2\sigma_{\tilde u,t-1}^2$, $\sigma_{\tilde v,t}^2=\gamma_0+\gamma_1\tilde v_{t-1}^2+\gamma_2\sigma_{\tilde v,t-1}^2$, where $\gamma_0=0.1$ and $\gamma_1=\gamma_2=0.4$, $t=1,\ldots,T$, $\vr_t$ is as in \textit{Case A}.

\item[{\it Case C:}] $u_t$ and $v_t\stackrel{IID}{\sim}N(0,1)$, $\cov(u_t,v_t)=0.5$, $t=1,\ldots,[T/3]$; $u_t$ and $v_t\stackrel{IID}{\sim}N(0,2)$, $\cov(u_t,v_t)=0.5$, $t=[T/3]+1,\ldots,T$.

\item[{\it Case D:}] $u_t$ and $v_t$ are as in \textit{Case D} and $\vr_t\stackrel{IID}{\sim}N(\vzeros_{4 \times 1},\mI_{4})$ for $t=1,\ldots,[3T/5]$, and  $\vr_t\stackrel{IID}{\sim}N(\vzeros_{4\times 1},1.5\mI_{4})$ for $t=[3T/5]+1,\ldots,T$.

\end{enumerate}
In {\it Case A}, the errors $u_t$ and $v_t$ are homoskedastic and the contemporaneous exogenous regressors $\vr_t$ are stable. In {\it Case B}, the errors are conditionally heteroskedastic. In {\it Case C} the errors have a contemporaneous  upward shift in the unconditional variance, while in {\it Case D} there is also an upward shift in the variance of $\vr_t$.

In our simulations we consider the behavior of the bootstrap tests both under their null and alternative hypotheses. For scenarios $(h,m)=(0,0)$ and $(h,m)=(1,0)$ we consider the behavior of the $\sup$-$Wald_T$. For scenarios $(h,m)=(0,1)$ and $(h,m)=(1,1)$ we consider the performance of the $\sup$-$Wald_T(2|1)$. In order to assess the power of our bootstrap tests we also consider the case when the null hypotheses are not true and there is an additional break in the SE parameters at $[T/2]$. More exactly we consider in all the four scenarios described above the following:
\begin{align}
y_t &= (\alpha_{y,(i)}+g) + x_t(\beta^0_{x,(i)}+g) + (\beta^0_{r_1,(i)}+g)r_{1,t}+(\beta^0_{y_1,(i)}+g)
y_{t-1} +u_t,\;\;\;\text{for}\,\,t=[T/2]+1,\ldots,\tilde T,\label{power}
\end{align}
with $g$ a constant; $i=1$ and $\tilde T=T$ for scenarios $(h,m)=(0,0)$ and $(h,m)=(1,0)$, and the equation for $y_t$ for $t<[T/2]+1$ is the same as that given in the two scenarios $(h,m)=(0,0)$ and $(h,m)=(1,0)$; $i=2$ and $\tilde T=[3T/4]$ for scenarios $(h,m)=(1,0)$ and $(h,m)=(1,1)$, and the equation for $y_t$ for $t<[T/2]+1$ and $t>[3T/4]$ is the same as that given in the two scenarios $(h,m)=(1,0)$ and $(h,m)=(1,1)$. When $g=0$, the null hypothesis is satisfied. We illustrate the behavior of the tests under the alternative hypothesis for the following values of $g$: $g=-0.007, -0.009$ for scenario $(h,m)=(0,0)$; $g=-0.05,-0.07$ for scenario $(h,m)=(1,0)$; $g=0.3,0.4$ for scenario $(h,m)=(0,1)$; $g=-0.5,0.5$ for scenario $(h,m)=(1,1)$.

For scenarios $(h,m)=(1,0)$ and $(h,m)=(1,1)$ we have tested for the presence of max $2$ breaks in the RF for $x_t$ (in \eqref{xtsim3i}-\eqref{xtsim3ii} and \eqref{xtsim5i}-\eqref{xtsim5ii} respectively) prior to testing for breaks in the SE. More exactly we tested the null hypothesis $H_0:h=\ell$ against $H_1:h=\ell+1$, $\ell=0,1$ using the WR and WF bootstrap $\sup$-$Wald$ for OLS. If the bootstrap $p$-value (given by the fraction of bootstrap statistics more extreme than the $\sup$-$Wald$ based on the original sample) was larger than $5\%$, then we imposed the $\ell$ breaks (assumed under null $H_0:h=\ell$) in the RF and estimated their locations which were subsequently accounted for in the estimation of the SE.

We now describe other features of the calculations before discussing the results. For the WR and the  WF bootstraps the auxiliary distribution (from Assumption \ref{aboot}) is the Rademacher distribution  proposed by \citet{Liu:1988} which assigns 0.5 probability to the value $\nu_t=-1$ and 0.5 probability to $\nu_t=1$, $t=1,\ldots,T$. The same $\nu_t$ is used to obtain both the bootstrap residuals $u_t^b=\hat u_t\nu_t$ and  $v_t^b=\hat v_t\nu_t$ in order to preserve the contemporaneous correlation between the error terms.  We consider $T=120,240,480$ for the sample size and $B=399$ for the number bootstrap replications. All results are calculated using  $N=1,000$ replications.

The reported rejection rates of the WR and WF bootstraps are calculated as: $N^{-1}\sum_{j=1}^N 1_{t_j\geq t^b_{1-\alpha_1,j}}$, where $\alpha_1=0.10,0.05,0.01$ are the nominal values of the tests; $t_j$ is the statistic (sup-$Wald$) computed from the original sample; $t^b_{1-\alpha_1,j}$ is $1-\alpha_1$ quantile of the bootstrap distribution calculated as $(1-\alpha_1)(B+1)$ bootstrap order statistic from the sample of bootstrap statistics in simulation $j=1,\ldots, N$.


For the WR bootstrap, the bootstrap samples were generated recursively with start-up values for $y^b_1$ and $x^b_1$ being given by  the first observations from the sample ($x_1,y_1$); see \citet{Davidson/MacKinnon:1993}.

In all settings, the bootstrap samples are generated by imposing the null hypothesis. The value of $\epsilon$, the trimming parameter in Assumption \ref{a6}, is set to $0.15$ which is a typical value used in the literature.

We now turn to our results. We present results for the $\sup$-$Wald$ test under both the null and alternative hypotheses in Tables \ref{tabhm=00}-\ref{tabhm=11} of the paper. In Tables \ref{tabhm=00F}-\ref{tabhm=11F} of the Supplementary Appendix we also present similar results for the $\sup$-$F$ test. The first two columns of these tables give the rejection rates of the tests under the null hypothesis, while columns 3-6 give the rejection rates of the tests under the alternative hypothesis.\footnote{ The rejection rates under the alternative are not level-adjusted, but since we have used the same sequence of random numbers for repetition $i$, $i=1,\ldots,N$, in the experiments under both null and the alternative hypotheses, one can always subtract (or add) the positive (or negative) size discrepancy (relative to the nominal size) from the rejection rate under the alternative in order to obtain the level-adjusted power of the test; see \citet{Davidson:1998}.}.

 From the first two columns of Tables \ref{tabhm=00}-\ref{tabhm=11}, it can be seen that the WR bootstrap works better in general than the WF bootstrap. The latter has large size distortions for scenarios $(h,m)=(0,0)$, $(h,m)=(0,1)$ and $(h,m)=(1,0)$ whether the errors are conditionally homoskedastic, are conditionally heteroskedastic or have a break in the unconditional variance. For scenario $(h,m)=(1,1)$, the WF bootstrap is only slightly undersized or oversized. Regarding the behavior of the $\sup$-$Wald$ test under the alternative hypothesis, the main conclusion that emerges from columns 3-6 of Tables \ref{tabhm=00}-\ref{tabhm=11} is that the power is influenced in small sample ($T=120$) by the number of breaks in RF and SE, the distribution of the errors $u_t$ and $v_t$, the distribution of $r_t$, as well as the number of breaks in the variance of the errors and in the variance of $\vr_t$. When there is a break in SE, we need a larger $g$ in \eqref{power} to be able to see an increase in the power of the test, compared with scenarios with no break in SE ($g=0.3, 0.4$  for scenario $(h,m)=(0,1)$, and $g=-0.5,0.5$ for scenario $(h,m)=(1,1)$, while $g=-0.007,-0.009$ for scenario $(h,m)=(0,0)$ and $g=-0.05,-0.07$ for scenario $(h,m)=(1,0)$). This can be explained by the fact that the second break in the SE is tested over smaller samples than the first break in the SE. Moreover, the power is lower for the smallest sample ($T=120$) when the error terms have an upward shift in the variance ($Case$ $C$ in Tables \ref{tabhm=00}-\ref{tabhm=11}) and the contemporaneous exogenous regressors also have an upward shift in their variance ($Case$ $D$). However, for $T=240, 480$ the power increases sharply in all cases.

In Tables \ref{tabhm=10} and \ref{tabhm=11} we have sequentially   tested for the presence of max $2$ breaks in the RF for $x_t$ (in \eqref{xtsim3i}-\eqref{xtsim3ii} and \eqref{xtsim5i}-\eqref{xtsim5ii} respectively) using the WR/WR $\sup$-$Wald$ for OLS, and the resulting number of RF breaks was imposed in each simulation prior to estimating the RF and SE and computing the test statistics for 2SLS.  The fraction of times that $0,1,2$ breaks were detected in RF (out of 1,000 replications of the scenarios), is given  in Tables \ref{tabhm=10fracRF}-\ref{tabhm=11fracRF} from Section \ref{supsec6} of the Supplemental Appendix. To assess the impact of the pre-testing in RF (in the first two columns of Tables \ref{tabhm=10} and \ref{tabhm=11}), we have obtained the rejection frequencies of the bootstrap tests when the number of breaks in the RF is held at the true number; see (the first two columns of) Tables \ref{tabhm=10true2} and \ref{tabhm=11true2} from Section \ref{supsec6} of the Supplemental Appendix.
To complement our results, we have also considered a break in RF of smaller size than the one mentioned after \eqref{xtsim3i}-\eqref{xtsim3ii} by taking $\vdelta_{r,(1)}^0=(1,1,1,1)'$ (and the rest of the parameter values are as mentioned after \eqref{xtsim3i}-\eqref{xtsim3ii}); see Tables \ref{tabhm=10small} and \ref{tabhm=11small} from Section \ref{supsec6} of the Supplemental Appendix.

Looking at the results for the $\sup$-$Wald$ our results suggest that in the smaller samples ($T=120$, $240$) the recursive bootstrap is clearly to be preferred over the fixed regressor bootstrap. In the larger sample ($T=480$), the case for the WR over the WF is more marginal as the latter yields  only slightly oversized tests. This relative ranking of the two methods is intuitive from the perspective of \citeauthor{Davidson:2016}'s (\citeyear{Davidson:2016}) first ``golden rule" of bootstrap, which states: ``The bootstrap DGP [...] must belong to the model [...] that represents the null hypothesis.'' The fixed regressor bootstraps treat the lagged dependent variables in the RF and SE as fixed across bootstrap samples, and as such do not seem to replicate the true model that represents the null hypothesis.  This would seem to point toward a recommendation to use the WR but it is important to note an important caveat to our results: our designs involve models for which both recursive and fixed bootstraps are valid. As discussed in Section \ref{asymval}, the fixed regressor bootstrap is asymptotically valid under weaker conditions than the recursive bootstrap. Therefore, while the recursive bootstrap works best in the settings considered here, there may be other settings of interest in which only the fixed bootstrap is valid and so would obviously be preferred.
\section{Concluding remarks}
In this paper, we analyse the use of bootstrap methods to test for parameter change in linear models estimated via Two Stage Least Squares (2SLS). Two types of test are considered:  one where the null hypothesis is of no change and the alternative hypothesis involves discrete change at $k$ unknown break-points in the sample; and a second test where the null hypothesis is that there is discrete parameter change at $l$ break-points in the sample against an alternative in which the parameters change at $l+1$ break-points. In both cases, we consider inferences based on a $\sup$-$Wald$-type statistic using either the wild recursive bootstrap or the wild fixed regressor bootstrap.   We establish the asymptotic validity of these bootstrap tests under a set of general conditions that allow the errors to exhibit conditional and/or  unconditional heteroskedasticity and the regressors to have breaks in their marginal distributions. While we focus on inferences based on $\sup$-$Wald$ statistics, our arguments are easily extended to establish the asymptotic validity of inferences based on bootstrap versions of the analogous tests based on $\sup$-$F$ statistics; see the Supplementary Appendix.

Our simulation results show that the wild recursive bootstrap is more reliable compared to the wild fixed regressor bootstrap, yielding $\sup$-$Wald$-type tests with empirical size equal or close to the nominal size. The gains from using the wild recursive bootstrap are quite clear in the smaller sample sizes, but are more marginal in the largest sample size ($T=480$) in our simulation study. This would seem to point  toward a recommendation to use the wild recursive bootstrap but it is important to note that the wild fixed bootstrap is asymptotically valid under less restrictive conditions than the wild recursive bootstrap. Thus, while both bootstraps are valid in our simulation design, there may be other circumstances when the recursive bootstrap is invalid and the fixed bootstrap would be preferred. The powers of the bootstrap tests are  affected in small sample by the characteristics of the error distribution, but in moderate sample sizes often encountered in macroeconomics, there is a very sharp increase in power.

Our analysis covers the cases where the first-stage estimation of 2SLS involves a model whose parameters are either constant or themselves subject to discrete parameter change. If the errors exhibit unconditional heteroscedasticity and/or the reduced form is unstable then the bootstrap methods are particularly attractive because the limiting distributions are non-pivotal. As a result, critical values have to be simulated on a case-by-case basis. In principle it may be possible to simulate these critical values directly from the limiting distributions presented in our Supplementary Appendix replacing unknown moments and parameters by their sample estimates but this would seem to require knowledge (or an estimate of) the function driving the unconditional heteroskedasticity.  In contrast, the bootstrap approach is far more convenient because it involves simulations of the estimated data generation process using the residuals and so does not require knowledge of the form of heteroskedasticity. Furthermore, our results indicate that the bootstrap approach yields reliable inferences in the sample sizes often encountered in macroeconomics.

\newpage
\appendix
\section{Tables}
\begin{table}[h!]
\caption{\label{tabhm=00}{\it Scenario:(h,m)=(0,0)} - rejection probabilities from testing $H_0:\,m=0$ vs. $H_1:\,m=1$ with bootstrap $\sup$-$Wald$ test.}
\begin{center}
\addtolength{\tabcolsep}{-2pt}
\begin{tabular}{@{}c ccc c ccc c ccc c ccc c ccc c ccc@{}}
\toprule
\midrule
&\multicolumn{3}{c}{WR bootstrap}&&\multicolumn{3}{c}{WF bootstrap}&&\multicolumn{3}{c}{WR bootstrap}&&
\multicolumn{3}{c}{WF bootstrap}&&\multicolumn{3}{c}{WR bootstrap}&&\multicolumn{3}{c}{WF bootstrap}\\
&\multicolumn{3}{c}{Size }&&\multicolumn{3}{c}{Size}&&\multicolumn{3}{c}{Power }&&
\multicolumn{3}{c}{Power}&&\multicolumn{3}{c}{Power}&&\multicolumn{3}{c}{Power }\\
&\multicolumn{3}{c}{$g$=0}&&\multicolumn{3}{c}{$g$=0}&&\multicolumn{3}{c}{$g=-0.007$}&&
\multicolumn{3}{c}{$g=-0.007$}&&\multicolumn{3}{c}{$g=-0.009$}&&\multicolumn{3}{c}{$g=-0.009$}
\\
\cmidrule{2-4}\cmidrule{6-8} \cmidrule{10-12} \cmidrule{14-16} \cmidrule{18-20} \cmidrule{22-24}\\
T&10\%&5\%&1\%&&10\%&5\%&1\%&&10\%&5\%&1\%&&10\%&5\%&1\%&&10\%&5\%&1\%&&10\%&5\%&1\%\\
\cmidrule{2-4}\cmidrule{6-8} \cmidrule{10-12} \cmidrule{14-16} \cmidrule{18-20} \cmidrule{22-24}
\multicolumn{24}{c}{{\it Case A}}\\
\midrule
120&11.8&       6.1&    1.6&&   15.1&   8.7&    2.4&&   59.2&   48.3&   25&&    61.1&   55.3&   31.5&&  79.4&   70.3&   49.1&&  84.5&   75&     56.3\\
240&9.3&        4&      0.8&&   12.9&   6.4&    0.9&&   99.7&   99.7&   99.6&&  99.8&   99.7&   99.7&&  100&    100&    99.9&&  100&    100&    100\\
480&10.08       &5.09&  1.15&&  9.76&   5.52&   1.04&&  100&    100&    100&&   100&    100&    100&&   100&    100&    100&&   100&    100&    100\\

\midrule
\multicolumn{24}{c}{{\it Case B}}\\
\midrule
120&12& 5.9&    0.7&&   14.4&   8.5&    1.7&&   65.5&   54.2&   32.6&&  71.3&   61.1&   38.6&&  83.2&   75.9&   54.7&&  87.1&   80.3&   62.6\\
240&9.5&        4.7&    1.1&&   11.9&   6.2&    1.4&&   99.8&   99.8&   99.5&&  99.9&   99.9&   99.7&&  100&    100&    100&&   100&    100&    100\\
480&10.1&       4.9&    0.5&&   11.5&   6.1&    1.3&&100   &100&100&&    100&    100&    100&&   100&    100&    100&&   100&    100&    100\\

\midrule
\multicolumn{24}{c}{{\it Case C}}\\

\midrule
120&9.9&        5.6&    1.5&&   15.7&   8.3&    1.9&&   46.8&   33.6&   12.5&&  58.9&   45&     22.1&&  70      &56.8&  29.4&&  78.4&   68.1&   43\\
240&9.9&        5.1&    0.7&&   15.2&   8.4&    1.2&&   99.6&   99.5&   99&&    99.6&   99.5&   99&&   99.8&   99.8&   99.5&&  99.9&   99.9    &99.7\\
480&8.8&        5.1&    1.3&&   12      &6.5&   2.1&&   100&100&100&&100&100&100&&        100&    100&    100&&   100&    100&    100\\

\midrule
\multicolumn{24}{c}{{\it Case D}}\\
\midrule
120&10.7&       5.3&    1       &&13.8& 7.3&    1.9&&   50.1&   37.7&   14.1&&  57.6&   44.2&   22.7&&  70.7&   58.7&   34.4&&  76.1&   65.8&   43.2\\
240&9.8&        4.5&    0.9&&   14.1&   7.2&    1.7&&   100&    99.8&   98.6&&  100&    100&    99.4&&  100&    100&    99.9&&  100&    100&    100\\
480&10.1&       4.3&    0.9&&   12.2&   6       &1.2&&100  &100&100&&    100&    100&    100&&    100&   100&    100&&   100&    100&    100\\
\midrule
\bottomrule
\multicolumn{24}{l}{\small{Notes. The first two columns refer to the case when $H_0: m=0$ is true ($g$=0 in equation \eqref{power}). The next columns refer }}\\
\multicolumn{24}{l}{\small{  to the case when we test for $H_0: m=0$, but $H_1:m=1$ is true ($g=-0.007, -0.009$ in equation \eqref{power}). Under the null}}\\ \multicolumn{24}{l}{\small{ and the alternative hypotheses we impose $h=0$ in the RF.}}
\end{tabular}%
\end{center}
\end{table}

\vspace*{2.5in}
\newpage
\begin{table}[h!]
\caption{\label{tabhm=01Wald}{\it Scenario:(h,m)=(0,1)} - rejection probabilities from testing $H_0:\,m=1$ vs. $H_1:\,m=2$ with bootstrap $\sup$-$Wald$ test.}
\begin{center}
\addtolength{\tabcolsep}{-2pt}
\begin{tabular}{@{}c ccc c ccc c ccc c ccc c ccc c ccc@{}}
\toprule
\midrule
&\multicolumn{3}{c}{WR bootstrap}&&\multicolumn{3}{c}{WF bootstrap}&&\multicolumn{3}{c}{WR bootstrap}&&
\multicolumn{3}{c}{WF bootstrap}&&\multicolumn{3}{c}{WR bootstrap}&&\multicolumn{3}{c}{WF bootstrap}\\
&\multicolumn{3}{c}{Size }&&\multicolumn{3}{c}{Size}&&\multicolumn{3}{c}{Power }&&
\multicolumn{3}{c}{Power}&&\multicolumn{3}{c}{Power}&&\multicolumn{3}{c}{Power }\\
&\multicolumn{3}{c}{$g$=0}&&\multicolumn{3}{c}{$g$=0}&&\multicolumn{3}{c}{$g$=0.3}&&
\multicolumn{3}{c}{$g$=0.3}&&\multicolumn{3}{c}{$g$=0.4}&&\multicolumn{3}{c}{$g$=0.4}
\\
\cmidrule{2-4}\cmidrule{6-8} \cmidrule{10-12} \cmidrule{14-16} \cmidrule{18-20} \cmidrule{22-24}\\
T&10\%&5\%&1\%&&10\%&5\%&1\%&&10\%&5\%&1\%&&10\%&5\%&1\%&&10\%&5\%&1\%&&10\%&5\%&1\%\\
\cmidrule{2-4}\cmidrule{6-8} \cmidrule{10-12} \cmidrule{14-16} \cmidrule{18-20} \cmidrule{22-24}
\multicolumn{24}{c}{{\it Case A}}
\\
\midrule
120&10.7&       5&      1.2&&   15.5&   9.9&    5.6&&   54.9&   36.9&   10.8&&  61.5&   45.6&   19.8&&  78.1&   60.3&   24.8&&  82.2&   67.8&   38.5\\
240&10.2&       4.9&    0.5&&   12.5&   7.1&    3.4&&   99.5&   98.9&   89.9&&  99.6&   98.6&   92.2&&  100&    100&    98.8&&  100&    100&    99.1\\
480&8   &4.5&   1&&     8.6&    4.4&    0.8&&   100&    100&    100&&   100&    100&    100&&   100&    100&    100&&   100&    100&    100\\

\midrule
\multicolumn{24}{c}{{\it Case B}}\\
\midrule
120&9.7&        4.6&    1       &&16&   10.2&   6.5&&   62.3&   44.8&   16.1&&  67.2&   53.9&   26.2&&  82.2&   67.6&   31.1&&  84.1&   73.6&   44.7\\
240&10.6&       5.2&    1.2&&   13.8&   8.1&    3       &&99.3& 97.5&   86.2&&  99.1&   93&     91.5&&  99.9&   99.6&   96.6&&  100&    99.9&   98.3\\
480&8   &4.2&   0.9&&   8.4&    4.8&    0.8&&   100&    99.8&   99.5&   &99.8&  99.7&   99.5&&  100&    100&    100&&   100&    100&    99.9\\

\midrule
\multicolumn{24}{c}{{\it Case C}}\\

\midrule
120&10.5        &5.2&   0.9&&   16.3&   11&     5.8&&   26.3&   14.8&   3.3&&   36.1&   21.4&   6.6&&   40.1&   24.5&   7.5&&   51.1&   34.6&   13\\
240&11& 4.8&    0.9&&   13.2&   8.3&    2.4&&   83.1&   68.7&   31.4&&  87.2&   77.6&   47.2&&  98.5&   93.4&   68.7&&  99&     97&     80.1\\
480&10.4&       5.6&    0.5&&   11.2&   6.1&    1.2&&   100&    99.9&   98.4&&  100&    99.9&   99.2&&  100&    100&    100&&   100&    100&    100\\

\midrule
\multicolumn{24}{c}{{\it Case D}}\\
\midrule
120&11.6&       5.8&    1.5&&   15.3&   9.5&    5.3&&   39.8&   24.1&   6.5&&   51.2&   33.2&   13.3&&  64.8&   43.33&  14&&    72.5&   54.8&   24.2\\
240&11.5&       6&      1       &&14.9& 9.1&    2.9&&   98.9&   94.6&   73&&    98.9&   97.1&   82.7&&  100&    99.8&   95.6&&  100&    99.9&   97.9\\
480&9.6&        4       &1.3&&  9.5&    5.3&    1.5&&   100&    100&    100&&   100&    100&    100&&   100&    100&    100&&   100&    100&    100\\

\midrule
\bottomrule
\multicolumn{24}{l}{\small{Notes. The first two columns refer to the case when $H_0: m=1$ is true ($g$=0 in equation \eqref{power}). The next columns refer }}\\ \multicolumn{24}{l}{\small{  to the case when we test for $H_0: m=1$, but $H_1:m=2$ is true ($g=0.3, 0.4$ in equation \eqref{power}). Under the null and }}\\ \multicolumn{24}{l}{\small {the alternative hypotheses we impose $h=0$ in the RF.}}
\end{tabular}%
\end{center}
\end{table}
\vspace*{2.5in}
\newpage

\begin{table}[h!]
\caption{\label{tabhm=10}{\it Scenario:(h,m)=(1,0)} - rejection probabilities from testing $H_0:\,m=0$ vs. $H_1:\,m=1$ with bootstrap $\sup$-$Wald$ test; number of breaks in the RF was estimated and imposed in each simulation using a sequential strategy based on the WR/WF $\sup$-$Wald$ for OLS}
\begin{center}
\addtolength{\tabcolsep}{-2pt}
\begin{tabular}{@{}c ccc c ccc c ccc c ccc c ccc c ccc@{}}
\toprule
\midrule
&\multicolumn{3}{c}{WR bootstrap}&&\multicolumn{3}{c}{WF bootstrap}&&\multicolumn{3}{c}{WR bootstrap}&&
\multicolumn{3}{c}{WF bootstrap}&&\multicolumn{3}{c}{WR bootstrap}&&\multicolumn{3}{c}{WF bootstrap}\\
&\multicolumn{3}{c}{Size }&&\multicolumn{3}{c}{Size }&&\multicolumn{3}{c}{Power }&&
\multicolumn{3}{c}{Power }&&\multicolumn{3}{c}{Power }&&\multicolumn{3}{c}{Power }\\
&\multicolumn{3}{c}{$g=0$}&&\multicolumn{3}{c}{$g=0$}&&\multicolumn{3}{c}{$g=-0.05$}&&
\multicolumn{3}{c}{$g=-0.05$}&&\multicolumn{3}{c}{$g=-0.07$}&&\multicolumn{3}{c}{$g=-0.07$}\\
\cmidrule{2-4}\cmidrule{6-8} \cmidrule{10-12} \cmidrule{14-16} \cmidrule{18-20} \cmidrule{22-24}\\
T&10\%&5\%&1\%&&10\%&5\%&1\%&&10\%&5\%&1\%&&10\%&5\%&1\%&&10\%&5\%&1\%&&10\%&5\%&1\%\\
\cmidrule{2-4}\cmidrule{6-8} \cmidrule{10-12} \cmidrule{14-16} \cmidrule{18-20} \cmidrule{22-24}
\multicolumn{24}{c}{{\it Case A}}\\
\midrule
120&10.2&       3.7&    0.9&&   15.3&   7.1&    1.3&&   52.3&   42.8&   22.5&&  60.5&   50.6&   29.3&&  67.9&   58.3&   38.1&&  74.9&   66.3&   45.6\\
240&10.8&       5.7&    0.8&&   14&     6.7&    1.2&&   94.6&   91.3&   84.5&&  95.1&   92.1&   86&&    98.1&   96.7&   91.7&&  98.7&   97.1&   92.7\\
480&10.9&       5.2&    0.9&&   12.5&   6&      0.8&&   99.9&   99.8&   99.3&&100  &99.8&99.5&&    100&    100&    99.8&&  100&    100&    99.7\\
\midrule
\multicolumn{24}{c}{{\it Case B}}\\
\midrule
120&10.1&       4.8&    1       &&13.3& 7.8&    1.5&&   54.5&   44.9&   28.1&&  63.1&   51.5&   33.7&&  68.9&   59.6&   43.4&&  77      &68     &48.9\\
240&10& 5.4&    1.2&&   12.2&   6.8&    1.5&&   94.5&   92.1&   83.7&&  95.8&   93.3&   85.9&&  97.9&   96.7&   92&&    98.8&   97.5&   93.3\\
480&11& 5.4&    0.7&&   12.8&   5.9&    1.2&&   100&    100&    100&&100&99.7&99.2&&100&   99.9&   99.8&&  100&    99.9&   99.8\\
\midrule
\multicolumn{24}{c}{{\it Case C}}\\
\midrule
120&9.6&        4.3&    0.9&&   15.6&   7.5&    1.7&&   39.6&   28.7&   11.3&&  50.8&   37.4&   18.4&&  54.9&   43.8&   22.9&&  66.6&   53.8&   33.5\\
240&11.8&       6&      0.6&&   15.6&   8.6&    1.4&&   88.8&   83.5&   71.8&&  91.3&   87.3&   76.2&&  94.3&   92&     83.5&&96  &93.7&86.7\\
480&10.8&       5.9&    1.1&&   12.6&   7.1&    1.4&&   99.9&   99.6&   98.6&&  99.9&99.5&98.5&&    99.9&   99.9&   99.8&&  100&    99.9&   99.4\\
\midrule
\multicolumn{24}{c}{{\it Case D}}\\
120&10.2&       4.8&    1.2&&   14.8&   6.7&    1.6&&   40.9&   29.9&   12.9&&  49&     37.3&   16.8&&  56&     45.1&   24.3&&  64.5&   52.3&   31.9\\
240&10.6        &5.7&   0.9&&   14.2&   7.5&    1.8&&   89.6&   85.2&   73.2&&  91.1&   87.2&   76.6&&  94.9&   92.4&   85.3&&  95.7&   93.7&   86.9\\
480&11.6&       6.2&    0.9&&   13.5&   7.4&    1       &&99.4& 99.1&   98&     &99.5&99.3&98&&   100&    99.8&   98.9&&99.8  &99.8&99.1\\

\midrule
\bottomrule
\multicolumn{24}{l}{\small{Notes. The first two columns refer to the case when $H_0: m=0$ is true ($g$=0 in equation \eqref{power}). The next columns refer }}\\ \multicolumn{24}{l}{\small{  to the case when we test for $H_0: m=0$, but $H_1:m=1$ is true ($g=-0.05, -0.07$ in equation \eqref{power}). Prior to testing    }}\\ \multicolumn{24}{l}{\small{$H_0:m=0$ vs $H_1:m=1$ (for all columns above), we tested sequentially for the presence of maximum two breaks in the }}\\ \multicolumn{24}{l}{\small{RF (we used the  WR/WF bootstrap $\sup$-$Wald$ for OLS to test $H_0:h=\ell$ vs. $H_1:\ell+1$, $\ell=0,1$). If breaks are detected}}\\ \multicolumn{24}{l}{\small{in RF, the number of breaks and the estimated locations are imposed when estimating the SE.}}
\end{tabular}%
\end{center}
\end{table}
\vspace*{2.5in}
\newpage

\begin{table}[h!]
\caption{\label{tabhm=11}{\it Scenario:(h,m)=(1,1)} - rejection probabilities from testing $H_0:\,m=1$ vs. $H_1:\,m=2$ with bootstrap $\sup$-$Wald$ test;  number of breaks in the RF was estimated and imposed in each simulation using a sequential strategy based on the WR/WF $\sup$-$Wald$ for OLS}
\begin{center}
\addtolength{\tabcolsep}{-2pt}
\begin{tabular}{@{}c ccc c ccc c ccc c ccc c ccc c ccc@{}}
\toprule
\midrule
&\multicolumn{3}{c}{WR bootstrap}&&\multicolumn{3}{c}{WF bootstrap}&&\multicolumn{3}{c}{WR bootstrap}&&
\multicolumn{3}{c}{WF bootstrap}&&\multicolumn{3}{c}{WR bootstrap}&&\multicolumn{3}{c}{WF bootstrap}\\
&\multicolumn{3}{c}{Size}&&\multicolumn{3}{c}{Size}&&\multicolumn{3}{c}{Power}&&
\multicolumn{3}{c}{Power}&&\multicolumn{3}{c}{Power}&&\multicolumn{3}{c}{Power}\\
&\multicolumn{3}{c}{$g=0$}&&\multicolumn{3}{c}{$g=0$}&&\multicolumn{3}{c}{$g=0.5$}&&
\multicolumn{3}{c}{$g=0.5$}&&\multicolumn{3}{c}{$g=-0.5$}&&\multicolumn{3}{c}{$g=-0.5$}\\
\cmidrule{2-4}\cmidrule{6-8} \cmidrule{10-12} \cmidrule{14-16} \cmidrule{18-20} \cmidrule{22-24}\\
T&10\%&5\%&1\%&&10\%&5\%&1\%&&10\%&5\%&1\%&&10\%&5\%&1\%&&10\%&5\%&1\%&&10\%&5\%&1\%\\
\cmidrule{2-4}\cmidrule{6-8} \cmidrule{10-12} \cmidrule{14-16} \cmidrule{18-20} \cmidrule{22-24}
\multicolumn{24}{c}{{\it Case A}}\\
\midrule
120&8.8&        4.7&    0.7&&   8.7&    4.5&    0.8&&   52&     40.7&   16&&    57.4&   45.4&   22.7&&  85&     71.9&   32.1&&  88.2&   74.8&   37.5\\
240&10.4&       5.7&    0.7&&   10.4&   5.2&    0.8&&   99.8&   99.4&   97.6&&  99.6&   99.4&   97.4&&  100&    100&    99.8&&  100&    100&    99.7\\
480&9.7 &4.2&   0.7&&   10.2&   4.6&    0.8&&   100&    100&    100&&   100&    99.8&   99.1&&  100&    100&    100&&   100&    100&    100\\

\midrule
\multicolumn{24}{c}{{\it Case B}}\\
\midrule
120&8.9&        3.7&    0.9&&   8.7&    3.4&    0.9&&   50.1&   40.1&   18.6&&  54.8&   45.4&   24.4&&  81.8&   70.9&   38.3&&  85.5&   73&     39.9\\
240&10.8&       4.7&    0.8&&   10.6&   5.3&    0.9&&   98.8&   98.3&   96&&    99.2&   98.7&   95.8&&  99.6&   99.5&   98.1&&  98&     99.6&   98.3\\
480&9.9&        4.1&    0.9&&   10.9&   5.4&    0.9&&   100&    100&    99.8&&  100&    99.8&   99.6&&100  &100&100&&100&100&100\\

\midrule
\multicolumn{24}{c}{{\it Case C}}\\

\midrule
120&9.1&        3.5&    1       &&9.4&  4       &0.4&&  30.7&   17.5&   3.1&&   38.3&   25.9&   8&&     45.1&   25.4&   6.9&&   49.7&   31.7&   8.8\\
240&10.3&       5.2&    1&&     10.2&   5&      1&&     98.6&   96.8&   86.2&&  99&     97.7&   88.4&&  99.3&   98.5&   86.9&&  100&    99.8&   90.8\\
480&11.3&       4.8&    1       &&12.1& 5.3&    0.6     &&100&  100&    99.9&&  100&    100&    99.2&& 99.9 &99.9&99.7&&100&100&99.8\\

\midrule
\multicolumn{24}{c}{{\it Case D}}\\
\midrule
120&10.1&       4.4&    1.6&&   8.5&    3.8&    0.6&&   36.8&   23.4&   6.3&&   42.1&   30.4&   12.4&&  69.3&   52.6&   16.4&&  76.8&   59.4&   25.4\\
240&10.9&       4.9&    0.8&&   11.8&   5.2&    0.8&&   99.2&   98.6&   94&&    99.6&   98.9&   94.5&&  99.5&   99.4&   98&&    99.9&   99.9&   98.6\\
480&10.2&       5.3&    1.4&&   11&     5.6&    1.2&&   100&    100&    100&&   100&    100&    98.1  &&100&100&100&&100&100&100\\
\midrule
\bottomrule
\multicolumn{24}{l}{\small{Notes. The first two columns refer to the case when $H_0: m=1$ is true ($g$=0 in equation \eqref{power}). The next columns refer }}\\ \multicolumn{24}{l}{\small{  to the case when we test for $H_0: m=1$, but $H_1:m=2$ is true ($g= -0.5, 0.5$ in equation \eqref{power}). Prior to testing    }}\\ \multicolumn{24}{l}{\small{$H_0:m=1$ vs $H_1:m=2$ (for all columns above), we tested sequentially for the presence of maximum two breaks in the }}\\ \multicolumn{24}{l}{\small{RF (we used the  WR/WF bootstrap $\sup$-$Wald$ for OLS to test $H_0:h=\ell$ vs. $H_1:\ell+1$, $\ell=0,1$). If breaks are detected}}\\ \multicolumn{24}{l}{\small{in RF, the number of breaks and the estimated locations are imposed when estimating the SE.}}
\end{tabular}%
\end{center}
\end{table}

\vspace*{2.5in}
\newpage

\section{Appendix: Proof of Theorems}
For the purposes of our analysis, it is convenient to write the system in (\ref{ah1}) as a VAR($1$) model.\footnote{For example, see Using \citet{Hamilton:1994}[p.259].} To this end, define:
\begin{eqnarray*}
\underbrace{\begin{array}{c}\vxi_t\end{array}}_{np\times 1}&\equiv& \left[\,\begin{array}{c} \tilde{\vz}_t\\\tilde{\vz}_{t-1}\\\vdots\\\tilde{\vz}_{t-p+1}\end{array}\,\right],\qquad
\underbrace{\begin{array}{c}\mF_s\end{array}}_{np\times np}\;\equiv\; \left[\,\begin{array}{cccccc} \mC_{1,s}& \mC_{2,s}& \mC_{3,s}&\hdots& \mC_{p -1,s}& \mC_{p,s}\\ \mI_n & \vzeros_{n\times n} & \vzeros_{n\times n} &\hdots & \vzeros_{n\times n} & \vzeros_{n\times n}\\ \vzeros_{n\times n} & \mI_{n} & \vzeros_{n\times n} & \hdots & \vzeros_{n\times n} & \vzeros_{n\times n}\\ \vdots & \vdots & \vdots & \hdots & \vdots & \vdots\\ \vzeros_{n\times n} & \vzeros_{n\times n} & \vzeros_{n\times n} & \hdots & \mI_{n} & \vzeros_{n\times n}\end{array} \,\right],\\
\underbrace{\begin{array}{c}\veta_t\end{array}}_{np\times 1}&\equiv& \left[\,\begin{array}{c} \ve_t\\\vzeros_{n}\\\vdots\\\vzeros_{n}\end{array}\,\right],\,\qquad \mbox{ and }\qquad\underbrace{\begin{array}{c}\vmu_s\end{array}}_{np\times 1}\;\equiv\; \left[\,\begin{array}{c} \vc_{\tilde{\vz},s}\\\vzeros_{n}\\\vdots\\\vzeros_{n}\end{array}\,\right].
\end{eqnarray*}
Then equation (\ref{ah1}) is the first $n$ entries of:
\begin{equation}
\label{xivar}
\vxi_t\;=\;\vmu_s\,+\,\mF_s\vxi_{t-1}\,+\,\veta_t,
\end{equation}
where we have suppressed the dependence of $\vxi_t$ and $\veta_t$ on $s$ for notational convenience.

From Assumption \ref{a8} it follows that $\veta_t$ is a vector m.d.s. relative to $\mathcal{F}_{t-1}$ with conditional covariance matrix
\begin{eqnarray}
\E(\veta_t\veta_j'\mid\mathcal{F}_{t-1})=\left\{\begin{array}{c}\mOmega_{t\mid t-1},\text{ for }t=j,\\\vzeros_{np\times np}\text{ otherwise,}\end{array}\right.
\end{eqnarray}
\begin{eqnarray*}
\underbrace{\begin{array}{c}\mOmega_{t\mid t-1}\end{array}}_{np\times np}&\equiv& \left[\,\begin{array}{cc}\mA_s^{-1}\overline\mSigma_{t\mid t-1}\mA_s^{-1'}& \vzeros_{n\times n(p-1)}\\
\vzeros_{n\times n(p-1)}'&\vzeros_{n(p-1)\times n(p-1)}\end{array}\,\right],
\end{eqnarray*}
where $\overline\mSigma_{t\mid t-1}=\mS \mD_t\mSigma_{t\mid t-1}\mD_t^\prime\mS^\prime$, and time-varying unconditional covariance matrix $$\underbrace{\begin{array}{c}\mOmega_t\end{array}}_{np\times np}\equiv\E(\veta_t \veta_t^\prime)=\left[\,\begin{array}{cc}\mA_s^{-1}\overline\mSigma_t\mA_s^{-1'}& \vzeros_{n\times n(p-1)}\\
\vzeros_{n\times n(p-1)}^\prime&\vzeros_{n(p-1)\times n(p-1)}\end{array}\,\right] $$
where $\overline\mSigma_t=\mS \mD_t\E(\vl_t\vl_t^\prime)\mD_t^\prime\mS^\prime$.

From (\ref{xivar}), it follows that within each regime we have, for $t=[\tau_{s-1}T]+1,[\tau_{s-1}T]+2,\ldots, [\tau_{s}T]$,
\begin{equation}
\vxi_t\;=\; \mF^{t-[\tau_{s-1}T]}_s \vxi_{[\tau_{s-1}T]} + \tilde\vxi_t+\left(\sum_{l=0}^{t-[\tau_{s-1}T]-1}\mF_s^l\right)\vmu_s,\label{xima}.
\end{equation}
where
 $\tilde\vxi_t=\sum_{l=0}^{t-[\tau_{s-1}T]-1}\mF_s^l\veta_{t-l}$, $\{\veta_t\}$ is a m.d.s. sequence, and, from Assumption \ref{a7}, all the eigenvalues of $\mF_s$ have modulus less than one.

The following lemmas are used in proofs; Lemmas \ref{lem2} and \ref{lem4}-\ref{lem8} are proven in the Supplementary Appendix, which also contains the asymptotic distributions of the sup Wald test statistics. The rest of the lemmas are proven below.

\begin{lem}
\label{lem1}
If $\{\vartheta_t,\mathcal{F}_t\}$ is a mean-zero sequence of  $L^1$-mixingale random variables with constants $\{c_t\}$ that satisfy $\overline{\lim}_{T\to\infty}T^{-1}\sum_{t=1}^Tc_t<\infty$, and $\sup_t \E|\vartheta_t|^{b} <\infty$ for some $b>1$, then $\sup_{s\in(0,1]}|T^{-1}\sum_{t=1}^{[Ts]} \vartheta_t|\stackrel{p}{\to} 0$.
\end{lem}
\noindent
This follows from applying the LLN in \citet{Andrews:1988}[Theorem 1], modified to be a uniform LLN in the proof of Lemma A2 of \citet{Andrews:1993}.

\begin{lem}\label{lem2}
\black{For $s=1,\ldots, N+1$, where $N$ is the total number of} \black{breaks in the coefficients of the VAR($p$) representation of $\tilde \vz_t$, define the following functions:
$
\mF(\tau)= \mF_s, \mA(\tau)= \mA_s,\vmu(\tau)=\vmu_s,\mUpsilon(\tau)= \mUpsilon_s \mbox{ for } \tau_{s-1}<\tau\leq \tau_s.
$. Also, define the function $\overline{\mSigma}(\tau)$ on $\tau \in [0,1]$ as follows $\overline{\mSigma}(0)=0$, and $\overline{\mSigma}(\tau)  = \mSigma_t$ for $\tau \in ((t-1)/T, t/T]$, $t=1,\ldots, T$. Let  $\mathcal{S}$ and $\mathcal{S}_{\vr}$ be the selection matrices such that $\vz_t = \vec(1,\mathcal{S}_{\vr}\vxi_t,\mathcal{S}\vxi_{t-1}) =\vec(1,\vr_t,\mathcal{S}\vxi_{t-1})$, and
$$
\mathbb{Q}_{\vz}(\tau)\;=\;\left[\,\begin{array}{ccc} 1 & \{\mathcal{S}_{\vr}\mathbb{Q}_1(\tau)\}^\prime
&\{\mathcal{S}\mathbb{Q}_1(\tau)\}^\prime\\
\mathcal{S}_{\vr}\mathbb{Q}_1(\tau)& \mathcal{S}_{\vr}\mathbb{Q}_2(\tau) \mathcal{S}_{\vr}^\prime & \mathcal{S}_{\vr}(\vmu(\tau) \mathbb Q_1'(\tau) +\mF(\tau) \mathbb{Q}_2(\tau)) \mathcal{S}^\prime
\\\mathcal{S}\mathbb{Q}_1(\tau) & (\mathcal{S}_{\vr}(\vmu(\tau) \mathbb Q_1'(\tau) +\mF(\tau) \mathbb{Q}_2(\tau)) \mathcal{S}^\prime)^\prime& \mathcal{S}\mathbb{Q}_2(\tau) \mathcal{S}^\prime \end{array}\,\right],
$$
where
$\mathbb{Q}_1(\tau)=\left\{\mI_{np}\,-\,\mF(\tau)\right\}\vmu(\tau)$ and \begin{eqnarray*}
\mathbb{Q}_2(\tau)&=& \sum_{l=0}^{\infty} \mF(\tau)^l \left[\begin{array}{cc} \mA(\tau)^{-1}\overline \mSigma(\tau)\mA(\tau)^{-1^\prime} &\vzeros_{n \times n(p-1)}\\\vzeros_{n \times n(p-1)}' & \vzeros_{n(p-1) \times n(p-1)}\end{array}\,\right] (\mF(\tau)^{l})'+ \, \mathbb{Q}_1(\tau) \mathbb{Q}_1'(\tau).
\end{eqnarray*}
Also, let $\mathbb Q_i = \int_{\lambda_{i-1}}^{\lambda_{i}}\mUpsilon'(\tau)\mathbb{Q}_{\vz}(\tau)\mUpsilon(\tau)\rd\tau$.

Under Assumptions \ref{a1}-\ref{a9},} $$
\hat{\mQ}_{(i)}\;=\; T^{-1} \sum_{t \in I_{i,\vlambda_k}} \hat{\mUpsilon}_t' \vz_t \vz_t^\prime\hat{\mUpsilon}_t\;\stackrel{p}{\to}\; \mathbb Q_i.
$$

\end{lem}

\begin{lem}\label{lem3}
If $(\va_t, \mathcal F_{t})$ is a $o\times 1$ vector of m.d.s. with $\sup_t \E |a_{t,j}|^{2+\delta^*} < \infty$ for some $\delta^*>0$ and all elements $a_{t,j}$ of the vector $\va_t$,  if $T^{-1} \sum_{t=1}^{[Tr]} [\E(\va_t \va_t'| \mathcal F_{t-1}) - \E(\va_t \va_t')]  \inp 0$ uniformly in $r$ and if $T^{-1} \sum_{t=1}^{[Tr]} \E(\va_t\va_t') \rightarrow r \mI_o$ uniformly in $r$, then $T^{-1/2} \sum_{t=1}^{[Tr]} \va_t \Rightarrow \mB(r)$, a $o\times 1$ vector of independent standard Brownian motions.
\end{lem}
Lemma \ref{lem3} provides sufficient conditions so that Theorem 3 in \citet{Brown:1971} is satisfied.

\begin{lem}\label{lem4} Under Assumption \ref{a8},\hfill \\
(i)
$T^{-1} \sum_{t=1}^{[Tr]} \E\, (\vl_t \vl_{t}'\, | \mathcal F_{t-1}) \inp r \mI_n$ uniformly in $r$. \\
(ii) $T^{-1} \sum_{t=1}^{[Tr]} \E\, ( (\vl_t \vl_t') \otimes \vl_{t-i} \, | \mathcal F_{t-1}) \inp  \vrho_i $ uniformly in $r$, for all $i\geq 0$. \\
(iii) $T^{-1} \sum_{t=1}^{[Tr]} \E\, ( (\vl_t \vl_t') \otimes( \vl_{t-i} \vl_{t-j}) | \mathcal F_{t-1}) \inp r \vrho_{i,j}$ uniformly in $r$, for all $i,j\geq 0$.
\end{lem}

For the following lemmas and the rest of the proofs, we need additional notation. Define $\tilde{\mathcal S}_{1}= [ \mI_{p_1+1}\,\,\,\, \,  \,\vzeros_{(p_1+1)\times p_2}] $ and $\tilde{\mathcal S}_2 = [\vzeros_{p_2 \times (p_1+1)} \,\,\,\,\,\,\mI_{p_2}]$. Also, define the following vectors of Brownian motions: $\mB_0(r)$, a $n\times 1$ vector with variance $r\mI_n$, $\mB_{l}(r)$, a $n^2\times 1$ vector with variance $r\vrho_{l,l}$ for each $l\geq 1$, $\mB_{\vzeta}(r) = \vec(\mB_{u\vzeta}(r),\mB_{\vv \vzeta}(r))$ with $\mB_{u\vzeta}(r)$ of dimension $p_2 \times 1$ and $\mB_{\vv \vzeta}(r)$ of dimension $p_1 p_2 \times 1$, where the variance of $\mB_{\vzeta}(r)$ is $r (\tilde{\mathcal S}_1 \otimes \tilde{\mathcal S}_2)\, \vrho_{0,0} (\tilde{\mathcal S}_1 \otimes \tilde{\mathcal S}_2)'=r\vrho_{\vxi,0,0}=r\left[\begin{array}{cc}\vrho_{u,\vxi,0,0}&\vrho_{u,\vv,\vxi,0,0}\\\vrho_{u,\vv,\vxi,0,0}'&\vrho_{\vv,\vxi,0,0} \end{array}\right]$, where $\vrho_{u,\vxi,0,0}$ is of dimension $p_2\times p_2$.
The covariances of these processes are:
$\cov(\mB_l(r_1), \mB_\kappa(r_2)) = \min(r_1,r_2)\,\vrho_{l,\kappa} \mbox{ for all } l,\kappa \geq 1, l\neq \kappa$, and $\cov(\mB_{\vzeta}(r_1),\mB_l(r_2))  = \min(r_1,r_2)(\tilde{\mathcal S}_1 \otimes \tilde{\mathcal S}_2)\,\vrho_{0,l} \mbox{ for all } l\geq 1$ and $\cov(\mB_{\vzeta}(r_1),\mB_0(r_2))=\min(r_1,r_2)(\tilde{\mathcal S}_1 \otimes \tilde{\mathcal S}_2)\,\vrho_{0}\,\mathcal{\tilde S}_1'=\min(r_1,r_2)\vrho_{\vxi,0}=\min(r_1,r_2)\vec(\vrho_{u,0},\vrho_{\vv,0})$, where $\vrho_{0,l}$ and $\vrho_{0}$ are given in Assumption \ref{a8}(v) and (iv) respectively, and $\vrho_{u,0}$ is of dimension $p_2 \times n$.

\begin{lem}\label{lem5}

For fixed $n^*$, under Assumption \ref{a8},\hfill
$$T^{-1/2} \sum_{t=1}^{[Tr]} \vec(\vl_t, \vl_t \otimes\vl_{t-1}, \ldots, \vl_t \otimes\vl_{t-n^*},l_{u,t}\vl_{\vzeta,t},\vl_{\vv,t} \otimes \vl_{\vzeta,t}) \Rightarrow \vec(\mB_{0}(r),\mB_{1}(r), \ldots \mB_{n^*}(r),\mB_{\vzeta}(r)),$$
where if $t-l<0$, the rest of the elements of this sum are artificially set to zero.
\end{lem}

Now define for $b=1,2$ and any $n^b \times 1$ vectors $\va$,  $\va_{\#}=\vec(\va,\vzeros_{n^b(p^b-1)})$, and for any $n^b\times n^b$ matrices $\mA$, let $\mA_{\#} = \diag(\mA,\vzeros_{n^b(p^b-1)\times n^b(p^b-1)})$, except for $\vbeta_{\vx,s,\#}$, which is $\vbeta_{\vx,s,\#}=\vec(0, \vbeta_{\vx,(s)}^0, \vzeros_{p_2+n(p-1)})$  and the subscript $s$ indicates the value of $\vbeta_{\vx,(i)}^0$ in the stable regime $\tilde I_s=[[\tau_{s-1}T]+1,[\tau_s T]]$. If $m=0$, then $\vbeta_{\vx,(s)}^0=\vbeta_{\vx}^0$, and $\vbeta_{\vx,\#}=\vec(0,\vbeta_{\vx}^0,\vzeros_{p_2+n(p-1)})$. Let $\mathcal{S}_u=\vec(1,\vzeros_{n-1},\vzeros_{n(p-1)})$  and $\mathcal{S}_{\dag}=\mathcal{S}_u$ or $\mathcal{S}_{\dag}=\vbeta_{\vx,s, \#}$, where the value  $\mathcal{S}_{\dag}$ takes is clarified in each context where the distinction between the two values is necessary. Let $\mS$, defined in Assumption \ref{a8}, and  $\mD(\tau)$, the function such that $\mD(\tau)=\mD_t$ for $\tau \in [\frac{t}{T},\frac{t+1}{T})$, be partitioned as follows:
\begin{eqnarray}\label{decompsd}
\mS&=&\left[ \begin{array}{ccc}1 &\vzeros_{1\times p_1}&\vzeros_{1\times p_2}\\
\vs_{p_1}&\mS_{p_1}&\vzeros_{p_1\times p_2}\\
\vzeros_{p_2\times 1}&\vzeros_{p_2\times p_1}&\mS_{p_2} \end{array}\right],\,\,\, \mD (\tau)=\left[\begin{array}{ccc}d_{u}(\tau)&\vzeros_{1\times p_1}&\vzeros_{1\times p_2} \\
\vzeros_{p_1 \times 1}&\mD_{\vv}(\tau)&\vzeros_{p_1\times p_2}\\
\vzeros_{p_2 \times 1}&\vzeros_{p_2\times p_1}&\mD_{\vzeta}(\tau)\end{array} \right],
\end{eqnarray} where $\vs_{p_1}$ is of dimension $p_1 \times 1$, $\mS_1$ and $\mD_{\vv}(\tau)$ are of dimension $p_1 \times p_1$, and $\mS_{p_2}$ and $\mD_{\vzeta}(\tau)$ are of dimension $p_2 \times p_2$. For any interval $[[\tau_{s-1}T]+1,[\tau_s T]]$ where the coefficients of the VAR representation in \eqref{ah1} are stable, let:
\begin{align*}
\mathbb{M}_{1}(\tau_{s-1},\tau_s) &= (\mathcal{S}_{\dag}'\mS_{\#})\left(\int_{\tau_{s-1}}^{\tau_s}\mD_{\#}(\tau)\rd \mB_{0,\#}(\tau)\right)  \\
\mathbb{M}_{2,1}(\tau_{s-1},\tau_s)&=  \sum_{l=0}^{\infty}((\mathcal{S}_{\dag}'\mS_{\#})\otimes (\mathcal{S}_{\vr}\mF_s^l)) \, \left(\left[\int_{\tau_{s-1}}^{\tau_s}  \mD_{\#}(\tau)\rd\mB_{0,\#}(\tau)\right] \otimes \vmu_s\right)\\
\mathbb{M}_{2,2}(\tau_{s-1},\tau_s) &= \sum_{l=0}^{\infty}  ((\mathcal{S}_{\dag}'\mS_{\#})\otimes (\mathcal{S}_{\vr}\mF_s^{l+1}\mA_{s,\#}^{-1}\mS_{\#}))\int_{\tau_{s-1}}^{\tau_s}  \left(\mD_{\#}(\tau)\otimes\mD_{\#}\left(\tau\right)\right)\rd\mB_{l+1,\#}(\tau) \\
\mathbb{M}_{2,3}^{(1)}(\tau_{s-1},\tau_s) &=\mS_{p_2} \int_{\tau_{s-1}}^{\tau_s} (d_{u}(\tau) \mD_{\vzeta}(\tau))\rd\mB_{u\vzeta}(\tau) \\
 \mathbb{M}_{2,3}^{(2)}(\tau_{s-1},\tau_s)& =((\vbeta_{\vx,(s)}^{0'}\vs_{p_1}) \otimes \mS_{p_2})\int_{\tau_{s-1}}^{\tau_s} (d_{u}(\tau)\otimes\vd_{\vzeta}(\tau))\rd\mB_{u\vzeta}(\tau)+((\vbeta_{\vx,(s)}^{0'}\mS_{p_1}) \otimes \mS_{p_2})\int_{\tau_{s-1}}^{\tau_s} (\mD_{\vv}(\tau)\otimes\mD_{\vzeta}(\tau))\rd\mB_{\vv\vzeta}(\tau)\\
\mathbb M_2 (\tau_{s-1},\tau_s) &= \mathbb{M}_{2,1}(\tau_{s-1},\tau_s) + \mathbb{M}_{2,2}(\tau_{s-1},\tau_s)+ \mathbb{M}_{2,3}^{(j)}(\tau_{s-1},\tau_s), \mbox{ where } j=1 \mbox{ if } \mathcal{S}_{\dag}=\mathcal{S}_u \mbox{ and } j=2 \mbox{ otherwise} \\
\mathbb{M}_3(\tau_{s-1},\tau_s) & = \sum_{l=0}^{\infty}  \left((\mathcal{S}_{\dag}'\mS_{\#}) \otimes (\mathcal{S}\mF_s^l)\right) \, \left(\left[\int_{\tau_{s-1}}^{\tau_s}  \mD_{\#}(\tau)\rd\mB_{0,\#}(\tau)\right] \otimes \vmu_s\right)\\
&+  \sum_{l=0}^{\infty}  ((\mathcal{S}_{\dag}'\mS_{\#})\otimes (\mathcal{S}\mF_s^l\mA_{s,\#}^{-1}\mS_{\#}))\int_{\tau_{s-1}}^{\tau_s}  \left(\mD_{\#}(\tau)\otimes\mD_{\#}\left(\tau\right)\right)\rd\mB_{l+1,\#}(\tau)\\
 \mathbb M(\tau_{s-1},\tau_s) & = \vec( \mathbb{M}_{1}(\tau_{s-1},\tau_s),\mathbb{M}_{2}(\tau_{s-1},\tau_s),\mathbb{M}_{3}(\tau_{s-1},\tau_s)),
\end{align*}
where $\mathcal S_{\vr}$ was defined in Lemma \ref{lem2}.
\begin{lem}\label{lem6}
Let the interval $I_i$ contain $N_i$ breaks from the total set of $N$ breaks. Then, under Assumptions \ref{a1}-\ref{a8}, \begin{align*}
T^{-1/2} \sum_{t \in I_i} \vz_t u_t \Rightarrow \widetilde{ \mathbb M}_i = \begin{cases} \mathbb M(\lambda_{i-1}, \tau_{s}) + \sum_{j=1}^{N_i}\mathbb M(\tau_{s+j-1}, \tau_{s+j}) + \mathbb M(\tau_{s+N_i},\lambda_i) & \mbox{ if } N_i\geq 2\\
 \mathbb M(\lambda_{i-1}, \tau_s) +  \mathbb M(\tau_s, \lambda_i) & \mbox{ if } N_i =1 \\
 \mathbb M(\lambda_{i-1}, \lambda_i) & \mbox{ if } N_i=0. \end{cases},
\end{align*}
with  $\mathcal{S}_{\dag}=\mathcal{S}_u$. Similarly, $ T^{-1/2} \sum_{t \in I_i} \vz_t \vv_t'\vbeta_{\vx,(i)}^ 0 \Rightarrow \widetilde{\mathbb{M}}_i$ but with $\mathcal{S}_{\dag}=\vbeta_{\vx,i,\#} = \vec(0, \vbeta_{\vx,(i)}^0,\vzeros_{p_2+n(p-1)})$. If $m=0$, then
$\mathcal{S}_{\dag}= \vbeta_{\vx,\#}$.

\end{lem}

\begin{lem}\label{lem7}
Under Assumptions \ref{a1}-\ref{a8}, \\
(i) if $h>0$, then $T(\hat \pi_i -\pi_i^0) = \Op(1), i=1,\ldots, h+1$;\\
(ii) $T^{1/2}(\hat \mDelta_{(i)}-\mDelta_{(i)}^0) = \Op(1)$ for $i=1,\ldots, h+1$;\\
(iii) if $m>0$, $T(\hat \lambda_i-\lambda_i^0)=\Op(1)$, $i=1, \ldots, m+1$.
\end{lem}

\begin{lem}\label{lem8}
Under Assumption \ref{a8}, uniformly in $r$,\hfill \\
(i) $T^{-1} \sum_{t=1}^{[Tr]} \left\{\vepsi_t \vepsi_t' -\E(\vepsi_t \vepsi_t')\right\} \inp 0
$, \\
(ii) $T^{-1} \sum_{t=1}^{[Tr]} \left\{(\vepsi_t \vepsi_t') \otimes \vepsi_{t-i} - \E[(\vepsi_t \vepsi_t') \otimes \vepsi_{t-i}]\right\} \inp 0$ for all $i\geq 0,$\\
(iii) $T^{-1} \sum_{t=1}^{[Tr]} \left\{(\vepsi_t \vepsi_t') \otimes (\vepsi_{t-i} \vepsi_{t-j}') - \E[(\vepsi_t \vepsi_t') \otimes (\vepsi_{t-i} \vepsi_{t-j}')]\right\} \inp 0$ for all $i,j\geq 0$\\
(iv) Parts (i)-(iii) hold with $\vl_t,\vl_{t-i},\vl_{t-j}$ replacing $\vepsi_t,\vepsi_{t-i},\vepsi_{t-j}$. \end{lem}

\begin{lem}\label{lem9}
Let $\hat \mQ^b_{(i)} = T^{-1}  \sum_{t\in I_i} \hat{\mUpsilon}_t' \vz_t^b \vz_t^{b'}\hat{\mUpsilon}_t$. Then, under Assumptions \ref{a1}-\ref{a8},
$
\hat{\mQ}^b_{(i)}=\mathbb Q_i + o_p^b(1),
$
where
\begin{equation}\label{defq}
\mathbb Q_i = \int_{\lambda_{i-1}}^{\lambda_i} \mUpsilon(\tau)' \mathbb Q_{\vz}(\tau)\mUpsilon(\tau) d\tau.
\end{equation}
\end{lem}
\begin{proof}[Proof of Lemma \ref{lem9}]\hfill\\
For the WF bootstraps, $\vz_t^b=\vz_t$, and therefore  Lemma \ref{lem9} holds by Lemma \ref{lem2}. Consider the WR bootstrap, first for $I_i = \tilde I_s$. Define $\tilde \vz_{t}^b=(y_t^b,\vx_t^{b\prime},\vr_t)'$,  and:
\begin{eqnarray*}
\hat \mQ^b_{(i)}=T^{-1} \sum_{t \in I_i} \hat\mUpsilon_t^{\prime} \vz_t^b \vz_t^{b\prime} \hat\mUpsilon_t=\hat\mUpsilon^{\prime}_s\,\left[\,\begin{array}{ccc} \Delta\tau_s &\mathcal{A}_1^{b\prime}\mathcal{S}_{\vr}^\prime&\mathcal{A}_2^{b\prime}\mathcal{S}'\\ \mathcal{S}_{\vr}\mathcal{A}_1^b & \mathcal{S}_{\vr}\mathcal{B}_1^b\mathcal{S}_{\vr}^\prime&\mathcal{S}_{\vr}\mathcal{B}_{2}^b\mathcal{S}'\\
\mathcal{S}\mathcal{A}_2^b&\mathcal{S}\mathcal{B}_{2}^{b\prime}\mathcal{S}_{\vr}'&\mathcal{S}\mathcal{B}_3^b\mathcal{S}'\end{array}\,\right]\,\hat\mUpsilon_s,\end{eqnarray*}
where
$$
\mathcal{A}_1^b\;=\;\,T^{-1} \sum_{t \in I_i}\vxi_t^b,\;\; \mathcal{A}_2^b\;=\;\,T^{-1} \sum_{t \in I_i}\vxi_{t-1}^b
$$
and
$$
\mathcal{B}_1^b\;=\;T^{-1} \sum_{t \in I_i} \vxi_t^b\vxi_t^{b\prime},\;\; \mathcal{B}_{2}^b\;=\;T^{-1} \sum_{t \in I_i} \vxi_t^b\vxi_{t-1}^{b\prime},\;\;\mathcal{B}_3\;=\;T^{-1} \sum_{t \in I_i} \vxi_{t-1}^b\vxi_{t-1}^{b\prime}.
$$

Note that, because $\vr_t$ is kept fixed, $\mathcal{S}_{\vr}\mathcal{A}_1^b = T^{-1} \sum_{t \in I_i}\mathcal{S}_{\vr}\vxi_{t}=\mathcal{S}_{\vr}\mathcal A_1$,  and $ \mathcal{S}_{\vr}\mathcal{B}_1^b\mathcal{S}_{\vr} = \mathcal{S}_{\vr}\mathcal{B}_1\mathcal{S}_{\vr}^\prime$, where $\mathcal A_1, \mathcal A_2$ are the sample counterparts of $\mathcal A_1^b, \mathcal B_1^b$ defined at the beginning of the proof of Lemma \ref{lem2}. By Lemma \ref{lem2}, the result in Lemma \ref{lem9} holds automatically for these terms.
We now analyze the rest of the terms. To that end, we first derive some preliminary results.

$\bullet$ \textbf{Preliminary results and bootstrap notation.} Note that in any stable subinterval $\tilde I_s$,
\begin{equation}
\label{ah2}
\tilde\vz_t^b \;=\;\hat\vc_{\tilde{\vz},s}\,+\, \sum_{i=1}^{p} \hat\mC_{i,s} \tilde\vz_{t-i}^b+ \ve_t^b,\qquad [\tau_{s-1}T]+1\,\leq t\,\leq [\tau_s T],\; s=1,2,\ldots, N+1,
\end{equation}
where $\ve^b_t=\hat\mA_s^{-1}\vepsi^b_t$,
$\vepsi_t^b=\vec(u_t^b,\vv_t^{b},\vzeta_t)$, of size $n\times 1$, and the elements of $\hat\mA_s$, $\hat\vc_{\tilde{\vz},s}$ and $\hat\mC_{i,s}$ corresponding to the equation for $\vr_t$ are the true parameters, not the estimated ones. Then, \begin{eqnarray}
\vxi_t^b\;&=&\;\hat \vmu_s^b\,+\,\hat\mF_s\vxi^b_{t-1}\,+\,\veta^b_t\\ \label{xib}
&=& \hat\mF^{t-[\tau_{s-1}T]}_s \vxi_{[\tau_{s-1}T]}^b + \left(\sum_{l=0}^{t-[\tau_{s-1}T]-1}\hat\mF_s^l\right)\hat\vmu_s +\sum_{l=0}^{t-[\tau_{s-1}T]-1}\hat\mF_s^l\veta^b_{t-l},
\end{eqnarray}
where $\tilde\vxi^b_t= \vec(\tilde{\vz}_t^b,\tilde{\vz}_{t-1}^b,\ldots, \tilde{\vz}_{t-p+1}^b)$, $\veta_t^b = \hat \mA_{s,\#}^{-1} \vepsi^b_{t,\#}$, and $\hat \mF_s,\hat \vmu_s$ are defined as $\mF_s,\vmu_s$, but replacing the true coefficients that are estimated by 2SLS with those estimated counterparts. Also, let $\hat\veta_t = \hat \ve_{t,\#} = \hat \mA_{s,\#}^{-1} \hat \vepsi_{t,\#}$, where $\hat\vepsi_t = \vec(\hat u_t,\hat \vv_t, \vzeta_t)$.

We now show two results that we repeatedly need in the proofs: $T^{-\alpha} \vxi_t^b =o_p^b(1)$ and $T^{-\alpha}\vxi_t^b\vxi_t^{b'}=\op^b(1)$ for any $\alpha>0$.

For this purpose, we first show that $\E^b(T^{-\alpha} \veta_t^b)=\op^b(1)$ and that $\var^b(T^{-\alpha} \veta_t^b)=\op^b(1)$. Then, by Markov's inequality, for any $C>0$, $P^b(T^{\alpha} \|\veta_t^b-E^b(\veta_b)\| \geq C) \leq C^{-2} T^{-2\alpha} \var^b \| \veta_t^b\| \inp 0$, completing the proof.

Let $\mathcal I = \vec(\vzeros_{p_1+1}, \viota_{p_2}, \vzeros_{p_2+n(p+1)})$ and $\mathcal J = [\diag(\mJ_{p_1+1},\mJ_{p_2+1})]_{\#}$, where $\viota_a$ is a $a\times 1$ vector of ones, and $\mJ_a = \viota_a\viota_a'$.
Let $\vnu_t = [\vec(\nu_t \viota_{p_1+1},\viota_{p_2})]_{\#}$. Then $E^b(\vnu_t) = \mathcal I$ and $E^b (\vnu_t \vnu_t')= \mathcal J$.

Also, let $\vg_t^b = \vepsi_{t,\#}^b = \hat \vepsi_{t,\#} \odot \vnu_t$, where $\odot$ is the element-wise multiplication. Then $\vg_t^b=\hat\mA_{s,\#}\;\veta_{t}^b$, and letting  $\hat \vg_t = \hat \vepsi_{t,\#}$,  it follows that $\vg_t^b = \hat \vg_t \,\odot\, \vnu_t$. Further, let $\hat \vg_{t,1} \equiv \vec(\hat u_t, \hat \vv_t,\vzeros_{p_2+n(p-1)})$ and  $\vg_{t,2} \equiv \vec(\vzeros_{(p_1+1)},\vzeta_t,\vzeros_{n(p-1)})$.
Also, note that  $\vg_{t}^b= \hat \mA_{s,\#}^{-1}(\hat \vg_t \odot \mathcal \vnu_t)$. Then:
\begin{align}\label{eb1}
\E^b(\veta_t^b) &= E^b(  \hat \mA_{s,\#}^{-1}(\hat \vg_t \odot \vnu_t)= (\hat \mA_{s,\#}^{-1}(\vg_t^b \, \odot  \mathcal I) =  \hat \mA_{s,\#}^{-1} \vec(\vzeros_{(p_1+1)},\vzeta_t,\vzeros_{n(p-1)}) = \hat \mA_{s,\#}^{-1} \vg_{t,2}\\ \nonumber
\E^b(\veta_t^b\veta_t^{b'}) &= E^b(  \hat \mA_{s,\#}^{-1}(\hat \vg_t \odot \vnu_t)(\hat \vg_t \odot \vnu_t) '\hat \mA_{s,\#}^{'-1})\\ \label{eb2}
&= \hat \mA_{s,\#}^{-1} \, [(\vg_t^b \vg_t^{b'}) \odot \mathcal J ]  \hat \mA_{s,\#}^{-1} = \hat \mA_{s,\#}^{-1} \begin{bmatrix}\hat u_t^2 & \hat u_t \hat \vv_t' & \vzeros_{1\times p_2}  \\
\hat \vv_t\, \hat u_t & \hat \vv_t \hat \vv_t'& \vzeros_{p_1 \times p_2} \\
\vzeros_{p_2} & \vzeros_{p_2 \times p_1} & \vzeta_t\vzeta_t' \end{bmatrix}_{\#} \hat \mA_{s,\#}^{-1} = \hat \mA_{s,\#}^{-1} (\hat \vg_{t} \hat \vg_{t}' \odot \mathcal J)\hat \mA_{s,\#}^{-1}\\   \label{varb}
\var^b(\veta_t^b) & = \E^b(\veta_t^b\veta_t^{b'}) - \E^b(\veta_t^b) \E^b(\veta_t^{b'}) = \hat \mA_{s,\#}^{-1} \begin{bmatrix}\hat u_t^2 & \hat u_t \hat \vv_t' & \vzeros_{1\times p_2}   \\
\hat \vv_t\, \hat u_t & \hat \vv_t \hat \vv_t'& \vzeros_{p_1 \times p_2} \\
\vzeros_{p_2} & \vzeros_{p_2 \times p_1} & \vzeros_{p_2\times p_2} \end{bmatrix}_{\#} \hat \mA_{s,\#}^{-1} = \hat \mA_{s,\#}^{-1} \hat \vg_{t,1} \hat \vg_{t,1}' \hat \mA_{s,\#}^{-1}.
\end{align}

By Lemma \ref{lem7}, Lemma \ref{lem8} and standard 2SLS theory, $ \hat \vg_{t,1} = \vec(\hat u_t, \hat \vv_t,\vzeros_{p_2+n(p-1)})= \Op(1)$,  and $\hat \mA_{s,\#} = \mA_{s,\#} +\op(1)$, therefore  $\hat \mA_{s,\#}^{-1} \,  \hat \vg_{t,1}=\Op(1)$, so $E^b [T^{-\alpha} \veta_t^b] =\op^b(1)$.

Next, we show that $T^{-\alpha} \vxi_{t} =\op^b(1)$ by induction. First, recall that $\vxi_0^b = \vxi_0$, and therefore, $T^{-\alpha}\vxi_1^b  =T^{-\alpha}\hat \vmu_1 + \hat \mF_1 T^{-\alpha}\vxi_0+ T^{-\alpha}\veta_1^b  =\op^b(1)$ because $\hat \vmu_s-\vmu_s = \op(1)$, $\hat \mF_s -\hat \mF_s = \op(1)$, and $T^{-\alpha} \veta_t^b = \op^b(1)$. Now let $T^{-\alpha}\vxi_{t-1}^b =\op(1)$; then for $t, t-1 \in \tilde I_s$,  $T^{-\alpha}\vxi_{t}^b = T^{-\alpha} \hat \vmu_s + \hat \mF_s  T^{-\alpha}\vxi_{t-1}^b + T^{-\alpha}\veta_{t}^b =  o_p(1)+ \hat \mF_s o_p^b(1) +  o_p^b(1) =  o_p^b(1)$.

Therefore, it follows that:
\begin{align}\label{rxib}
T^{-\alpha} \vxi_{t}^b = \op^b(1).
\end{align}

Next, we show that $T^{-\alpha} \vxi_t^b \vxi_t^{b'} = o_p^b(1)$, also by mathematical induction. Note that, from the results above,
\begin{align}\nonumber
T^{-\alpha} \vxi_t^b \vxi_t^{b'}  &= T^{-\alpha} (\hat \vmu_s + \hat \mF_s \vxi_{t-1}^b + \veta_t^b)(\hat \vmu_s+ \hat \mF_s \vxi_{t-1}^b + \veta_t^b)'\\ \nonumber
& = T^{-\alpha} \hat \vmu_s \hat \vmu_s'  + \hat \mF_s ( T^{-\alpha} \vxi_{t-1}^b \vxi_{t-1}^{b'}) \hat \mF_s' +  T^{-\alpha} \veta_t^b \veta_t^{b'} + T^{-\alpha} \hat \vmu_s \vxi_{t-1}^{b'}\hat \mF_s' + (T^{-\alpha} \hat \vmu_s \vxi_{t-1}^{b'}\hat \mF_s')' \\ \nonumber
& +  T^{-\alpha} \hat \vmu_s \veta_t^{b'} + (T^{-\alpha} \hat \vmu_s \veta_t^{b'})' + \hat \mF^s T^{-\alpha/2}\vxi_{t-1}^b T^{-\alpha/2} \veta_t^{b'} +(\hat \mF^s T^{-\alpha/2}\vxi_{t-1}^b T^{-\alpha/2} \veta_t^{b'})'\\ \label{xi2}
& =  \hat \mF_s ( T^{-\alpha} \vxi_{t-1}^b \vxi_{t-1}^{b'}) \hat \mF_s'  +  T^{-\alpha} \veta_t^b \veta_t^{b'}+ o_p^b(1).
\end{align}
Now consider $\vec(T^{-\alpha} \veta_t^b \veta_t^{b'})=T^{-\alpha} \veta_t^b \otimes \veta_t^{b}$.  We have:
\begin{align*}
&\veta_t^b \otimes \veta_t^{b}  =  (\hat \mA_{s,\#}^{-1} \otimes \hat \mA_{s,\#}^{-1}) (\vg_{t}^b \otimes \vg_t^b) \mbox{ and } E^b( \vg_t^b \otimes \vg_t^{b})( \vg_t^b \otimes \vg_t^{b})' = E^b( (\vg_t^b \vg_t^{b'}) \otimes (\vg_t^b \vg_t^{b'})).
\end{align*}
Since $\vg_t^b \vg_t^{b'} = (\hat \vg_t \hat \vg_t') \odot \vnu_t\vnu_t' =\Op(1)\odot (\vnu_t\vnu_t') $, the typical non-zero elements of $E^b( (\vg_t^b \vg_t^{b'}) \otimes (\vg_t^b \vg_t^{b'}))$ are $\Op(1) E^b(\nu_t^{j})$, with $j=0,1,\ldots, 4$. By Assumption \ref{aboot}, $\sup_t E^b(\nu_t^4) <\infty$, it follows that $ E^b( (\vg_t^b \vg_t^{b'}) \otimes (\vg_t^b \vg_t^{b'})) = \Op(1)$, which implies that  $T^{-\alpha} E^b( (\veta_t^b \veta_t^{b'}) \otimes (\veta_t^b \veta_t^{b'})) = \Op(T^{-\alpha})=\op(1)$. By Markov's inequality, for any $C>0$, $P^b( T^{-\alpha} \| \veta_t^b \otimes \veta_t^{b'})\| \geq C) \leq T^{-2\alpha} C^{-2} E^b\|\veta_t^b \otimes \veta_t^{b'}\|^2 \leq T^{-2\alpha} C^{-2}\| E^b( (\veta_t^b \veta_t^{b'}) \otimes (\veta_t^b \veta_t^{b'}))\| \inp 0$,  it follows that $T^{-\alpha} \veta_t^b \veta_t^{b'}=\op^b(1)$.

Using this result in \eqref{xi2}, by a similar mathematical induction argument as  for $T^{-\alpha} \vxi_{t}^b=o_p^b(1)$, it follows  that
\begin{align}\label{rxibx}
T^{-\alpha} \vxi_{t}^b \vxi_{t}^{b'} = \op^b(1).
\end{align}

Besides \eqref{rxib} and \eqref{rxibx}, in the proof below we will  assume that  $\left|\mI_{n}\,-\,\hat \mC_{1,s}a\, -\, \hat \mC_{2,s} a^2\, -\, \cdots\, -\, \hat \mC_{p,s} a^{ p}\right|\neq 0$, for all $s=1,\ldots, N+1$, and all $\left|a\right|\leq 1$; otherwise the estimated system is not stationary. Then we show in the Supplementary Appendix, Section \ref{supsec2}, that $\sum_{l=0}^{\infty} \| \mF_s^l\| <\infty$, and similarly, it can be shown that $\sum_{l=0}^{\infty} \| \hat \mF_s^l\| <\infty$. Moreover, the results in the Supplementary Appendix  Section \ref{supsec3} show that  $\mathcal R_{s,l} =\hat \mF_s^l - \mF_s^l $ is such that
\begin{equation}\label{mathcalR}
\sum_{l=0}^{\infty} \| \mathcal R_{s,l} \| =\| \hat \mF_s - \mF_s\| \, \,\Op(1) = \op(1),
\end{equation} an argument which will be used repeatedly in the proofs.

$\bullet$ Now consider the case where $I_i = \tilde I_s$ first, and analyze $\mathcal{A}_2^b$. From \eqref{rxib},
\begin{align} \label{a2b}
\mathcal{A}_2^b = T^{-1} \sum_{t \in \tilde I_s} \vxi_{t-1}^b  = T^{-1}\vxi_{[\tau_{s-1} T]-1}^b -  T^{-1} \vxi_{[\tau_{s} T]}^b + T^{-1} \sum_{t \in \tilde I_s} \vxi_{t}^b = T^{-1} \sum_{t \in \tilde I_s} \vxi_{t}^b+ \op^b(1) = \mathcal A_1^b+\op(1).  \end{align}
Therefore, we now derive the limit of $\mathcal A_1^b$. Note that
\begin{eqnarray*}
\vxi_t^b\;&=&\;\hat\vmu_s\,+\,\hat\mF_s\vxi^b_{t-1}\,+\,\veta^b_t \;=\;\hat\mF^{t-[\tau_{s-1}T]}_s \vxi_{[\tau_{s-1}T]}^b + \tilde\vxi_t^b+\left(\sum_{l=0}^{t-[\tau_{s-1}T]-1}\hat\mF_s^l\right)\hat\vmu_s,
\end{eqnarray*}
where $\tilde\vxi^b_t=\sum_{l=0}^{t-[\tau_{s-1}T]-1}\hat\mF_s^l\veta^b_{t-l}$. Therefore,
$
\mathcal{A}_1^b\;=\;\sum_{i=1}^4 \mathcal{A}_{1,i}^b,
$
where $\Delta\tau_sT = [\tau_{s}T]-[\tau_{s-1}T]$, and
\begin{eqnarray*}
\mathcal{A}_{1,1}^b&=&T^{-1}\sum_{t=[\tau_{s-1}T]+1}^{[\tau_sT]}\tilde{\vxi}_t^b, \, \,
\mathcal{A}_{1,2}^b \;=\;
T^{-1}\Delta\tau_sT\sum_{l=0}^{\Delta \tau_sT -1}\hat\mF_s^l\,\hat\vmu_s,\\
\mathcal{A}_{1,3}^b
&=&T^{-1}\sum_{l=1}^{\Delta \tau_s T} \hat\mF_s^l\,\vxi_{[\tau_{s-1}T]}^b,\,\,
\mathcal{A}_{1,4}^b\;=\;-T^{-1}\left(\sum_{l=1}^{\Delta \tau_s T -1} l\hat\mF_s^l\right)\,\hat\vmu_s.
\end{eqnarray*}

We show that $\mathcal{A}_{1,1}^b=\op^b(1)$. First, we show $\E^b(\mathcal{A}_{1,1}^b)=o_p(1)$. Second, we show $\var^b(\mathcal{A}_{1,1}^b)=\op(1)$ which by Markov's inequality implies that $\mathcal{A}_{1,1}^b=\op^b(1)$. Consider $\E^b(\mathcal{A}_{1,1}^b)$ with $\E^b(\tilde{\vxi}_t^b)=\sum_{l=0}^{\tilde t-1}\hat\mF_s^l\E^b(\veta^b_{t-l})=\sum_{l=0}^{\tilde t-1}\hat\mF_s^l\hat\mA_{s,\#}^{-1}\left(\hat\vg_{t-l}\odot\mathcal{I}\right)$.

We have $\vxi_t=\vmu_s+\mF_s\vxi_{t-1}+\veta_t=\hat\vmu_s+\hat\mF_s\vxi_{t-1}+\hat\veta_t$. Then,
\begin{eqnarray}
\label{etahat}\hat\veta_t&=&\veta_t+(\vmu_s-\hat\vmu_s)+(\mF_s-\hat\mF_s)\vxi_{t-1},\\
\hat\vg_t=\hat\mA_{s,\#}\;\hat\veta_t&=& \hat\mA_{s,\#}\veta_t+\hat\mA_{s,\#}(\vmu_s-\hat\vmu_s)+\hat\mA_{s,\#}(\mF_s-\hat\mF_s)\vxi_{t-1}.
\end{eqnarray}
Therefore,
$
\veta_t^b = (\veta_t+(\vmu_s-\hat\vmu_s)+(\mF_s-\hat\mF_s)\vxi_{t-1})\odot \vnu_t $.

Note that $\hat \vmu_s - \vmu_s = (\hat \vc_{\tilde z,s}-\vc_{\tilde z,s})_{\#}  = \vec( \hat{d}_{s}, \hat \vd_s,  \vzeros_{p_2})_{\#}$, where $\hat d_s, \hat \vd_s$ are of dimension $1$ and $p_1 \times 1$, respectively, and this holds because the rows $p_2+1:n$ are not estimated since the equation for $\vr_t$ is not estimated. Let $\hat \va_{1,\bcdot}, \hat \mA_{p_1,\bcdot}, \hat \mA_{p_2,\bcdot}$ be rows $1$, $2:p_1+1$ and $p_1+2:n$ of the matrix $\hat \mA_s^{-1}$ respectively. Note that like $\mA_{s}^{-1}$, $\hat \mA_s^{-1}$ is upper triangular  with $\hat \mA_{p_2,\bcdot} =[\vzeros_{p_2\times (p_1+1)}, \mI_{p_2}]$ because the equation for $\vr_t$ is not estimated, and $\vr_t$ is assumed contemporaenously exogenous. Therefore,
\begin{eqnarray}\label{decompA}
\hat \mA_{s,\#}^{-1}\;(\vmu_s-\hat\vmu_s)=\left[\begin{array}{c} \hat \mA_s^{-1} (\hat \vc_{\tilde z,s}-\vc_{\tilde z,s})\\ \vzeros_{n(p-1)}\end{array}\right]=\left[\begin{array}{c} \hat \va_{1,\bcdot}\; (\hat \vc_{\tilde z,s}-\vc_{\tilde z,s})\\ \hat \mA_{p_1,\bcdot}\; (\hat \vc_{\tilde z,s}-\vc_{\tilde z,s})\\ \hat \mA_{p_2,\bcdot}\; (\hat \vc_{\tilde z,s}-\vc_{\tilde z,s})\\\vzeros_{n(p-1)}\end{array}\right]=\left[\begin{array}{c} \hat \va_{1,\bcdot}\;(\hat \vc_{\tilde z,s}-\vc_{\tilde z,s})\\ \hat \mA_{p_1,\bcdot}\;(\hat \vc_{\tilde z,s}-\vc_{\tilde z,s})\\ \vzeros_{p_2}\\\vzeros_{n(p-1)}\end{array}\right],
\end{eqnarray}
so
\begin{equation}\label{first0}
(\hat \mA_{s,\#}(\hat \vmu_s - \vmu_s)) \odot \mathcal I =\vzeros_{np}.
\end{equation}

By similar arguments, because the $p_1+2:n$ rows of $\hat \mF_s$ are equal to the corresponding rows of $\mF_s$, $\hat \mA_{s,\#}(\hat \mF_s - \mF_s)$, the rows $p_1+2:n$ of   $\hat \mA_{s,\#}(\hat \mF_s - \mF_s)$ are equal to zero, therefore
\begin{equation}\label{second0}
(\hat \mA_{s,\#}(\hat \mF_s - \mF_s) \vxi_{t-1}) \odot \mathcal I =\vzeros_{np}.
\end{equation}
Using \eqref{first0}-\eqref{second0}, and letting $\tilde t = t-[\tau_{s-1}T]$, we have:
 \begin{eqnarray*}
\E^b(\mathcal{A}_{1,1}^b)&=& \sum_{i=1}^3\mathcal{H}_i,\,\,\text{where:}\\
\mathcal{H}_1&=&T^{-1}\sum_{t\in \tilde I_s}\sum_{l=0}^{\tilde t-1}\hat\mF_s^l\hat\mA_{s,\#}^{-1}\left((\hat\mA_{s,\#}\veta_{t-l})\odot\mathcal{I}\right)=T^{-1}\sum_{t\in \tilde I_s}\sum_{l=0}^{\tilde t-1}\hat\mF_s^l\hat\mA_{s,\#}^{-1}\left((\hat\mA_{s,\#}\mA_{s,\#}^{-1}\;\vg_{t-l})\odot\mathcal{I}\right),\\
\mathcal{H}_2&=& T^{-1}\sum_{t\in \tilde I_s}\sum_{l=0}^{\tilde t-1}\hat\mF_s^l\hat\mA_{s,\#}^{-1}\left((\hat\mA_{s,\#}(\vmu_s-\hat\vmu_s))\odot\mathcal{I}\right) =\vzeros_{np}, \\
\mathcal{H}_3&=& T^{-1}\sum_{t\in \tilde I_s}\sum_{l=0}^{\tilde t-1}\hat\mF_s^l\hat\mA_{s,\#}^{-1}\left((\hat\mA_{s,\#}(\mF_s-\hat\mF_s)\vxi_{t-1})\odot\mathcal{I}\right)=\vzeros_{np}. \end{eqnarray*}

 Since $\hat\mA_{s,\#}\mA_{s,\#}^{-1}=\mI_{n,\#}+\op(1)$, it follows that:
\begin{eqnarray*}
\mathcal{H}_1&=&T^{-1}\sum_{t\in \tilde I_s}\sum_{l=0}^{\tilde t-1}\mF_s^l\mA_{s,\#}^{-1}\left(\vg_{t-l}\odot\mathcal{I}\right)+T^{-1}\sum_{t\in \tilde I_s}\sum_{l=0}^{\tilde t-1}\mathcal{R}_{s,l}\mA_{s,\#}^{-1}\left(\vg_{t-l}\odot\mathcal{I}\right)+\op(1)\\
&=&\mathcal{H}_{1}^{(1)}+\mathcal{H}_{1}^{(2)}+\op(1).
\end{eqnarray*}
From Assumptions \ref{a7} and \ref{a8}, and using $\sum_{l=0}^{\infty} \|\mF_s^l\| <\infty$ and  $\sum_{l=0}^{\infty} \|\mathcal R_{s,l}\| <\infty$, both proven in the Supplementary Appendix Sections \ref{supsec2}-\ref{supsec3}, it can be shown that  $\sum_{l=0}^{\tilde t-1}\mF_s^l\mA_{s,\#}^{-1}\left(\vg_{t-l}\odot\mathcal{I}\right)$ and $T^{-1}\sum_{t\in \tilde I_s}\sum_{l=0}^{\tilde t-1}\mathcal{R}_{s,l}\mA_{s,\#}^{-1}\left(\vg_{t-l}\odot\mathcal{I}\right)$ are $L^1$-mixingales satisfying the conditions of Lemma \ref{lem1}, therefore $\mathcal{H}_{1}^{(1)}=o_p(1)$ and $\mathcal{H}_{1}^{(2)}=o_p(1)$. Hence $\mathcal{H}_1=\op(1)$, so $\E^b(\mathcal{A}_{1,1}^b)=\op(1)$.

Second, we show that $\var^b(\mathcal{A}_{1,1}^b)=\op(1)$. To that end, note that
\begin{eqnarray*}
\E^b(\veta^b_{t-l}\veta^{b\prime}_{t-\kappa})&=&\hat\mA_{s,\#}^{-1}\E^b((\hat \vg_{t-l}\odot\vnu_{t-l})(\hat \vg_{t-\kappa}\odot\vnu_{t-\kappa})')(\hat\mA_{s,\#}^{-1})'\;=\;\hat\mA_{s,\#}^{-1} ((\hat\vg_{t-l}\hat\vg_{t-\kappa}')\odot\E^b(\vnu_{t-l}\vnu_{t-\kappa}'))(\hat\mA_{s,\#}^{-1})',
\end{eqnarray*}
For $l\neq \kappa$, $\E^b(\vnu_{t-l}\vnu_{t-\kappa}^\prime)= \mathcal I \mathcal I' = [\diag( \vzeros_{(p_1+1)\times (p_1+1)}, \mJ_{p_2})]_{\#} = \mathcal{J}_2$. Therefore, exploiting the upper triangular structure of $\hat \mA_s^{-1}$ with $p_2 \times p_2$ lower right block equal to $\mI_{p_2}$, for $l\neq \kappa$,  \begin{align*}
\E^b(\veta^b_{t-l}\veta^{b\prime}_{t-\kappa})=\hat \mA_{s,\#}^{-1} ((\hat\vg_{t-l}\hat\vg_{t-\kappa}')\odot\mathcal J_2)(\hat\mA_{s,\#}^{-1})' = ((\hat \mA_{s,\#}^{-1} \hat \vg_{t-l})\odot\mathcal I)((\hat \mA_{s,\#}^{-1} \hat \vg_{t-\kappa})\odot\mathcal I)' =\veta_{t-l,2} \veta_{t-\kappa,2}',
 \end{align*}
where $\veta_{t,2}=\vg_{t,2}=[\vec(\vzeros_{p_1+1},\vzeta_t)]_{\#}$.
For
$l=\kappa$, $\E^b(\veta^b_{t-l}\veta^{b\prime}_{t-l})=\hat\mA_{s,\#}^{-1} ((\hat\vg_{t-l}\hat\vg_{t-l}')\odot \E^b(\vnu_{t-l}\vnu^{\prime}_{t-l}))(\hat\mA_{s,\#}^{-1})'$, where $\E^b(\vnu_{t-l}\vnu_{t-\kappa}^\prime)=\mathcal{J}$, so $T^{-1}  \sum_{t\in \tilde I_s}(\hat\vg_{t-l}\hat\vg_{t-l}')\odot \mathcal{J} = T^{-1}  \sum_{t\in \tilde I_s}\vg_{t-l}\vg_{t-l}'+\op(1)$  by Assumption \ref{a8}(ii), Lemma \ref{lem8} followed by standard 2SLS theory. Therefore, $T^{-1}  \sum_{t\in \tilde I_s} \E^b(\veta^b_{t-l}\veta^{b\prime}_{t-l}) = \hat\mA_{s,\#}^{-1} \left(T^{-1}  \sum_{t\in \tilde I_s} \vg_{t-l} \vg_{t-l}' \right)(\hat\mA_{s,\#}^{-1})'+\op(1) = T^{-1}  \sum_{t\in \tilde I_s}  \veta_{t-l} \veta_{t-l}'+\op(1)$. Hence,
\begin{eqnarray}\label{a11b}\var^b(\mathcal{A}_{1,1}^b)&=&T^{-2}\sum_{t\in \tilde I_s}\E^b(\tilde{\vxi}_t^b\tilde{\vxi}_t^{b'})=\mathcal{V}_1+\mathcal{V}_2+\op(1)\\
\mathcal{V}_1&=&T^{-2}\sum_{t\in \tilde I_s}\sum_{l=0}^{\tilde t-1}\hat\mF_s^l \veta_{t-l} \veta_{t-l}' (\hat\mF_s^l)'\label{vv1}\\
\mathcal{V}_2&=&T^{-2}\sum_{t\in \tilde I_s}\sum_{l,\kappa=0,l\neq \kappa}^{\tilde t-1}\hat\mF_s^l \veta_{t-l,2} \veta_{t-\kappa,2}'(\hat\mF_s^\kappa)'\label{vv2}
\end{eqnarray}
Consider $\mathcal{V}_1$. We analyze first $\mathcal V_1^{*}$, where
\begin{eqnarray*}
\mathcal{V}_{1}^*&=&T^{-1}\sum_{t\in \tilde I_s}\sum_{l=0}^{\tilde t-1}\mF_s^l\left(\mA_{s,\#}^{-1} \vg_{t-l}\vg_{t-l}'(\mA_{s,\#}^{-1})'\right)(\mF_s^l)' =T^{-1}\sum_{t\in \tilde I_s}\sum_{l=0}^{\tilde t-1}\mF_s^l \veta_{t-l}\veta_{t-l}' \mF_s^{l'} \\
& = &\mathcal B_{1,1}^{(1)}+ \mathcal B_{1,1}^{(2)}+\op(1) = \mathbb B_1(\tau_{s-1},\tau_s) +\op(1),
\end{eqnarray*}
where the last three quantities above were already defined and analyzed in the proof of Lemma \ref{lem2} in Supplementary Appendix, Section \ref{supsec1}, where it was shown that $\mathcal B_{1,1}^{(1)} \inp \mathbb B_1(\tau_{s-1},\tau_s)$ and $\mathcal B_{1,1}^{(2)} =\op(1)$. By similar arguments that were employed to analyze those terms,
\begin{eqnarray*}
\mathcal{V}_{1}^{**}&=&T^{-1}\sum_{t\in \tilde I_s}\sum_{l=0}^{\tilde t-1}\mF_s^l \veta_{t-l} \veta_{t-l}' (\mF_s^l)' = \mathcal B_{1,1}^{(1)} + \op(1)=  \mathbb B_1(\tau_{s-1},\tau_s) +\op(1).
\end{eqnarray*}
Now consider $\mathcal{V}_1$, where
\begin{eqnarray*}
\mathcal{V}_1&=&T^{-1}(\mathcal{V}_1^{(1)}+\mathcal{V}_1^{(2)}+(\mathcal{V}_1^{(2)})^{\prime}+\mathcal{V}_1^{(3)})+\op(1),\\
\mathcal{V}_1^{(1)}&=&T^{-1}\sum_{t\in I_i}\sum_{l=0}^{\tilde t-1}\mF_s^l \veta_{t-l} \veta_{t-l}' (\mF_s^l)' =  \mathcal V_1^{**} = \Op(1)\\
\mathcal{V}_1^{(2)}&=&T^{-1}\sum_{t\in I_i}\sum_{l=0}^{\tilde t-1}\mF_s^l \veta_{t-l} \veta_{t-l}'(\mathcal{R}_{s,l})'\\
\mathcal{V}_1^{(3)}&=&T^{-1}\sum_{t\in I_i}\sum_{l=0}^{\tilde t-1}\mathcal{R}_{s,l} \veta_{t-l} \veta_{t-l}' (\mathcal{R}_{s,l})',
\end{eqnarray*}

Similarly to $\mathcal{V}_{1}^{**}$, because $\hat \mF_s^l-\mF_s^l =\mathcal R_{s,l} =\op(1)$, we can show that $\mathcal{V}_1^{(2)}=\op(1)$. From \eqref{mathcalR}, $\mathcal R_{s,l}$ is such that $\sum_{l=0}^{\infty} \mathcal R_{s,l} =\| \hat \mF_s - \mF_s\| \, \Op(1)=\op(1)$.    Therefore, by the same arguments as for $\mathcal{V}_{1}^{(2)}=\op(1)$, one can show that $\mathcal{V}_1^{(3)}=\op(1)$, therefore  $T \mathcal{V}_1 =B_1(\tau_{s-1},\tau_s) +\op(1)$, and  $\mathcal{V}_1 =\op(1)$.

 By similar arguments to  the analysis of the term $\mathcal B_{1,1}$ in the Supplementary Appendix Section \ref{supsec1}, proof of Lemma \ref{lem2},  $T \mathcal{V}_2=\op(1)$. Substituting $\mathcal{V}_2=\op(1)$ and $\mathcal{V}_1=\op(1)$ into \eqref{a11b}, it follows that $\var^b(\mathcal{A}_{1,1}^b)=\mathcal{V}_1+\mathcal{V}_2=o_p(1)$, and, by Markov's inequality, that $\mathcal{A}_{1,1}^b=o_p^b(1)$. It also follows that:
\begin{align}\label{a11b2}
T\var^b(\mathcal{A}_{1,1}^b) = B_1(\tau_{s-1},\tau_s) +\op(1),
\end{align}
a stronger result that we need later in this proof.

Consider now $\mathcal{A}_{1,2}^b, \mathcal{A}_{1,3}^b, \mathcal{A}_{1,4}^b$. We have $\mathcal{A}_{1,2}^b=T^{-1}\Delta\tau_sT\sum_{l=0}^{\Delta \tau_sT -1}\mF_s^l\,\vmu_s+\op(1)$ and by Assumption \ref{a7}, it follows that $\mathcal{A}_{1,2}^b\stackrel{p}{\to}\Delta \tau_s(\mI_{np}-\mF_s)^{-1}\vmu_s = \int_{\tau_{s-1}}^{\tau_s} \mathbb Q_1(\tau)d\tau$.
Now consider $\mathcal A_{1,3}^b$. Because $\| T^{-1} \vxi_{[\tau_{s-1}T]}^b \| =o_p^b(1)$  and $\| \sum_{l=1}^{\infty} \hat \mF_s^l \| = \Op(1)$,
\begin{align*}
\|\mathcal A_{1,3}^b\|  = \| T^{-1}\sum_{l=1}^{\Delta \tau_s T} \hat\mF_s^l\,\vxi_{[\tau_{s-1}T]}^b\|   \leq  \| \sum_{l=1}^{\infty} \hat \mF_s^l \| \,\, \| T^{-1} \vxi_{[\tau_{s-1}T]}^b \|  = \Op(1) o_p^b(1) = o_p^b(1).
\end{align*}

Since   $\sum_{l=1}^{\Delta \tau_s T -1} l\hat\mF_s^l=O_p(1)$ and $\hat\vmu_s- \vmu_s= o_p(1)$, $\mathcal{A}_{1,4}^b=\op(1)$. Combining these results, we obtain:
$$
\mathcal{A}_1^b\;=\;\Delta \tau_s(\mI_{np}-\mF_s)^{-1}\vmu_s\,+\,o_p^b(1)= \int_{\tau_{s-1}}^{\tau_s} \mathbb Q_1(\tau)d\tau +\op^b(1) = \mathcal{A}_{1} +\op^b(1);
$$
where $\mathcal{A}_{1} =\int_{\tau_{s-1}}^{\tau_s} \mathbb Q_1(\tau)d\tau +\op(1)$ from \eqref{abb} in Supplementary Appendix, Section \ref{supsec1}, and $\mathcal A_i, i=1,2,$ are the sample equivalents of $\mathcal A_i^b$. From \eqref{a2b}, it follows that $\mathcal{A}_{2}^b = \mathcal{A}_{1} +\op(1) = \mathcal{A}_{2} +\op(1) \inp \int_{\tau_{s-1}}^{\tau_s} \mathbb Q_1(\tau)d\tau$.

$\bullet$ Next, analyze $\mathcal{B}_3^b$.  First, note that because $T^{-1}\vxi_t^b \vxi_t^{b'} =o_p^b(1)$ as shown in the preliminaries of this proof,
\begin{align}\label{ind2}
\mathcal B_1^b = T^{-1} \sum_{t\in \tilde I_s}  \vxi_t^b \vxi_t^{b'} = T^{-1}\sum_{t\in \tilde I_s} \vxi_{t-1}^b \vxi_{t-1}^{b'} +o_p^b(1) = \mathcal B_3^b\ +o_p^b(1),
\end{align}
so we analyze instead $\mathcal B_1^b$. Note that
$
\mathcal{B}_1^b\;=\;\sum_{i=1}^3\mathcal{B}_{1,i}^b\,+\,\sum_{j=1}^3 \left\{\,\mathcal{B}_{1,3+j}^b + \mathcal{B}_{1,3+j}^{b\prime}\,\right\},
$
where\begin{eqnarray*}
\mathcal{B}_{1,1}^b&=&T^{-1}\sum_{t\in \tilde I_s} \tilde{\vxi}_t^b\tilde{\vxi}_t^{b\prime}=T^{-1}\sum_{t\in \tilde I_s}\left(\,\sum_{l=0}^{\tilde{t}-1}\hat\mF_s^l\veta_{t-l}^b\,\right)\left(\,\sum_{l=0}^{\tilde{t}-1}\hat\mF_s^l\veta_{t-l}^b\,\right)^\prime\nonumber\\[0.1in]
\mathcal{B}_{1,2}^b&=&T^{-1}\sum_{t\in \tilde I_s} \left( \sum_{l=0}^{\tilde{t}-1}\hat\mF_s^l\hat\vmu_s\right)\,\left(\,\sum_{l=0}^{\tilde{t}-1}\hat\mF_s^l\hat\vmu_s\right)^\prime\nonumber\\[0.1in]
\mathcal{B}_{1,3}^b&=&T^{-1}\sum_{t\in \tilde I_s} \hat\mF_s^{\tilde{t}}\vxi_{[\tau_{s-1}T]}^b\vxi_{[\tau_{s-1}T]}^{b\prime} (\hat\mF_s^{\tilde{t}})^{\prime}\\[0.1in]
\mathcal{B}_{1,4}^b&=&T^{-1}\sum_{t\in \tilde I_s} \tilde{\vxi}_t^b \left(\,\sum_{l=0}^{\tilde{t}-1}\hat\mF_s^l\hat\vmu_s\right)^\prime\\[0.1in]
\mathcal{B}_{1,5}^b&=&T^{-1}\sum_{t\in \tilde I_s} \tilde{\vxi}_t^b \vxi_{[\tau_{s-1}T]}^{b\prime} (\hat\mF_s^{\tilde{t}})^\prime\\[0.1in]
\mathcal{B}_{1,6}^b&=&T^{-1}\sum_{t\in \tilde I_s} \left(\,\sum_{l=0}^{\tilde{t}-1}\hat\mF_s^l\hat\vmu_s\right)\vxi_{[\tau_{s-1}T]}^{b\prime} (\hat{\mF}_s^{\tilde{t}})^{\prime}.
\end{eqnarray*}

Let $$ \mathbb{B}_1(\tau_{s-1},\tau_s) = \sum_{l=0}^{\infty} \mF_s^l \left( \mA_{s}^{-1}\int_{\tau_{s-1}}^{\tau_s}\overline \mSigma(\tau)\rd \tau \mA_s^{-1'}\right)_{\#}\mF_s^{l\prime}.$$ We show  $\mathcal{B}_{1,1}^b-\mathbb{B}_1(\tau_{s-1},\tau_s)=o^b_p(1)$ by showing that $\E^b(\mathcal{B}_{1,1}^b-\mathbb{B}_1(\tau_{s-1},\tau_s))=o_p(1)$ and $\var^b(\vecc\mathcal{B}_{1,1}^b)=o_p(1)$.
\begin{eqnarray*}
\mathcal{B}_{1,1}^b&=& T^{-1}\sum_{t\in \tilde I_s}\sum_{l,\kappa=0}^{\tilde{t}-1}\hat\mF_s^l\veta_{t-l}^b \veta_{t-\kappa}^{b\prime}(\hat\mF_s^\kappa)'= T^{-1}\sum_{t\in \tilde I_s}\sum_{l,\kappa=0}^{\tilde{t}-1}\hat\mF_s^l\hat\mA_{s,\#}^{-1}\vg_{t-l}^b \vg_{t-\kappa}^{b\prime}(\hat\mA_{s,\#}^{-1})'(\hat\mF_s^\kappa)'\\
\E^b(\mathcal{B}_{1,1}^b)&=& T^{-1}\sum_{t\in \tilde I_s}\sum_{l,\kappa=0}^{\tilde{t}-1}\hat\mF_s^l\E^b(\veta_{t-l}^b\veta_{t-\kappa}^{b\prime})(\hat\mF_s^{\kappa})^{\prime}
=  T\var^b(\mathcal{A}_{1,1}^b)=\mathbb{B}_{1}(\tau_{s-1},\tau_s)+\op(1), \end{eqnarray*}
where the last equality above follows from \eqref{a11b2}.
We have:
\begin{eqnarray*}
\vecc\mathcal{B}_{1,1}^b&=&T^{-1}\sum_{t\in \tilde I_s}\sum_{l,\kappa=0}^{\tilde{t}-1}\left((\hat\mF_s^\kappa\hat\mA_{s,\#}^{-1})\otimes(\hat\mF_s^l\hat\mA_{s,\#}^{-1})\right)\left(\vg^b_{t-\kappa}\otimes\vg^b_{t-l}\right)\\
&=& T^{-1}\sum_{t\in \tilde I_s}\sum_{l,\kappa=0}^{\tilde{t}-1}\left((\hat\mF_s^\kappa\hat\mA_{s,\#}^{-1})\otimes(\hat\mF_s^l\hat\mA_{s,\#}^{-1})\right)\left((\hat\vg_{t-\kappa}\odot\vnu_{t-\kappa})\otimes(\hat\vg_{t-l}\odot\vnu_{t-l})\right)\\
&=& T^{-1}\sum_{t\in \tilde I_s}\sum_{l,\kappa=0}^{\tilde{t}-1}\left((\hat\mF_s^\kappa\hat\mA_{s,\#}^{-1})\otimes(\hat\mF_s^l\hat\mA_{s,\#}^{-1})\right)\left((\hat\vg_{t-\kappa}\otimes\hat\vg_{t-l})\odot(\vnu_{t-\kappa}\otimes\vnu_{t-l})\right).\\
\var^b(\vecc\mathcal{B}_{1,1}^b)&=&\E^b(\vecc\mathcal{B}_{1,1}^b(\vecc\mathcal{B}_{1,1}^{b})^\prime)-\E^b(\vecc\mathcal{B}_{1,1}^b)\E^b(\vecc\mathcal{B}_{1,1}^b)^\prime\\
&=&\E^b(\vecc\mathcal{B}_{1,1}^b(\vecc\mathcal{B}_{1,1}^{b})^\prime)-\vecc\mathbb{B}_1(\tau_{s-1},\tau_s)\left(\vecc\mathbb{B}_1(\tau_{s-1},\tau_s)\right)'+o_p(1). \end{eqnarray*}
We need to show that $\E^b(\vecc\mathcal{B}_{1,1}^b(\vecc\mathcal{B}_{1,1}^{b})^\prime)\stackrel{p}{\to}\vecc\mathbb{B}_1(\tau_{s-1},\tau_s)\left(\vecc\mathbb{B}_1(\tau_{s-1},\tau_s)\right)'$. Letting $\tilde t^*=t^*-[\tau_{s-1}T]$,
\begin{eqnarray*}
\E^b(\vecc\mathcal{B}_{1,1}^b(\vecc\mathcal{B}_{1,1}^{b})^\prime)&=&\left[T^{-1}\sum_{t\in \tilde I_s}\sum_{l,\kappa=0}^{\tilde{t}-1}\left((\hat\mF_s^\kappa\hat\mA_{s,\#}^{-1})\otimes(\hat\mF_s^l\hat\mA_{s,\#}^{-1})\right)\left(\vg^b_{t-\kappa}\otimes\vg^b_{t-l}\right)\right]\\
&\times&\left[T^{-1}\sum_{t^*\in \tilde I_s}\sum_{l^*,\kappa^*=0}^{\tilde{t}-1}\left((\hat\mF_s^{\kappa^*}\hat\mA_{s,\#}^{-1})\otimes(\hat\mF_s^{l^*}\hat\mA_{s,\#}^{-1})\right)\left(\vg^b_{t^*-\kappa^*}\otimes\vg^b_{t^*-l^*}\right)\right]\\
&=& T^{-2}\sum_{t,t^*\in \tilde I_s}\sum_{l,\kappa=0}^{\tilde{t}-1}\sum_{l^*,\kappa^*=0}^{\tilde{t}^*-1}((\hat\mF_s^\kappa\hat\mA_{s,\#}^{-1})\otimes(\hat\mF_s^l\hat\mA_{s,\#}^{-1}))\mathcal{G}((\hat\mF_s^{\kappa^*}\hat\mA_{s,\#}^{-1})'\otimes(\hat\mF_s^{l^*}\hat\mA_{s,\#}^{-1})')\\
&=&\sum_{i=1}^{9}\mathcal{O}_i,\\
\mathcal{G}&=&\E^b(\left(\vg^b_{t-\kappa}\otimes\vg^b_{t-l}\right)\left(\vg^b_{t^*-\kappa^*}\otimes\vg^b_{t^*-l^*}\right)^\prime),
\end{eqnarray*}
where $\mathcal{O}_i$ are the terms corresponding to nine cases when $\mathcal{G}\neq \mO_{(n^2p^2)\times(n^2p^2)}$. Case (1) is when  $t-\kappa=t-l, t^*-\kappa^*=t^*-l^*, t-\kappa\neq t^*-\kappa^*$; we show below that $\mathcal{O}_1=\vecc\mathbb{B}_1(\tau_{s-1},\tau_s)(\vecc\mathbb{B}_1(\tau_{s-1},\tau_s))'+\op(1)$. For brevity, the rest of the cases are defined and analyzed in Supplementary Appendix, Section \ref{supsec4}, where we show that
\begin{equation}\label{mathcaloi}
\mathcal O_i = \op(1) \mbox{ for } i=2,\ldots, 9.
\end{equation}

By Assumption \ref{aboot}, $\E^b[(\vnu_t\vnu_t')\otimes (\vnu_{t-l}\vnu_{t-l}')]=[\E^b(\vnu_t\vnu_t')] \otimes [\E^b(\vnu_{t-l}\vnu_{t-l})']=\E^b(\vnu_t\otimes\vnu_t)\E^b(\vnu_{t-l}\otimes\vnu_{t-l})$ since $\E^b(\nu_t^2\nu_{t-l}^2)=\E^b(\nu_t^2)\E^b(\nu_{t-l}^2)=1$, $\E^b(\nu_{t}\nu_{t-l})=0$ and $\E^b(\nu_{t}^2\nu_{t-l})=0$ (these are elements of $\E^b(\vnu_t\vnu_t'\otimes \vnu_{t-l}\vnu_{t-l}')$). Hence, conditional on the data, we have, by Assumption \ref{aboot},
\begin{eqnarray*}
\mathcal{G}&=&\E^b(\left(\vg^b_{t-\kappa}\otimes\vg^b_{t-\kappa}\right)\left(\vg^b_{t^*-\kappa^*}\otimes\vg^b_{t^*-\kappa^*}\right)^\prime)\\
&=&\E^b[((\hat\vg_{t-\kappa}\odot\vnu_{t-\kappa})\otimes (\hat\vg_{t-\kappa}\odot\vnu_{t-\kappa}))[(\hat\vg_{t^*-\kappa^*}\odot\vnu_{t^*-\kappa^*})(\hat\vg_{t^*-\kappa^*}\odot\vnu_{t^*-\kappa^*})]']\\
&=&\E^b[[(\hat\vg_{t-\kappa}\otimes \hat\vg_{t-\kappa})\odot(\vnu_{t-\kappa}\otimes \vnu_{t-\kappa})][(\hat\vg_{t^*-\kappa^*}\otimes \hat\vg_{t^*-\kappa^*})\odot(\vnu_{t^*-\kappa^*}\otimes \vnu_{t^*-\kappa^*})]]\\
&=& [(\hat\vg_{t-\kappa}\otimes \hat\vg_{t-\kappa})\odot\E^b(\vnu_{t-\kappa}\otimes \vnu_{t-\kappa})][(\hat\vg_{t^*-\kappa^*}\otimes \hat\vg_{t^*-\kappa^*})\odot\E^b(\vnu_{t^*-\kappa^*}\otimes \vnu_{t^*-\kappa^*})]\\
&=& (\hat\vg_{t-\kappa}\otimes \hat\vg_{t-\kappa})(\hat\vg_{t^*-\kappa^*}\otimes \hat\vg_{t^*-\kappa^*}),
\end{eqnarray*}
hence
\begin{align*}
\mathcal{O}_{1}&= \left[T^{-1}\sum_{t\in \tilde I_s}\sum_{\kappa=0}^{\tilde{t}-1}\left((\hat\mF_s^\kappa\hat\mA_{s,\#}^{-1})\otimes(\hat\mF_s^\kappa\hat\mA_{s,\#}^{-1})\right)\left(\vg^b_{t-\kappa}\otimes\vg^b_{t-\kappa}\right)\right]\\
&\times\left[T^{-1}\sum_{t^*\in \tilde I_s}\sum_{\kappa^*=0}^{\tilde{t}-1}\left((\hat\mF_s^{\kappa^*}\hat\mA_{s,\#}^{-1})\otimes(\hat\mF_s^{\kappa^*}\hat\mA_{s,\#}^{-1})\right)\left(\vg^b_{t^*-\kappa^*}\otimes\vg^b_{t^*-\kappa^*}\right)\right]'\\
&=\sum_{\kappa,\kappa^*=0}^{\Delta\tau_sT-1}\left((\hat\mF_s^\kappa\hat\mA_{s,\#}^{-1})\otimes(\hat\mF_s^\kappa\hat\mA_{s,\#}^{-1})\right)\left(T^{-1}\sum_{t\in\tilde I_s}(\hat\vg_{t-\kappa}\otimes \hat\vg_{t-\kappa})\right)\\
&\times\left(T^{-1}\sum_{t\in\tilde I_s}(\hat\vg_{t^*-\kappa^*}\otimes \hat\vg_{t^*-\kappa^*})\right)\left((\hat\mF_s^{\kappa^*}\hat\mA_{s,\#}^{-1})\otimes(\hat\mF_s^{\kappa^*}\hat\mA_{s,\#}^{-1})\right)'+\op(1)\\
&=(\vecc\mathbb{B}_1(\tau_{s-1},\tau_s)+\op(1))((\vecc\mathbb{B}_1(\tau_{s-1},\tau_s))'+\op(1))=\vecc\mathbb{B}_1(\tau_{s-1},\tau_s)(\vecc\mathbb{B}_1(\tau_{s-1},\tau_s))'+\op(1),
\end{align*}
where the last two lines follows because $\vg_{t-\kappa}\otimes\vg_{t-\kappa}=\vecc(\vg_{t-\kappa}\vg_{t-\kappa}')$, and
\begin{align*}
T^{-1}\sum_{t\in\tilde I_s}(\hat\vg_{t-\kappa}\otimes \hat\vg_{t-\kappa})=T^{-1}\sum_{t\in\tilde I_s}\vecc(\hat\vg_{t-\kappa}\hat\vg_{t-\kappa}')=\pim\vecc T^{-1}\sum_{t\in\tilde I_s}\vg_{t-\kappa}\vg_{t-\kappa}'+\op(1),
\end{align*}
which follows by standard 2SLS theory and Lemma \ref{lem8}. So, $\mathcal{O}_1=\vecc\mathbb{B}_1(\tau_{s-1},\tau_s)(\vecc\mathbb{B}_1(\tau_{s-1},\tau_s))'+\op(1)$.

Therefore, $\var^b(\vecc\mathcal{B}_{1,1}^b)=\op(1 $, so by Markov's inequality,
\begin{eqnarray*}
\mathcal{B}_{1,1}^b=\mathbb{B}_{1}(\tau_{s-1},\tau_s)+\op^b(1).
\end{eqnarray*}
Next, because $\hat \vmu_s = \vmu_s+\op(1)$, and $\hat \mF_s = \mF_s +\op(1)$, and $ \sum_{l=0}^{\infty} \| \hat \mF_s^l - \mF_s^l\| = \op(1)$ as shown in Supplementary Appendix, Section \ref{supsec3}, $\mathcal{B}_{1,2}^b=\mathcal{B}_{1,2}+\op(1)=\mathbb{B}_2(\tau_{s-1},\tau_s)+\op(1)$, where $\mathbb{B}_2(\tau_{s-1},\tau_s) = \int_{\tau_{s-1}}^{\tau_s} \mathbb Q_1(\tau) \mathbb Q_1'(\tau) d\tau$, and $\mathcal B_{1,2}$ is the sample equivalent of $\mathcal B_{1,2}^b$ (and in general, $\mathcal B_{1,i}, i=1,\ldots,6$ are the sample equivalents of $\mathcal B_{1,i}^b$,  defined in the proof of Lemma \ref{lem2} in Supplementary Appendix, Section \ref{supsec1}). Also, we have $\mathcal{B}_{1,3}^b=\mathcal{B}_{1,3}+\op(1)=\op(1),$ because, as shown in the preliminaries $T^{-\alpha} \vxi_t^b \vxi_t^{b'} =o_p^b(1)$, and $\hat \mF_s^l$ is exponentially decaying with $l$.

Consider $\mathcal{B}_{1,4}^b$.
\begin{eqnarray*}
\mathcal{B}_{1,4}^b&=&T^{-1}\sum_{t\in \tilde I_s} \left(\,\sum_{l=0}^{\tilde{t}-1}\hat\mF_s^l\veta_{t-l}^b\,\right) \left(\,\sum_{l=0}^{\tilde{t}-1}\hat\mF_s^l\hat\vmu_s\right)^\prime=T^{-1}\sum_{t\in \tilde I_s} \left(\,\sum_{l=0}^{\tilde{t}-1}\hat\mF_s^l((\hat\mA_{s,\#}^{-1}\hat\vg_{t-l})\odot\vnu_{t-l})\,\right) \left(\,\sum_{l=0}^{\tilde{t}-1}\hat\mF_s^l\hat\vmu_s\right)^\prime.
\end{eqnarray*}
We show $\mathcal{B}_{1,4}^b=\op^b(1)$. To that end note that
\begin{eqnarray*}
\E^b(\mathcal{B}_{1,4}^b)&=& T^{-1}\sum_{t\in \tilde I_s} \left(\,\sum_{l=0}^{\tilde{t}-1}\mF_s^l\veta_{t-l,2}\,\right) \left(\,\sum_{l=0}^{\tilde{t}-1}\mF_s^l\vmu_s\right)^\prime+\op(1) =\op(1),
\end{eqnarray*}
by similar arguments as for its sample equivalent $\mathcal B_{1,4}$ defined in the proof of Lemma \ref{lem2}.
So, $\E^b(\mathcal{B}_{1,4}^b)=\op(1)$. Moreover, by similar arguments as before, it can be shown that
$\|\vecc\var^b(\mathcal{B}_{1,4}^b)\| =\op(1)$. Hence, by Markov's inequality,  $\mathcal{B}_{1,4}^b=\op^b(1)$. Similarly, because $T^{-\alpha} \vxi_{[\tau_{s-1}T]} = \Op(1)$ for any $\alpha>0$, we can show that $\mathcal{B}_{1,5}^b=\op^b(1)$, and  $\mathcal{B}_{1,6}^b=\op^b(1)$. Putting all the results for $\mathcal{B}_{1,i}^b$ together, $i=1,\ldots,6$ we conclude $\mathcal{B}_{1}^b=\mathbb{B}_1(\tau_{s-1},\tau_s)+\mathbb{B}_2(\tau_{s-1},\tau_s)+\op^b(1)=\mathbb{B}(\tau_{s-1},\tau_s)+\op^b(1) = \mathcal B_1+\op(1)$, where $\mathbb{B}(\tau_{s-1},\tau_s) = \mathbb{B}_1(\tau_{s-1},\tau_s)+\mathbb{B}_2(\tau_{s-1},\tau_s)$, and $\mathcal B_1$ is the sample equivalents of $\mathcal B_1^b$ defined in the proof of Lemma \ref{lem2}.

Because $\mathcal B_3^b = \mathcal B_1^b +\op(1)$, it follows that $\mathcal B_3^b =\mathbb{B}(\tau_{s-1},\tau_s)+\op^b(1)$, where the same result was shown to hold for $\mathcal B_3$ defined in the proof of Lemma \ref{lem2}.
Now consider $\mathcal B_2^b$. Using $\vxi_t^b = \hat \vmu_s + \hat \mF_s \vxi_{t-1}^b + \veta_t^b$, it follows that:
\begin{align*}
\mathcal B_2^b=  \hat \vmu_s \mathcal A_2^{b'} + \hat \mF_s \mathcal B_3^b + T^{-1} \sum_{t\in \tilde I_s} \vxi_{t-1}^b \veta_t^{b'} + o_p^b(1).
\end{align*}
By similar arguments as for some elements of $\mathcal B_1$, it can be shown that  $ T^{-1} \sum_{t\in \tilde I_s} \vxi_{t-1}^b \veta_t^{b'}= o_p^b(1)$. Therefore,
\begin{align*}
\mathcal B_2^b =   \int_{\tau_{s-1}}^{\tau_s} \vmu(\tau) \mathbb Q_1(\tau)d\tau + \int_{\tau_{s-1}}^{\tau_s} \mF(\tau) \mathbb Q_2(\tau)d\tau+ o_p^b(1).
\end{align*}
Therefore, for $I_i=\tilde I_s$,
$$
\hat \mQ^b_{(i)}\;=\;\int_{\tau_{s-1}}^{\tau_s} \mUpsilon^{\prime}(\tau)\,\mathbb Q_{\vz}(\tau) \mUpsilon(\tau)d\tau+\op^b(1) = \mathbb Q_{(i)} +\op(1).
$$
For other regimes, by similar arguments as in the end of the proof of Lemma \ref{lem2},
$$
\hat \mQ^b_{(i)}\;=\;\int_{\lambda_{i-1}}^{\lambda_i} \mUpsilon^{\prime}(\tau)\,\mathbb Q_{\vz}(\tau) \mUpsilon(\tau)d\tau+\op^b(1) = \mathbb Q_{(i)} +\op(1),
$$
concluding the proof.

\end{proof}

\begin{lem}
        \label{lemzuv_WR} $T^{-1/2}\sum_{t\in I_i}\vz_t^b\vg_t^{b'}\mathcal{ S}_{\dag}^b\weaks \widetilde{ \mathbb M}_i$ in probability uniformly in $\vlambda_k$, where $\widetilde{ \mathbb M}_i$ is defined as in Lemma \ref{lem6}, and $\mathcal{S}^b_{\dag}=\mathcal{S}_u$ or $\mathcal{S}_{\dag}^b=(\hat\vbeta_{\vx,(i)})_{\#}$. If $m=0$, then $\mathcal{S}_{\dag}^b=\mathcal{S}_u$ or $\mathcal{S}_{\dag}^b=\hat\vbeta_{\vx,\#}$.
\end{lem}

\begin{lem}
        \label{lemzuv_WF} $T^{-1/2}\sum_{t\in I_i}\vz_t\vg_t^{b'}\mathcal{S}^b_{\dag}\weaks \widetilde{ \mathbb M}_i$ in probability uniformly in $\vlambda_k$.
\end{lem}

For the proofs of Lemma \ref{lemzuv_WR}-\ref{lemzuv_WF}, it suffices to consider $\mathcal{S}^b_{\dag}=\mathcal{S}_u$ or $\mathcal{S}_{\dag}^b=\hat\vbeta_{\vx,\#}$, therefore considering $m=0$. If  $\mathcal{S}_{\dag}^b=(\hat\vbeta_{\vx,(i)})_{\#}$, by Lemma \ref{lem7} followed by standard 2SLS theory, $\hat\vbeta_{\vx,(i)}=\vbeta_{\vx,(i)}^0+\Op(T^{-1/2})$ so $\mathcal{S}_{\dag}^b= \mathcal S_{\dag} +\Op(T^{-1/2})$, and the results follow in a similar fashion.

Additionally to the notation already defined at the beginning of the proof of Lemma \ref{lem9}, we use the following results and notation, some relevant for both Lemma \ref{lemzuv_WR} and \ref{lemzuv_WF}.
 Consider the partition $\tilde I_s$, then for the WR bootstrap we have $\tilde\vz_t^b \;=\;\hat\vc_{\tilde{\vz},s}\,+\, \sum_{i=1}^{p} \hat\mC_{i,s} \tilde\vz_{t-i}^b+ \ve_t^b$, and for both WR and WF bootstraps, we have $\ve^b_t=\hat\mA_s^{-1}\vepsi^b_t$. We have for the WR bootstrap: \begin{align}
\vxi_t^b\;&=\;\hat\vmu_s\,+\,\hat\mF_s\vxi^b_{t-1}\,+\,\veta^b_t= \hat\mF^{t-[\tau_{s-1}T]}_s \vxi_{[\tau_{s-1}T]}^b + \sum_{l=0}^{t-[\tau_{s-1}T]-1}\hat\mF_s^l\veta^b_{t-l}+\left(\sum_{l=0}^{t-[\tau_{s-1}T]-1}\hat\mF_s^l\right)\hat\vmu_s\label{58}
\end{align}
except that in \eqref{58} when $s=1$ and we are in the first regime $\tilde I_1=[1,\ldots,[\tau_1 T]]$, we have that $\vxi_0^b=\vxi_0$, where $\vxi^b_t=\vec_{j=0:(p-1)}( \tilde{\vz}_{t-j}^b)$.
Let $\mathcal{F}_{t}^b=\{\nu_{t},\nu_{t-1},\ldots,\nu_1\}$. Recall that, by Assumption \ref{aboot},
\begin{eqnarray}
\E^b(\vnu_t)&=&\E^b(\vnu_t|\mathcal{F}^b_{t-1})=\vec(\vzeros_{p_1+1},\viota_{p_2},\vzeros_{n(p-1)\times 1})=\mathcal{I}\label{enu}\\
\E^b(\vnu_t\vnu_t^\prime)&=&\E^b(\vnu_t\vnu_t^\prime|\mathcal{F}^b_{t-1})=
\left(\diag(\mJ_{p_1+1},\mJ_{p_2})\right)_{\#}=\mathcal{J}.\label{varnu}\\
\E^b(\vnu_t\vnu_{t-j}^\prime)&=&\E^b(\vnu_t\vnu_{t-j}^\prime|\mathcal{F}^b_{t-1})=
(\diag(\vzeros_{p_1+1},\mJ_{p_2}))_{\#}=\mathcal{J}_2.\label{varnu2}
\end{eqnarray}
Furthermore, recall that $\vxi_t=\vmu_s+\mF_s\vxi_{t-1}+\veta_t=\hat\vmu_s+\hat\mF_s\vxi_{t-1}+\hat\veta_t$, and therefore $\hat\veta_t=\veta_t+(\vmu_s-\hat\vmu_s)+(\mF_s-\hat\mF_s)\vxi_{t-1}.
$ By backward substitution of $\vxi_t=\vmu_s+\mF_s\vxi_{t-1}+\veta_t$, we have that: $\vxi_{t-1}=\mF^{\tilde t-1}_s \vxi_{[\tau_{s-1}T]} + \sum_{l=0}^{\tilde t-2}\mF_s^l\veta_{t-l-1}+\left(\sum_{l=0}^{\tilde t-2}\mF_s^l\right)\vmu_s$, where $\tilde t =t- [\tau_{s-1}T]$. We also have $\veta_t^b=\hat\mA_{s,\#}^{-1}\;\vg_t^b=\hat\mA_{s,\#}^{-1}\;(\hat\vg_t\odot\vnu_t)$, where recall that $\odot$ is the element-wise multiplication. Hence:
\begin{eqnarray}
\hat\vg_t&=&\hat\mA_{s,\#}\;\hat\veta_t= \hat\mA_{s,\#}\veta_t+\hat\mA_{s,\#}(\vmu_s-\hat\vmu_s)+\hat\mA_{s,\#}(\mF_s-\hat\mF_s)\vxi_{t-1}\nonumber\\
&=& \hat\mA_{s,\#}\veta_t+\hat\mA_{s,\#}(\vmu_s-\hat\vmu_s)+\hat\mA_{s,\#}(\mF_s-\hat\mF_s)\mF^{\tilde t-1}_s \vxi_{[\tau_{s-1}T]}\nonumber   \\
&+&\hat\mA_{s,\#}(\mF_s-\hat\mF_s)(\sum_{l=0}^{\tilde t-2}\mF_s^l\veta_{t-l-1})+\hat\mA_{s,\#}(\mF_s-\hat\mF_s)\left(\left(\sum_{l=0}^{\tilde t-2}\mF_s^l\right)\vmu_s\right)\label{hatgt}\\\label{gtabc}
\vg_t^b&=& ((\hat\mA_{s,\#}\veta_t)\odot\vnu_t) + ((\hat\mA_{s,\#}(\vmu_s-\hat\vmu_s))\odot\vnu_t)+((\hat\mA_{s,\#}(\mF_s-\hat\mF_s)\vxi_{t-1})\odot\vnu_t)\nonumber\\
&=&\vg_{t,A}^b+\vg_{t,B}^b+\vg_{t,C}^b,\\
\veta_t^b&=& \hat\mA_{s,\#}^{-1}\;((\hat\mA_{s,\#}\veta_t)\odot\vnu_t) + \hat\mA_{s,\#}^{-1}\;((\hat\mA_{s,\#}(\vmu_s-\hat\vmu_s))\odot\vnu_t)+\hat\mA_{s,\#}^{-1}\;((\hat\mA_{s,\#}(\mF_s-\hat\mF_s)\vxi_{t-1})\odot\vnu_t)\nonumber\\
&=&\veta_{t,A}^b+\veta_{t,B}^b+\veta_{t,C}^b,
\end{eqnarray}
where
$\| \hat \mA_{s,\#} \| \leq \| \mA_{s,\#} \| + \| \hat \mA_{s,\#}- \mA_{s,\#} \| = \| \mA_{s,\#} \|+\op(1)$ and hence $\hat\mA_{s,\#}\mA_{s,\#}^{-1}=\mI_{n,\#}+\op(1)$ .
Moreover, $\| \hat \mA_{s,\#}^{-1} \| \leq \| \mA_{s,\#}^{-1} \| + \| \hat \mA_{s,\#}^{-1}- \mA_{s,\#}^{-1} \| = \| \mA_{s,\#}^{-1} \|+\op(1)$.

Finally, for a vector $\vo$, we denote $\vo^{(j_1:j_2)}$ its sub-vector with elements $j_1$ to $j_2$ selected in order, and for a matrix $\mO$, we denote by $\mO^{(j_1:j_2,j_1^*:j_2^*)}$ its sub-matrix consisting of rows $j_1$ to $j_2$ and columns $j_1^*$ to $j_2^*$.

\begin{proof}[Proof of Lemma \ref{lemzuv_WR}]

 As for the proof of Lemma \ref{lem9}, consider the interval $I_i=\tilde I_s$. Let $\mathcal{S}_{\dag}^b=\mathcal{S}_u$ or $\mathcal{S}_{\dag}^b=\hat\vbeta_{\vx,\#}$. We derive the asymptotic distribution of $T^{-1/2} \sum_{t\in \tilde I_s} \vz_t^b\vg_t^{b\prime}\mathcal{S}_{\dag}^b$,
        \begin{align}\label{defineste}
        T^{-1/2} \sum_{t\in \tilde I_s}\vz_t^b\vg_t^{b\prime}\;\mathcal{S}_{\dag}^b
        =  \begin{bmatrix}T^{-1/2} \sum_{t \in \tilde I_s}\vg_t^{b\prime}\;\mathcal{S}_{\dag}^b\\ T^{-1/2} \sum_{t \in \tilde I_s}\mathcal{S}_{\vr}\vxi_{t} \vg_t^{b\prime}\;\mathcal{S}_{\dag}^b\\T^{-1/2} \sum_{t \in \tilde I_s}\mathcal{S}\vxi^b_{t-1} \vg_t^{b\prime}\;\mathcal{S}_{\dag}^b
        \end{bmatrix} \equiv \begin{bmatrix} \mathcal E_1^b \\ \mathcal E_2^b \\ \mathcal E_3^b \end{bmatrix}.
        \end{align}

        \bigskip

        $\bullet$ Consider first $\mathcal E_1^b$. By \eqref{gtabc},
        \begin{eqnarray*}
                \mathcal{E}_1^b&=&T^{-1/2} \sum_{t \in \tilde I_s}\vg^{b\prime}_t\mathcal{S}_{\dag}^b= T^{-1/2} \sum_{t \in I_i}\vg_{t,A}^{b\prime}\;\mathcal{S}_{\dag}^b+T^{-1/2} \sum_{t \in \tilde I_s}\vg_{t,B}^{b\prime}\;\mathcal{S}_{\dag}^b+T^{-1/2} \sum_{t \in \tilde I_s}\vg_{t,C}^{b\prime}\;\mathcal{S}_{\dag}^b\\
&=& T^{-1/2} \sum_{t \in \tilde I_s}\mathcal{S}_{\dag}^{b'}((\hat\mA_{s,\#}\veta_t)\odot\vnu_t)  +T^{-1/2} \sum_{t \in \tilde I_s} \mathcal{S}_{\dag}^{b'}((\hat\mA_{s,\#}(\vmu_s-\hat\vmu_s))\odot\vnu_t)\\
 &+&T^{-1/2} \sum_{t \in \tilde I_s}\mathcal{S}_{\dag}^{b'}((\hat\mA_{s,\#}(\mF_s-\hat\mF_s)\vxi_{t-1})\odot\vnu_t)\\
                &=&\mathcal{E}_{1,1}^b+\mathcal{E}_{1,2}^b+\mathcal{E}_{1,3}^b.
        \end{eqnarray*}
Because $\|\mathcal{S}_{\dag}^b-\mathcal S_{\dag}\| = \op(1)$, $\|\hat \vmu_s - \vmu_s\| =\op(1)$, $\|\hat \mA_s - \mA_s\| =\op(1)$, $\|\hat \mA_s^{-1} - \mA_s^{-1}\| =\op(1)$ and $\sum_{l=0}^{\infty} \| \hat \mF_s^l - \hat \mF_s^l\| =\op(1)$, whenever a $\Op^b(1)$ term is written with the estimated quantities instead of the true one, the difference is $o_p^b(1)$, so asymptotically negligible. Therefore, we proceed in the rest of the proof by replacing the estimated parameters mentioned above with their true values, and denote the remainder by $o_p^b(1)$.

Using these replacements, one can show that $\mathcal E_{1,1}^b = \op^b(1)$ and $\mathcal E_{1,2}^b = \op^b(1)$.
So, $\mathcal{E}_{1}^b = \mathcal{S}_{\dag}' \, T^{-1/2}\sum_{t\in \tilde I_s}((\mA_{s,\#}\veta_t)\odot\vnu_t)+\op^b(1)$. Since $\mA_{s,\#}\veta_t = \mA_{s,\#}\mA_{s,\#}^{-1} \vg_t = \vg_t=\vepsi_{t,\#}$, it follows that
\begin{equation}\label{defe1b}
\mathcal{E}_{1}^b = \mathcal{S}_{\dag}' \, T^{-1/2}\sum_{t\in \tilde I_s}(\vg_t\odot\vnu_t)+\op^b(1).
\end{equation}

$\bullet$ First, let $\mathcal{S}_{\dag} = \mathcal S_u$. Then $\mathcal{E}_{1}^b = T^{-1/2}\sum_{t\in \tilde I_s} u_t \nu_t+ \op^b(1) = T^{-1/2}\sum_{t\in \tilde I_s} d_{u,t} l_{u,t} \nu_t +\op^b(1)= \mathcal E_{1,1}^b+ \op^b(1)$, where recall that $d_{u,t}=d_{1,t}$ and $l_{u,t}$ is the first element of $\vl_t$.

 We now derive the limiting distribution of $\mathcal E_{1,1}^b$, in two parts: in part (i), we show that Lemma \ref{lem3} holds  for $\tilde{\mathcal E}_{1,1}^b = T^{-1/2}\sum_{t=1}^{[Tr]} l_{u,t} \nu_t$, i.e. $\tilde{\mathcal E}_{1,1}^b \stackrel{d_{p}^b}{\Rightarrow} \mB_{0}^{(1)}(r)$ in probability, where $\mB_{0}^{(1)}(r)$ is the first element of $\mB_{0}(r)$ defined just before Lemma \ref{lem6}; in part (ii), we show that the condition of Theorem 2.1 of \citet{Hansen:1992} holds, that is, the bootstrap unconditional variance of $\mathcal E_{1,1}^b$ converges in probability to the unconditional variance of $T^{-1/2}\sum_{t\in \tilde I_s} d_{u,t} l_{u,t} =\mathcal E_1$ (or equivalently, to the variance of $\mB_{0}^{(1)}(r)$). Note that here $\mathcal E_1 = \mathcal{S}_{\dag}' \left( T^{-1/2}\sum_{t\in \tilde I_s} \vg_t\right)$ is the sample equivalent of $\mathcal E_1^b$, defined in the proof of Lemma \ref{lem6} in the Supplementary Appendix, Section \ref{supsec1}.

        \textbf{Part (i). }First, $\nu_t$ is i.i.d, so conditional on the data, $\E^b(l_{u,t} \nu_t |\mathcal F_{t-1}^b) = l_{u,t} \E^b (\nu_t|\mathcal F_{t-1}^b)=0$, so $l_{u,t} \nu_t$ is a m.d.s.  Second, for some $C>0$, $\sup_t \E(\E^b |l_{u,t} \nu_t|^{2+\delta^*})= \sup_t \E|l_{u,t}|^{2+\delta^*} \sup_t \E^b|\nu_t|^{2+\delta^*}<C$ by Assumption \ref{a8}(iii) and Assumption \ref{aboot}(ii), so $\sup_t \E^b |l_{u,t} \nu_t|^{2+\delta^*}<\op(1)+C$. Third, by Lemma \ref{lem8}(iv), \begin{align*}
\E^b(\mathcal E_{1,1}^b)^2 &= T^{-1} \sum_{t=1}^{[Tr]} \E^b(l_{u,t}^2 \nu_t^2)=T^{-1} \sum_{t=1}^{[Tr]} l_{u,t}^2 -r  \inp 0.
\end{align*}
Forth, because $l_{u,t} \nu_t$ is i.i.d conditional on the data, the conditional and unconditional bootstrap second moments are the same, so $\E^b[(\mathcal E_{1,1}^b)^2 |\mathcal F_{t-1}^b]-\E^b(\mathcal E_{1,1}^b)^2 =0$. This shows that $\tilde{\mathcal E}_{1,1}^b = T^{-1/2}\sum_{t=1}^{[Tr]} l_{u,t} \nu_t \stackrel{d_{p}^b}{\Rightarrow} \mB_{0}^{(1)}(r)$ in probability (uniformly in $r$).

\textbf{Part (ii). } By Assumption \ref{a8}(ii),  $\E(d_{u,t}^2 l_{u,t}^2)=d_{u,t}^2$. Therefore, by Lemma \ref{lem8}(iv), uniformly in $r$,
\begin{align*}
\E^b(\mathcal E_{1,1}^b)^2 - \E(\mathcal E_1^2)&= T^{-1} \sum_{t=1}^{[Tr]} [d_{u,t}^2 l_{u,t}^2 -\E (d_{u,t}^2 l_{u,t}^2)] \inp 0.
\end{align*}
Therefore, by Theorem 2.1 in \citet{Hansen:1992}, $T^{-1/2}\sum_{t=1}^{[Tr]} d_{u,t} l_{u,t} \nu_t \stackrel{d_{p}^b}{\Rightarrow} \int_{0}^r d_{u}(\tau) d\mB_{0}^{(1)}(\tau) = \mathbb M_1(\tau_{s-1},\tau_s)$ in probability, where $\mathbb M_1(\tau_{s-1},\tau_s)$ is defined  just before Lemma \ref{lem6}. So for $\mathcal S_{\dag} = S_u$,
$
\mathcal E_1^b \stackrel{d_{p}^b}{\Rightarrow}  \mathbb M_1(\tau_{s-1},\tau_s).
$

Now let $\mathcal S_{\dag} = \vbeta_{x,\#}$ and note that $\mathcal{E}_{1}^b = \vbeta_{\vx}^{0'} T^{-1/2}\sum_{t\in \tilde I_s} \vv_t \nu_t$. Recall that by the decomposition of $\mS$ and a decomposition of $\mD_t$ exactly as $\mD(\tau)$ in \eqref{decompsd}, we have:
\begin{align}\label{decompvepsi}
\vg_t \odot \vnu_t = \vepsi_{t,\#} \odot \vnu_t = ( \mS \mD_t \vl_t)_{\#} \odot \vnu_t = \left[\vec( d_{u,t} l_{u,t} \nu_t \, , \, \vs_{p_1} d_{u,t} l_{u,t} \nu_t + \mS_{p_1} \mD_{\vv,t} \vl_{\vv, t}\nu_t \, , \, \mS_{p_2} \vl_{\vzeta,t})\right]_{\#}.
\end{align}

Therefore, $\mathcal{E}_{1}^b = \vbeta_{\vx}^{0'} T^{-1/2}\sum_{t\in \tilde I_s} \vv_t \nu_t = (\vbeta_{\vx}^{0'}\vs_{p_1}) \left( T^{-1/2}\sum_{t\in \tilde I_s}  d_{u,t} l_{u,t} \nu_t \right) + (\vbeta_{\vx}^{0'}\mS_{p_1})\left( T^{-1/2}\sum_{t\in \tilde I_s} \mD_{\vv,t} \vl_{\vv,t} \nu_t \right)$.

Because $E(\vl_{\vv,t} \vl_{\vv,t}')= \mI_{p_1}$, by similar arguments as for $T^{-1/2}\sum_{t=1}^{[Tr]} d_{u,t} l_{u,t} \nu_t \stackrel{d_{p}^b}{\Rightarrow} \int_{0}^r d_{u}(\tau) d\mB_{0}^{(1)}(\tau)$ in probability, it can be shown that
$(\vbeta_{\vx}^{0'}\mS_{p_1})\left( T^{-1/2}\sum_{t=1}^{[Tr]} \mD_{\vv,t} \vl_{\vv,t} \nu_t \right) \stackrel{d_{p}^b}{\Rightarrow} (\vbeta_{\vx}^{0'}\mS_{p_1}) \int_{0}^r \mD_v(\tau) d\mB_{0}^{(2:p_1+1)}(\tau)$ in probability, where $\mB_{0}^{(2:p_1+1)}(\cdot)$ refers to selecting elements $2:p_1+1$ in order from $\mB_0(\cdot)$.  Moreover, because $u_t\nu_t,\vv_t \nu_t$ share the same $\nu_t$ which is i.i.d and for which $E^b(\nu_t^2)=1$, $ (\vbeta_{\vx}^{0'}\vs_{p_1}) \left( T^{-1/2}\sum_{t= 1}^{[Tr]}  d_{u,t} l_{u,t} \nu_t \right)$ and $(\vbeta_{\vx}^{0'}\mS_{p_1})\left( T^{-1/2}\sum_{t= 1}^{[Tr]}  \mD_{\vv,t} \vl_{\vv,t} \nu_t \right)$  also jointly converge, and their unconditional bootstrap covariance converges to the unconditional covariance of their respective limits. Therefore, also for $\mathcal S_{\dag} = \vbeta_{\vx,\#}$,
\begin{align}\nonumber
\mathcal{E}_{1}^b &\weaks (\vbeta_{\vx}^{0'}\vs_{p_1}) \int_{\tau_{s-1}}^{\tau_s} d_{u}^2(\tau)\rd\mB_{0}^{(1)}(\tau) + (\vbeta_{\vx}^{0'}\mS_{p_1}) \int_{\tau_{s-1}}^{\tau_s} \mD_{\vv}(\tau) \rd\mB_{0}^{(2:p_1+1)}(\tau)\\ \label{distre1b}
& = (\mathcal S_{\dag}' \mS_{\#}) \int_{\tau_{s-1}}^{\tau_s} \mD(\tau) \rd \, \mB_{0,\#}(\tau) = \mathbb M_1(\tau_{s-1},\tau_s),
\end{align}
 with variance matrix $\mV_{\mathbb{M}_1(\tau_{s-1},\tau_s)}$ given in the Supplementary Appendix, Section \ref{supsec1}, Proof of Lemma \ref{lem6}.

        \bigskip
        $\bullet$ Next, consider $\mathcal{E}_3^b$. From \eqref{58} we have that: $\vxi_{t-1}^b=\hat\mF^{\tilde t-1}_s \vxi_{[\tau_{s-1}T]}^b + \sum_{l=0}^{\tilde t-2}\hat\mF_s^l\veta_{t-l-1}^b+\left(\sum_{l=0}^{\tilde t-2}\hat\mF_s^l\right)\hat\vmu_s$. Then, replacing estimated parameters with the true ones and denoting the remainder by $\op^b(1)$ for reasons discussed earlier,
  \begin{align} \nonumber
        \mathcal{E}^b_3&=T^{-1/2} \sum_{t \in \tilde I_s}\mathcal{S}\vxi_{t-1}^b\vg_t^{b\prime}\;\mathcal{S}_{\dag}
        \\ \nonumber
        &=T^{-1/2}(\mathcal{S}_{\dag}^{\prime}\vg_{[\tau_{s-1}T]+1}^b)(\mathcal{S} \vxi_{[\tau_{s-1}T]}^b)  + T^{-1/2} \sum_{t \in \tilde I_s^{-}}(\mathcal{S}_{\dag}^{\prime}\vg_{t}^b)\left[\mathcal{S}\mF^{\tilde t-1}_s \vxi_{[\tau_{s-1}T]}^b\right] \\\nonumber
        &+T^{-1/2} \sum_{t \in \tilde I_s^{-}}(\mathcal{S}_{\dag}^{\prime}\vg_{t}^b)\left[ \mathcal{S}\left(\sum_{l=0}^{\tilde t-2}\mF_s^l\right)\vmu_s\right]\\ \label{mathcale3b}
        &+ T^{-1/2} \sum_{t \in \tilde I_s^{-}} (\mathcal{S}_{\dag}^{\prime}\vg_{t}^b)\left[ \mathcal{S}\sum_{l=0}^{\tilde t-2}\mF_s^l \veta_{t-l-1}^b\right] +\op^b(1)=\mathcal{ E}^b_{3,1}+\mathcal{E}^b_{3,2}+\mathcal{E}^b_{3,3}+\mathcal{ E}^b_{3,4}+\op^b(1).
        \end{align}
First, note that by \eqref{rxib}, $\|\vxi_{[\tau_{s-1}T]}^b \|=\Op^b(T^{-\alpha})$, and also that $\|\vg_{[\tau_{s-1}T]+1}^b \|=  \Op^b(T^{-\alpha})$, for any $\alpha>0$. Therefore,  $\mathcal{E}_{3,1}^b=\op^b(1)$. For the same reason and by the fact that $\| \mF_s^l\|$ is exponentially decaying with $l$,
\begin{align*}
\|\mathcal{E}_{3,2}^b \| \leq  \| T^{-1/2} \sum_{t \in \tilde I_s^{-}}\mathcal{S}_{\dag}^{\prime}\vg_{t}^b\| \,\, \left(\| \mathcal{S}\|\,  \sup_l \| \mF^{l}_s \|\,\, \ \| \vxi_{[\tau_{s-1}T]}^b\| \right) = \| \mathcal E_{1}^{b} \| \left(\, \| \mathcal{S}\| \, \sup_l  \| \mF^{l}_s \|\,\right) \op^b(1) = \op^b(1).
\end{align*}
Next, note that by similar derivations as for \eqref{nstar1} in Supplementary Appendix, Section \ref{supsec1}, and artificially setting $\vg_{t-l}= 0, \vnu_{t-l}=0$ for all $t<l$ (as in \citet{Boswijketal:2016}) we have, for $\tilde n =[\tau_{s-1}T]+2$, and $\tilde I_s^- = [[\tau_{s-1}T]+2, [\tau_s T]]$, \begin{align}\nonumber
\mathcal{E}^b_{3,4}& = \sum_{l=1}^{\Delta \tau_s T -2}  \mathcal{S} \mF_s^l \, \left(T^{-1/2}\sum_{t\in \tilde I_s^-}   (\mathcal{S}_{\dag}^{\prime}\vg_t^b)  \veta_{t-l-1}^b \right)- T^{-1/2} \sum_{l=1}^{\Delta \tau_s T -2}\mathcal{S} \mF_s^l \, \sum_{j=0}^{l-1}\, (\mathcal{S}_{\dag}^{\prime} \vg_{\tilde n+j}^b ) \veta_{\tilde n+j-(\black{l+1})}^b\\\label{shortcuts}
&\equiv\mathcal{\tilde E}^b_{3,4}(\Delta \tau_s T-2)-\mathcal{L}.
\end{align}

We now show that $\mathcal L=\op^b(1)$. Let $\mathcal S_{\dag}'\vg_t = u_t$. Then,
 \begin{align}\nonumber
&(\mathcal S_{\dag}'\vg_t \nu_t) ( \vg_{t-l} \odot \vnu_{t-l})( \vg_{t-l} \odot \vnu_{t-l})'(\mathcal S_{\dag}'\vg_t \nu_t)  =\nu_t^2 u_t^2 \begin{bmatrix} u_{t-l}^2 \nu_{t-l}^2 & u_{t-l} \vv_{t-l}' \nu_{t-l}^2 & u_{t-l} \vzeta_{t-l}' \nu_{t-l} \\
(u_{t-l} \vv_{t-l}' \nu_{t-l}^2)' & \vv_{t-l} \vv_{t-l}' \nu_{t-l}^2 & \vv_{t-l} \vzeta_{t-l}' \nu_{t-l} \\\label{derivcross}
(u_{t-l} \vzeta_{t-l}' \nu_{t-l})' & (\vv_{t-l} \vzeta_{t-l}' \nu_{t-l})' & \vzeta_{t-l} \vzeta_{t-l}'
  \end{bmatrix}_{\#}\\
  & =
\begin{bmatrix} u_t^2 u_{t-l}^2 & u_t^2 u_{t-l} u_t^2 \vv_{t-l}'  & u_t^2 u_{t-l} \vzeta_{t-l}'  \\
u_t^2 (u_{t-l} \vv_{t-l}' )' & u_t^2 \vv_{t-l} \vv_{t-l}'  & u_t^2 \vv_{t-l} \vzeta_{t-l}'  \\
u_t^2 (u_{t-l} \vzeta_{t-l}')'  & (u_t^2 \vv_{t-l} \vzeta_{t-l}')' & u_t^2  \vzeta_{t-l} \vzeta_{t-l}'
  \end{bmatrix}_{\#} \odot  \begin{bmatrix} \nu_t^2 \nu_{t-l}^2 & (\nu_t^2 \nu_{t-l}^2) \viota_{p_1}' &(\nu_t^2  \nu_{t-l})\viota_{p_2}' \\
(\nu_t^2 \nu_{t-l}^2) \viota_{p_1} & (\nu_t^2  \nu_{t-l}^2) \mJ_{p_1} & (\nu_t^2 \nu_{t-l}) \viota_{p_1}\viota_{p_2}'\\
(\nu_t^2 \nu_{t-l}) \viota_{p_2} & (\nu_t^2  \nu_{t-l}) \viota_{p_2}\viota_{p_1}' & \nu_t^2 \mJ_{p_2}
  \end{bmatrix}_{\#}
 \end{align}
Therefore, for $l\geq 1$,
\begin{align}\label{ebforl}
\sup_t \E(\E^b[(\mathcal S_{\dag}'\vg_t \nu_t )( \vg_{t-l} \odot \vnu_{t-l})( \vg_{t-l} \odot \vnu_{t-l})'(\mathcal S_{\dag}'\vg_t \nu_t)]) = \sup_t \E( u_t^2 ((\vg_{t-l} \vg_{t-l}') \odot \mathcal J)).
\end{align}
By Assumption \ref{a8}, the non-zero elements of $\E( u_t^2 \otimes  ((\vg_{t-l} \vg_{t-l}')) \odot \mathcal J$, do not depend on $t$, and are elements of linear functions $\vrho_{0,l}$, so they are uniformly bounded in $l$,
Therefore, for element $\mathcal L^{(a,b)}$ of the matrix $\mathcal L$, and constants $c>0,c_1>0$,
\begin{align*}
&\sup_t \E (\E^b| \mathcal L^{(a,b)}|)\\
& \leq T^{-1/2} \sum_{l=1}^{\Delta \tau_s T -2}| (\mathcal S \mF_s^l \mA_s^{-1})^{(a,b)}\, | \,\,\sum_{j=0}^{l-1}\, \sup_t \E \E^b|\{(\mathcal{S}_{\dag}^{b\prime} \vg_{\tilde n+j} ) \nu_{\tilde n+j}\, [\vg_{\tilde n+j-(\black{l+1})}) \odot \vnu_{\tilde n+j-(\black{l+1})} ]\}^{(a,b)}| \\
& \leq T^{-1/2} \sum_{l=1}^{\Delta \tau_s T -2}\| \mA_s^{-1}\|\,\, \|\mathcal S\|\,\, \| \mF_s^l \| \, \, \sum_{j=0}^{l-1} c \leq c_1 T^{-1/2} \sum_{l=0}^{\infty} l \| \mF_s^l \| \rightarrow 0.
\end{align*}
Therefore, $\mathcal L=\op^b(1)$ for $\mathcal S_{\dag}=\mathcal S_u$, and by similar arguments, $\mathcal L=\op^b(1)$ for $\mathcal S_{\dag}=\vbeta_{\vx,\#}$.

Next, we analyze $\tilde {\mathcal E}_{3,4}^b(\Delta_{\tau_s}T-2)$. To that end, let for now $\mathcal S_{\dag}=\mathcal S_u$, and note that a crucial term in  $\tilde {\mathcal E}_{3,4}^b(\Delta_{\tau_s}T-2)$ is $
\mathcal L_1^b(l)=T^{-1/2} \sum_{t=1}^{[Tr]} u_t \nu_t (\vg_{t-l} \odot \vnu_{t-l})
$
for $l\geq 1$.
By the structure of $\mS$ and $\mD_t$ in \eqref{decompsd}, letting $\tilde \vnu_t = \vec(\nu_t \viota_{p_1+1},\viota_{p_2+1})$,
\begin{align*}
\vg_{t-l} \odot \vnu_{t-l} = \begin{bmatrix} d_{u,t-l} l_{u,t-l} \nu_{t-l} \\
\vs_{p_1} d_{u,t} l_{u,t-l} \nu_{t-l} + \mS_{p_1} \mD_{\vv,t} \vl_{\vv,t-l} \nu_{t-l} \\
\mS_{p_2} \mD_{\vzeta,t-l} \vl_{\vzeta,t-l} \end{bmatrix}_{\#} = \mS_{\#} \mD_{t-l,\#}  (\vl_{t-l} \odot \tilde \vnu_{t-l})_{\#}.
\end{align*}
Then, letting $\mathcal E_{t,l} =l_{u,t}\vl_{t}$ and $\mathcal E_{t,l}^b =  l_{u,t}  \vl_{t-l} \odot \vec(\nu_t\nu_{t-l} \viota_{p_1+1}, \nu_t\viota_{p_2})$.
\begin{align}\label{l1bl}
\mathcal L_1^b(l) = T^{-1/2} \sum_{t=1}^{[Tr]} u_t \nu_t (\vg_{t-l} \odot \vnu_{t-l}) = T^{-1/2} \sum_{t=1}^{[Tr]} ( d_{ut} \mS_{\#} \mD_{t-l,\#}) (\mathcal E_{t,l}^b)_{\#}.
\end{align}

We now proceed as for $\mathcal E_1^b$, in two parts: in part (i), we derive the limiting distribution of $\mB_{l,T,A}^b(r)= T^{-1/2} \sum_{t=1}^{[Tr]}  \mathcal E_{t,l}^b$ and its equivalent for $\mathcal S_{\dag} =\vbeta_{\vx,\#}$  for each $l$, by verifying Lemma \ref{lem3} (we verify this for both definitions of $\mathcal S_{\dag}$ and therefore replace $ \mathcal E_{t,l}^b$ with the appropriate quantities when $\mathcal S_{\dag} = \vbeta_{\vx,\#}$); in part (ii), we derive the limiting distribution of $\tilde {\mathcal E}_{3,4}^b(n^*)$ using Theorem 2.1 in \citet{Hansen:1992} for fixed $n^*$. Then we take the limit as $n^*\rightarrow \infty$.

\textbf{Part (i). }Let $\mathcal S_{\dag}=\mathcal S_u$. First, we apply Lemma \ref{lem3} to $\mB_{l,T,A}^b(r) = T^{-1/2} \sum_{t=1}^{[Tr]} \mathcal E_{t,l}^b$, for $l\geq 1$, where note that even though $\mV_{\mB_{l,T,A}^b(r)}=\plim_{T\rightarrow \infty} \var^b (\mB_{l,T,A}^b(r))$ does not converge to $r \mI_n$ as one condition in Lemma \ref{lem3} requires, it is symmetric and positive semi-definite, so by a decomposition of $\mV_{\mB_{l,T,A}^b} = \mE^{1/2}  \mE^{1/2'}$,  $ \mE^{-1/2}\mB_{l,T,A}^b(r)$ converges to a process whose limiting variance is $r$ times the identity matrix, where $\mE^{-1}$ is the generalized inverse. Therefore, in the rest of the analysis, we no longer need to verify this condition, except for deriving the limit of the unconditional bootstrap variance, and proceed to verify the rest of the conditions.
First, $\E^b(\mathcal E_{t,l}^b) = \mathcal E_{t,l} \odot \vec(\E^b (\nu_t\nu_{t-l}|\mathcal F_{t-1}^b) \viota_{p_1+1}, \E^b (\nu_t|\mathcal F_{t-1}^b)\viota_{p_2})=\vzeros_{n}$, so $\mathcal E_{t,l}^b$ is a m.d.s. Second, for $\phi_t^b$ denoting a typical element of $\mathcal E_{t,l}^b$, and $\phi_t$ denoting the corresponding element of $\mathcal E_{t,l}$, we have that $\sup_t \E(\E^b| \phi_t^b |^{2+\delta^*}) = \sup_t \E (|\phi_t|^{2+\delta^*}\sup_t \E^b |\nu_t \nu_{t-l}|^{2+\delta^*} <\infty$ by Assumptions \ref{a8}(iii) and Assumption \ref{aboot}(ii) for the first $p_1+1$ elements of $\mathcal E_{t,l}^b$, or we have that $\sup_t \E(\E^b| \phi_t^b |^{2+\delta^*}) = \sup_t \E|\phi_t|^{2+\delta^*} \sup_t \E^b |\nu_t |^{2+\delta^*} <\infty$ by the same assumptions, for the rest of the elements.

Third, to facilitate showing that $\var^b(\mB_{l,T,A}^b(r)|\mathcal F_{t-1}^b) - \var^b(\mB_{l,T,A}^b(r)) \inp 0$, note that, from \eqref{ebforl}, $\var^b(\mathcal E_{t,l}^b)=(\mathcal E_{t,l}\mathcal E_{t,l}') \odot \diag(\mJ_1,\mJ_2) =\var^b(\mathcal E_{t,l}^b)$, where recall that $\mJ_1, \mJ_2$ are matrices of ones of dimension $(p_1+1)\times (p_1+1)$, respectively $p_2 \times p_2$. Therefore, by \ref{A8prime}(iii) and Lemma \ref{lem8},
\begin{align}\nonumber
&\var^b(\mB_{l,T,A}^b(r)) = T^{-1} \sum_{t=1}^{[Tr]} (l_{u,t}^2 \vl_{t-l} \vl_{t-l}' ) \odot  \diag(\mJ_1,\mJ_2) \\
&= T^{-1} \sum_{t=1}^{[Tr]} E^b\begin{bmatrix} l_{u,t}^2 l_{u,t-l}^2 & l_{u,t}^2 l_{u,t-l} \vl_{\vv,t-l}' & \vzeros_{1\times p_2} \\
l_{u,t}^2 l_{u,t-l} \vl_{\vv,t-l} & l_{u,t}^2 \vl_{\vv,t-l} \vl_{\vv,t-l}' &\vzeros_{p_1 \times p_2} \\ \label{refeq0}
\vzeros_{p_2} & \vzeros_{p_2 \times p_1} & l_{u,t}^2 \vl_{\vzeta,t-l}\vl_{\vzeta,t-l}' \end{bmatrix} \inp \var(\mB_l(r)^{(n+1:2n)}) = \vrho_{i,i}^{(1:n,1:n)},
\end{align}
where $\mB_l(r)$ was defined just before Lemma \ref{lem5}, $ \mB_l(r)^{(1:n)}$ is the vector stacking elements $1:n$ of $\mB_l(r)$ in order, and $\vrho_{i,i}^{(1:n,1:n)}$ is the left upper $n\times n$ block of $\vrho_{i,i}$.

Note that so far, all the proofs went through using Assumptions \ref{a1}-\ref{aboot}. In the last equation above, the need for  \ref{A8prime}(iii), which imposes the block diagonal structure for $\rho_{i,i}^{(1:n,1:n)}$ (equivalently, it imposes $\E[ l_{u,t}^2 \vn_{t-l} \vl_{\vzeta,t-l}']=\vzeros_{p_2 \times p_2}$  for $l\geq 1$ and $\vn_t =\vec(l_{u,t},\vl_{\vv,t})$).

To verify the last condition in Lemma \ref{lem3}, because $\nu_{t-l},\nu_{t-l}^2$ is i.i.d. and $\sup_t \E^b|\nu_{t}|^{4+\delta^*}<\infty$, by Lemma \ref{lem1} and Lemma \ref{lem8}(iv), (proceeding element-wise, first conditional on the sample, then unconditionally),
\begin{align*}
&\var^b(\mB_{l,T,A}^b(r)| \mathcal F_{t-1}^b) -\var^b(\mB_{l,T,A}^b(r)) \\
&= T^{-1} \sum_{t=1}^{[Tr]} (l_{u,t}^2 \vl_{t-l} \vl_{t-l}' ) \odot  \begin{bmatrix}  (\nu_{t-l}^2-1) & (\nu_{t-l}^2-1) \viota_{p_1}' &(\nu_{t-l})\viota_{p_2}' \\
(\nu_{t-l}^2-1) \viota_{p_1} & (\nu_{t-l}^2-1) \mJ_{p_1} & (\nu_{t-l}) \viota_{p_1}\viota_{p_2}'\\
(\nu_{t-l}) \viota_{p_2} & (\nu_{t-l}) \viota_{p_2}\viota_{p_1}' & \vzeros_{p_2\times p_2}
  \end{bmatrix} =\op^b(1).
\end{align*}
Therefore,
\begin{equation}\label{blb}
\mB_{l,T,A}^b(r) = T^{-1/2} \sum_{t=1}^{[Tr]} \mathcal E_{t,l}^b \stackrel{d_p^b}{\Rightarrow} \mB_{l}^{(1:n)}(r),
\end{equation}
where $\mB_l$ was defined just before Lemma \ref{lem6}.\\
Now let $\mathcal S_{\dag} = \vbeta_{\vx_,\#}$. Then:
\begin{align} \nonumber
\mathcal L_2^b(l)&\equiv  T^{-1/2} \sum_{t=1}^{[Tr]} \vbeta_{\vx}'\vv_t \nu_t (\vg_{t-l} \odot \vnu_{t-l}) =(\vbeta_{\vx}' \vs_{p_1}) T^{-1/2} \sum_{t=1}^{[Tr]}  d_{ut} \,  [ l_{u,t} \nu_t  (\vg_{t-l} \odot \vnu_{t-l}) ]\,  \\ \nonumber
 & +  T^{-1/2} \sum_{t=1}^{[Tr]} \vbeta_{\vx}' \mS_{p_1} \mD_{\vv,t} \vl_{\vv,t} \nu_t (\vg_{t-l} \odot \vnu_{t-l}) ] \\ \nonumber
& = (\vbeta_{\vx}' \vs_{p_1}) \mathcal L_1^b(l) + (\vbeta_{\vx}' \otimes \mI_{n,\#}) (\mS_{p_1} \otimes \mS_{\#}) T^{-1/2} \sum_{t=1}^{[Tr]} (\mD_{\vv} \otimes \mD_{t-l,\#})  ( \vl_{\vv,t} \nu_t) \otimes (\vl_{t-l,\#} \odot \vnu_{t-l}) \\ \label{l2bl}
&= \mathcal L_{2,1}^b + \mathcal L_{2,2}^b.
\end{align}
The distribution $\mathcal L_{2,1}^b$  follows from \eqref{blb} and part (ii) below. Following similar steps as for $\mathcal L_1^b(l)$ above (where $\mathcal S_{\dag}=\mathcal S_u$), it can be shown that under Assumptions \ref{a1}-\ref{aboot} and \ref{A8prime} (iii), which imposes $\E[ b_t \vn_{t-l} \vl_{\vzeta,t-l}']=\vzeros_{p_1 \times p_2}$ for $b_t$ being any element of $\vl_{\vv,t}\vl_{\vv,t}'$, we have:
\begin{align}\label{blb2}
\mB_{l,T,B}^b(r) = T^{-1/2} \sum_{t=1}^{[Tr]}   ( \vl_{\vv,t} \nu_t) \otimes (\vl_{t-l} \odot \tilde \vnu_{t-l}) \stackrel{d_p^b}{\Rightarrow} \mB_{l}^{(n+1:n(p_1+1))}(r). \end{align}
To facilitate presentation, define for a matrix $\mO$ whose rows and columns are multiples of $n$, the operation
\begin{equation}\label{defblock}
\textbf{block}_{\kappa,\kappa^*}(\mO) = \mO^{(n(\kappa-1)+1: n \kappa, n(\kappa^*-1)+1:n\kappa^*)},
\end{equation} that is, the operation that selects the $(\kappa,\kappa^*)$ $n\times n$ sub-matrix of the matrix $\mO$. Also, for a $\bar{n} \times \bar{n}$ square matrix $\mO_1$, define the operation that makes $\mO_1$ block diagonal at row $j$ as follows:
\begin{equation}
\textbf{blockdiag}_{j} (\mO_1) = \diag( \mO_1^{(1:j, 1:j)}, \mO_1^{(j+1:\bar{n}, j+1:\bar{n})}).
\end{equation}

Then, by similar arguments to $\mathcal S_{\dag} = \mathcal S_u$,
\begin{align}\label{blab}
\vec(\mB_{l,T,A}^b(r), \mB_{l,T,B}^b(r)) \stackrel{d_p^b}{\Rightarrow} \mB_l^{(1:n(p_1+1))}(r),
\end{align}
where the relevant bootstrap condition we have to show, by analogy to $\mathcal S_{\dag}=\mathcal S_u$, is that $\E^b(\mB_{l,T,A}^b(r_1) \mB_{l,T,B}^b(r_2)') - \min(r_1,r_2) \E(\mB_l^{(1:n)}(r_1) (\mB_l^{(n+1:n(p_1+1))}(r_2))^{\prime})=\op(1)$, a condition proven below for $r_1=r_2=r$, because when $r_1\neq r_2$, the proof follows in a similar fashion. Letting $\mB_{l,T}^* = \vec(\mB_{l,T,A}^b(r_1),\mB_{l,T,B}^b(r_1))$,  note that
\begin{equation*}
\E^b(\mB_{l,T,A}^b(r) \mB_{l,T,B}^{b'}(r))= [\textbf{block}_{1,2}(\var^b(\mB_{l,T}^*)), \textbf{block}_{1, 3}(\var^b(\mB_{l,T}^*)), \ldots, \textbf{block}_{1, p_1+1}(\var^b(\mB_{l,T}^*)) ],
\end{equation*}
so we proceed with each block $2,\ldots, p_1+1$, and let $b_t =l_{u,t} l_{\vv_\kappa,t}$ for $\kappa = 1,\ldots, p_1$. Then, noting that $l_{u,t} l_{\vv_{\kappa},t}\nu_t (\vl_{t-l} \odot \tilde \vnu_{t-l})$ is a m.d.s. with respect to $\mathcal F_{t-1}^b$ for $l\geq 1$,  we have:
\begin{align*}
&\textbf{block}_{1,\kappa+1} \var^b(B_{l,T}^*) = T^{-1} \sum_{t=1}^{[Tr]} \E^b [b_t^2 \nu_t^2 (\vl_{t-l} \odot \tilde{\vnu}_{t-l}) (\vl_{t-l} \odot \tilde\vnu_{t-l})'] \\
&= T^{-1} \sum_{t=1}^{[Tr]} \textbf{blockdiag}_{p_1+1} ( b_t \vl_{t-l} \vl_{t-l}') \stackrel{p^b}{\rightarrow} \textbf{blockdiag}_{p_1+1} (\textbf{block}_{1,\kappa+1}(\vrho_{l,l})) =\textbf{block}_{1,\kappa+1}(\vrho_{l,l})),
\end{align*}
where the last line follows by replacing $l_{u,t}$ by $b_t$ in \eqref{refeq0}, by Lemma \ref{lem8}(iv) and by \ref{A8prime}(iii).

\textbf{Part (ii).} First, let $\mathcal S_{\dag} = \mathcal S_u$, and recall that, from \eqref{l1bl}, we need the distribution of \begin{align*}
\mathcal L_1^b(l) = T^{-1/2} \sum_{t=1}^{[Tr]} u_t \nu_t (\vg_{t-l} \odot \vnu_{t-l}) = T^{-1/2} \sum_{t=1}^{[Tr]} ( d_{u,t} \mS_{\#} \mD_{t-l,\#}) (\mathcal E_{t,l}^b)_{\#}.
\end{align*}
By \citet{Hansen:1992}, Theorem 2.1, because $\| d_{u,t} \mS_{\#} \mD_{t-l,\#}\|$ is bounded by Assumption \ref{a8}(ii), and $\mD(\tau-\frac{l}{T}) = \mD_{t-l}$ when $\tau \in [\frac{t}{T},\frac{t+1}{T})$, we have:
\begin{eqnarray} \nonumber
\mathcal L_1^b(l) &= \int_{0}^r d_u(\tau) \mS_{\#} \mD (\tau-\frac{l}{T}) d\mB_{l,T,A,\#}(\tau) \stackrel{d_p^b}{\Rightarrow} \int_{0}^r d_u(\tau) \mS_{\#} \mD_{\#}(\tau) \rd\mB_{l,\#}^{(1:n)}(\tau)\\ \label{distribl1}
&=((\mathcal{S}_{\dag}'\mS_{\#})\otimes \mS_{\#}) \int_{0}^{r}  \left(\mD_{\#}(\tau)\otimes\mD_{\#}\left(\tau\right)\right)\rd\mB_{l,\#}(\tau).
\end{eqnarray}
where the convergence follows because $\var^b(\mathcal L_1^b(l))- \mathcal L_1(l) \stackrel{p^b}{\rightarrow} 0$, which can be shown by similar arguments to \eqref{refeq0}, and using Lemma \ref{lem8}(iii) instead of \ref{lem8}(iv). Similarly, for $\mathcal S_{\dag} = \vbeta_{\vx,\#}$, we have:
\begin{eqnarray} \nonumber
\mathcal L_2^b(l) &=& (\vbeta_{\vx}' \vs_{p_1}) \mathcal L_1^b(l) + (\vbeta_{\vx}' \otimes \mI_{n,\#}) (\mS_{p_1} \otimes \mS_{\#}) T^{-1/2} \sum_{t=1}^{[Tr]} (\mD_{\vv,t} \otimes \mD_{t-l,\#})  ( \vl_{\vv,t} \nu_t) \otimes (\vl_{t-l,\#} \odot \vnu_{t-l})\\ \nonumber
&\stackrel{d_p^b}{\Rightarrow}& (\vbeta_{\vx}' \vs_{p_1}) \int_{0}^r d_u(\tau) \mS_{\#} \mD_{\tau,\#} \rd\mB_{l,\#}^{(1:n)}(\tau) + (\vbeta_{\vx}' \mS_{p_1} \otimes \mS_{\#}) \int_{0}^r [(\mD_{\vv}(\tau) \otimes \mD_{\#}(\tau)] \rd\mB_{l,\#}^{(n+1:n(p_1+1))}(\tau) \\ \label{distribl2}
& =& [(\mathcal{S}_{\dag}'\mS_{\#})\otimes \mS_{\#}] \int_{0}^{r}  \left(\mD_{\#}(\tau)\otimes\mD_{\#}\left(\tau\right)\right)\rd\mB_{l,\#}(\tau).
\end{eqnarray}

Next, we derive the distribution of $\mathcal E_{3,4}(n^*)$. This follows by similar arguments as above if we can verify that the off-diagonal elements of the bootstrap covariance $\cov^b(\mB_{l,T}^*(r), \mB_{l^*,T}^*(r))$ converge in probability to the counterpart elements of the covariance $\cov(\mB_l^{(1:n(p_1+1))}(r),\mB_{l^*}^{(1:n(p_1+1))}(r))$ for $l\neq l^*$. We only do so for $\textbf{block}_{1,1}(\cov^b(\mB_{l,T}^*(r), \mB_{l^*,T}^*(r))$; the rest follows by similar reasoning.
\begin{align*}
&\textbf{block}_{1,1}(\cov^b(\mB_{l,T}^*(r), \mB_{l^*,T}^*(r)) = T^{-1} \sum_{t=1}^{[Tr]} l_{u,t}^2 (\vl_{t-l} \vl_{t-l^*}') \odot \E^b ( \nu_t^2 \tilde \vnu_{t-l} \tilde \vnu_{t-l^*}') \\
&+ T^{-1} \sum_{t,t^*=1, t\neq t^*}^{[Tr]} l_{u,t} l_{u,t^*} (\vl_{t-l} \vl_{t^*-l^*}') \odot \E^b ( \nu_t \nu_{t^*} \tilde \vnu_{t-l} \tilde \vnu_{t^*-l^*}')  = T^{-1} \sum_{t=1}^{[Tr]} l_{u,t}^2 (\vl_{t-l} \vl_{t-l^*}') \odot \diag(\vzeros_{p_1+1,p_1+1}, \mJ_2) \\
&\inp \textbf{block}_{1,1}(\vrho_{l,l^*})=\textbf{block}_{1,1}(\cov(\mB_l(r),\mB_{l^*}(r))),
\end{align*}
because of \ref{A8prime}(ii) we imposed that  $\E[l_{u,t}^2 \vn_{t-l} \vn_{t-l^*}] =\vzeros_{(p_1+1)\times (p_1+1)}$ for $l,l^*\geq 1, l\neq l^*$, and \ref{A8prime}(iii) we imposed that  $\E[l_{u,t}^2 \vn_{t-l} \vl_{\vzeta,t-l^*}'] =\vzeros_{(p_1+1)\times p_2}$ for $l,l^*\geq 1, l\neq l^*$. In the general setting, for $\mathcal S_{\dag} = \mathcal\ S_u$ or $\mathcal S_{\dag} = \vbeta_{\vx,\#}$, by analogy we need $\E[(\vn_t \vn_t') \otimes (\vn_{t-l} \vn_{t-l^*}')]=\vzeros_{(p_1+1)^2 \times(p_1+1)^2}$ for $l,l^* \geq 1, l\neq l^*$,  and $\E[(\vn_t \vn_t') \otimes (\vn_{t-l} \vl_{\vzeta,t-l^*}')]=\vzeros_{(p_1+1)^2 \times ((p_1+1)p_2)}$ for $l,l^* \geq 1, l\neq l^*$,  which are also satisfied by \ref{A8prime}(ii)-(iii).

Using \eqref{distribl1}-\eqref{distribl2} in the expression for $\mathcal E_{3,4}^b(n^*) = \sum_{l=0}^{n^*}  \mathcal{S} \mF_s^l \mA_s^{-1} \, \left(T^{-1/2}\sum_{t\in \tilde I_s^-}   (\mathcal{S}_{\dag}^{\prime}\vg_t) \nu_t (\vg_{t-l-1} \odot \vnu_{t-l-1})\right)$, it follows that, for a fixed $n^*$,
\begin{align*}
\tilde {\mathcal E}_{3,4}^b(n^*) \stackrel{d_p^b}{\Rightarrow} \sum_{l=0}^{n^*}  \mathcal{S} \mF_s^l \mA_s^{-1} ((\mathcal{S}_{\dag}'\mS_{\#})\otimes \mS_{\#})  \int_{\tau_{s-1}}^{\tau_s}  \left(\mD_{\#}(\tau)\otimes\mD_{\#}\left(\tau\right)\right)\rd\mB_{l+1,\#}(\tau).
\end{align*}
Now as in the proof of Lemma \ref{lem6}, setting $n^*=T^{\alpha}$ for some $\alpha \in(0,1)$, and noting that the remainder $\tilde {\mathcal E}_{3,4}^b(\Delta \tau_s T-2) - \tilde {\mathcal E}_{3,4}^b(n^*)=\op^b(1)$, it can be shown that:
\begin{align*}
\tilde {\mathcal E}_{3,4}^b(\Delta \tau_s T-2) \stackrel{d_p^b}{\Rightarrow} \sum_{l=0}^{\infty}   (\mathcal{S}_{\dag}'\mS_{\#})\otimes (\mathcal{S} \mF_s^l \mA_s^{-1} \mS_{\#}) \int_{\tau_{s-1}}^{\tau_s}  \left(\mD_{\#}(\tau)\otimes\mD_{\#}\left(\tau\right)\right)\rd\mB_{l+1,\#}(\tau) = \mathbb M_{3,2}(\tau_{s-1},\tau_s),
\end{align*}
where $\mathbb M_{3,2}(\tau_{s-1},\tau_s)$ as defined above is also the asymptotic distribution of the sample counterpart of $\tilde {\mathcal E}_{3,4}^b(\Delta \tau_s T-2)$, that is $\tilde {\mathcal E}_{3,4}(\Delta \tau_s T-2)$, featuring  in Supplementary Appendix, Section \ref{supsec1} in \eqref{mathcalM3}, whose variance exists and is derived right after that equation. Therefore,
\begin{equation}\label{mathcale34b}
\mathcal E_{3,4}^b \stackrel{d_p^b}{\Rightarrow} \mathbb M_{3,2}(\tau_{s-1},\tau_s).
\end{equation}
Now consider $\mathcal E_{3,3}^b =  T^{-1/2} \sum_{t \in \tilde I_s^{-}}(\mathcal{S}_{\dag}^{b\prime}\vg_{t}^b)\left[ \mathcal{S}\left(\sum_{l=0}^{\tilde t-2}\mF_s^l\right)\vmu_s\right]$ in \eqref{mathcale3b}. By similar analysis as for $\mathcal E_{3,3}^b$, it can be shown that:
\begin{equation}\label{mathcale33b}
\mathcal E_{3,3}^b \stackrel{d_p^b}{\Rightarrow} \sum_{l=0}^{\infty}   (\mathcal{S}_{\dag}'\mS_{\#})\otimes (\mathcal{S} \mF_s^l) \int_{\tau_{s-1}}^{\tau_s} \mD_{\#}(\tau) \rd \mB_{0,\#}(\tau) = \mathbb M_{3,1}(\tau_{s-1},\tau_s).
\end{equation}
Also, $\mathcal E_{3,3}^b +\mathcal E_{3,4}^b$ can be shown to jointly converge to $ \mathbb M_{3}(\tau_{s-1},\tau_s) =\mathbb M_{3,1}(\tau_{s-1},\tau_s)+ \mathbb M_{3,2}(\tau_{s-1},\tau_s)$, provided that $\cov^b( \mB_{0,T}^b(r),\mB_{l,T}^*(r)) -\cov ( \mB_{0}(r),\mB_{l}^{(1:n(p_1+1))}(r)) \inp 0$ uniformly in $r$ for all $l\geq 1$, where $\mB_{0,T}^b = T^{-1/2} \sum_{t=1}^{[Tr]}\vl_{t}\odot \tilde \vnu_t$, and recall that $\mB_{l,T}^*(r) = \vec(\mB_{l,T,A}^b(r),\mB_{l,T,B}^b(r))$. Consider $\cov^b( \mB_{0,T}^{b,(1)}(r),\mB_{l,T,A}^b(r)) = (\cov^b( \mB_{0,T}(r),\mB_{l,T}^*(r)))^{(1:n,1:1)}$ (so for $\mathcal S_{\dag} = \mathcal S_u$); the proof for the rest of the elements is similar.
\begin{align*}
\cov^b( \mB_{0,T}^{b,(1)}(r),\mB_{l,T,A}^b(r))  &=  T^{-1/2} \sum_{t=1}^{[Tr]} (l_{u,t} \nu_t )( T^{-1/2} \sum_{t=1}^{[Tr]} l_{u,t} \nu_t (\vl_{t-l} \odot \tilde \vnu_{t-l}))  \\
& = T^{-1} \sum_{t=1}^{[Tr]} l_{u,t}^2 \nu_t^2 (\vl_{t-l} \odot \tilde \vnu_{t-l}) + T^{-1} \sum_{t,t^*=1,t\neq t^*}^{[Tr]} l_{u,t} l_{u,t^*} \nu_t \nu_{t^*} (\vl_{t^*-l} \odot \tilde \vnu_{t^*-l})\\
& =\mathcal L_{3,1}^b + \mathcal L_{3,2}^b.
\end{align*}
Now note that by Lemma \ref{lem6}(iv),
\begin{align*}
&\E^b(\mathcal L_{3,1}^b)=  T^{-1} \sum_{t=1}^{[Tr]} l_{u,t}^2 \nu_t^2 (\vl_{t-l} \odot \tilde \vnu_{t-l}) = T^{-1} \sum_{t=1}^{[Tr]} \begin{bmatrix} l_{u,t}^2 l_{u,t-l} \\
l_{u,t}^2 \vl_{\vv,t-l} \\
l_{u,t}^2 \vl_{\vzeta,t-l} \end{bmatrix} \odot \E^b \begin{bmatrix} \nu_t^2 \nu_{t-l}  \\
\nu_t^2 \nu_{t-l} \viota_{p_1}\\
\nu_t^2\viota_{p_2}\end{bmatrix} \\
&= T^{-1} \sum_{t=1}^{[Tr]} l_{u,t}^2 \vec (\vzeros_{1+p_1},\vl_{\vzeta,t-l})  \inp r [ \vrho_{l}^{(1:n,1:1)} \odot \mathcal I], \\
&\E^b(\mathcal L_{3,2}^b)=  T^{-1} \sum_{t,t^*=1,t\neq t^*}^{[Tr]} \begin{bmatrix} l_{u,t} l_{u,t^*} l_{u,t^*-l}\\
l_{u,t} l_{u,t^*}  \vl_{\vv,t^*-l}\viota_{p_1} \\
l_{u,t} l_{u,t^*} \vl_{\vzeta,t^*-l}\viota_{p_2}
 \end{bmatrix} \odot \E^b \begin{bmatrix} \nu_t \nu_{t^*} \nu_{t^*-l} \\
\nu_t \nu_{t^*} \nu_{t^*-l} \\
\nu_t \nu_{t^*} \end{bmatrix} = \vzeros_{n}.
\end{align*}
Therefore,
$
\cov^b( \mB_{0,T}^{b,(1)}(r),\mB_{l,T,A}^b(r))   -\cov ( \mB_{0}^{(1)}(r),\mB_{l}^{(1:n)}(r)) =\op(1),
$
by the restriction in \ref{A8prime}(i), which ensures that $\vrho_{l}^{(1:n,1:1)} =\vrho_{l}^{(1:n,1:1)} \odot \mathcal I$; also note that in the general definition of $\mathcal S_{\dag}$, by analogy, we need $\E^b[(\vn_t \vn_t')\otimes \vn_{t-l}]=\vzeros_{(p_1+1)^2\times (p_1+1)}$ for $l\geq 1$, imposed in \ref{A8prime}(i). So,
\begin{equation*}
\mathcal E_{3,3}^b +\mathcal E_{3,4}^b \stackrel{d_p^b}{\Rightarrow} \mathbb M_{3}(\tau_{s-1},\tau_s) =\mathbb M_{3,1}(\tau_{s-1},\tau_s)+ \mathbb M_{3,2}(\tau_{s-1},\tau_s).
\end{equation*}
Because we showed that $\mathcal E_{3,1}^b =\op^b(1)$ and $\mathcal E_{3,2}^b =\op(1)$, it follows that:
$
\mathcal E_{3}^b \stackrel{d_p^b}{\Rightarrow}  \mathbb M_{3}(\tau_{s-1},\tau_s),
$ and that
$
\vec (\mathcal E_{1}^b, \mathcal E_3^b) \stackrel{d_p^b}{\Rightarrow}  \vec( \mathbb M_{1}(\tau_{s-1},\tau_s),\mathbb M_{3}(\tau_{s-1},\tau_s)).
$

$\bullet$ Now consider $\mathcal E_2^b$. Note that $\vr_t = \mathcal S_{\vr} \vxi_{t}$ is not bootstrapped, and recall that
$
\vxi_t= \vmu_s + \veta_t+\mF_s \vxi_{t-1}.
$
Therefore, replacing again estimated parameters by the true values, because the rest of the terms are $\op^b(1)$ (therefore also replacing, as before, $\vg_t^b$ with $\vg_t \odot \vnu_t$),
\begin{align*}
        \mathcal E_2^b & = \left\{ T^{-1/2} \sum_{t \in \tilde I_s} [\mathcal{S}_{\dag}'(\vg_t \odot \vnu_t)] \right\} \mathcal S_{\vr} \vmu_s +  T^{-1/2} \sum_{t \in \tilde I_s} [\mathcal{S}_{\dag}'(\vg_t \odot \vnu_t)] \mathcal S_{\vr} \mF_s \vxi_{t-1}+T^{-1/2} \sum_{t \in \tilde I_s} [\mathcal{S}_{\dag}'(\vg_t \odot \vnu_t)] \mathcal S_{\vr} \veta_t\\
        &=\mathcal{ E}_{2,1}^b+\mathcal{ E}_{2,2}^b+\mathcal{ E}_{2,3}^b+\op^b(1).
        \end{align*}
Now consider $\mathcal{ E}_{2,1}^b$. From \eqref{defe1b} and \eqref{distre1b}, without any restrictions on $\vrho_{i},\vrho_{ij}$ except those in  Assumption \ref{a8},
\begin{align}\nonumber
\mathcal{ E}_{2,1}^b & = \mathcal S_{\vr}[\mathcal E_1^b +\op^b(1)] \vmu_s \weaks  \left\{(\mathcal S_{\dag}' \mS_{\#}) \int_{\tau_{s-1}}^{\tau_s} \mD(\tau) d \, \mB_{0,\#}(\tau) \right\} \mathcal S_{\vr}  \vmu_s \\ \label{e21b}
& = ((\mathcal{S}_{\dag}'\mS_{\#})\otimes (\mathcal{S}_{\vr}\mF_s^0)) \, \left(\left[\int_{\tau_{s-1}}^{\tau_s}  \mD_{\#}(\tau)\rd\mB_{0,\#}(\tau)\right] \otimes \vmu_s\right),
\end{align}
where the latter is the first term in  $\mathbb M_{2,1}(\tau_{s-1},\tau_s) = \sum_{l=0}^{\infty}((\mathcal{S}_{\dag}'\mS_{\#})\otimes (\mathcal{S}_{\vr}\mF_s^l)) \, \left(\left[\int_{\tau_{s-1}}^{\tau_s}  \mD_{\#}(\tau)\rd\mB_{0,\#}(\tau)\right] \otimes \vmu_s\right)$ (the rest of the terms appear from the distribution of $\mathcal E_{2,2}^b$ as seen below).

Now consider $\mathcal{E}^b_{2,3}= T^{-1/2} \sum_{t \in \tilde I_s}[\mathcal{S}_{\dag}'(\vg_t \odot \vnu_t)]\mathcal{S}_{\vr}\veta_t $. Recall from the arguments above \eqref{decompA} that $\mA_s^{-1}$ is also upper triangular with rows $p_1+2:n$ equal to $[\vzeros_{p_2}\;\;\vzeros_{p_2\times p_1}\;\;\mI_{p_2}].$ Therefore, it can be shown that $\mathcal{S}_{\vr}\veta_t = \mathcal{S}_{\vr}\mA_{s,\#}^{-1}\;\vg_{t} = \vzeta_t $, so, for $\mathcal S_{\dag} = \mathcal S_u$,
\begin{align*}
\mathcal{ E}_{2,3}^b = T^{-1/2} \sum_{t \in \tilde I_s}[\mathcal{S}_{\dag}'(\vg_t \odot \vnu_t)]  \mathcal{S}_{\vr} \veta_t^b  =    T^{-1/2} \sum_{t \in \tilde I_s} u_t \nu_t \vzeta_t +\op^b(1),
\end{align*}
where the last equality follows from replacing $\mathcal{S}_{\dag}'\mA_{s,\#}\veta_t^b = \mathcal{S}_{\dag}'\mA_{s,\#}((\hat \mA_s^{-1} \hat \vg_t) \odot \vnu_t)$ with $ \mathcal{S}_{\dag}' (\vg_t \odot \vnu_t)$, as for $\mathcal E_1^b$, since the rest of the terms are of lower order in all the relevant sums.
Note that:
\begin{align*}
\mathcal{ E}_{2,3}^b =T^{-1/2} \sum_{t \in \tilde I_s} (\mS_{p_2} \mD_{\vzeta,t}  d_{u,t})  l_{u,t}\nu_t \vl_{\vzeta,t},
\end{align*}
and consider first $\mB_{u\vzeta,T}^b(r)= T^{-1/2} \sum_{t=1}^{[Tr]}  l_{u,t} \vl_{\vzeta,t} \nu_t$.

\textbf{Part (i).} Since $\nu_t$ is i.i.d, $\E^b( l_{u,t} l_{\vzeta_{\kappa},t}\nu_t|\mathcal F_{t-1}^b)=0$, for any element $l_{\vzeta_{\kappa},t}$ of $\vl_{\vzeta,t}$. Also, for some $c>0$, $\sup_t \E \E^b |l_{u,t} l_{\vzeta_{\kappa},t}\nu_t|^{2+\delta^*} \leq \sup_t \E |l_{u,t} l_{\vzeta_{\kappa},t}|^{2+\delta^*} \sup_t \E^b|\nu_t|^{2+\delta^*}  <  C$.   Because $\nu_t$ is i.i.d, the conditional and unconditional bootstrap moments are the same, and it remains to verify that $\var^b(\mB_{u\vzeta,T}^b(r)) - \var(\mB_{u\vzeta}(r)) = \op(1)$, where $\mB_{u\vzeta}(r)$ was defined just before Lemma \ref{lem6}.
\begin{align*}
\var^b(\mB_{u\vzeta,T}^b(r)) &=  T^{-1} \sum_{t=1}^{[Tr]} (l_{u,t}^2 \vl_{\vzeta,t} \vl_{\vzeta,t}') \E^b (\nu_t^2)  = T^{-1} \sum_{t=1}^{[Tr]} (l_{u,t}^2 \vl_{\vzeta,t} \vl_{\vzeta,t}') \inp r\vrho_{u,\vzeta,0,0}=\var(\mB_{u\vzeta}(r)),
\end{align*}
where $\vrho_{u,\vzeta,0,0}$ was defined before Lemma \ref{lem5}, and the convergence follows by Lemma \ref{lem8}(iv) and Assumption \ref{a8}.

\textbf{Part (ii).} Because $\var^b(\mB_{u\vzeta,T}^b(r))\inp \var(\mB_{u\vzeta}(r))$, using Lemma \ref{lem8}(iii), it follows by \citet{Hansen:1992}, Theorem 2.1, that:
\begin{align}\label{defe23b}
\mathcal{ E}_{2,3}^b =T^{-1/2} \sum_{t \in \tilde I_s} (\mS_{p_2} \mD_{\vzeta,t}  d_{u,t})  l_{u,t}\nu_t \vl_{\vzeta,t} \weaks \mS_{p_2} \int_{\tau_{s-1}}^{\tau_s} d_{u}(\tau) \mD_{\vzeta} (\tau) d \mB_{u\vzeta}(\tau) = \mathbb M_{2,3}^{(1)}(\tau_{s-1},\tau_s),
\end{align}
where $\mathbb M_{2,3}^{(1)}(\tau_{s-1},\tau_s)$ was defined right before Lemma \ref{lem6}. Similarly, it can be shown that for $\mathcal S_{\dag} = \vbeta_{\vx,\#}$, without restrictions on $\vrho_{0,0}$ besides those imposed in Assumption \ref{a8},
\begin{align}\label{defe23b2}
\mathcal{ E}_{2,3}^b = \mathbb M_{2,3}^{(2)}(\tau_{s-1},\tau_s).
\end{align}

Next, consider $\mathcal E_{2,2}^b$. By backward substituting $\xi_{t-1}$, \begin{align} \nonumber
\mathcal E_{2,2}^b &= \mathcal S_{\vr} \mF_s  T^{-1/2} \sum_{t\in \tilde I_s} [\mathcal{S}_{\dag}'(\vg_t \odot \vnu_t)]  \vxi_{t-1} \\  \nonumber
        & =  \mathcal S_{\vr} \mF_s T^{-1/2} [\mathcal{S}_{\dag}'(\vg_t \odot \vnu_t)] \vxi_{[\tau_{s-1}T]} + S_{\vr} \mF_s T^{-1/2} \sum_{t\in \tilde I_s^-} \mF^{\tilde t-1}_s [\mathcal{S}_{\dag}'(\vg_t \odot \vnu_t)] \left[ \vxi_{[\tau_{s-1}T]}\right] \\  \nonumber
&+S_{\vr} \, T^{-1/2} \sum_{t\in \tilde I_s^-} [\mathcal{S}_{\dag}'(\vg_t \odot \vnu_t)] \left[ \left(\sum_{l=0}^{\tilde t-2}\mF_s^l\right)\vmu_s\right] + S_{\vr} \, T^{-1/2} \sum_{t\in \tilde I_s^-} [\mathcal{S}_{\dag}'(\vg_t \odot \vnu_t)] \left[ \sum_{l=0}^{\tilde t-2}\mF_s^l \veta_{t-l-1}\right] \\ \label{decompe22b}
     & = \sum_{i=1}^4 \mathcal E_{2,2,i}^b.
\end{align}

First, note that $\mathcal E_{2,2,1}^b =\op^b(1)$ because, as shown before, $(\vg_t \odot \vnu_t) =\Op^b(1)$ and $\vxi_{[\tau_{s-1}T]}=\Op(1)$ from the proof of Lemma \ref{lem2} in the Supplementary Appendix. Next, because $\sum_{l=0}^{\infty} \| \mF_s^l\|$ is bounded, because $\var^b(\vg_t \odot \vnu_t) =\Op(1)$ and $\vxi_{[\tau_{s-1}T]}=\Op(1)$, $\mathcal E_{2,2,2}^b =\op^b(1)$.

Next, by similar  arguments as for $\mathcal E_{3,3}^b$, and noting that no restrictions are needed on $\vrho_i, \vrho_{ij}$ besides those in Assumption \ref{a8} (because $\mathcal{ E}_{2,3}^b$ has at the basis the same random process as $\mathcal E_1^b$,)
\begin{align}\nonumber
\mathcal{ E}_{2,2,3}^b &= \mathcal S_r \mF_s T^{-1/2} \sum_{t\in \tilde I_s^-} [\mathcal{S}_{\dag}'(\vg_t \odot \vnu_t)] \left[ \left(\sum_{l=0}^{\tilde t-2}\mF_s^l\right)\vmu_s\right] \weaks \\ \nonumber
& =\sum_{l=1}^{\infty}((\mathcal{S}_{\dag}'\mS_{\#})\otimes (\mathcal{S}_{\vr}\mF_s^l)) \, \left(\left[\int_{\tau_{s-1}}^{\tau_s}  \mD_{\#}(\tau)\rd\mB_{0,\#}(\tau)\right] \otimes \vmu_s\right) \\ \label{distribe223b}
&= \mathbb M_1(\tau_{s-1},\tau_s) - ((\mathcal{S}_{\dag}'\mS_{\#})\otimes (\mathcal{S}_{\vr}\mF_s^0)) \, \left(\left[\int_{\tau_{s-1}}^{\tau_s}  \mD_{\#}(\tau)\rd\mB_{0,\#}(\tau)\right] \otimes \vmu_s\right).
\end{align}
Now consider $\mathcal{ E}_{2,2,4}^b$.
\begin{align*}
\mathcal{ E}_{2,2,4}^b = S_{\vr} \, T^{-1/2} \sum_{t\in \tilde I_s^-} [\mathcal{S}_{\dag}'(\vg_t \odot \vnu_t)] \left[ \sum_{l=0}^{\tilde t-2}\mF_s^l \veta_{t-l-1}\right].
\end{align*}

By similar arguments as for $\mathcal E_{3,4}^b$ in \eqref{shortcuts} and just below it,
\begin{align*}
\mathcal{ E}_{2,2,4}^b = S_{\vr} \,\sum_{l=0}^{\Delta\tau_sT-2}  \mF_s^l \, \left(T^{-1/2}\sum_{t\in \tilde I_s^-}   \mathcal{S}_{\dag}'(\vg_t \odot \vnu_t)  \veta_{t-l-1} \right) + \op^b(1) = \mathcal E_{5}^b(\Delta \tau_s T-2)+\op^b(1).
\end{align*}

Next, we analyze $\mathcal E_{5}^b(n^*)$, first for a fixed $n^*$. Note that a crucial term in  $\mathcal E_{5}^b(n^*)$ is $
\mathcal L_5^b(l)=T^{-1/2} \sum_{t=1}^{[Tr]} u_t \nu_t \vg_{t-l}
$
for $l\geq 1$, because $\veta_{t-l} =\mA_s^{-1}\vg_{t-l}$.
By the structure of $\mS$ and $\mD_t$ in \eqref{decompsd}, recalling that $\tilde \vnu_t = \vec(\nu_t \viota_{p_1+1},\viota_{p_2+1})$,
\begin{align*}
\vg_{t-l}  = \begin{bmatrix} d_{u,t-l} l_{u,t-l}  \\
\vs_{p_1} d_{u,t-l} l_{u,t-l} \nu_{t-l} + \mS_{p_1} \mD_{\vv,t-l}  \nu_{t-l} \\
\mS_{p_2} \mD_{\vzeta,t-l} \vl_{\vzeta,t-l} \end{bmatrix}_{\#} = \mS_{\#} \mD_{t-l,\#} ( \vl_{t-l} \odot \tilde \vnu_{t-l})_{\#}.
\end{align*}
Let $\mathcal S_{\dag}=\mathcal S_u$. Then, letting $\mathcal E_{t,l,5} =l_{u,t}\vl_{t-l}$ and $\mathcal E_{t,l,5}^b =  l_{u,t}  \vl_{t-l} \nu_t$,
\begin{align}\label{l5bl}
\mathcal L_5^b(l) = T^{-1/2} \sum_{t=1}^{[Tr]} d_{u,t} \mS_{\#} \mD_{t-l,\#} l_{u,t} \vl_{t-l,\#} \nu_t = T^{-1/2} \sum_{t=1}^{[Tr]} ( d_{ut} \mS_{\#} \mD_{t-l,\#}) (l_{u,t}  \vl_{t-l} \nu_t)_{\#}.
\end{align}

\textbf{Part (i). }First, consider $\mB_{l,T,C}^b(r) = T^{-1/2} \sum_{t=1}^{[Tr]} l_{u,t}  \vl_{t-l} \nu_t$, for $l\geq 1$. Because $\nu_t$ is i.i.d, it is m.d.s under the bootstrap measure conditional on the data, so by arguments similar to before, $\mB_{l,T,C}^b(r) \weaks \mB_{l}^{(1:n)}(r)$, provided that $\var^b(\mB_{l,T,C}^b(r)) \inp \var(\mB_{l}^{(1:n)}(r))$, which we verify below:
\begin{align*}
\var^b(\mB_{l,T,C}^b(r)) = T^{-1} \sum_{t=1}^{[Tr]} l_{u,t}^2 \vl_{t-l} \vl_{t-l}' \E^b(\nu_t^2) =  T^{-1} \sum_{t=1}^{[Tr]} l_{u,t}^2 \vl_{t-l} \vl_{t-l}' \inp \var(\mB_{l}^{(1:n)}(r)) =\textbf{block}_{1,1}(\vrho_{l,l}).
\end{align*}
The previous to last statement above follows by Lemma \ref{lem8}(iv) without restrictions on the form of $\vrho_{l,l}$ besides the ones in Assumption \ref{a8}. Therefore, $\mB_{l,T,C}^b(r) \weaks \mB_{l}^{(1:n)}(r))$.

\textbf{Part (ii). } By \citet{Hansen:1992}, Theorem 2.1, and Lemma \ref{lem8}(iii), $\mathcal L_5^b(l)$ defined in \eqref{l5bl} is such that:\begin{align*}
\mathcal L_5^b(l) \weaks  \int_{0}^r d_u(\tau) \mS_{\#} \mD_{\#}(\tau) d\mB_{l,\#}^{(1:n)}(\tau)=((\mathcal{S}_{\dag}'\mS_{\#})\otimes \mS_{\#}) \int_{0}^{r}  \left(\mD_{\#}(\tau)\otimes\mD_{\#}\left(\tau\right)\right)\rd\mB_{l,\#}(\tau).
\end{align*}
For $\mathcal S_{\dag} = \vbeta_{\vx_,\#}$, the same result can be shown by similar arguments, and with no restrictions on $\vrho_{l,l}$ besides being finite.

Now let $\mathcal S_{\dag} = \mathcal S_u$ again. To derive the limiting distribution of $\mathcal E_{5}^b(n^*)$, we need not only that  $\mB_{l,T,C}^b(r)  \weaks \mB_{l}^{(1:n)}(r))$, but also that $\vec(\mB_{l,T,C}^b(r), \mB_{l^*,T,C}^b(r))  \weaks \vec(\mB_{l}^{(1:n)}(r), \mB_{l^*}^{(1:n)}(r))$, which can be  shown using Lemma \ref{lem3} and Lemma \ref{lem8}(iv), because $\cov^b(\mB_{l,T,C}^b(r),\mB_{l^*,T,C}^b(r) \inp \cov( \mB_{l}^{(1:n)}(r),\mB_{l^*}^{(1:n)}(r) =\textbf{block}_{1,1} (\vrho_{l,l^*})$; the latter condition holds because:
\begin{align*}
\E^b(\mB_{l,T,C}^b(r) (\mB_{l^*,T,C}^b(r))') & = T^{-1} \sum_{t=1}^{[Tr]} l_{u,t}^2  \vl_{t-l} \vl_{t-l^*}' \E^b(\nu_t^2) +T^{-1} \sum_{t,t^*=1,t\neq t^*}^{[Tr]} l_{u,t}  l_{u,t^*} \vl_{t-l} \vl_{t^*-l^*}'\E^b(\nu_t \nu_t^*) \\ & = T^{-1} \sum_{t=1}^{[Tr]} l_{u,t}^2  \vl_{t-l} \vl_{t-l^*}' \inp \textbf{block}_{1,1} (\vrho_{l,l^*}),
\end{align*}
where the last statement follows by Lemma \ref{lem8}(iv), and under Assumptions \ref{a1}-\ref{aboot}. By analogy, no other restrictions besides Assumptions \ref{a1}-\ref{aboot} are needed also when $\mathcal S_{\dag} = \vbeta_{\vx,\#}$.

 Therefore, by \citet{Hansen:1992} and Lemma \ref{lem8}(iii), for a fixed $n^*$,
\begin{align*}
\mathcal E_{5}^b(n^*) & = S_{\vr} \,\sum_{l=0}^{n^*}  \mF_s^l \mA_s^{-1} \, \left(T^{-1/2}\sum_{t\in \tilde I_s^-}   (\mathcal{S}_{\dag}^{\prime}\vg_t^b)  \vg_{t-l-1} \right) \\
& \stackrel{d_p^b}{\Rightarrow} \sum_{l=0}^{n^*}  ((\mathcal{S}_{\dag}'\mS_{\#})\otimes (\mathcal{S}_r \mF_s^{l+1} \mA_s^{-1} \mS_{\#}))  \int_{\tau_{s-1}}^{\tau_s}  \left(\mD_{\#}(\tau)\otimes\mD_{\#}\left(\tau\right)\right)\rd\mB_{l+1,\#}(\tau).
\end{align*}
Letting as before $n^*=T^{\alpha}$, it can be shown that $\mathcal E_{5}^b(\Delta \tau_s T-2) =\mathcal E_{5}^b(n^*) +\op^b(1)$, and therefore
\begin{eqnarray}\nonumber
\mathcal E_{2,2,4}^b &=& \mathcal E_{5}^b(\Delta \tau_s T-2) +\op^b(1) \stackrel{d_p^b}{\Rightarrow} \sum_{l=0}^{\infty}  ((\mathcal{S}_{\dag}'\mS_{\#})\otimes (\mathcal{S}_r \mF_s^{l+1} \mA_s^{-1} \mS_{\#}))  \int_{\tau_{s-1}}^{\tau_s}  \left(\mD_{\#}(\tau)\otimes\mD_{\#}\left(\tau\right)\right)\rd\mB_{l+1,\#}(\tau)\\
&=&\mathbb M_{2,2}(\tau_{s-1},\tau_s),\label{distre5b}
\end{eqnarray}
where $\mathbb M_{2,2}(\tau_{s-1},\tau_s)$ is defined just before Lemma \ref{lem6}. Substituting $\mathcal E_{2,2,1}^b = \op^b(1)$, $\mathcal E_{2,2,2}^b = \op^b(1)$, and \eqref{distribe223b} and \eqref{distre5b} into \eqref{decompe22b}, and then using \eqref{e21b}, it follows that:
\begin{align}\nonumber
&\mathcal E_{2,2}^b \weaks \mathbb M_{2,1}(\tau_{s-1},\tau_s)- ((\mathcal{S}_{\dag}'\mS_{\#})\left(\left[\int_{\tau_{s-1}}^{\tau_s}  \mD_{\#}(\tau)\rd\mB_{0,\#}(\tau)\right] \otimes \vmu_s\right)\otimes (\mathcal{S}_{\vr}\mF_s^0)) +\mathbb M_{2,2}(\tau_{s-1},\tau_s) \\  \label{m21plusm22}
&\mathcal E_{2,1}^b+ \mathcal E_{2,2}^b \weaks  \mathbb M_{2,1}(\tau_{s-1},\tau_s)+\mathbb M_{2,2}(\tau_{s-1},\tau_s),
\end{align}
because the joint convergence of $\mathcal E_{2,3}^b, \mathcal E_{2,2,4}^b$  can be shown as above under Assumptions \ref{a1}-\ref{aboot}. Because all these terms share the same $\nu_t$, it can be shown that they also jointly converge with $\mathcal E_{2,3}^b$ and their bootstrap covariance to the covariances of the relevant limits, under Assumptions \ref{a1}-\ref{aboot}.

Therefore, for $\mathcal S_{\dag}=\mathcal S_u$,
\begin{align*}
\mathcal E_2^b = \sum_{i=1}^3 \mathcal E_{2,i}^b  \weaks \mathbb M_{2,1}(\tau_{s-1},\tau_s)+\mathbb M_{2,2}(\tau_{s-1},\tau_s)+\mathbb M_{2,3}^{(1)}(\tau_{s-1},\tau_s)= \mathbb M_{2}(\tau_{s-1},\tau_s).
\end{align*}
Similarly, for $\mathcal S_{\dag}=\vbeta_{\vx,\#}$, $\mathcal E_2^b \weaks \mathbb M_{2}(\tau_{s-1},\tau_s)$, completing the proof for the distribution of $\mathcal E_2^b$, which note that we proved only under Assumptions \ref{a1}-\ref{aboot}.

Now note that because $\mathcal E_1^b$ featured as part of $\mathcal E_2^b$, their joint convergence was already shown, and recall that it also followed under Assumptions \ref{a1}-\ref{aboot}. It remains to verify the condition:$$\cov^b(\vec(\mathcal E_2^b,\mathcal E_3^b)) -\cov(\vec(\mathbb M_2(\tau_{s-1},\tau_s),\mathbb M_{3}(\tau_{s-1},\tau_s)) \inp 0,$$ because then $\vec_{i=1:3}(\mathcal E_i^b) \weaks \vec_{i=1:3}(\mathbb M_i(\tau_{s-1},\tau_s))=\mathbb M(\tau_{s-1},\tau_s)$. This condition follows by similar arguments as before, if we show that (C 1) $\cov^b( \mathcal E_{2,3}^b,\mathcal E_{3,4}^b)$ converges to the joint covariance of their respective limits, and that (C 2) $\cov^b(\mathcal E_{2,2,4}^b,\mathcal E_{3,4}^b)$ converges to the joint covariance of their respective limits.  For (C 1), by arguments as before, it suffices to show $\cov^b(\mB_{l,T,A}^b(r), \mB_{u\vzeta,T}^b(r)) - \cov(\mB_{l}^{(1:n)}(r), \mB_{u\vzeta}(r)) \inp 0$ (here, we set $\mathcal S_{\dag} = \mathcal S_u$ for all terms and that is why we consider the first $n\times 1$ elements of $\mB_{l}^{b}(r)$; the proofs for the case $\mathcal S_{\dag} = \vbeta_{\vx,\#}$ is similar and is briefly discussed below). Note:
\begin{align*}
&\cov^b(\mB_{l,T,A}^b(r), \mB_{u\vzeta,T}^b(r)) = \E^b( T^{-1} \sum_{t,t^*=1}^{[Tr]} [ (l_{u,t}\vl_{t-l}) \odot \vec(\nu_t\nu_{t-l} \viota_{p_1+1}, \nu_t\viota_{p_2})] \, [l_{u,t} \vl_{\vzeta,t}' \nu_t]  \\
& = T^{-1} \sum_{t=1}^{[Tr]} (l_{u,t}^2 \vl_{t-l} \vl_{\vzeta,t}')\odot  \E^b\begin{bmatrix}
\nu_t^2 \nu_{t-l} \viota_{p_1+1} \viota_{p_2}' \\
\nu_t^2 \viota_{p_2} \viota_{p_2}'
\end{bmatrix}
+ T^{-1} \sum_{t,t^*=1,t\neq t^*}^{[Tr]} (l_{u,t} l_{u,t^*} \vl_{t-l} \vl_{\vzeta,t^*}') \E^b\begin{bmatrix}
\nu_t \nu_{t^*} \nu_{t-l} \viota_{p_1+1} \viota_{p_2}' \\
\nu_t \nu_{t^*} \viota_{p_2} \viota_{p_2}'
\end{bmatrix}\\
& =  T^{-1} \sum_{t=1}^{[Tr]} (l_{u,t}^2 \vl_{t-l} \vl_{\vzeta,t}') \odot \vec(\vzeros_{(p_1+1) \times p_2}, \mJ_2) \inp  T^{-1} \sum_{t=1}^{[Tr]} \E(l_{u,t}^2 \vl_{t-l} \vl_{\vzeta,t}') \odot \vec(\vzeros_{(p_1+1) \times p_2}, \mJ_2) +\op(1),
\end{align*}
which shows why we need $\E( l_{u,t}^2 \vn_{t-l} \vl_{\vzeta,t}')=\vzeros_{(p_1+1) \times p_2}$, imposed in \ref{A8prime}(iii). In the general case, for $b_t$ as defined before, by analogy, the condition needed and imposed in \ref{A8prime}(iii) is that for $l\geq 1$, $\E( (\vn_t \vn_t') \otimes( \vn_{t-l} \vl_{\vzeta,t}')=\vzeros_{(p_1+1)^2 \times (p_1+1)p_2}$.

For (C 2), notice that from \eqref{distre5b},
\begin{align}
\mathcal E_{2,2,4}^b & \weaks \mathbb M_{3,2}(\tau_{s-1},\tau_s)  = \sum_{l=0}^{\infty}  ((\mathcal{S}_{\dag}'\mS_{\#})\otimes (\mathcal{S}_r \mF_s^{l+1} \mA_s^{-1} \mS_{\#}))  \int_{\tau_{s-1}}^{\tau_s}  \left(\mD_{\#}(\tau)\otimes\mD_{\#}\left(\tau\right)\right)\rd\mB_{l+1,\#}(\tau) \\ \nonumber
& = [\viota_{np}' \otimes (\mathcal S_r \mF_s) ] \, \sum_{l=0}^{\infty}  ((\mathcal{S}_{\dag}'\mS_{\#})\otimes (\mF_s^{l} \mA_s^{-1} \mS_{\#}))  \int_{\tau_{s-1}}^{\tau_s}  \left(\mD_{\#}(\tau)\otimes\mD_{\#}\left(\tau\right)\right)\rd\mB_{l+1,\#}(\tau) \\\label{simil1}
& = [\viota_{np}' \otimes (\mathcal S_r \mF_s) ] \, \mathbb P (\tau_{s-1},\tau_s), \mbox{ say},
\end{align}
while from \eqref{mathcale34b},
\begin{align*}
\mathcal E_{3,4}^b & \weaks \sum_{l=0}^{\infty}  ((\mathcal{S}_{\dag}'\mS_{\#})\otimes (\mathcal{S} \mF_s^{l} \mA_s^{-1} \mS_{\#})  \int_{\tau_{s-1}}^{\tau_s}  \left(\mD_{\#}(\tau)\otimes\mD_{\#}\left(\tau\right)\right)\rd\mB_{l+1,\#}(\tau) = [\viota_{np}' \otimes \mathcal S  ] \mathbb P (\tau_{s-1},\tau_s).
\end{align*}
Therefore,  they jointly converge. It follows that for $I_i = \tilde I_s$,
\begin{equation*}
T^{-1/2}\sum_{t\in I_i}\vz_t^b\vg_t^{b'}\mathcal{S}^b_{\dag} = \vec_{i=1:3}(\mathcal E_i^b)\weaks \mathbb M(\tau_{s-1},\tau_s).
\end{equation*}
Using exactly the same arguments as in the end of the proof of Lemma \ref{lem6}, $ T^{-1/2}\sum_{t\in I_i}\vz_t^b\vg_t^{b'}\mathcal{S}^b_{\dag} \weaks \tilde{\mathbb M}_i$ for $I_i \neq \tilde I_s$, completing the proof.

\end{proof}

\begin{proof}[Proof of Lemma \ref{lemzuv_WF}]
        As for the proof of Lemma \ref{lem6}, consider the interval $I_i=\tilde I_s$.
        Let $\mathcal{S}_{\dag}^b=\mathcal{S}_u$ or $\mathcal{S}_{\dag}^b=\hat\vbeta_{\vx,\#}$. We need the asymptotic distribution of $\mathcal{Z}_T^b=T^{-1/2} \sum_{t\in \tilde I_s} \vz_t\vg_t^{b\prime}\mathcal{S}_{\dag}^b$.

        \begin{align*}
        \mathcal{Z}_T^b=T^{-1/2} \sum_{t\in \tilde I_s}\vz_t\vg_t^{b\prime}\;\mathcal{S}_{\dag}^b
        =  \begin{bmatrix}T^{-1/2} \sum_{t \in \tilde I_s}\vg_t^{b\prime}\;\mathcal{S}_{\dag}^b\\ T^{-1/2} \sum_{t \in \tilde I_s}\mathcal{S}_{\vr}\vxi_{t}\vg_t^{b\prime}\;\mathcal{S}_{\dag}^b\\T^{-1/2} \sum_{t \in \tilde I_s}\mathcal{S}\vxi_{t-1} \vg_t^{b\prime}\;\mathcal{S}_{\dag}^b
        \end{bmatrix} \equiv \begin{bmatrix} \mathcal F_1^b \\ \mathcal F_2^b \\ \mathcal F_3^b \end{bmatrix}.
        \end{align*}
Note that $\mathcal F_1^b = \mathcal E_1^b$, and $\mathcal F_2^b = \mathcal E_2^b$, defined in \eqref{defineste} and analyzed in the proof of Lemma \ref{lemzuv_WR}. Also note that, using \eqref{simil1} in the proof of Lemma \ref{lemzuv_WR}, and replacing  as in the proof of Lemma \ref{lemzuv_WR}, estimated parameters with true values because their difference is asymptotically negligible,
\begin{align*}
\mathcal F_3^b &= \mathcal S  \left\{T^{-1/2} \sum_{t\in \tilde I_s} [\mathcal{S}_{\dag}'(\vg_t \odot \vnu_t)]  \vxi_{t-1} \right\} +\op^b(1)\\
\mathcal E_{2,2}^b & = \mathcal S_{\vr} \mF_s \left\{T^{-1/2} \sum_{t\in \tilde I_s} [\mathcal{S}_{\dag}'(\vg_t \odot \vnu_t)]  \vxi_{t-1} \right\} +\op^b(1).
\end{align*}
Since they involve the same underlying random quantity, just scaled differently ($\mathcal S$ versus $\mathcal S_{\vr} \mF_s$), the desired distribution for $\mathcal F_3^b$ follows directly from the analysis of $\mathcal E_{2,2}^b$ in Lemma \ref{lemzuv_WR}. Careful inspection of the proof of Lemma \ref{lemzuv_WR} (focusing on the analysis of $\mathcal E_1^b$ and $\mathcal E_2^b$ only) also shows that $\mathcal{Z}_T^b \weaks \mathbb M(\tau_{s-1},\tau_s)$, and indicates that this result holds under Assumptions \ref{a1}-\ref{aboot}, without the need for \ref{A8prime}. By a similar argument as for the proof of Lemma \ref{lem6} in the Supplemental Appendix, Section \ref{supsec1}, when $I_i \neq \tilde I_s$, $\mathcal{Z}_T^b \weaks \tilde{\mathbb M}_i$, completing the proof.
\end{proof}

\begin{proof}[Proof of Theorem \ref{theo1boot}]\hfill\\
We consider only the WR bootstrap; for the WF bootstrap, the results follows in a similar fashion. Let for simplicity $I_i = I_{i,\vlam_k}$. From \eqref{supwald2}-\eqref{hatMi} and for the Eicker-White estimator $\hat \mM_{(i)}$,
\begin{align}\label{supwalda}
&Wald_{T\vlam_k}= T \, \hat \vbeta_{\vlam_k}' \,
        \mR_k' \,\left(\mR_k \hat \mV_{\vlam_k} \mR_k'\right)^{-1}\mR_k \,\hat \vbeta_{\vlam_k}, \mbox{ where } \hat{\mV}_{\vlam_k}=\diag_{i=1:k+1}(\hat \mQ_{(i)}^{-1}  \ \hat \mM_{(i)} \ \hat \mQ_{(i)}^{-1})\, \\ \nonumber
        & \hat \mQ_{(i)} = T^{-1} \sum_{t \in I_{i}} \hat \mUpsilon_t' \vz_t \vz_t'\hat \mUpsilon_t\, ,  \mbox{ and }\, \hat{\mM}_{(i)} = \widehat{EW}\left[\,\hat\mUpsilon_t^\prime \vz_t(\hat u_t+\hat\vv_t^\prime \hat \vbeta_{\vx,(i)});\,I_{i}\,\right].
         \end{align}
From \eqref{supwald2b}-\eqref{hatMib},
\begin{align}\label{supwaldab}
&Wald^b_{T\vlam_k}= T \, \hat \vbeta_{\vlam_k}^{b'} \,
        \mR_k' \,\left(\mR_k \hat \mV_{\vlam_k}^b \mR_k'\right)^{-1}\mR_k \,\hat \vbeta_{\vlam_k}^b \mbox{ where }\hat{\mV}_{\vlam_k}^b=\diag_{i=1:k+1}((\hat \mQ_{(i)}^b)^{-1}   \hat{\mM}_{(i)}^b \ (\hat \mQ_{(i)}^b)^{-1}) \\ \nonumber
        &\hat \mQ_{(i)}^b = T^{-1} \sum_{t \in I_{i}} \hat \mUpsilon_{t}^{b\prime} \vz_t^b \vz_t^{b\prime}\hat \mUpsilon_{t}^b\, ,  \mbox{ and }\hat{\mM}_{(i)}^b = \widehat{EW}\left[\,\hat\mUpsilon_t^{b\prime} \vz_t^b(u_t^b+\vv_t^{b\prime} \hat \vbeta_{\vx,(i)}^b);\,I_{i}\,\right].
   \end{align}
From Lemma \ref{lem2}, $\hat \mQ_{(i)} \inp \mathbb Q_i$ and from Lemmas \ref{lem9} and \ref{lemzuv_WF} \, $\hat \mQ_{(i)}^b \stackrel{p^b}{\rightarrow} \mathbb Q_i$.

 Now consider $\hat\vbeta_{\vlam_k}=\vec(\hat\vbeta_{i,\vlam_k})$. Let $\hat \mQ_{j^*} =T^{-1}\sum_{t\in I_j^*}\vz_t\vz_t^\prime$. By Lemma \ref{lem2}, $\hat \mQ_{j^*} \inp \int_{\pi_{j-1}^0}^{\pi_j^0} \mathbb Q_{\vz}(\tau)\rd\tau =\mathbb Q_{\vz,j^*} $. Therefore, from the proof of Theorem B \ref{theo_0vsk} in the Supplementary Appendix, Section \ref{supsec1},
                \begin{eqnarray}\nonumber
                        T^{1/2} (\hat\vbeta_{i,\vlam_k}-\vbeta^0) &=& \mathbb Q_i^{-1}\mUpsilon_t^{0'} \left(T^{-1/2} \sum_{t\in I_i} \vz_t u_t + T^{-1/2} \sum_{t\in I_i} \vz_t \vv_t^\prime \vbeta_{\vx}^0\right.\\ \label{beta1}
                        &-&\left.\;\; T^{-1} \sum_{t\in I_i}  \vz_t\vz_t^\prime\left\{\,\sum_{j=1}^{h+1}1_{t\in I_j^*} \mathbb Q_{\vz,j^*}^{-1} T^{-1/2} \sum_{t\in I_j^*} \vz_t \vv_t^\prime\vbeta_{\vx}^0\,\right\}\right)\,+\,o_p(1).
        \end{eqnarray}
From Lemma \ref{lem9} and \ref{lemzuv_WR},  $\hat \mUpsilon_t^b = \mUpsilon_t^{0}+\op^b(1)$. Also, $\hat \mQ_{j^*}^b =T^{-1}\sum_{t\in I_j^*}\vz_t^b\vz_t^{b\prime} \ins \mathbb Q_{\vz,j^*} $ (in probability) by the proof of Lemma \ref{lem9}, therefore:
\begin{eqnarray}\nonumber
                        T^{1/2} (\hat\vbeta_{i,\vlam_k}^b-\vbeta^0) &=& \mathbb Q_i^{-1} \mUpsilon_t^{0'} \left(T^{-1/2} \sum_{t\in I_i} \vz_t^b (u_t^b +\vv_t^{b\prime} \vbeta_{\vx}^0)\right.\\\label{beta2}
                        &-&\left.\;\; T^{-1} \sum_{t\in I_i} \vz_t^b\vz_t^{b\prime}\left\{\,\sum_{j=1}^{h+1}1_{t\in I_j^*}\mathbb Q_{\vz,j^*}^{-1} T^{-1/2} \sum_{t\in I_j^*} \vz_t^b \vv_t^{b\prime}\vbeta_{\vx}^0\,\right\}\right)\,+\,o_p^b(1).
        \end{eqnarray}
From Lemma \ref{lem6} and \ref{lemzuv_WR}, we have that $T^{-1/2} \sum_{t\in I_i} \vz_t^b \vv_t^{b\prime} \vbeta_{\vx}^0  - T^{-1/2} \sum_{t\in I_i} \vz_t \vv_t^{\prime} \vbeta_{\vx}^0 =\op^b(1)$ and $T^{-1/2} \sum_{t\in I_i} \vz_t^b u_t^b - T^{-1/2} \sum_{t\in I_i} \vz_t u_t =\op^b(1)$. Therefore, from \eqref{beta1}-\eqref{beta2}, $T^{1/2} (\hat\vbeta_{i,\vlam_k}^b- \hat\vbeta_{i,\vlam_k})=\op^b(1)$.

Because $\hat \mM_{(i)}$ and $\hat \mM_{(i)}^b$ estimate the same part of the variance of $T^{1/2} (\hat\vbeta_{i,\vlam_k}-\vbeta_{\vx}^0)$, and $T^{1/2} (\hat\vbeta_{i,\vlam_k}^b-\vbeta_{\vx}^0)$ respectively, from Lemma \ref{lem2}, \ref{lem6}, \ref{lem9}, and \ref{lemzuv_WR}, it follows that $\hat \mM_{(i)}^b-\hat \mM_{(i)} =\op^b(1)$. Putting these results together,
$\sup_{c\in\mathbb{R}}\left| P^{b} \left(\sup\text{-}Wald^b_T\leq c\right)-P(\sup\text{-}Wald_T\leq c)\right|\inp 0$ as $T\rightarrow\infty$.
\end{proof}
 \begin{proof}[Proof of Theorem \ref{theo2boot}]\hfill\\
Inspecting the alternative representation of the $sup\text{-}Wald_T(\ell+1|\ell)$ in the proof of Theorem B \ref{theo_ellvsell+1} in the Supplementary Appendix, Section \ref{supsec1}, and defining the same representation for $sup\text{-}Wald_T^b(\ell+1|\ell)$, the desired result follows using the same steps as in the proof of Theorem \ref{theo1boot}.
\end{proof}

\bibliography{hallhan5}
\newpage
\noindent
%

%

\begin{center}
\Large{Supplementary Appendix for $``$Bootstrapping Structural Change Tests$"$}\\
\vspace{0.2in}
\large{Otilia Boldea,} \small{Tilburg University}\\ \large{Adriana Cornea-Madeira,} \small{University of York}\\\large{Alastair R. Hall,} \small{University of Manchester}\\
\end{center}
\vspace{0.3in}
 \thispagestyle{empty}

\normalsize
In the proofs below we refer to the equation numbers from the paper and to
the equation numbers from this Supplemental Appendix by (eqnnumber) and respectively
(sectionnumber.eqnnumber). Also the tables from this Supplementary Appendix are referred to by sectionnumber.tablenumber, and the tables from the paper are referred to by tablenumber. All the sections referenced below refer to Supplementary Appendix sections.
\setcounter{section}{0}
\setcounter{equation}{0}
\setcounter{table}{0}
\renewcommand{\thesection}{\arabic{section}}
\renewcommand{\thetable}{\arabic{section}.\arabic{table}}
\renewcommand{\theequation}{\arabic{section}.\arabic{equation}}

\section{Proofs of Lemmas \ref{lem2}, \ref{lem4}-\ref{lem8} and asymptotic distribution of the $\sup\text{-}Wald$ tests}\label{supsec1}

\begin{proof}[Proof of Lemma \ref{lem2}]\hfill\\
        To begin, consider the case where $I_i=\left[\,[\tau_{s-1}T]+1,[\tau_s T]\right]$\black{, where $[[\tau_{s-1}T]+1,[\tau_s T]]$ are the intervals for which the coefficients in the VAR($p$) representation of $\tilde \vz_t$ are stable, and $s=1,\ldots,N$, with $N$ being the total number of breaks in the slope coefficients of the VAR($p$) representation for $\tilde \vz_t$ following from Assumptions \ref{a1}-\ref{a5}}. Then, letting $I_i=\tilde{I}_s \stackrel{d}{=}\left[\,[\tau_{s-1}T]+1,[\tau_s T]\right]$, we have:
        $$
        \hat \mQ_{\vz,(i)}  \;=\; T^{-1} \sum_{t \in \tilde I_s} \vz_t \vz_t^\prime \;=\;\mathcal{A}\,+\,o_p(1),
        $$
        where
        \begin{eqnarray*}
                \mathcal{A}&=&\,T^{-1} \sum_{t \in \tilde I_s}  \left[\,\begin{array}{ccc}1 &(\mathcal{S}_{\vr}\vxi_t)'&(\mathcal{S}\vxi_{t-1})^\prime\\
                        \mathcal{S}_{\vr}\vxi_t&\mathcal{S}_{\vr}\vxi_t(\mathcal{S}_{\vr}\vxi_t)'&\mathcal{S}_{\vr}\vxi_t(\mathcal{S}\vxi_{t-1})\prime\\\mathcal{S}\vxi_{t-1} &\mathcal{S}\vxi_{t-1}(\mathcal{S}_{\vr}\vxi_t)'&\mathcal{S}\vxi_{t-1}(\mathcal{S}\vxi_{t-1})^\prime \end{array}\,\right]\,=\left[\,\begin{array}{ccc} \Delta\tau_s &\mathcal{A}_1^\prime\mathcal{S}_{\vr}^\prime&\mathcal{A}_2'\mathcal{S}'\\ \mathcal{S}_{\vr}\mathcal{A}_1 & \mathcal{S}_{\vr}\mathcal{B}_1\mathcal{S}_{\vr}^\prime&\mathcal{S}_{\vr}\mathcal{B}_{2}\mathcal{S}'\\
                        \mathcal{S}\mathcal{A}_2&\mathcal{S}\mathcal{B}_{2}'\mathcal{S}_{\vr}'&\mathcal{S}\mathcal{B}_3\mathcal{S}'\end{array}\,\right],
        \end{eqnarray*}
        where
        $$
        \mathcal{A}_1\;=\;\,T^{-1} \sum_{t \in \tilde I_s}\vxi_t,\;\; \mathcal{A}_2\;=\;\,T^{-1} \sum_{t \in \tilde I_s}\vxi_{t-1},\;\;
        \mathcal{B}_1\;=\;T^{-1} \sum_{t \in \tilde I_s} \vxi_t\vxi_t^\prime,\;\; \mathcal{B}_{2}\;=\;T^{-1} \sum_{t \in \tilde I_s} \vxi_t\vxi_{t-1}^\prime,\;\;\mathcal{B}_3\;=\;T^{-1} \sum_{t \in \tilde I_s} \vxi_{t-1}\vxi_{t-1}^\prime.
        $$

        Consider $\mathcal{A}_1$. We have
        $
        \mathcal{A}_1\;=\;\sum_{i=1}^4 \mathcal{A}_{1,i},
        $
        \black{where below we denote $\Delta\tau_sT=[\tau_sT]-[\tau_{s-1}T]$,} and:
        \begin{eqnarray*}
                \mathcal{A}_{1,1}&=&T^{-1}\sum_{t=[\tau_{s-1}T]+1}^{[\tau_sT]}\tilde{\vxi}_t\\[0.1in]
                \mathcal{A}_{1,2}&=&T^{-1}\sum_{t=[\tau_{s-1}T]+1}^{[\tau_sT]}\left(\sum_{l=0}^{t-[\tau_{s-1}T]-1}\mF_s^l\right)\vmu_s-\mathcal{A}_{1,4}\\
                &=& T^{-1}\left(\sum_{l=0}^0\mF_s^l+\sum_{l=0}^1\mF_s^l+\ldots+\sum_{l=0}^{\Delta\tau_sT-1}\mF_s^l\right)\vmu_s-\mathcal{A}_{1,4}\\
                &=& T^{-1}\left(\Delta\tau_sT+(\Delta\tau_sT-1)\mF_s+(\Delta\tau_sT-2)\mF_s^2+\ldots+(\Delta\tau_sT-(\Delta\tau_sT-1))\mF_s^{\Delta\tau_sT-1}\right)\vmu_s-\mathcal{A}_{1,4}\\
                &=&T^{-1}\Delta\tau_sT\sum_{l=0}^{\Delta \tau_sT -1}\mF_s^l\,\vmu_s,\\[0.1in]
                \mathcal{A}_{1,3}&=&T^{-1}\left(\sum_{t=[\tau_{s-1}T]+1}^{[\tau_sT]}  \mF^{t-[\tau_{s-1}T]}_s \right)\vxi_{[\tau_{s-1}T]}\\
                &=&T^{-1}\sum_{l=1}^{\Delta \tau_s T} \mF_s^l\,\vxi_{[\tau_{s-1}T]},\\[0.1in]
                \mathcal{A}_{1,4}&=&-T^{-1}\left(\sum_{l=1}^{\Delta \tau_s T -1} l\mF_s^l\right)\,\vmu_s.
        \end{eqnarray*}
        Let $\tilde{\vxi}_{t,i}$ be the $i^{th}$ element of $\tilde{\vxi}_t$. From Assumptions \ref{a7} and \ref{a8}, $\{\tilde{\vxi}_{t,i},\mathcal{F}_t\}$ is a $L^1$-mixingale satisfying the conditions of Lemma \ref{lem1} and so $\mathcal{A}_{1,1}\stackrel{p}{\to}0$.  From Assumption \ref{a7}, it follows that $\mathcal{A}_{1,2}\to \Delta \tau_s(\mI_{np}-\mF_s)^{-1}\vmu_s$. From Assumptions \ref{a7} and \ref{a8}, it follows that $\var[\vxi_{[\tau_{s-1}T]}]=O(1)$ and so, again using Assumption \ref{a7}, it follows that $\mathcal{A}_{1,3}=o_p(1)$.  From Assumption \ref{a7}, it follows that $\sum_{l=1}^{\Delta \tau_s T -1} l\mF_s^l=O(1)$ and hence $\mathcal{A}_{1,4}=o(1)$. Combining these results, we obtain:
        $
        \mathcal{A}_1\;=\;\mathbb{A}(\tau_{s-1},\tau_s)\,+\,o_p(1),
        $
        where
        \begin{equation}
        \mathbb{A}(\tau_{s-1},\tau_s)\;=\;\Delta \tau_s(\mI_{np}-\mF_s)^{-1}\vmu_s = \int_{\tau_{s-1}}^{\tau_s} \mathbb Q_1(\tau).\label{abb}
        \end{equation}
        \vspace*{0.1in}
        By similar arguments, we have
        $
        \mathcal{A}_2\;=\;\mathbb{A}(\tau_{s-1},\tau_s)\,+\,o_p(1).      $
        Now consider $\mathcal{B}_1$. Since
        \begin{eqnarray*}
                \mathcal{B}_1&=&T^{-1} \sum_{t \in \tilde I_s}\left(\mF^{t-[\tau_{s-1}T]}_s \vxi_{[\tau_{s-1}T]} + \tilde\vxi_t+\left(\sum_{l=0}^{t-[\tau_{s-1}T]-1}\mF_s^l\right)\vmu_s\right)\left(\mF^{t-[\tau_{s-1}T]}_s \vxi_{[\tau_{s-1}T]} + \tilde\vxi_t+\left(\sum_{l=0}^{t-[\tau_{s-1}T]-1}\mF_s^l\right)\vmu_s\right)^\prime,
        \end{eqnarray*}

        we have
      $
        \mathcal{B}_1\;=\;\sum_{i=1}^3\mathcal{B}_{1,i}\,+\,\sum_{j=1}^3 \left\{\,\mathcal{B}_{1,3+j} + \mathcal{B}_{1,3+j}^\prime\,\right\},
 $
        where, setting $\tilde{t}=t-[\tau_{s-1}T]$,
        \begin{eqnarray}
        \mathcal{B}_{1,1}&=&T^{-1}\sum_{t\in I_i} \tilde{\vxi}_t\tilde{\vxi}_t^\prime=T^{-1}\sum_{t\in I_i}\left(\,\sum_{l=0}^{\tilde{t}-1}\mF_s^l\veta_{t-l}\,\right)\left(\,\sum_{l=0}^{\tilde{t}-1}\mF_s^l\veta_{t-l}\,\right)^\prime\nonumber\\[0.1in]
        \mathcal{B}_{1,2}&=&T^{-1}\sum_{t\in I_i} \left( \sum_{l=0}^{\tilde{t}-1}\mF_s^l\vmu_s\right)\,\left(\,\sum_{l=0}^{\tilde{t}-1}\mF_s^l\vmu_s\right)^\prime\nonumber\\[0.1in]
        \mathcal{B}_{1,3}&=&T^{-1}\sum_{t\in I_i} \mF_s^{\tilde{t}}\vxi_{[\tau_{s-1}T]}\vxi_{[\tau_{s-1}T]}^\prime \mF_s^{\tilde{t}\prime} \nonumber \\[0.1in]
        \mathcal{B}_{1,4}&=&T^{-1}\sum_{t\in I_i} \tilde{\vxi}_t \left(\,\sum_{l=0}^{\tilde{t}-1}\mF_s^l\vmu_s\right)^\prime \nonumber \\[0.1in] \nonumber
        \mathcal{B}_{1,5}&=&T^{-1}\sum_{t\in I_i} \tilde{\vxi}_t \vxi_{[\tau_{s-1}T]}^\prime \mF_s^{\tilde{t}'}\\[0.1in] \nonumber
        \mathcal{B}_{1,6}&=&T^{-1}\sum_{t\in I_i} \left(\,\sum_{l=0}^{\tilde{t}-1}\mF_s^l\vmu_s\right)\vxi_{[\tau_{s-1}T]}^\prime \mF_s^{\tilde{t}'}\label{b16}.
        \end{eqnarray}

Note that
        $
        \mathcal{B}_{1,1}\;=\;\sum_{i=1}^3 \mathcal{B}_{1,1}^{(i)}
        $
        where
        \begin{eqnarray*}
                \mathcal{B}_{1,1}^{(1)}&=&T^{-1}\sum_{t\in I_i}\sum_{l=0}^{\tilde{t}-1}\mF_s^l \mOmega_{t-l\mid t-l-1}\mF_s^{l\prime},\\[0.1in]
                \mathcal{B}_{1,1}^{(2)}&=&T^{-1}\sum_{t\in I_i}\sum_{l=0}^{\tilde{t}-1}\mF_s^l \left(\,\veta_{t-l}\veta_{t-l}^\prime -\mOmega_{t-l\mid t-l-1}\,\right)\mF_s^{l\prime},\\[0.1in]
                \mathcal{B}_{1,1}^{(3)}&=&T^{-1}\sum_{t\in I_i}\sum_{l,j=0,l\neq j}^{\tilde{t}-1} \mF_s^l \veta_{t-l}\veta_{t-j}^\prime \mF_s^{j\prime}.
        \end{eqnarray*}
        It can be shown that         $
        \mathcal{B}_{1,1}^{(1)}\;=\;\mathcal{C}_1\,+\,\mathcal{C}_2,
        $
        where
        \begin{eqnarray*}
                \mathcal{C}_1&=&\sum_{l=0}^{(\Delta \tau_sT)-1} \mF_s^l\,\left\{\, T^{-1}\sum_{t\in I_i}\mOmega_{t\mid t-1}\,\right\}\,\mF_s^{l\prime}\\[0.1in]
                \mathcal{C}_2&=&-T^{-1}\sum_{l=1}^{(\Delta \tau_sT)-1} \mF_s^l\,\left(\,\sum_{j=[\tau_s T] -l+1}^{[\tau_s T]}\mOmega_{j\mid j-1} \,\right)\mF_s^{l\prime}.
        \end{eqnarray*}
        Under our Assumptions, \black{it is shown in Section \ref{supsec2} that} $\mathcal{C}_2=o_p(1)$. Furthermore, we have:
        \begin{eqnarray}
        \mathcal{C}_1\;\stackrel{p}\to\;\sum_{l=0}^{\infty} \mF_s^l \left[\begin{array}{cc} \mA_s^{-1}\int_{\tau_{s-1}}^{\tau_s}\overline \mSigma(\tau)\rd \tau \mA_s^{-1'} &\vzeros_{n\times n(p-1)}\\\vzeros_{n\times n(p-1)}' & \vzeros_{n(p-1)\times n(p-1)}\end{array}\,\right] \mF_s^{l\prime}=\mathbb{B}_1(\tau_{s-1},\tau_s),\label{bb1}
        \end{eqnarray}
        where $\overline \mSigma (a)=\mS \{\mD(a)\}^2 \mS^\prime$, and $\mD(a) = \diag_{i=1:n}(d_i(a))$, with $d_i(a)$ defined in Assumption \ref{a8}.

        Now consider $\mathcal{B}_{1,1}^{(2)}$. Recall that the only non-zero elements of $\veta_{t-l}\veta_{t-l}^\prime -\mOmega_{t-l\mid t-l-1}$ are in the upper left-hand block and take the form:
        $$
        \mA_s^{-1}\left(\vepsi_{t-l} \vepsi_{t-l}^\prime\,-\, \overline{\mSigma}_{t-l\mid t-l-1}\right)\mA_s^{-1'}.
        $$
        Since each element of the matrix $\vepsi_{t} \vepsi_{t}^\prime-\E[\vepsi_{t} \vepsi_{t}^\prime|\mathcal{F}_{t-1}]$ is a mean-zero m.d.s., and has uniformly bounded $(2+\delta)$ moments (Assumption \ref{a8}), it follows that each element of the matrix $\sum_{l=0}^{\tilde{t}-1}\mF_s^l \left(\,\veta_{t-l}\veta_{t-l}^\prime -\mOmega_{t-l\mid t-l-1}\,\right)\mF_s^{l\prime}$ is a $L^1$-mixingale with constants that are uniformly bounded. Therefore, by Lemma \ref{lem1}, it follows that $\mathcal{B}_{1,1}^{(2)}\stackrel{p}{\to}0$.

        Now consider $\mathcal{B}_{1,1}^{(3)}$.  Under Assumptions \ref{a7} and \ref{a8}, each element of the matrix $\sum_{\ell,j=0,\ell\neq j}^{\tilde{t}-1}\mF_s^\ell\veta_{t-\ell}\veta_{t-j}^\prime \mF_s^{j\prime}$ is a $L^1$-mixingale satisfying the conditions of Lemma \ref{lem1}, so $\mathcal{B}_{1,1}^{(3)}\stackrel{p}{\to}0$.
        Therefore,
        \begin{equation}\label{bb1a}
        \mathcal B_{1,1} \inp \mathbb{B}_1(\tau_{s-1},\tau_s).
        \end{equation}

        From Assumption \ref{a7} and Section \ref{supsec2}, it follows that:
        \begin{equation}
        \label{bb2}
        \mathcal{B}_{1,2}\;\to\;\Delta \tau_s (\mI_{np}-\mF_s)^{-1}\vmu_s\vmu_s^\prime(\mI_{np}-\mF_s)^{-1^\prime} = \int_{\tau_{s-1}}^{\tau_s} \mathbb Q_1(\tau)\mathbb Q_1'(\tau) =\mathbb{B}_2(\tau_{s-1},\tau_s).
        \end{equation}

        Now consider $\mathcal{B}_{1,3}$. Under Assumptions \ref{a7} and \ref{a8}, it follows that $\vxi_{[\tau_{s-1}T]}\vxi_{[\tau_{s-1}T]}^\prime=O_p(1)$ and so from Assumption \ref{a7},  $\mathcal{B}_{1,3}=o_p(1)$.
        Next, consider $\mathcal{B}_{1,4}$. Under Assumptions \ref{a7} and \ref{a8},  $\tilde{\vxi}_t \left(\,\sum_{l=0}^{\tilde{t}-1}\mF_s^l\vmu_s\right)^\prime$ is a m.d.s. \black{with uniformly bounded $(2+\delta)$ moments}, so $\mathcal{B}_{1,4}=\op(1)$.         Next, consider $\mathcal{B}_{1,5}$. Under Assumptions \ref{a7} and \ref{a8}, it can be shown that $\tilde\vxi_t$ is a $L^1$-mixingale satisfying Lemma \ref{lem1}. Hence $$\mathcal{B}_{1,5}=T^{-1}\sum_{t\in I_i}   \sum_{l=0}^{\Delta\tau_{s}T-1}\mF_s^l\veta_{t-l}O_p(1)\mF_s^{\tilde{t}'}=\op(1).$$
        \noindent
        Finally, it is shown in Section \ref{supsec2} that under Assumption \ref{a7}, $\mathcal{B}_{1,6}=\op(1)$.
        %
        %
        %
        %
        %
        %
        %

        Combining these results, we obtain:
        $$
        \mathcal{B}_1\;\stackrel{p}{\to}\; \mathbb{B}(\tau_{s-1},\tau_s),
        $$
        where
        \begin{equation}
        \mathbb{B}(\tau_{s-1},\tau_s)\;=\;\mathbb{B}_1(\tau_{s-1},\tau_s)+\mathbb{B}_2(\tau_{s-1},\tau_s).\label{mathbbb}
        \end{equation}
        \black{Next, define $\tilde I_s^{-}=\{[\tau_{s-1}T]+2,[\tau_{s-1}T]+3,\ldots,[\tau_sT]\}$, and note that $\tilde I_s=\tilde I_s^{-}\cup\{[\tau_{s-1}T]+1\}$. Similarly to $\mathcal{B}_1$,}
        \begin{align}\nonumber
        \mathcal{B}_3&= T^{-1} \sum_{t \in \tilde I_s} \vxi_{t-1}\vxi_{t-1}^\prime = T^{-1} \sum_{t\in \tilde I_s^{-}} \vxi_{t-1}\vxi_{t-1}^\prime + T^{-1} \vxi_{[\tau_{s-1}T]-1}\vxi_{[\tau_{s-1}T]-1}' \\ \nonumber
        & =T^{-1} \sum_{t \in \tilde I_s} \vxi_{t}\vxi_{t}^\prime - T^{-1} \vxi_{[\tau_sT]}\vxi_{[\tau_sT]}'+T^{-1} \vxi_{[\tau_{s-1}T]-1}\vxi_{[\tau_{s-1}T]-1}' \\ \label{eqB3}
        & =  \mathbb{B}(\tau_{s-1},\tau_s) +\op(1)- T^{-1} \vxi_{[\tau_sT]}\vxi_{[\tau_sT]}'+T^{-1} \vxi_{[\tau_{s-1}T]-1}\vxi_{[\tau_{s-1}T]-1}.
        \end{align}
        \black{We now show that $\sup_t \E \|\vxi_t \vxi_t' \| <c$ for some $c>0$. By backward substitution of the first regime $s=1$, $\vxi_t = \mF^{t}_s \vxi_{0} + \left(\sum_{l=0}^{t-1}\mF_s^l\right)\vmu_s+ \tilde \vxi_t$. Let $\vc_t = \left(\sum_{l=0}^{t-1}\mF_s^l\right)\vmu_s$, and note that $ \|\vc_t\|< c_0$ for some $c_0>0$, by arguments similar to Section \ref{supsec2}. Also,
                \begin{align*}
                \E\| \vxi_t \vxi_t' \| & \leq \E\| \mF^{t}_s \vxi_{0}\vxi_{0}' \mF^{t'}_{s}\| + \| \vc_t \vc_t'\| +  \E\| \tilde \vxi_t \tilde \vxi_t'\| + 2 \E \|\mF^{t}_s \vxi_{0} \vc_t' \| + 2 \E\| \vc_t \tilde \vxi_t'\| + 2\E\|\tilde \vxi_t  \vxi_{0}'\mF^{t'}_s \|.
                \end{align*}
                First, $ \sup_t \E\| \mF^{t}_s \vxi_{0}\vxi_{0}' \mF^{t'}_{s}\| \leq \sup_t \| \mF_s^t\|^2\, \E \| \vxi_0\|^2 < c_1$ by Assumptions \ref{a7} and \ref{a8} for some $c_1>0$. Second, $\sup_t \|\vc_t\vc_t'\| < \sup_t \|\vc_t\|^2< c_0^2=c_2$. Third, $\sup_t \E \|\veta_{t-l} \veta_{t-j}\| \leq \sup_t (\E\|\veta_{t-l}\|^2)^{1/2} \sup_t(\E\|\veta_{t-j}\|^2)^{1/2}< c^*$ for some $c^*>0$ by Assumption \ref{a8}. Therefore, $ \sup_t \E\| \tilde \vxi_t \tilde \vxi_t'\| \leq  \sup_t \left( \sum_{l,j=0}^{t-1} \| \mF_s^l \| \, \, \| \mF_s^j\| \, \, \sup_t \E\| \veta_{t-l} \veta_{t-j}'\|  \right) \leq c^* \sup_t \left(\sum_{l,j=0} ^{t-1}\| \mF_s^l \| \, \, \| \mF_s^j\| \right)<c_3 $ for some $c_3>0$. Similarly,  $\sup_t \left( 2 \E \|\mF^{t}_s \vxi_{0} \vc_t' \| + 2 \E\| \vc_t \tilde \vxi_t'\| + 2\E\|\tilde \vxi_t  \vxi_{0}'\mF^{t'}_s \| \right)<c_4$ for some $c_4>0$. Therefore,
                \begin{equation}\label{momvxi}
                \sup_t \E\| \vxi_t \vxi_t' \| <c_1+c_2+c_3+c_4 = c <\infty,
                \end{equation}so,  by Markov's inequality, $ T^{-1} \vxi_{[\tau_sT]}\vxi_{[\tau_sT]} =\op(1)$ and $T^{-1} \vxi_{[\tau_{s-1}T]-1}\vxi_{[\tau_{s-1}T]-1} =\op(1)$ for $s=1$. For the other regimes $s>1$, repeated backward substitution yields the same result. Substituting this result into \eqref{eqB3}, it follows that $\mathcal{B}_3\;\stackrel{p}{\to}\; \mathbb{B}(\tau_{s-1},\tau_s)$. }

        \black{Now consider $\mathcal{B}_{2}$:
                \begin{eqnarray} \nonumber
                \mathcal{B}_{2}&=& T^{-1} \sum_{t\in \tilde I_s} \vxi_t \vxi_{t-1}' =  T^{-1} \sum_{t\in \tilde I_s} (\vmu_s+ \mF_s \vxi_{t-1} +\veta_t) \vxi_{t-1}' \\ \label{eqB22}
                &=& \vmu_s T^{-1} \sum_{t\in \tilde I_s} \vxi_{t-1}' + T^{-1} \sum_{t\in \tilde I_s} \veta_t \vxi_{t-1} + \mF_s \mathbb{B}(\tau_{s-1},\tau_s) +\op(1).
                \end{eqnarray}
                First, note that $\vmu_s T^{-1} \sum_{t\in \tilde I_s} \vxi_{t-1}' = \vmu_s \mathcal A_2' = \vmu_s \mathbb A'(\tau_{s-1},\tau_s)+\op(1)$. Next, $\veta_t \vxi_{t-1}'$ is a m.d.s. sequence, so each of its element is $L^1$-mixingale with bounded constants as defined in Lemma \ref{lem1}. Also, $\sup_t \E \| \veta_t \vxi_{t-1}'\|^{1+\delta/2} \leq \sup_t (\E \| \veta_t\|^{2+\delta})^{\frac{1}{2+\delta}} \sup_t (\E \| \vxi_{t}\|^{2+\delta})^{\frac{1}{2+\delta}} <\infty$, so $\veta_t \vxi_{t-1}$ satisfies Lemma \ref{lem1} element-wise, so $T^{-1} \sum_{t\in \tilde I_s} \veta_t \vxi_{t-1}' \inp 0$. Substituting these results into \eqref{eqB22}, we obtain:
                \begin{eqnarray}\label{eqB2}
                \mathcal{B}_{2}&\inp & \vmu_s \mathbb A'(\tau_{s-1},\tau_s)+ \mF_s \mathbb{B}(\tau_{s-1},\tau_s).
                \end{eqnarray}
                Therefore,}
        $$
        \mathcal{A}\;\stackrel{p}{\to}\;\left[\,\begin{array}{ccc} \Delta\tau_s &\mathbb{A}'(\tau_{s-1},\tau_s)\mathcal{S}_{\vr}^\prime &\mathbb{A}'(\tau_{s-1},\tau_s)\mathcal{S}^\prime\\\mathcal{S}_{\vr}\mathbb{A}(\tau_{s-1},\tau_s)&\mathcal{S}_{\vr}\mathbb{B}(\tau_{s-1},\tau_s)\mathcal{S}_{\vr}^\prime&\mathcal{S}_{\vr}\left( \vmu_s \mathbb A'(\tau_{s-1},\tau_s)+ \mF_s \mathbb{B}(\tau_{s-1},\tau_s)\right)\mathcal{S}^\prime \\
        \mathcal{S}\mathbb{A}(\tau_{s-1},\tau_s) &\mathcal{S}\left( \vmu_s \mathbb A'(\tau_{s-1},\tau_s)+ \mF_s \mathbb{B}(\tau_{s-1},\tau_s)\right)' \mathcal{S}_{\vr}^\prime& \mathcal{S}\mathbb{B}(\tau_{s-1},\tau_s)\mathcal{S}^\prime\end{array}\,\right]\,.
        $$
        Now consider the case of $I_i$ containing $N_i$ breaks from the total set of $N$ breaks, that is, there is an $s$ such that $\tau_{s-1}<\lambda_{i-1}\leq \tau_s$ and $\tau_{s+N_i}\leq\lambda_i<\tau_{s+N_i+1}$. Then, generalizing the previous results which were for $I_i=\tilde I_s=[[\tau_{s-1} T]+1, [\tau_s T]]$, we have by similar arguments that:
        $
        \hat{\mQ}_{\vz,(i)}\;\stackrel{p}{\to}\; \int_{\lambda_{i-1}}^{\lambda_{i}}\mathbb{Q}_{\vz}(\tau)\rd\tau .
        $

By Lemma \ref{lem7}, $\hat \mUpsilon_t = \mUpsilon_t^0 +\op(1)$, therefore
    $$
        \hat{\mQ}_{(i)} \;\stackrel{p}{\to}\; \int_{\lambda_{i-1}}^{\lambda_{i}}\mUpsilon'(\tau)\mathbb{Q}_{\vz}(\tau)\mUpsilon(\tau)\rd\tau = \mathbb Q_i.
        $$

\end{proof}

\begin{proof}[Proof of Lemma \ref{lem4}]
        Let $\phi_t$ denote the $(a,b)$ element of the matrix $[\E\,(\vl_t\vl_t'|\mathcal{F}_{t-1})-\mI_n]$, $[\E\,( (\vl_t \vl_t') \otimes \vl_{t-i} | \mathcal F_{t-1})-\vrho_i]$ or $[\E\,( (\vl_t \vl_t') \otimes \vl_{t-i} \vl_{t-j} | \mathcal F_{t-1})-\vrho_{i,j}]$. Note that $\phi_t$ is a m.d.s., so it is a $L^1$-mixingale with constants satisfying Lemma \ref{lem1}. We now show that $\sup_t\E\left|\phi_t\right|^b <\infty$, letting $b=1+\delta/4>1$ (for example).

        From Jensen's inequality, , we have:
        \begin{eqnarray}
        \left|\E(\phi_t|\mathcal{F}_{t-1})\right|^b\leq \E\left(|\phi_t|^b \, |\mathcal{F}_{t-1} \right),\,\,\,\left| \E(\phi_t)\right|^b\leq \E \left| \phi_t\right|^b.
        \label{Jen}
        \end{eqnarray}

        Consider $\E\left(\left|\E(\phi_t|\mathcal{F}_{t-1})-\E(\phi_t)\right|^b\right)$. We have:
        \begin{eqnarray*}
                \left| \E(\phi_t|\mathcal{F}_{t-1})-\E(\phi_t)\right|^b\leq 2^b \max\left( \left| \E(\phi_t|\mathcal{F}_{t-1})\right|^b,\left|\E(\phi_t)\right|^b\right).
        \end{eqnarray*}

        From \eqref{Jen}, $\sup_t \E(\left| \E(\phi_t|\mathcal{F}_{t-1})\right|^b)\leq \sup_t\E( \E(\left|\phi_t\right|^b\,|\mathcal{F}_{t-1}))=\sup_t \E\left| \phi_t\right|^b$, so $\sup_t \E\left(\left| \E(\phi_t|\mathcal{F}_{t-1})-\E(\phi_t)\right|^b\right)\leq 2^b \sup_t \E\left| \phi_t\right|^b<\infty$.  Therefore, the conditions of Lemma \ref{lem1} are satisfied, so $T^{-1}\sum_{t=1}^{[Tr]}\E(\phi_t|\mathcal{F}_{t-1})-\E(\phi_t)\inp 0$ uniformly in $r$.
\end{proof}

\begin{proof}[Proof of Lemma \ref{lem5}]\hfill\\
By Assumption \ref{a8}(i) and Lemma \ref{lem4}, Assumption 2 in \citet{Boswijketal:2016} is satisfied; therefore, the proof follows by exactly the same arguments as in the proof of their Lemma \ref{lem2}, page 79, paragraphs 1-2.

\end{proof}

\begin{proof}[Proof of Lemma \ref{lem6}]\hfill \\
        \vspace*{0.1in} We first derive the asymptotic distribution of $T^{-1/2} \sum_{t\in I_i} \vz_tu_t$ and $T^{-1/2} \sum_{t\in I_i} \vz_t\vv_t^\prime \vbeta_x^0$ for interval  $I_i=\tilde I_s=[[\tau_{s-1}T]+1,\ldots,[\tau_{s}T]]$ when all the coefficients of the VAR are stable. Let $\vg_t=\vepsi_{t,\#} = \vec(\vepsi_t,\vzeros_{n(p-1)})=\vec( u_t, \vv_t,\vzeta_t,\vzeros_{n(p-1)}).$

        Then $\veta_t=\mA_{s,\#}^{-1}\;\vg_t$ and $\vg_t=\mA_{s,\#}\;\veta_t$.  Hence, $T^{-1/2} \sum_{t\in \tilde I_s} \vz_tu_t=T^{-1/2} \sum_{t\in \tilde I_s} \vz_t\vg_t^\prime\mathcal{S}_u$ and $T^{-1/2} \sum_{t\in \tilde I_s} \vz_t\vv_t'\vbeta_x^{0}=T^{-1/2} \sum_{t\in \tilde I_s} \vz_t\vg_t'\vbeta_{\vx,\#}$.
        We now derive the limit of $\mathcal{Z}_T=T^{-1/2} \sum_{t\in \tilde I_s} \vz_t\vg_t^\prime\mathcal{S}_{\dag}$.
        \begin{align*}
        \mathcal{Z}_T=T^{-1/2} \sum_{t\in \tilde I_s}\vz_t\veta_t^\prime\mA_{s,\#}^\prime\;\mathcal{S}_{\dag}
        =  \begin{bmatrix}T^{-1/2} \sum_{t \in \tilde I_s}\veta_t^\prime\mA_{s,\#}^\prime\;\mathcal{S}_{\dag}\\ T^{-1/2} \sum_{t \in \tilde I_s}\mathcal{S}_{\vr}\vxi_{t}\veta_t^\prime\mA_{s,\#}^\prime\;\mathcal{S}_{\dag}\\T^{-1/2} \sum_{t \in \tilde I_s}\mathcal{S}\vxi_{t-1} \veta_t^\prime\mA_{s,\#}^\prime\;\mathcal{S}_{\dag}
        \end{bmatrix} \equiv \begin{bmatrix} \mathcal E_1 \\ \mathcal E_2 \\ \mathcal E_3 \end{bmatrix}.
        \end{align*}

        Note that $\veta_t=\mA_{s,\#}^{-1}\;\mS_{\#}\;\mD_{t,\#}\;\vl_{t,\#}$.

        $\bullet$ Consider first $\mathcal E_1$. Notice that
        $\mathcal{E}_1=\mathcal E_1'=T^{-1/2} \sum_{t \in \tilde I_s}\mathcal{S}_{\dag}'\mA_{s,\#}\veta_t \;=\vecc(\mathcal{E}_1)= \mathcal{S}_{\dag}'\mA_{s,\#}\mathcal{\tilde E}_1$, where $\tilde{\mathcal E}_1 = T^{-1/2} \sum_{t \in \tilde I_s}\veta_t$.  Hence, consider $\mathcal{\tilde E}_1$. By Lemma \ref{lem5},  $\mB_{0,T}(r)=T^{-1/2}\sum_{t=1}^{[rT]}\vl_t\Rightarrow\mB_0(r)$. By Assumption \ref{a8} and Lemma \ref{lem4}, the conditions in Lemmas 1-2 of \citet{Boswijketal:2016} are satisfied, so: \begin{eqnarray*}
                \mathcal{\tilde E}_1&=&T^{-1/2}\sum_{t\in \tilde I_s}\veta_t = \mA_{s,\#}^{-1}\mS_{\#}T^{-1/2}\sum_{t\in \tilde I_s}(\mD_{t,\#}\vl_{t,\#})=\mA_{s,\#}^{-1}\mS_{\#}\left(\int_{\tau_{s-1}}^{\tau_s}\mD_{\#}(\tau)\rd \mB_{0,T,\#}(\tau)\right)\\
                &\Rightarrow& \mA_{s,\#}^{-1}\mS_{\#}\left(\int_{\tau_{s-1}}^{\tau_s}\mD_{\#}(\tau)\rd \mB_{0,\#}(\tau)\right).
        \end{eqnarray*}
        Using the fact that $\mA_{s,\#}\mA_{s,\#}^{-1}=\mI_{\#}$, $\mI_{\#}\mS_{\#}=\mS_{\#}$,
        \begin{align}
        \label{mathcalM1}
        \mathcal{E}_{1} = \mathcal{S}_{\dag}'\mA_{s,\#}\mathcal{\tilde E}_1 \Rightarrow (\mathcal{S}_{\dag}'\mS_{\#})\left(\int_{\tau_{s-1}}^{\tau_s}\mD_{\#}(\tau)\rd \mB_{0,\#}(\tau)\right)=\mathbb{M}_{1}(\tau_{s-1},\tau_s),
        \end{align}
        with variance
        \begin{align}
        \label{varmathcalM1}
        \mV_{\mathbb{M}_{1}(\tau_{s-1},\tau_s)} =(\mathcal{S}_{\dag}'\mS_{\#})\left(\int_{\tau_{s-1}}^{\tau_s}\mD_{\#}(\tau)\mD_{\#}^{\prime}(\tau)\rd \tau\right)(\mathcal{S}_{\dag}'\mS_{\#})'.\end{align}

        $\bullet$ Next, consider $\mathcal{E}_3$. Note that $\mathcal{E}_3=\vecc(\mathcal{E}_3)=((\mathcal{S}_{\dag}'\mA_{s,\#})\otimes\mathcal{S})\;T^{-1/2} \sum_{t \in \tilde I_s}(\veta_t\otimes\vxi_{t-1})=((\mathcal{S}_{\dag}'\mA_{s,\#})\otimes\mathcal{S})\mathcal{\tilde E}_3$, where $\tilde{\mathcal E}_3 =T^{-1/2} \sum_{t \in \tilde I_s}(\veta_t\otimes\vxi_{t-1})$.  Hence consider $\mathcal{\tilde E}_3$.  Letting $\tilde t = t-[\tau_{s-1}T]$ and recall that $\tilde I_s^- = [[\tau_{s-1}T]+2, [\tau_s T]]$. Then:
        \begin{align*}
        \tilde{\mathcal E}_3 & = T^{-1/2} \sum_{t\in \tilde I_s} \veta_t \otimes \vxi_{t-1} \\
        & =   T^{-1/2} \veta_{[\tau_{s-1}T]+1} \vxi_{[\tau_{s-1}T]} + T^{-1/2} \sum_{t\in \tilde I_s^-} \veta_t \otimes \left[\mF^{\tilde t-1}_s \vxi_{[\tau_{s-1}T]}\right] + T^{-1/2} \sum_{t\in \tilde I_s^-} \veta_t \otimes \left[ \left(\sum_{l=0}^{\tilde t-2}\mF_s^l\right)\vmu_s\right]  \\
        & +  T^{-1/2} \sum_{t\in \tilde I_s^-} \veta_t \otimes \left[ \sum_{l=0}^{\tilde t-2}\mF_s^l \veta_{t-l-1}\right] \\
        &  \equiv \sum_{i=1}^4 \tilde {\mathcal E}_{3,i}.
        \end{align*}
        Note that $\E \|\tilde{\mathcal E}_{3,1} \| \leq T^{-1/2} \sup_t (\E \| \veta_{t}\|^2)^{1/2}\sup_t (\E \| \vxi_{t}\|^2)^{1/2} < c_{\veta} \sup_t (\E \| \vxi_{t}\|^2)^{1/2}$ for some $c_{\veta}>0$ by Assumption \ref{a8}. Since $ \sup_t \E \| \vxi_{t} \vxi_t'\|< c$ by \eqref{momvxi}, $\tilde{\mathcal E}_{3,1} \inp 0$.

        Next, from above, $\vxi_{[\tau_{s-1}T]} =\Op(1)$, and
        $$
        \tilde {\mathcal E}_{3,2} = \left( T^{-1/2} \sum_{t\in \tilde I_s^-} (\mI_{np} \otimes \mF_s^{t-[\tau_{s-1}T]-1} )\, (\veta_t \otimes \mI_{np} ) \right) \vxi_{[\tau_{s-1}T]},  $$ where for some $c^*>0$,
        \begin{align*}
        &  \E\| T^{-1/2} \sum_{t\in \tilde I_s^-} (\mI_{np} \otimes \mF_s^{t-[\tau_{s-1}T]-1} )\, (\veta_t \otimes \mI_{np})\|  \leq   T^{-1/2} \sum_{t\in \tilde I_s^-} \| \mI_{np} \otimes \mF_s^{t-[\tau_{s-1}T]-1} \| \, \,  \E \|\veta_t \otimes \mI_{np}\| \\
        & = T^{-1/2} \sum_{l=0}^{\Delta \tau_s T-2} \| \mI_{np} \| \, \,\| \mF_s^{l+1}\| \, \E\| \veta_t \| \,\, \| \mI_{np} \|  \leq c^*  T^{-1/2} \sum_{l=0}^{\Delta \tau_s T-2} \| \mF_s^{l+1}\|  \rightarrow 0,
        \end{align*}
        where we used $\| \mA \otimes \mB \| = \|\mA\|\,\, \|\mB\|$, and the last statement follows by Assumptions \ref{a7} and \ref{a8} and the derivations in Section \ref{supsec2} below. Therefore, by the Markov inequality, it follows that
        $
        \tilde{\mathcal E}_{3,2} =\op(1).
        $\\[0.1in]
        Next, we show that $\tilde{\mathcal E}_{3,4} = \sum_{l=0}^{\Delta\tau_sT-2}  [\mI_{np} \otimes \mF_s^l] \, \left(T^{-1/2}\sum_{t\in \tilde I_s^-}  [ \veta_t \otimes \veta_{t-l-1}]\right) +\op(1) \equiv \tilde{\mathcal E}_{3,4}^{(1)}+\op(1)$. To that end, let $\tilde n=[\tau_{s-1}T]$.

        \begin{align}\nonumber
        &T^{-1/2} \sum_{t\in \tilde I_s^-} \veta_t \otimes \left[ \sum_{l=0}^{t-\tilde n-2}\mF_s^l \veta_{t-l-1}\right]   =  T^{-1/2} \veta_{\tilde n+2} \otimes [\mF_s^0 \veta_{\tilde n+1}] + T^{-1/2} \veta_{\tilde n+3} \otimes [\mF_s^0 \veta_{\tilde n+2}+\mF_s^1 \veta_{\tilde n+1}] \\ \nonumber
        &+  T^{-1/2} \veta_{\tilde n+4} \otimes [\mF_s^0 \veta_{\tilde n+3}+\mF_s^1 \veta_{\tilde n+2} +\mF_s^2 \veta_{\tilde n+1}] + \ldots  \\ \nonumber
        & +   T^{-1/2} \veta_{\tilde n+ \Delta \tau_s T} \otimes [\mF_s^0 \veta_{\tilde  n+ \Delta \tau_s T-1}+\mF_s^1 \veta_{\tilde n+ \Delta \tau_s T -2}+ .. +\mF_s^{\Delta \tau_s T-2} \veta_{\tilde n+1}]  \\ \nonumber
        & = (T^{-1/2} \veta_{\tilde n+2} \otimes [\mF_s^0 \veta_{\tilde n+1}] + T^{-1/2} \veta_{\tilde n+3} \otimes [\mF_s^0 \veta_{\tilde n+2}] + \ldots + T^{-1/2} \veta_{\tilde n +\Delta \tau_s T} \otimes [\mF_s^0 \veta_{\tilde  n+ \Delta \tau_s T-1}]) + \\ \nonumber
        & + (T^{-1/2} \veta_{\tilde n+2} \otimes [\mF_s^1 \veta_{\tilde n}] + T^{-1/2} \veta_{\tilde n+3} \otimes [\mF_s^1 \veta_{\tilde n+1}] + T^{-1/2} \veta_{\tilde n+4} \otimes [\mF_s^0 \veta_{\tilde n+2}] + \ldots + T^{-1/2} \veta_{\tilde n +\Delta \tau_s T} \otimes [\mF_s^1 \veta_{\tilde  n+ \Delta \tau_s T-2}]  ) \\ \nonumber
        & - T^{-1/2} \veta_{\tilde n+2} \otimes [\mF_s^1 \veta_{\tilde n}] \\ \nonumber
        & + (T^{-1/2} \veta_{\tilde n+2} \otimes [\mF_s^2 \veta_{\tilde n-1}] + T^{-1/2} \veta_{\tilde n+3} \otimes [\mF_s^2 \veta_{\tilde n}] + T^{-1/2} \veta_{\tilde n+4} \otimes [\mF_s^2 \veta_{\tilde n+1}] + \ldots + T^{-1/2} \veta_{\tilde n +\Delta \tau_s T} \otimes [\mF_s^2 \veta_{\tilde  n+ \Delta \tau_s T-3}]  ) \\  \nonumber
        & - (T^{-1/2} \veta_{\tilde n+2} \otimes [\mF_s^2 \veta_{\tilde n-1}]+ T^{-1/2} \veta_{\tilde n+3} \otimes [\mF_s^2 \veta_{\tilde n}]) + \ldots \\ \nonumber
        & + (T^{-1/2} \veta_{\tilde n+2} \otimes [\mF_s^{\Delta \tau_s T-2} \veta_{\tilde n+3 - \Delta \tau_s T}] + T^{-1/2} \veta_{\tilde n+3} \otimes [\mF_s^{\Delta \tau_s T-2} \veta_{\tilde n+4 - \Delta \tau_s T}] + \ldots + T^{-1/2} \veta_{\tilde n +\Delta \tau_s T} \otimes [\mF_s^{\Delta \tau_s T-2} \veta_{\tilde  n+ 1}]  )\\ \nonumber
        & - ( T^{-1/2} \veta_{\tilde n+2} \otimes [\mF_s^{\Delta \tau_s T-2} \veta_{\tilde n+3 - \Delta \tau_sT}] + T^{-1/2} \veta_{\tilde n+3} \otimes [\mF_s^{\Delta \tau_s T-2} \veta_{\tilde n+4 - \Delta \tau_s T}] + \ldots + T^{-1/2} \veta_{\tilde n +\Delta \tau_s T-1} \otimes [\mF_s^{\Delta \tau_s T-2} \veta_{\tilde  n+ 2}]  )\\ \label{nstar1}
        & =  \tilde{\mathcal E}_{3,4}^{(1)} - T^{-1/2} \sum_{l=1}^{\Delta \tau_s T -2}[ \mI_{np} \otimes \mF_s^l ] \, \sum_{j=0}^{l-1}\, \veta_{\tilde n+2+j} \otimes \veta_{\tilde n+2+j-(\black{l+1})}.
        \end{align}

        Note that by Assumption \ref{a7}, for some $c,c^*>0$,
        \begin{align}\nonumber
        &\E \|T^{-1/2} \sum_{l=1}^{\Delta \tau_s T -2}[ \mI_{np} \otimes \mF_s^l ] \, \sum_{j=0}^{l-1}\, \veta_{\tilde n+2+j} \otimes \veta_{\tilde n+2+j-\black{(l+1)}} \, \| \\ \nonumber
        &\leq T^{-1/2} \sum_{l=1}^{\Delta \tau_s T -2}\| \mI_{np} \otimes \mF_s^l \| \sum_{j=0}^{l-1}\, \E\|\veta_{\tilde n+2+j} \otimes \veta_{\tilde n+2+j-\black{(l+1)}}\|
        \\ \nonumber
        & \leq T^{-1/2} \sum_{l=1}^{\Delta \tau_s T -2}\| \mI_{np} \otimes \mF_s^l \| \sum_{j=0}^{l-1}  (\| \E([\veta_{\tilde n+2+j} \veta_{\tilde n+2+j}'] \otimes [\veta_{\tilde n+2+j-(l-1)}\veta_{\tilde n+2+j-\black{(l+1)}}^\prime])\|)^{1/2}\\ \nonumber
        &\leq T^{-1/2} \sum_{l=1}^{\Delta \tau_s T -2}\| \mI_{np} \otimes \mF_s^l \| \sum_{j=0}^{l-1}  \|\mA_s^{-1} \otimes \mA_s^{-1}\|^2 \,\, (\| \E([\vepsi_{\tilde n+2+j} \vepsi_{\tilde n+2+j}'] \otimes [\vepsi_{\tilde n+2+j-\black{(l+1)}}\vepsi_{\tilde n+2+j-\black{(l+1)}}^\prime]\|)^{1/2}\\ \nonumber
        & \leq T^{-1/2} \sum_{l=1}^{\Delta \tau_s T -2}\sum_{j=0}^{l-1} \| \mI_{np} \otimes \mF_s^l \|c \,\, (\| \E([\vl_{\tilde n+2+j} \vl_{\tilde n+2+j}^\prime] \otimes [\vl_{\tilde n+2+j-\black{(l+1)}}\vl_{\tilde n+2+j-\black{(l+1)}}^\prime\|)^{1/2} \\ \label{nstar2}
        &  \leq c T^{-1/2} \sum_{l=1}^{\Delta \tau_s T -2} np \,\,  l \|\mF_s^l \|  \sup_j (||\vrho_{jj}\|)^{1/2} \leq c^* \, \sum_{l=1}^{\Delta \tau_s T -2} l\| \mF_s^l \| \rightarrow 0,
        \end{align}
        where the last statement follows from Assumptions \ref{a7} and \ref{a8} and arguments similar to the ones in Section \ref{supsec2}. Therefore, by Markov's inequality, $T^{-1/2} \sum_{l=1}^{\Delta \tau_s T -2}[ \mI_{np} \otimes \mF_s^l ] \, \sum_{j=0}^{l-1}\, \veta_{\tilde n+2+j} \otimes \veta_{\tilde n+2+j-(l-1)} =\op(1)$, so

        $$\tilde{\mathcal E}_{3,4} = \tilde{\mathcal E}_{3,4}^{(1)} + \op(1).$$

        By similar arguments, $$\tilde{\mathcal E}_{3,3} =  \tilde{\mathcal E}_{3,3}^{(1)} + \op(1),$$
        where $\tilde{\mathcal E}_{3,3}^{(1)} =\sum_{l=0}^{\Delta\tau_sT-2}  [\mI_{np} \otimes \mF_s^l] \, \left(T^{-1/2}\sum_{t\in \tilde I_s^-}  [ \veta_t \otimes \vmu_s]\right)$.

        Putting the results about $\tilde{\mathcal E}_{3,i}$ for $i=1,2,3,4$ together, and letting $$\tilde{\mathcal E}_4(l) = [\mI_{np} \otimes \mF_s^l] \, \left(T^{-1/2}\sum_{t\in \tilde I_s^-}  \veta_t \otimes [ \vmu_s + \veta_{t-l-1}]\right),$$  and $\bar {\mathcal E}_4 (\Delta\tau_s T-2) = \sum_{l=0}^{\Delta\tau_s T-2} \tilde{\mathcal E}_4(l)$, we have:
        \begin{equation}
        \tilde{\mathcal E}_3 = \sum_{l=0}^{\Delta\tau_s T-2} \tilde{\mathcal E}_4(l) + \op(1) =\bar {\mathcal E}_4 (\Delta\tau_s T-2) +\op(1).
        \end{equation}

        We are interested in the asymptotic distribution of $\bar{\mathcal E}_4 (\Delta\tau_s T-2)$. First, consider the asymptotic distribution of $\tilde{\mathcal E}_4(l)$ for a given $l$.  First, $\mB_{0,T,\#}(r)=T^{-1/2}\sum_{t=1}^{[rT]}\vl_{t,\#}\Rightarrow\mB_{0,\#}(r)$.
        We also have that: $\veta_t\otimes\veta_{t-l-1}=(\mA_{s,\#}^{-1}\otimes \mA_{s,\#}^{-1})(\vg_t\otimes\vg_{t-l-1})=(\mA_{s,\#}^{-1}\otimes \mA_{s,\#}^{-1})((\mS_{\#}\mD_{t,\#})\otimes(\mS_{\#}\mD_{t-l-1,\#}))(\vl_{t,\#}\otimes\vl_{t-l-1,\#})$. By Lemma \ref{lem5} , for $l\geq 0$, $\mB_{l+1,T}(r)=T^{-1/2}\sum_{t=1}^{[rT]}(\vl_t\otimes\vl_{t-l-1})\Rightarrow\mB_{l+1}(r)$. Also, we have that $\mB_{l+1,T,\#}(r)=T^{-1/2}\sum_{t=1}^{[rT]}(\vl_{t,\#}\otimes\vl_{t-l-1,\#})\Rightarrow\mB_{l+1,\#}(r)=\left[\begin{array}{c}\mB_{l+1}(r)\\\vzeros_{n^2(p^2-1)}\end{array}\right]$, with variance matrix  $\vrho_{l+1,l+1,\#}^{perm} =\mP_1 \vrho_{l+1,l+1,\#}\mP_2$, where $\mP_1, \mP_2$ are permutation matrices defined as in \citet{Tracy/Jinadasa:1989}, Theorem 7, equation (18); note that this result holds because $\vl_{t\#} \otimes \vl_{t,\#}$ is a permutation of $(\vl_t \otimes \vl_t)_{\#}$.

        Now let  $\vu_t^{(n^*)} =  \vec(\vl_t, \vl_t \otimes\vl_{t-1}, \vl_t \otimes\vl_{t-2},\ldots, \vl_t \otimes\vl_{t-n^*})$ , and $\mB_T^{(n^*)}(r)  =T^{-1/2} \sum_{t=1}^{[Tr]} \vu_t^{(n^*)}$. Then, by Lemma \ref{lem5},
        $
        \mB_T^{(n^*)}(r) =  T^{-1/2} \sum_{t=1}^{[Tr]} \vu_t^{(n^*)} \Rightarrow \mB^{(n^*)}(r)=\vec_{i=0:n^*}(\mB_{i}(r)),
   $
        with variance
        \begin{align*}
        \mV_{\mB^{n^*}}(r) =  r\begin{bmatrix} \mI_n & \vrho_{0,1} & \ldots & \vrho_{0,n^*}\\
        \vrho_{0,1}' & \vrho_{1,1} & \ldots &\vrho_{1,n^*} \\
        \ldots &\ldots & \ldots &\ldots \\
        \vrho_{0,n^*}'& \vrho_{1,n^*}' & \ldots &\vrho_{n^*,n^*}.\end{bmatrix}
        \end{align*}

        We have:
        \begin{eqnarray*}
                \bar{\mathcal E}_4( [\Delta \tau_s T]-2) &=& \bar{\mathcal E}_4(n^*)+ (\bar{\mathcal E}_4( [\Delta \tau_s T]-2) -\bar{\mathcal E}_4(n^*))\\
                &=& \sum_{l=0}^{n^*}  [\mI_{np} \otimes \mF_s^l] \, \left(T^{-1/2}\sum_{t\in \tilde I_s^-}  (\veta_t \otimes \vmu_s)\right)+ \sum_{l=0}^{n^*}  [\mI_{np} \otimes \mF_s^l] \, \left(T^{-1/2}\sum_{t\in \tilde I_s^-}  (\veta_t \otimes \veta_{t-l-1})\right)\\
                &+& (\bar{\mathcal E}_4( [\Delta \tau_s T]-2) -\bar{\mathcal E}_4(n^*))+ \op(1) \\
                &=&\mathcal{\tilde E}_{*,1}(n^*)+\mathcal{\tilde E}_{*,2}(n^*)+(\bar{\mathcal E}_4( [\Delta \tau_s T]-2) -\bar{\mathcal E}_4(n^*))+\op(1)\\
                \mathcal{\tilde E}_{*,1}(n^*)&=& \sum_{l=0}^{n^*}  [\mI_{np} \otimes \mF_s^l] \, \left(\left[\mA_{s,\#}^{-1}\mS_{\#}\int_{\tau_{s-1}}^{\tau_s}  \mD_{\#}(\tau)\rd\mB_{0,T,\#}(\tau)\right] \otimes \vmu_s\right)\\
                &=& \sum_{l=0}^{n^*}  [\mI_{np} \otimes \mF_s^l]((\mA_{s,\#}^{-1}\mS_{\#})\otimes \mI_{np}) \, \left(\left[\int_{\tau_{s-1}}^{\tau_s}  \mD_{\#}(\tau)\rd\mB_{0,T,\#}(\tau)\right] \otimes \vmu_s\right)\\
                &=& \sum_{l=0}^{n^*}  ((\mA_{s,\#}^{-1}\mS_{\#})\otimes \mF_s^l) \, \left(\left[\int_{\tau_{s-1}}^{\tau_s}  \mD_{\#}(\tau)\rd\mB_{0,T,\#}(\tau)\right] \otimes \vmu_s\right)\\
                \mathcal{\tilde E}_{*,2}(n^*)&=& \sum_{l=0}^{n^*}  [\mI_{np} \otimes \mF_s^l] \, ((\mA_{s,\#}^{-1}\mS_{\#})\otimes (\mA_{s,\#}^{-1}\mS_{\#}))\int_{\tau_{s-1}}^{\tau_s}  \left(\mD_{\#}(\tau)\otimes\mD_{\#}\left(\tau-\frac{l+1}{T}\right)\right)\rd\mB_{l+1,T,\#}(\tau)
                \\
                &=& \sum_{l=0}^{n^*}  ((\mA_{s,\#}^{-1}\mS_{\#})\otimes (\mF_s^l\mA_{s,\#}^{-1}\mS_{\#}))\int_{\tau_{s-1}}^{\tau_s}  \left(\mD_{\#}(\tau)\otimes\mD_{\#}\left(\tau-\frac{l+1}{T}\right)\right)\rd\mB_{l+1,T,\#}(\tau).
        \end{eqnarray*}

        Using the convergence $ \mB_{l,T}(r)\Rightarrow \mB_l(r)$ , $l=0,1,\ldots,n^*$, and Theorem 2.1 in \citet{Hansen:1992}, it follows that, as $T\rightarrow \infty$, with $n^*=T^{\alpha}$ and $\alpha \in (0,1)$,
        \begin{align*}
        \mathcal{\tilde E}_{*,1}(n^*)+\mathcal{\tilde E}_{*,2}(n^*)&\Rightarrow\sum_{l=0}^{\infty} ((\mA_{s,\#}^{-1}\mS_{\#})\otimes \mF_s^l) \, \left(\left[\int_{\tau_{s-1}}^{\tau_s}  \mD_{\#}(\tau)\rd\mB_{0,\#}(\tau)\right] \otimes \vmu_s\right)\\
        &+  \sum_{l=0}^{\infty}  ((\mA_{s,\#}^{-1}\mS_{\#})\otimes (\mF_s^l\mA_{s,\#}^{-1}\mS_{\#}))\int_{\tau_{s-1}}^{\tau_s}  \left(\mD_{\#}(\tau)\otimes\mD_{\#}\left(\tau\right)\right)\rd\mB_{l+1,\#}(\tau).
        \end{align*}
        where the variance of the right hand side quantity exists by Assumption \ref{a7} and \ref{a8}. Next, note that under the same assumptions, for $n^*=T^{\alpha}$ and $T\rightarrow \infty$,
        \begin{align*}
        \|\bar{\mathcal E}_4( [\Delta \tau_s T]-2) -\bar{\mathcal E}_4(n^*)  \| & \leq  \sum_{l=n^*+1}^ {[\Delta \tau_s T]-2}  \| ((\mA_{s,\#}^{-1}\mS_{\#})\otimes \mF_s^l) \, \left(\left[\int_{\tau_{s-1}}^{\tau_s}  \mD_{\#}(\tau)\rd\mB_{0,\#}(\tau)\right) \otimes \vmu_s\right)\| \\
        &+  \sum_{l=n^*+1}^{ [\Delta \tau_s T]-2}  \| ((\mA_{s,\#}^{-1}\mS_{\#})\otimes (\mF_s^l\mA_{s,\#}^{-1}\mS_{\#}))\int_{\tau_{s-1}}^{\tau_s}  \left(\mD_{\#}(\tau)\otimes\mD_{\#}\left(\tau\right)\right)\rd\mB_{l+1,\#}(\tau) \| \\
        & \leq \sum_{l=n^*+1}^ {[\Delta \tau_s T]-2}  \| \mF_s^l \| \,\, \| \mA_{s,\#}^{-1}\mS_{\#}\| \,  \Op(1) + \sum_{l=n^*+1}^ {[\Delta \tau_s T]-2}  \| \mF_s^l \| \,\, \| \mA_{s,\#}^{-1}\mS_{\#}\|^2 \ \Op(1) \inp 0.
        \end{align*}
        Therefore,
        \begin{align*}
        \tilde{\mathcal{E}}_{3}& = \bar{\mathcal E}_4( [\Delta \tau_s T]-2)+ \op(1) \Rightarrow \sum_{l=0}^{\infty} ((\mA_{s,\#}^{-1}\mS_{\#})\otimes \mF_s^l) \, \left(\left[\int_{\tau_{s-1}}^{\tau_s}  \mD_{\#}(\tau)\rd\mB_{0,\#}(\tau)\right] \otimes \vmu_s\right)\\
        &+  \sum_{l=0}^{\infty}  ((\mA_{s,\#}^{-1}\mS_{\#})\otimes (\mF_s^l\mA_{s,\#}^{-1}\mS_{\#}))\int_{\tau_{s-1}}^{\tau_s}  \left(\mD_{\#}(\tau)\otimes\mD_{\#}\left(\tau\right)\right)\rd\mB_{l+1,\#}(\tau),
        \end{align*}

        and because $\mathcal E_3 = ((\mathcal{S}_{\dag}'\mA_{s,\#})\otimes\mathcal{S}) \, \tilde{\mathcal{E}}_{3}$, we have:
        \begin{align}
        \mathcal{E}_{3}&  \Rightarrow\sum_{l=0}^{\infty}  \left((\mathcal{S}_{\dag}'\mS_{\#}) \otimes (\mathcal{S}\mF_s^l)\right) \, \left(\left[\int_{\tau_{s-1}}^{\tau_s}  \mD_{\#}(\tau)\rd\mB_{0,\#}(\tau)\right] \otimes \vmu_s\right)\nonumber\\
        &+  \sum_{l=0}^{\infty}  ((\mathcal{S}_{\dag}'\mS_{\#})\otimes (\mathcal{S}\mF_s^l\mA_{s,\#}^{-1}\mS_{\#}))\int_{\tau_{s-1}}^{\tau_s}  \left(\mD_{\#}(\tau)\otimes\mD_{\#}\left(\tau\right)\right)\rd\mB_{l+1,\#}(\tau)\nonumber\\ &=\mathbb{M}_{3,1}(\tau_{s-1},\tau_s)+\mathbb{M}_{3,2}(\tau_{s-1},\tau_s)=\mathbb{M}_3(\tau_{s-1},\tau_s).\label{mathcalM3}
        \end{align}
        Note that $\mathbb{M}_3(\tau_{s-1},\tau_s)$ depends on the parameters in the interval $\tilde I_s$, $\mA_s$ and $\vmu_s$, so when the interval over which $\mathbb M_3(\cdot)$ is evaluated changes, these coefficients also change. Also note that the variance matrix of $\mathbb{M}_3(\tau_{s-1},\tau_s)$ is \begin{align}
        \label{varmathcalM3}
        \mV_{\mathbb{M}_3(\tau_{s-1},\tau_s)}=\mV_{\mathbb{M}_{3,1}(\tau_{s-1},\tau_s)}+\mV_{\mathbb{M}_{3,2}(\tau_{s-1},\tau_s)}+\mV_{\mathbb{M}_{3,1}(\tau_{s-1},\tau_s),\mathbb{M}_{3,2}(\tau_{s-1},\tau_s)}+\mV_{\mathbb{M}_{3,1}(\tau_{s-1},\tau_s),\mathbb{M}_{3,2}(\tau_{s-1},\tau_s)}^\prime,
        \end{align}
        where using the fact that the variance of $\rd\mB_{0,\#}(\cdot)$ is $\mI_{\#}$ for all $l$, \begin{eqnarray}
        \mV_{\mathbb{M}_{3,1}(\tau_{s-1},\tau_s)}&=&\sum_{l,\kappa=0}^{\infty}  \left((\mathcal{S}_{\dag}'\mS_{\#}) \otimes (\mathcal{S}\mF_s^l)\right) \, \left(\mV_0 \otimes \vmu_s\vmu_s^\prime\right)  \left((\mathcal{S}_{\dag}'\mS_{\#}) \otimes (\mathcal{S}\mF_s^\kappa)\right)'\label{varmathcal31}\\
        \mV_0&=& \int_{\tau_{s-1}}^{\tau_s}  \mD_{\#}(\tau)\mD_{\#}^\prime(\tau)\rd\tau\label{V0}\\
        \mV_{\mathbb{M}_{3,2}(\tau_{s-1},\tau_s)}&=& \sum_{l,\kappa=0}^{\infty}  ((\mathcal{S}_{\dag}'\mS_{\#})\otimes (\mathcal{S}\mF_s^l\mA_{s,\#}^{-1}\mS_{\#}))\,\mV_{l,\kappa}\,((\mathcal{S}_{\dag}'\mS_{\#}) \otimes(\mathcal{S}\mF_s^\kappa\mA_{s,\#}^{-1}\mS_{\#}))'\label{varmathcalM32}\\
        \mV_{l-1,\kappa-1}&=& \int_{\tau_{s-1}}^{\tau_s}  (\mD_{\#}(\tau)\otimes\mD_{\#}(\tau))\,\vrho_{l-1,\kappa-1,\#}^{perm}\,(\mD_{\#}(\tau)\otimes \mD_{\#}(\tau))^\prime\rd\tau\label{Vll}\\
        \mV_{\mathbb{M}_{3,1}(\tau_{s-1},\tau_s),\mathbb{M}_{3,2}(\tau_{s-1},\tau_s)}&=& \sum_{l,\kappa=0}^{\infty} ((\mathcal{S}_{\dag}'\mS_{\#})\otimes (\mathcal{S}\mF_s^l\mA_{s,\#}^{-1}\mS_{\#}))\,\mV_{l}\,(\mI_{np} \otimes\vmu_s') \left((\mathcal{S}_{\dag}'\mS_{\#}) \otimes (\mathcal{S}\mF_s^\kappa)\right)'\label{covmathcal3132}\\
        \mV_{l-1}&=&\int_{\tau_{s-1}}^{\tau_s}  (\mD_{\#}(\tau)\otimes\mD_{\#}(\tau))\vrho_{l-1,\#}^{perm}(\mD_{\#}(\tau)^\prime\otimes1)\rd\tau\label{Vl}
        \end{eqnarray}
        where $\vrho_{l,\kappa,\#}^{perm} = \mP_1 \vrho_{l,\kappa,\#} \mP_2$, with $\mP_1,\mP_2$ the previously defined permutations, and the only non-zero block of $\vrho_{l,\kappa,\#}$ is given by the upper-left block $\vrho_{l,\kappa}$ in Assumption \ref{a8}(v); $\vrho_{l,\#}^{perm}$ is a permutation of $\vrho_{l,\#}$ which ensures that $\vrho_{l,\#}^{perm}= \E(\rd\mB_{l,\#}(\cdot) \rd \mB_{0,\#}(\cdot)) = \E((\vl_{t,\#} \otimes \vl_{t-l,\#}') \vl_{t,\#}')$; and the only non-zeros block of $\vrho_{l,\#}$ is given by upper-left block $\vrho_{l}$ in Assumption \ref{a8}(iv). \\[0.1in]

        $\bullet$ Consider $
        \mathcal E_2$. Substituting into the expression for $\mathcal E_2$ the expression $\vxi_t = \vmu_s+\veta_t + \mF_s\vxi_{t-1}$, we have:
        \begin{align*}
        \mathcal E_2 & =\vecc(\mathcal{E}_2)=((\mathcal{S}_{\dag}'\mA_{s,\#})\otimes\mathcal{S}_{\vr})\;T^{-1/2} \sum_{t \in \tilde I_s}(\veta_t\otimes\vxi_{t})=((\mathcal{S}_{\dag}'\mA_{s,\#})\otimes\mathcal{S}_{\vr})T^{-1/2} \sum_{t \in \tilde I_s}(\veta_t\otimes \vmu_s)\\
        &+((\mathcal{S}_{\dag}'\mA_{s,\#})\otimes\mathcal{S}_{\vr})T^{-1/2} \sum_{t \in \tilde I_s}(\veta_t\otimes (\mF_s\vxi_{t-1}))+((\mathcal{S}_{\dag}'\mA_{s,\#})\otimes\mathcal{S}_{\vr})T^{-1/2} \sum_{t \in \tilde I_s}(\veta_t\otimes \veta_t)\\
        &=\mathcal{ E}_{2,1}+\mathcal{ E}_{2,2}+\mathcal{ E}_{2,3}.
        \end{align*}
        Consider $\mathcal{ E}_{2,1}$ and $\mathcal{ E}_{2,2}$ first. Using similar arguments as before,
        \begin{align}
        \mathcal{E}_{2,1}&\Rightarrow   \sum_{l=0}^{\infty}((\mathcal{S}_{\dag}'\mS_{\#})\otimes (\mathcal{S}_{\vr}\mF_s^l)) \, \left(\left[\int_{\tau_{s-1}}^{\tau_s}  \mD_{\#}(\tau)\rd\mB_{0,\#}(\tau)\right] \otimes \vmu_s\right)=\mathbb{M}_{2,1}(\tau_{s-1},\tau_s)\label{mathcalM21}\\
        \mathcal{E}_{2,2}&\Rightarrow   \sum_{l=0}^{\infty}  ((\mathcal{S}_{\dag}'\mS_{\#})\otimes (\mathcal{S}_{\vr}\mF_s^{l+1}\mA_{s,\#}^{-1}\mS_{\#}))\int_{\tau_{s-1}}^{\tau_s}  \left(\mD_{\#}(\tau)\otimes\mD_{\#}\left(\tau\right)\right)\rd\mB_{l+1,\#}(\tau)=\mathbb{M}_{2,2}(\tau_{s-1},\tau_s),\label{mathcalM22}
        \end{align}
        with variances derived similarly to $\mV_{\mathbb{M}_{3,1}(\tau_{s-1},\tau_s)}$ and $\mV{\mathbb{M}_{3,2}(\tau_{s-1},\tau_s)}$ above:
        \begin{align}
        \mV_{\mathbb{M}_{2,1}(\tau_{s-1},\tau_s)}&=   \sum_{l,\kappa=0}^{\infty}((\mathcal{S}_{\dag}'\mS_{\#})\otimes (\mathcal{S}_{\vr}\mF_s^l))\,\mV_0\,  ((\mathcal{S}_{\dag}'\mS_{\#})\otimes (\mathcal{S}_{\vr}\mF_s^\kappa))'\label{varmathcal21}\\
        \mV_{\mathbb{M}_{2,2}(\tau_{s-1},\tau_s)}&=   \sum_{l,\kappa=0}^{\infty} ((\mathcal{S}_{\dag}'\mS_{\#})\otimes (\mathcal{S}_{\vr}\mF_s^{l+1}\mA_{s,\#}^{-1}\mS_{\#}))\,\mV_{l,l}\, ((\mathcal{S}_{\dag}'\mS_{\#})\otimes (\mathcal{S}_{\vr}\mF_s^{\kappa+1}\mA_{s,\#}^{-1}\mS_{\#}))'\label{varmathcal22}
        \end{align}
        and covariance derived similarly to $\mV_{\mathbb{M}_{3,1}(\tau_{s-1},\tau_s),\mathbb{M}_{3,2}(\tau_{s-1},\tau_s)}$ above:
        \begin{align}
        \mV_{\mathbb{M}_{2,1}(\tau_{s-1},\tau_s),\mathbb{M}_{2,2}(\tau_{s-1},\tau_s)}=  \sum_{l,\kappa=0}^{\infty} ((\mathcal{S}_{\dag}'\mS_{\#})\otimes (\mathcal{S}_{\vr}\mF_s^{l+1}\mA_{s,\#}^{-1}\mS_{\#}))\,\mV_l\,((\mathcal{S}_{\dag}'\mS_{\#})\otimes (\mathcal{S}_{\vr}\mF_s^\kappa))'\label{covmathcal2122}.
        \end{align}
        Consider next $\mathcal{E}_{2,3}$, where $\mathcal{E}_{2,3}=((\mathcal{S}_{\dag}'\mA_{s,\#})\,\otimes \,\mathcal{S}_{\vr})\, T^{-1/2} \sum_{t \in \tilde I_s}(\veta_t\,\otimes \,\veta_t)=T^{-1/2} \sum_{t \in \tilde I_s}\mathcal{S}_{\vr}\veta_t\vg_t'\mathcal{S}_{\dag}$. Note that $\mA_s$ given in \eqref{AS} is upper triangular with ones on the main diagonal. Hence, $\mA_s^{-1}$ is also upper triangular and has ones on the main diagonal. Denote by $\va_{1,\bcdot}$ the first row of $\mA_{s}^{-1}$, by $\mA_{p_1,\bcdot}$ the subsequent $p_1$ rows, and by $\mA_{p_2,\bcdot}=[\vzeros_{p_2}\;\;\vzeros_{p_2\times p_1}\;\;\mI_{p_2}]$ the subsequent $p_2$ rows. Then:
        \begin{eqnarray*}\mathcal{S}_{\vr}\veta_t=\mathcal{S}_{\vr}\mA_{s,\#}^{-1}\;\vg_{t}=\mathcal{S}_{\vr}\left[\begin{array}{c} \mA_s^{-1}\vepsi_t\\\vzeros_{n(p-1)}\end{array}\right]=\mathcal{S}_{\vr}\left[\begin{array}{c} \va_{1,\bcdot}\;\vepsi_t\\\mA_{p_1,\bcdot}\;\vepsi_t\\\mA_{p_2,\bcdot}\;\vepsi_t\\\vzeros_{n(p-1)}\end{array}\right]=\mathcal{S}_{\vr}\left[\begin{array}{c} \va_{1,\bcdot}\;\vepsi_t\\\mA_{p_1,\bcdot}\;\vepsi_t\\\vzeta_t\\\vzeros_{n(p-1)}\end{array}\right]=\vzeta_{t},
        \end{eqnarray*}
        and $\vg_t'\mathcal{S}_{\dag}$ is equal to either $u_t$ or $\vv_t'\vbeta_{\vx}^0$. Let $\mathcal{E}_{2,3}=\mathcal{E}_{2,3}^{(1)}=T^{-1/2} \sum_{t \in \tilde I_s}\vzeta_t u_t$ (when $\vg_t'\mathcal{S}_{\dag}=u_t$) and let $\mathcal{E}_{2,3}=\mathcal{E}_{2,3}^{(2)}=T^{-1/2}\sum_{t\in \tilde I_s}\vzeta_t\vv_t'\vbeta_{\vx}^0$ (when $\vg_t'\mathcal{S}_{\dag}=\vv_t'\vbeta_{\vx}^{0}$).

        Consider first $\mathcal{E}_{2,3}^{(1)}$. To that end, by Assumption \ref{a8} we have:
        \begin{eqnarray*}
                \mD_t=\left[\begin{array}{ccc}d_{u,t}&\vzeros_{p_1}'&\vzeros_{p_2}' \\
                        \vzeros_{p_1}&\mD_{\vv,t}&\vzeros_{p_1\times p_2}\\
                        \vzeros_{p_2}&\vzeros_{p_1\times p_2}'&\mD_{\vzeta,t}\end{array} \right],\;\;\mS=\left[ \begin{array}{ccc}1 &\vzeros_{p_1}'&\vzeros_{p_2}'\\
                        \vs_{p_1}&\mS_{p_1}&\vzeros_{p_1\times p_2}\\
                        \vzeros_{p_2}&\vzeros_{p_1\times p_2}'&\mS_{p_2} \end{array}\right],\;\;\vl_t=\left[ \begin{array}{c} l_{u,t}\\\vl_{\vv,t}\\\vl_{\vzeta,t}\end{array}\right ],
        \end{eqnarray*}
        where $\mD_t$ is partitioned exactly as $\mD(\tau)$ in the notation preceding the statement of Lemma \ref{lem6}, $\mS$ is partitioned exactly the same way in the notation section before Lemma \ref{lem6}, $l_{u,t}$ is a scalar and $\vl_{v,t}$ is a $p_1\times 1$.  Therefore,
        \begin{eqnarray}
        \left[\begin{array}{c}u_t\\\vv_t\\\vzeta_t \end{array}\right]=\vepsi_t=\mS \mD_t\vl_t=\left[\begin{array}{c} d_{u,t}\;l_{u,t}\\\vs_{p_1}\;d_{u,t} \;l_{u,t}+\mS_{p_1}\;\mD_{\vv,t}\;\vl_{\vv,t}\\\mS_{p_2}\;\mD_{\vzeta,t}\;\vl_{\vzeta,t}\end{array} \right].\label{lut}
        \end{eqnarray}
        Let $\mB_{u\vzeta,T}(r)=T^{-1/2}\sum_{t=1}^{[Tr]} l_{u,t} \vl_{\vzeta,t}$. By Lemma \ref{lem5}, $\mB_{u\vzeta,T}(\cdot) \Rightarrow \mB_{u\vzeta}(\cdot)$. By Assumption \ref{a8}, $\mS, \mD_t$ are bounded, so by Theorem 2.1 in \citet{Hansen:1992}, we have:
        \begin{eqnarray}
        \mathcal{E}_{2,3}^{(1)}&=&T^{-1/2}\sum_{t\in \tilde I_s} u_t \vzeta_t=T^{-1/2}\sum_{t\in \tilde I_s} \mS_{p_2}(d_{u,t} \mD_{\vzeta,t})(l_{u,t} \vl_{\vzeta,t})\nonumber\\
        &=&\mS_{p_2} \int_{\tau_{s-1}}^{\tau_s} d_{u}(\tau) \mD_{\vzeta}(\tau) \rd\mB_{u\vzeta,T}(\tau)\nonumber\\
        &\Rightarrow& \mS_{p_2} \int_{\tau_{s-1}}^{\tau_s} (d_{u}(\tau) \mD_{\vzeta}(\tau))\rd\mB_{u\vzeta}(\tau)=\mathbb{M}_{2,3}^{(1)}(\tau_{s-1},\tau_s).\label{mathcalM231}
        \end{eqnarray}
        with variance
        \begin{align}
        \mV_{\mathbb{M}_{2,3}^{(1)}(\tau_{s-1},\tau_s)}=\mS_{p_2}\int_{\tau_{s-1}}^{\tau_s}d_{u}^2(\tau_s)\mD_{\vzeta}(\tau)\vrho_{u,\vxi,0,0}\,\mD_{\vzeta}'(\tau)\mS_{p_2}'\rd\tau\label{varmathcal231}
        \end{align}

        Consider $\mathcal{E}_{2,3}^{(2)}$. Similarly to $\mathcal{E}_{2,3}^{(1)}$,  $\mB_{\vv\vzeta,T}(r)=T^{-1/2}\sum_{t=1}^{[Tr]} \vl_{\vv,t}\otimes\vl_{\vzeta,t}\Rightarrow \mB_{\vv\vzeta}(r)$, where recall that $\mB_{\vv\vzeta}(\cdot)$ is a $p_2 \times 1$ vector of Brownian motions with variance $\vrho_{\vv,\vzeta,0,0}=\E((\vl_{\vv,t}\vl_{\vv,t}')\otimes(\vl_{\vzeta,t}\vl_{\vzeta,t}'))$ which is just the $(p_1 p_2 )\times (p_1 p_2)$ lower-right  block of $\vrho_{\vxi,0,0}$. Also, by Lemma \ref{lem5}, $\mB_{\vzeta,T}(\cdot)=\vec(\mB_{u\vzeta,T}(\cdot),\mB_{\vv\vzeta,T}(\cdot))$ also jointly converge to $\mB_{\vzeta}(\cdot)$, a process defined just before Lemma \ref{lem5}. Therefore, \begin{eqnarray}
        \mathcal{E}_{2,3}^{(2)}&=& \vbeta_{\vx,(s)}^{0'}\left(T^{-1/2}\sum_{t\in \tilde I_s}(\vv_t\otimes\vzeta_t)\right)\nonumber\\
        &=& (\vbeta_{\vx,(s)}^{0'}\otimes 1)\left(T^{-1/2}\sum_{t\in \tilde I_s}((\vs_{p_1}\;d_{u,t} \;l_{u,t})\otimes(\mS_{p_2}\;\mD_{\vzeta,t}\;\vl_{\vzeta,t}))\right)\nonumber\\
        &+&  (\vbeta_{\vx,(s)}^{0'}\otimes 1 )\left(T^{-1/2}\sum_{t\in \tilde I_s}((\mS_{p_1}\;\mD_{\vv,t}\;\vl_{\vv,t})\otimes(\mS_{p_2}\;\mD_{\vzeta,t}\;\vl_{\vzeta,t}))\right)\nonumber\\
        &\Rightarrow&((\vbeta_{\vx,(s)}^{0'}\vs_{p_1}) \otimes \mS_{p_2})\int_{\tau_{s-1}}^{\tau_s} (d_{u}(\tau)\otimes\mD_{\vzeta}(\tau))\rd\mB_{u\vzeta}(\tau)+((\vbeta_{\vx,(s)}^{0'}\mS_{p_1}) \otimes \mS_{p_2})\int_{\tau_{s-1}}^{\tau_s} (\mD_{\vv}(\tau)\otimes\mD_{\vzeta}(\tau))\rd\mB_{\vv\vzeta}(\tau)\nonumber\\
        &=& \mathbb{M}_{2,3}^{(2)}(\tau_{s-1},\tau_s),\label{mathcalM232}
        \end{eqnarray}
        with variance
        \begin{align}
        &\mV_{\mathbb{M}_{2,3}^{(2)}(\tau_{s-1},\tau_s)}=\mV^{(1)}+\mV^{(2)}+\mV^{(3)}+\left(\mV^{(3)}\right)'\label{varmathcal232}\\
        \mV^{(1)}&=((\vbeta_{\vx,(s)}^{0'}\vs_{p_1}) \otimes \mS_{p_2})\int_{\tau_{s-1}}^{\tau_s} (d_{u}(\tau)\otimes\mD_{\vzeta}(\tau))\vrho_{u,\vxi,0,0}(d_{u}(\tau)\otimes\mD_{\vzeta}(\tau))'\rd\tau((\vbeta_{\vx,(s)}^{0'}\vs_{p_1}) \otimes \mS_{p_2})'\nonumber\\
        \mV^{(2)}&= ((\vbeta_{\vx,(s)}^{0'}\mS_{p_1}) \otimes \mS_{p_2})\int_{\tau_{s-1}}^{\tau_s} (\mD_{\vv}(\tau)\otimes\mD_{\vzeta}(\tau))\vrho_{\vv,\vxi,0,0}(\mD_{\vv}(\tau)\otimes\mD_{\vzeta}(\tau))'\rd\tau ((\vbeta_{\vx,(s)}^{0'}\mS_{p_1}) \otimes \mS_{p_2})'\nonumber\\
        \mV^{(3)}&= ((\vbeta_{\vx,(s)}^{0'}\vs_{p_1}) \otimes \mS_{p_2})\int_{\tau_{s-1}}^{\tau_s} (d_{u}(\tau)\otimes\mD_{\vzeta}(\tau))\vrho_{u,\vv,\vxi,0,0}(\mD_{\vv}(\tau)\otimes\mD_{\vzeta}(\tau))'\rd\tau ((\vbeta_{\vx,(s)}^{0'}\mS_{p_1}) \otimes \mS_{p_2})',\nonumber
        \end{align}
        where $\vrho_{u,\vv,\vxi,0,0}$ is the upper-right block of $\vrho_{\vxi,0,0}$.




        In conclusion, we have:
        \begin{eqnarray*}
                \mathcal{E}_{2}\Rightarrow \mathbb{M}_{2,1}(\tau_{s-1},\tau_s)+\mathbb{M}_{2,2}(\tau_{s-1},\tau_s)+\mathbb{M}_{2,3}(\tau_{s-1},\tau_s)=\mathbb{M}_{2}(\tau_{s-1},\tau_s),
        \end{eqnarray*}
        where $\mathbb{M}_{2,2}(\tau_{s-1},\tau_s)=\mathbb{M}_{2,2}^{(1)}(\tau_{s-1},\tau_s)$ when $\vg_t'\mathcal{S}_{\dag}=u_t$ and $\mathbb{M}_{2,2}(\tau_{s-1},\tau_s)=\mathbb{M}_{2,2}^{(2)}(\tau_{s-1},\tau_s)$ when $\vg_t'\mathcal{S}_{\dag}=\vv_t'\vbeta_{\vx,(s)}^{0}$.
        with asymptotic variance $\mV_{\mathbb{M}_2(\tau_{s-1},\tau_s)}$ which can be derived by similar arguments as those used for $\mV_{\mathbb{M}_1(\tau_{s-1},\tau_s)}$ and $\mV_{\mathbb{M}_3(\tau_{s-1},\tau_s;s)}$. We have:
        \begin{align}
        \mV_{\mathbb{M}_2(\tau_{s-1},\tau_s)}&=\mV_{\mathbb{M}_{2,1}(\tau_{s-1},\tau_s)}+\mV_{\mathbb{M}_{2,2}(\tau_{s-1},\tau_s)}+\mV_{\mathbb{M}_{2,3}^{(i)}(\tau_{s-1},\tau_s)}+\mV_{\mathbb{M}_{2,1}(\tau_{s-1},\tau_s),\mathbb{M}_{2,2}(\tau_{s-1},\tau_s)}\label{varmathcal2_1}\\&+\mV_{\mathbb{M}_{2,1}(\tau_{s-1},\tau_s),\mathbb{M}_{2,2}(\tau_{s-1},\tau_s)}'+\mV_{\mathbb{M}_{2,1}(\tau_{s-1},\tau_s),\mathbb{M}_{2,3}^{(i)}(\tau_{s-1},\tau_s)}+\mV_{\mathbb{M}_{2,1}(\tau_{s-1},\tau_s),\mathbb{M}_{2,3}^{(i)}(\tau_{s-1},\tau_s)}'\\
        &+\mV_{\mathbb{M}_{2,2}(\tau_{s-1},\tau_s),\mathbb{M}_{2,3}^{(i)}(\tau_{s-1},\tau_s)}+\mV_{\mathbb{M}_{2,2}(\tau_{s-1},\tau_s),\mathbb{M}_{2,3}^{(i)}(\tau_{s-1},\tau_s)}'\label{varmathcal2_2}
        \end{align}
        where the terms in \eqref{varmathcal2_1} are given in \eqref{varmathcal21}, \eqref{varmathcal22}, \eqref{varmathcal231} and \eqref{covmathcal2122} respectively, and $\mV_{\mathbb{M}_{2,1}(\tau_{s-1},\tau_s),\mathbb{M}_{2,3}^{(i)}(\tau_{s-1},\tau_s)}$
        and $\mV_{\mathbb{M}_{2,2}(\tau_{s-1},\tau_s),\mathbb{M}_{2,3}^{(i)}(\tau_{s-1},\tau_s)}$
        can be obtained similarly using $\vrho_{u,0}$ (the upper block of $\vrho_{\vxi,0}$ derived before the proof of Lemma \ref{lem6}) and $\vrho_{\vv,0}$ (the lower block of $\vrho_{\vxi,0}$) respectively for $i=1,2$.

        Now note that $\mathcal E_1, \mathcal E_2, \mathcal E_3$ are functions of the same underlying Brownian motions which were shown to jointly converge, therefore they also jointly converge:
        \begin{align*}
        \mathcal Z_T = \begin{bmatrix} \mathcal E_1 \\
        \mathcal E_2 \\
        \mathcal E_3 \end{bmatrix} \Rightarrow \begin{bmatrix} \mathbb{M}_1(\tau_{s-1},\tau_s) \\
        \mathbb{M}_2(\tau_{s-1},\tau_s)\\
        \mathbb{M}_3(\tau_{s-1},\tau_s)
        \end{bmatrix}\equiv  \mathbb M(\tau_{s-1},\tau_s).
        \end{align*}
        This completes the proof for $I_i = \tilde I_s = [ [\tau_{s-1}T]+1, [\tau_s T]]$. Now consider the case of $I_i$ containing $N_i$ breaks from the total set of $N$ breaks, that is, there is an $s$ such that $\tau_{s-1}<\lambda_{i-1}\leq \tau_s$ and $\tau_{s+N_i-1}\leq\lambda_i<\tau_{s+N_i}$. Then, generalizing the previous results which were for $I_i=\tilde I_s=[[\tau_{s-1} T]+1, [\tau_s T]]$,
        \begin{align*}
        \mathcal Z_T \Rightarrow \begin{cases} \mathbb M(\lambda_{i-1}, \tau_{s}) + \sum_{j=1}^{N_i}\mathbb M(\tau_{s+j-1}, \tau_{s+j}) + \mathbb M(\tau_{s+N_i},\lambda_i) & \mbox{ if } N_i\geq 2\\
        \mathbb M(\lambda_{i-1}, \tau_s) +  \mathbb M(\tau_s, \lambda_i) & \mbox{ if } N_i =1 \\
        \mathbb M(\lambda_{i-1}, \lambda_i) & \mbox{ if } N_i=0. \end{cases}
        \end{align*}

\end{proof}

\begin{proof}[Proof of Lemma \ref{lem7}]\hfill \\
        \black{Because of Assumptions \ref{a4} of fixed breaks, this lemma fits the setting of \citet{Hall/Han/Boldea:2012} and \citet{Bai/Perron:1998}.
                If the reduced form is stable, then $h=0$, and Lemma \ref{lem7}(ii) follows from Lemma \ref{lem2} and Lemma \ref{lem6}. Lemma \ref{lem7}(iii) follows from  \citet{Hall/Han/Boldea:2012} Theorems 1-2, where their Assumptions 6-11 are automatically satisfied by our assumptions and Lemma \ref{lem2} as follows: their Assumption 6 is our Assumption \ref{a1}, their Assumption 7 is our Assumption \ref{a6}, their Assumption 8 is automatically satisfied by our Assumption \ref{a8}, their Assumption 9 by our Assumption \ref{a9}, and their Assumptions 10-11 hold by Lemma \ref{lem2} and Assumption \ref{a7}. If $h>0$, Lemma \ref{lem7}(i) is the same as \citet{Hall/Han/Boldea:2012} Assumption 19(i), and is a special case of Proposition 1 in \citet{Bai/Perron:1998}, where the \citet{Bai/Perron:1998} Assumptions 1-3 hold by Lemma \ref{lem2} and our Assumption \ref{a7}, the \citet{Bai/Perron:1998} Assumption 4 encompasses Assumption \ref{a7} as a special case, and the \citet{Bai/Perron:1998} Assumption 5 is exactly our Assumption \ref{a3}. Lemma \ref{lem7}(ii) follows from Lemma \ref{lem2} and \ref{lem6}. Lemma \ref{lem7}(iii) follows from \citet{Hall/Han/Boldea:2012} Theorem 8(i)-(ii). In particular, their Assumption 6 is the same as our Assumption \ref{a1}, their Assumption 8 is automatically satisfied by our Assumption \ref{a8}, their Assumptions 10-11 hold by Lemma \ref{lem2}, their Assumption 17 is our Assumption \ref{a3}, their Assumption 18 is our Assumption \ref{a6}, and their Assumption 19(ii) holds by our Assumption \ref{a9}.
        }

\end{proof}

\begin{proof}[Proof of Lemma \ref{lem8}]\hfill \\
Let $\phi_t$ be the $(a,b)$ element of $\vepsi_t \vepsi_t'$, $(\vepsi_t \vepsi_t') \otimes \vepsi_{t-i}$, or $(\vepsi_t \vepsi_t') \otimes (\vepsi_{t-i}\vepsi_{t-j})$, for $i,j \geq 0$.
\begin{align*}
T^{-1} \sum_{t=1}^{[Tr]} \phi_t  = T^{-1} \sum_{t=1}^{[Tr]} (\phi_t -\E(\phi_t|\mathcal F_{t-1})) + T^{-1} \sum_{t=1}^{[Tr]} (\E(\phi_t|\mathcal F_{t-1}) - \E(\phi_t)).
\end{align*}
Note that  $(\phi_t -\E(\phi_t|\mathcal F_{t-1}))$ is a m.d.s. so it is a $L^1$-mixingale with uniformly bounded constants. Moreover, by similar arguments  as in the proof of Lemma \ref{lem4}, for $b=1+\delta/4>1$ (for example), $\sup_t \E|\phi_t|^b < \infty$.
Therefore, by Lemma \ref{lem1},  $T^{-1} \sum_{t=1}^{[Tr]} (\phi_t -\E(\phi_t|\mathcal F_{t-1})) \inp 0$ uniformly in $r$. By similar arguments, $T^{-1} \sum_{t=1}^{[Tr]} (\E(\phi_t|\mathcal F_{t-1}) - \E(\phi_t)) \inp 0$ uniformly in $r$, completing the proof.
\end{proof}

\begin{theoA}
        \label{theo_0vsk}
        Under Assumption \ref{a1}-\ref{a9} and the null hypothesis $k=0$,
        \begin{align*}
        Wald_{T\vlam_k}\; =\;  T \, \hat \vbeta_{\vlam_k}' \,
        \mR_k' \,\left(\mR_k \hat \mV_{\vlam_k} \mR_k'\right)^{-1}\mR_k \,\hat \vbeta_{\vlam_k} \Rightarrow \mathbb N(\vlam_k),
        \end{align*}
        where
        $
        \mathbb N(\vlam_k) = [\vec_{i=1:k+1}(\mathbb Q_i^{-1} \mathbb N_i)]' \mR_k' \, \left(\mR_k \, \diag_{i=1:k+1}(\mathbb Q_i^{-1} \mathbb V_i  \mathbb Q_i^{-1}) \mR_k' \, \right)^{-1} [\vec_{i=1:k+1}(\mathbb Q_i^{-1} \mathbb N_i)],
        $
        and $\mathbb Q_i, \mathbb V_i, \mathbb N_i$ are defined in equations \eqref{defq} in the paper, and \eqref{defv} and \eqref{defn} in this Appendix.
        \end{theoA}


\begin{proof}[Proof of Theorem B \ref{theo_0vsk}] Recall that $\hat \mDelta_{t} = \hat{\mDelta}_{(j)}$ if $t \in \left[ [\hat \pi_{j-1}^*T]+1,[\hat \pi_j^*T]\right]$ for $ j=1,\ldots, h+1$, and $\hat \mUpsilon_{t} = (\hat \mDelta_{t}, \mPi)$.

        For notation ease, set $I_i=I_{i,\vlam_k}$, and recall that $Wald_{T\vlam_k}\; =\;  T \, \hat \vbeta_{\vlam_k}' \,
        \mR_k' \,\left(\mR_k \hat \mV_{\vlam_k} \mR_k'\right)^{-1}\mR_k \,\hat \vbeta_{\vlam_k}$, where \begin{align*}
        &\hat{\mV}_{\vlam_k}=\diag_{i=1:k+1}(\hat \mV_{(i)}) , \qquad \hat \mV_{(i)}\;=\;\hat \mQ_{(i)}^{-1}  \ \hat \mM_{(i)} \ \hat \mQ_{(i)}^{-1}\, , \qquad \hat \mQ_{(i)} = T^{-1} \sum_{t \in I_{i}} \hat \vw_t \hat \vw_t'\, ,  \\
        &\hat{\mM}_{(i)} \inp \lm \var\left(T^{-1/2}\sum_{t \in I_{i}} \mUpsilon_t^{0'} \vz_t\left(u_t+\vv_t'\vbeta_{\vx}^0\right)\right).
        \label{hatMix}
        \end{align*}

        By Lemma \ref{lem2},
        $
        \hat \mQ_{(i)}  \inp \mathbb Q_i$.
        Letting $\widetilde{\mathbb M}(\tau)=  \mathbb M(\tau)|_{ \{\mathcal{S}_{\dag}=\mathcal{S}_u\}} + \mathbb M(\tau) |_{\{\mathcal{S}_{\dag}= \vbeta_{\vx,\#}^0\}}$, and $\mathbb \mathbb M(\tau_{s-1},\tau_s) = \int_{\tau_{s-1}}^{\tau_s} d \mathbb M(\tau)$, we have:
        \begin{eqnarray}\label{defv}
        \lm \var\left(T^{-1/2}\sum_{t \in I_{i}} \mUpsilon_t^{0'} \vz_t\left(u_t+\vv_t'\vbeta_{\vx}^0\right)\right)= \int_{\lambda_{i-1}}^{\lambda_i} \mUpsilon^{'} (\tau) \, \var(d\tilde {\mathbb M}(\tau)) \mUpsilon^{'} (\tau) = \mathbb V_{i}.
        \end{eqnarray}

        Now consider $\hat\vbeta_{\vlam_k}=\vec(\hat\vbeta_{i,\vlam_k})$. Under the null hypothesis, $\vbeta_{(i)}^0=\vbeta^0$, and so therefore we write:
        \begin{equation}
        \tilde{u}_t  = y_t - \hat \vw_t^{\prime} \vbeta^0 = u_t + (\vx_t - \hat{\vx}_t)^{\prime}
        \vbeta_{\vx}^0  = u_t + \vv_t^{\prime} \vbeta_x^0 - \vz_t^{\prime}
        \ \left(\hat{\mDelta}_t - \mDelta^0_t\right)\vbeta_x^0.\label{tildeut}
        \end{equation}
        where, by Lemma \ref{lem7},
        \begin{align*}
        T^{1/2} (\hat \mDelta_t -\mDelta_t^0) = \sum_{j=1}^{h+1}1_{t\in I_j^*}\left\{T^{-1}\sum_{t\in I_j^*}\vz_t\vz_t^\prime\right\}^{-1} T^{-1/2} \sum_{t\in I_j^*} \vz_t \vv_t^\prime +\op(1),
        \end{align*}
        where $I_j^*=\{[\pi_{j-1}^0T]+1,[\pi_{j-1}^0T]+2,\ldots,[\pi_{j}^0T]\}$. \black{From the definition of $\hat\vbeta_{i,\vlam_k}$, it follows that under $H_0$, we have
                \begin{eqnarray*}
                        T^{1/2} (\hat\vbeta_{i,\vlam_k}-\vbeta^0) &=&\left( T^{-1} \sum_{t\in I_i} \hat \vw_t\hat \vw_t^\prime\right)^{-1}\left(T^{-1/2} \sum_{t\in I_i} \hat \vw_t \tilde u_t\right)\\
                        &=&\left( T^{-1} \sum_{t\in I_i} \hat \vw_t\hat \vw_t^\prime\right)^{-1}\left(T^{-1/2} \sum_{t\in I_i} \hat \vw_t(u_t +\vv_t^\prime \vbeta_{\vx}^0)\right.\\
                        &\;&\left.\;-\; T^{-1} \sum_{t\in I_i} \hat \vw_t\vz_t^\prime\left\{\,\sum_{j=1}^{h+1}1_{t\in I_j^*}\left\{T^{-1}\sum_{t\in I_j^*}\vz_t\vz_t^\prime\right\}^{-1} T^{-1/2} \sum_{t\in I_j^*} \vz_t \vv_t^\prime\vbeta_{\vx}^0\,\right\}\right)\,+\,o_p(1) \\
                        &=& \left( T^{-1} \sum_{t\in I_i} \hat \mUpsilon_t' \vz_t \vz_t' \hat \mUpsilon_t\right)^{-1} \left(T^{-1/2} \sum_{t\in I_i} \hat \mUpsilon_t' \vz_t(u_t +\vv_t^\prime \vbeta_{\vx}^0)\right.\\
                        &-&\left.\;\; T^{-1} \sum_{t\in I_i} \hat \mUpsilon_t' \vz_t\vz_t^\prime\left\{\,\sum_{j=1}^{h+1}1_{t\in I_j^*}\left\{T^{-1}\sum_{t\in I_j^*}\vz_t\vz_t^\prime\right\}^{-1} T^{-1/2} \sum_{t\in I_j^*} \vz_t \vv_t^\prime\vbeta_{\vx}^0\,\right\}\right)\,+\,o_p(1).
        \end{eqnarray*}}
        where $I_i=\{[\lambda_{i-1}T]+1,[\lambda_{i-1}T]+2,\ldots,[\lambda_{i}T]\}$. Therefore, letting $I_i$ contain $N_i$ true breaks of the VAR representation in \eqref{ah1}, letting $s_0=\lambda_{i-1}$,  $s_{N_i+1} = \lambda_i$, as well as denoting $\mathbb Q_{\vz,j^*} = \int_{\pi_{j-1}^0}^{\pi_j^0} \mathbb Q_{\vz}(\tau),$
        \begin{align} \nonumber
        T^{1/2} (\hat\vbeta_{i,\vlam_k}-\vbeta^0)& \Rightarrow  \mathbb Q_i^{-1} \int_{\lambda_{i-1}}^{\lambda_{i}} \mUpsilon'(\tau) d\widetilde {\mathbb M}(\tau) -\mathbb Q_i^{-1} \left(\int_{\lambda_{i-1}}^{\lambda_{i}}\mUpsilon'(\tau)\mathbb{Q}_{\vz}(\tau)d\tau \right)  \sum_{j=1}^{h+1}1_{t\in I_j^*} \mathbb Q_{\vz,j^*} \tilde{\mathbb M}_{j^*} |_{ \{\mathcal{S}_{\dag}= \vbeta_{\vx,\#}^0\}}\\ \label{defn}
        & = \mathbb Q_i^{-1}\mathbb N_i (\lambda_{i-1},\lambda_i) \equiv\mathbb Q_i^{-1}\mathbb N_i ,
        \end{align}
        where $\tilde{\mathbb M}_{j^*}$ is defined as $\tilde{\mathbb M}_{i}$, but with $\pi_{j-1}^0$ replacing $\lambda_{i-1}$ and $\pi_j^0$ replacing $\lambda_{i}$.
\end{proof}
\vspace*{0.1in}
\begin{theoA}\label{theo_ellvsell+1}
        Under Assumptions \ref{a1}-\ref{a9} and the null hypothesis $H_0:m=\ell$,
        $$
        \sup\text{-}Wald_T(\ell+1\,|\,\ell)\;=\;\max_{i=1,2,\ldots\ell+1}\,\left\{\,\sup_{\varpi_i\in \mathcal{N}(\hat{\vlambda}_\ell)}\mathcal{W}_{i,1}(\varpi_i)^\prime \{\mathcal{W}_{i,2}(\varpi_i)\}^{-1}\mathcal{W}_{i,1}(\varpi_i)\right\},
        $$
        where $\mathcal W_{i,1} (\varpi_i)$ and $\mathcal W_{i,2}(\varpi_i)$ are defined in \eqref{defw1}-\eqref{defw2}.
\end{theoA}

\begin{proof}[Proof of Theorem B \ref{theo_ellvsell+1}] We begin by deriving an alternative representation of the $\sup\text{-}Wald_T(\ell+1\,|\,\ell)$. Define $\hat{\mE}_i(\varpi)$ to be the $(\hat{T}_i-\hat{T}_{i-1})\times d_\beta$ matrix with $\bar{t}^{th}$ row given by
        \begin{eqnarray*}
                \{\hat \mE_i(\varpi)\}_{\bar{t},\cdot}&=&\hat{\vw}_t^\prime,\qquad \mbox{ for } t\,=\,\hat{T}_{i-1}+ \bar{t}, \; \bar{t}=1,2,\ldots, [\varpi T],\\
                &=&\vzeros_{d_\beta}^\prime, \qquad \mbox{ for } \bar t=[\varpi T]+1,[\varpi T]+2,\ldots, \hat{T}_i.
        \end{eqnarray*}

        Then we can re-parameterize the model in (\ref{2sls_step2_mod_caseii}) as
        \begin{equation}
        \label{waldreparm}
        \vy_i\;=\;\hat{\mW}_i\vgamma\,+\,\hat{\mE}_i(\varpi_i)\valpha\,+\,\mbox{error}
        \end{equation}
        where $\vy_i$ is the $(\hat{T}_i-\hat{T}_{i-1})\times 1$ vector with $\bar{t}^{th}$ element $y_{\hat{T}_{i-1}+\bar t}$ and $\hat{\mW}_i=\vec_{\hat{T}_{i-1}+1:\hat{T}_i}(\hat{\vw}_t)$. If $\hat{\valpha}(\varpi_i)$ denotes the OLS estimator of $\valpha$ based on (\ref{waldreparm}) then it follows by straightforward arguments that $\mR_1\hat{\vbeta}(\varpi_i)=\hat{\valpha}(\varpi_i)$. Using the Frisch-Waugh theorem, we have
        $$
        \mR_1\hat{\vbeta}(\varpi_i)=\hat{\valpha}(\varpi_i)\;=\;\left\{\hat \mE_i(\varpi)^\prime \mM_{\hat{\mW}_i}\hat \mE_i(\varpi)\right\}^{-1}\hat \mE_i(\varpi)^\prime \mM_{\hat{\mW}_i}\vy_i,
        $$
        where $\mM_{\hat{\mW}_i}=\mI_{\hat{T}_i-\hat{T}_{i-1}}\,-\,{\hat{\mW}_i}\left({\hat{\mW}_i}^\prime {\hat{\mW}_i}\right)^{-1}{\hat{\mW}_i}^\prime$.

        Let $\mE_i(\varpi)$ and $\mW_i$ be defined analogously to $\hat \mE_i(\varpi)$ and $\hat \mW_i$ only replacing $\hat{T}_{i-1}$ and $\hat{T}_i$ by $T_{i-1}^0$ and $T_{i}^0$ respectively. Further define $\tilde{\vu}_i$ to be the $(T_i^0-T_{i-1}^0)\times 1$ vector with $\bar{t}^{th}$ element $\tilde{u}_{T_{i-1}^0+\bar{t}}$, where
        $$
        \tilde{u}_t  \equiv y_t - \hat \vw_t^{\prime} \vbeta_{(i)}^0 = u_t + (\vx_t - \hat{\vx}_t)^{\prime}
        \vbeta_{\vx,(i)}^0  = u_t + \vv_t^{\prime} \vbeta_{\vx,(i)}^0 - \vz_t^{\prime}
        \ \left(\hat{\mDelta}_t - \mDelta^0_t\right)\vbeta_{\vx,(i)}^0.
        $$
        Under $H_0$, from Lemma \ref{lem7}, it follows that $\hat{T}_i-T_i^0=O_p(1)$ for $i=1,2,\ldots, \ell$, so we can treat $T_i^0$ as known for the rest of this proof. Additionally, $\vbeta_{\vx,(i)}^0$ is constant in interval $[T_{i-1}^0+1,T_i^0]$, so $\mR_1 \vec(\vbeta_{\vx,(i)}^0, \vbeta_{\vx,(i)}^0)=\vzeros_{2d_{\beta}}$, therefore:
        $$
        T^{1/2}\mR_1\hat{\vbeta}(\varpi_i)\;=\;T^{1/2}\left\{\mE_i(\varpi)^\prime \mM_{\mW_i} \mE_i(\varpi)\right\}^{-1} \mE_i(\varpi)^\prime \mM_{\hat{\mW}_i}\tilde \vu_i\,+\,o_p(1),
        $$
        uniformly in $\varpi_i$. Therefore,
        \begin{equation}\label{supWald_ellv2}
        \sup\text{-}Wald_T(\ell+1\,|\,\ell)\;=\;\max_{i=1,2,\ldots\ell+1}\,\left\{\,\sup_{\varpi_i\in \mathcal{N}(\hat{\vlambda}_\ell)}Wald_T(\varpi_i;\,\ell)\right\}\,+\,o_p(1)
        \end{equation}
        where
        \begin{eqnarray}
        Wald_T(\varpi_i;\,\ell)&=&T\tilde \vu_i^\prime \mM_{\mW_i}\mE_i(\varpi)\left\{\mE_i(\varpi)^\prime \mM_{\mW_i} \mE_i(\varpi)\right\}^{-1}\left(\mR_1\bar{\mV}(\varpi_i)\mR_1^\prime\right)^{-1}\nonumber\\[0.1in]
        &\;&\;\times \left\{\mE_i(\varpi)^\prime \mM_{\mW_i} \mE_i(\varpi)\right\}^{-1} \mE_i(\varpi)^\prime \mM_{\hat{\mW}_i}\tilde \vu_i,
        \label{Waldella}
        \end{eqnarray}
        where
        \begin{align*}
        &\hat{\mV}(\varpi_i)\;=\;\diag\left(\,\hat \mV_{1}(\varpi_i),\hat \mV_{2}(\varpi_i)\,\right) , \qquad \hat \mV_{j}(\varpi_i)\;=\;\{\hat \mQ_{j}(\varpi_i)\}^{-1}  \ \hat \mM_{j}(\varpi_i) \ \{\hat \mQ_{j}(\varpi_i)\}^{-1},\\
        & \hat \mQ_{j}(\varpi_i) = T^{-1} \sum_{t \in  I_{i}^{(j)}} \hat \vw_t \hat \vw_t^\prime,\qquad
        \hat{\mM}_{j}(\varpi_i) \inp \lm \var\left(T^{-1/2}\sum_{t \in I_{i}^{(j)}} \mUpsilon_t^{0\prime} \vz_t\left(u_t+\vv_t^\prime\vbeta_{\vx,(i)}^0\right)\right),
        \end{align*}
        and, $I_i^{(1)}(\varpi_i)=\{t:\,[\lambda_{i-1}^0T]+1,[\lambda_{i-1}^0T]+2,\ldots,[\varpi_iT]\}$ and  $I_i^{(2)}(\varpi_i)=\{t:\,[\varpi_iT]+1,[\varpi_iT]+2,\ldots,[\lambda_i^0T]\}$.

        \vspace*{0.1in}
        To derive the limit of the $Wald_T(\varpi_i;\,\ell)$, note that by Lemma \ref{lem7} and $\hat \mUpsilon_t \inp \mUpsilon_t^0$, $T^{-1/2}\sum_{t=T_{i-1}^0+1}^{[T\varpi_i]}\hat{\vw}_t\tilde{u}_t \Rightarrow \mathbb N (\lambda_{i-1}^0,\varpi_i) $, where the definition of $\mathbb N (\cdot,\cdot)$ for this entire section is as in equation \eqref{defn}  but with $\vbeta_{\vx}^0$ replaced by $\vbeta_{\vx,(i)}^0$ when $\mathbb N(\cdot,\cdot)$ is evaluated in (a subset of) the interval $[\lambda_{i-1}^0,\lambda_i^0]$. Similarly, $T^{-1/2}\sum_{T_{i-1}^0+1}^{T_i^0}\hat{\vw}_t\tilde{u}_t \Rightarrow \mathbb N (\lambda_{i-1}^0,\lambda_i^0)$. Similar to \eqref{defv}, one can derive
        $ \mathbb V_{i,j} (\varpi_i)= \lm \var\left(T^{-1/2}\sum_{t \in I_{i}^{(j)}} \mUpsilon_t^{0\prime} \vz_t\left(u_t+\vv_t^\prime\vbeta_{\vx,(i)}^0\right)\right)$.

        Using these results, Lemma \ref{lem2}, and letting $\mQ(\tau)=\mUpsilon'(\tau)\mathbb Q_{\vz}(\tau)\mUpsilon(\tau)$, we have:
        \begin{eqnarray*}
                T^{-1}\mW_i^\prime \mW_i &\stackrel{p}{\to}&\int_{\lambda_{i-1}^0}^{\lambda_i^0}
                \mathbb{Q}(\tau) d\tau \equiv \mathbb Q_{i}^0\\[0.1in]
                T^{-1} \mE_i(\varpi_i)^\prime \mW_i \;=\;\hat \mQ_{1}(\varpi_i)&\stackrel{p}{\to}&\int_{\lambda_{i-1}^0}^{\lambda_{i-1}^0+\varpi_i}\mathbb{Q}(\tau)d\tau\;\equiv\;\mathbb{Q}_{i,1}(\varpi_i),\\[0.1in]
                \hat \mQ_{2}(\varpi_i)&\stackrel{p}{\to}&\int_{\lambda_{i-1}^0+\varpi_i}^{\lambda_i^0}\mathbb{Q}(\tau)d\tau\;\equiv\;\mathbb{Q}_{i,2}(\varpi_i),\\[0.1in]
                \hat{\mM}_{1}(\varpi_i)&\stackrel{p}{\to}& \mathbb V_{i,1}(\varpi_i)\\[0.1in]
                \hat{\mM}_{2}(\varpi_i)&\stackrel{p}{\to}& \mathbb V_{i,2}(\varpi_i),\\[0.1in]
                T^{-1/2}\mE_i(\varpi_i)^\prime\tilde{u}_i&\Rightarrow&\mathbb N(\lambda_{i-1}^0,\varpi_i),\\[0.1in]
                T^{-1/2}\mW_i^\prime\tilde{u}_i&\Rightarrow& \mathbb N(\lambda_{i-1}^0,\lambda_i^0).
        \end{eqnarray*}
        It then follows by the continuous mapping theorem (CMT) that:
        $$
        Wald_T(\varpi_i;\,\ell)\;\Rightarrow\; \mathcal{W}_{i,1}'(\varpi_i) \{\mathcal{W}_{i,2}(\varpi_i)\}^{-1}\mathcal{W}_{i,1}'(\varpi_i)
        $$
        where
        \begin{eqnarray}\label{defw1}
        \mathcal{W}_{i,1}&=&\left( \mathbb Q_{i,1}(\varpi_i)\, \{\mathbb Q_i^0\}^{-1} \, \mathbb Q_{i,1}(\varpi_i)\right)^{-1}\left(\mathbb Q_{i,1}(\varpi_i)\,\, \{\mathbb Q_i^0\}^{-1} \,\, \mathbb N(\lambda_{i-1}^0,\lambda_i^0) \right),\\[0.1in] \label{defw2}
        \mathcal{W}_{i,2}(\varpi_i)&=&\mR_1\,\diag_{j=1:2}\left(
        \mathbb V_{i,j}(\varpi_i)\right) \mR_1^\prime.
        \end{eqnarray}

        Therefore, by the CMT, we have
        $$
        sup\text{-}Wald_T(\ell+1\,|\,\ell)\;=\;\max_{i=1,2,\ldots\ell+1}\,\left\{\,\sup_{\varpi_i\in \mathcal{N}(\hat{\vlambda}_\ell)}\mathcal{W}_{i,1}(\varpi_i)^\prime \{\mathcal{W}_{i,2}(\varpi_i)\}^{-1}\mathcal{W}_{i,1}(\varpi_i)\right\}.
        $$
\end{proof}

\section{Analysis of $\mathcal{C}_{2}$, $\mathcal{B}_6$ and $\mathcal{B}_{1,2}$ defined in Section \ref{supsec1}, Proof of Lemma \ref{lem2}} \label{supsec2}
This section considers $\mathcal{C}_{2}$, $\mathcal{B}_6$ and $\mathcal{B}_{1,2}$ in more detail, but it also shows that  $\sum_{l=0}^{\infty} \| \mF_s^l \| < \infty$.

Let $\{\gamma_i;\,i=1,2,\ldots np\}$ be the eigenvalues of $\mF_s$. We consider four cases relating to the roots of the characteristic equation of VAR in (\ref{a1}): (i) real and distinct; (ii) real and repeated; (iii) complex and distinct; (iv) complex and repeated. (The case of repeated complex and real roots is easily deduced from Case's (i) and (iv) and so is omitted for brevity.) \\[0.1in]

\noindent
{\it Case (i): real and distinct roots }

We have, for $\mOmega_j =\mA_s^{-1}\mSigma_j\{\mA_s^{-1}\}$ and $\mSigma,\overline \mSigma$ defined in the main paper, at the beginning of its appendix,
\begin{eqnarray*}
\|\mathcal{C}_{2}\|&=&T^{-1}\bigg\| \sum_{l=1}^{(\Delta\tau_sT)-1} \mF_s^l\left\{\,\sum_{j=[\tau_sT]-l+1}^{[\tau_sT]}\mOmega_{j}\,\right\} \mF_s^{l\prime}\bigg\|,\\[0.1in]
&\leq&T^{-1} \sum_{l=1}^{(\Delta\tau_sT)-1} \bigg\|\mF_s^l\left\{\,\sum_{j=[\tau_sT]-l+1}^{[\tau_sT]}\mOmega_{j}\,\right\} \mF_s^{l\prime}\bigg\|,\\[0.1in]
&\leq& T^{-1}\,\sum_{l=1}^{(\Delta\tau_sT)-1}\|\mF_s^l\|\,\bigg\|\left\{\,\sum_{j=[\tau_sT]-l+1}^{[\tau_sT]}\mOmega_{j}\,\right\}\,\bigg\|\| \mF_s^{l\prime}\|\\[0.1in]
&\leq& T^{-1}\,\sum_{l=1}^{(\Delta\tau_sT)-1}\|\mF_s^l\|^2\,\bigg\|\left\{\,\sum_{j=[\tau_sT]-l+1}^{[\tau_sT]} \mOmega_{j}\,\right\}\bigg\|\,
\end{eqnarray*}

Given the definition of $\mOmega_{j}$, we have
$$
\|\mOmega_j\|\;=\;\|\mA_s^{-1}\mSigma_j\{\mA_s^{-1}\}^\prime\|\;\leq\;\|\mA_s^{-1}\|^2\,\bigg\|\,\sum_{j=[\tau_sT]-l+1}^{[\tau_sT]}\overline \mSigma_j\,\bigg\|
$$
Under our assumptions, we have  $\|\mA_s^{-1}\|^2=O(1)$. \black{Since} $\sup\|\overline \mSigma_t\|<C_\Sigma$ for some finite positive constant $C_\Sigma$, we have have for $l=1,2,\ldots,(\Delta\tau_sT)-1$:
$$
\bigg\|\,\sum_{j=[\tau_sT]-l+1}^{[\tau_sT]}\overline \mSigma_j\,\bigg\|\;\leq \; l C_\Sigma
$$

Now consider $\|\mF_s^l\|$. From \citet{Hamilton:1994} (p.259) and \citet{Lutkepohl:1993} (p.460), it follows that under the assumptions about the roots,\footnote{The dependence of both $\mP$ and $\mGamma$ on $s$ is suppressed for ease of notation.}
$$
\mF_s\;=\;\mP^{-1}\mGamma \mP
$$
where $\mGamma=\diag(\gamma_1,\gamma_2,\ldots,\gamma_{np})$ and $\mP$ is a real matrix with $\mP=O(1)$, and hence that
$$
\mF_s^l\;=\;\mP^{-1}\mGamma^l \mP.
$$
Therefore, we have:
\begin{equation}
\label{normFl}
\|\mF_s^l\|\;\leq\;\|\mP^{-1}\|\|\mGamma^l\|\| \mP\|\;\leq\; \|\mP^{-1}\|\,\| \mP\| \sqrt{np}\, \gamma_{max, \mGamma^2}^l = \sqrt{np}\, \gamma_{max, \mGamma^2}^l,
\end{equation}
where $\gamma_{max,\mGamma^2}$ is the positive square root of the maximum eigenvalue of $\mGamma^2$.
\black{Therefore,}
$$
\|\mathcal{C}_{2}\| \leq T^{-1}\,\sum_{l=1}^{(\Delta\tau_sT)-1} \sqrt{np}\, l \gamma_{max, \gamma^2}^l C_\Sigma
$$


Since $\sum_{i=1}^\infty i \gamma^i=\gamma(1-\gamma)^{-2}=O(1)$, it follows that $\mathcal{C}_{2}=o(1)$.\\[0.1in]

Now consider $\mathcal{B}_6$. We have
\begin{eqnarray*}
\|\mathcal{B}_6\|&=&T^{-1}\bigg\|\sum_{t\in I_i} \left(\,\sum_{i=0}^{\tilde{t}-1}\mF_s^l\vmu_s\right)\vxi_{[\tau_{s-1}T]}^\prime \mF_s^{\tilde{t}\prime}\bigg\|\\[0.1in]
&\leq &T^{-1}\sum_{t\in I_i} \,\sum_{i=0}^{\tilde{t}-1}\,\|\,\mF_s^l\vmu_s\vxi_{[\tau_{s-1}T]}^\prime \mF_s^{\tilde{t}\prime}\,\|\\[0.1in]
&\leq &T^{-1}\|\vmu_s\|\,\|\vxi_{[\tau_{s-1}T]}^\prime\|\,\sum_{t\in I_i} \,\sum_{l=0}^{\tilde{t}-1}\,\|\mF_s^l\|\,\|\mF_s^{\tilde{t}\prime}\,\|.
\end{eqnarray*}
Using (\ref{normFl}), we have:

$$
\sum_{t\in I_i} \,\sum_{i=0}^{\tilde{t}-1}\,\|\mF_s^l\|\,\|\mF_s^{\tilde{t}\prime}\,\|\;\leq\; K \sum_{t\in I_i} \,\sum_{l=0}^{\tilde{t}-1} \gamma^{l+\tilde{t}}
$$
where $\gamma=\gamma_{max, \gamma^2}$ and $K$ is a bounded positive constant. Expanding the sums and using $0<\gamma<1$, it can be shown that $\sum_{t\in I_i} \,\sum_{l=0}^{\tilde{t}-1} \gamma^{l+\tilde{t}}\leq \sum_{l=1}^\infty l\gamma^l=(1-\gamma)^{-2}=O(1)$ and so
$$
\|\mathcal{B}_6\|\;\leq T^{-1} K\|\vmu_s\|\,\|\vxi_{[\tau_{s-1}T]}^\prime\| \,(1-\gamma)^{-2}\;=\;o_p(1)\;\stackrel{p}{\to}\;0.
$$

Now consider $\mathcal{B}_{1,2}$. Recall that
$$
\mathcal{B}_{1,2}\;=\;T^{-1}\sum_{t\in I_i} \left( \sum_{l=0}^{\tilde{t}-1}\mF_s^l\vmu_s\right)\,\left(\,\sum_{l=0}^{\tilde{t}-1}\mF_s^l\vmu_s\right)^\prime,
$$
where $\tilde{t}=t-[\tau_{s-1}T]$.
$$
\sum_{l=0}^{\tilde{t}-1}\mF_s^l\vmu_s\;=\;\sum_{l=0}^{\Delta\tau_{s}T-1}\mF_s^l\vmu_s\,-\,\sum_{\tilde{t}}^{\Delta\tau_{s}T-1}\mF_s^l\vmu_s
$$
we have
\begin{eqnarray}
\mathcal{B}_{1,2}&=&T^{-1}\sum_{t\in I_i}\left(\,\sum_{l=0}^{\Delta\tau_{s}T-1}\mF_s^l\vmu_s\,\right)\left(\,\sum_{l=0}^{\Delta\tau_{s}T-1}\mF_s^l\vmu_s\,\right)^\prime\nonumber\\
&\;&-\;T^{-1}\sum_{t\in I_i}\left(\,\sum_{l=0}^{\Delta\tau_{s}T-1}\mF_s^l\vmu_s\,\right)\left(\,\sum_{l=\tilde{t}}^{\Delta\tau_{s}T-1}\mF_s^l\vmu_s\,\right)^\prime\nonumber\\
&\;&-\;T^{-1}\sum_{t\in I_i}\left(\,\sum_{l=\tilde{t}}^{\Delta\tau_{s}T-1}\mF_s^l\vmu_s\,\right)\left(\,\sum_{l=0}^{\Delta\tau_{s}T-1}\mF_s^l\vmu_s\,\right)^\prime\nonumber\\
&\;&+\;T^{-1}\sum_{t\in I_i}\left(\,\sum_{l=\tilde{t}}^{\Delta\tau_{s}T-1}\mF_s^l\vmu_s\,\right)\left(\,\sum_{l=\tilde{t}}^{\Delta\tau_{s}T-1}\mF_s^l\vmu_s\,\right)^\prime\label{AppBB12}.
\end{eqnarray}
The first term on the right-hand side (rhs) of (\ref{AppBB12}) converges to $\mathbb{B}_2(\tau_{s-1},\tau_2)$.
Since $\sum_{t\in I_i}\sum_{\tilde{t}}^{\Delta\tau_{s}-1}\mF_s^l\vmu_s=\sum_{l=1}^{\Delta \tau_sT-1}l \mF_s^l\vmu_s$ and is $O(1)$ by same arguments as above, it follows that the second and third terms on the rhs of  (\ref{AppBB12}) are $o(1)$. Now consider the fourth term.
Noting that
$$
\sum_{l=j}^{\Delta\tau_{s}T-1}\mF_s^l\vmu_s=\mF_s^j\vmu_s\,+\,\sum_{l=j+1}^{\Delta\tau_{s}T-1}\mF_s^l\vmu_s
$$
and so
\begin{eqnarray*}
\left(\sum_{l=j}^{\Delta\tau_{s}T-1}\mF_s^l\vmu_s\right)\left(\sum_{l=j}^{\Delta\tau_{s}T-1}\mF_s^l\vmu_s\right)^\prime&=&\mF_s^j\vmu_s\vmu_s^\prime \mF_s^j\,+\,\mF_s^j\vmu_s\left(\sum_{l=j+1}^{\Delta\tau_{s}T-1}\mF_s^l\vmu_s\right)^\prime\,+\,\left(\sum_{l=j+1}^{\Delta\tau_{s}T-1}\mF_s^l\vmu_s\right)\vmu_s^\prime \mF_s^j\\
&\;&\;+\;\left(\sum_{l=j+1}^{\Delta\tau_{s}T-1}\mF_s^l\vmu_s\right)\left(\sum_{l=j+1}^{\Delta\tau_{s}T-1}\mF_s^l\vmu_s\right)^\prime,
\end{eqnarray*}
it follows that
\begin{eqnarray}
\sum_{t\in I_i}\left(\,\sum_{l=\tilde{t}}^{\Delta\tau_{s}T-1}\mF_s^l\vmu_s\,\right)\left(\,\sum_{l=\tilde{t}}^{\Delta\tau_{s}T-1}\mF_s^l\vmu_s\,\right)^\prime&=&\sum_{l=1}^{\Delta\tau_s T-1}l\mF_s^l\vmu_s\vmu_s^\prime\mF_s^{l\prime}\,+\,\sum_{l=1}^{\Delta\tau_s T-2}l\mF_s^l\vmu_s\vmu_s^\prime\sum_{j=l+1}^{\Delta\tau_s T-1}\mF_s^{j\prime}\nonumber\\&\;&\,+\,\sum_{l=1}^{\Delta\tau_s T-2}\sum_{j=l+1}^{\Delta\tau_s T-1}\mF_s^{j}\vmu_s\vmu_s^\prime l\mF_s^l.\label{AppBB12a}
\end{eqnarray}
Using similar arguments to above, it can be shown that the rhs of (\ref{AppBB12a}) is $O(1)$ and so the fourth term on the rhs of (\ref{AppBB12}) is $o(1)$. This completes the derivation.

\vspace*{0.1in}
\noindent
{\it Case (ii) roots are real but not distinct}\\
Using the Jordan decomposition, Magnus \& Neudecker (1991, p.17), there is a nonsingular matrix $P$ such that
$$
\mF\;=\;\mP^{-1}\mGamma \mP
$$
where
$$
\mGamma=blockdiag(\mGamma_1,\mGamma_2,\ldots,\mGamma_k)
$$
and $\mGamma_j$ is the $n_j\times n_j$ matrix
$$
\mGamma_j\;=\;\left[\,\begin{array}{ccccc} \gamma_j & 1 & 0 & \hdots & 0\\ 0 & \gamma_j & 1 & \hdots & 0\\\vdots & \vdots & \vdots & {}&\vdots\\ 0 & 0 & 0 &\hdots & 1\\0 & 0 & 0 & \hdots & \gamma_j\end{array}\,\right],
$$
and $\{\gamma_j\}$ are the eigenvalues of $F$.\footnote{An eigenvalue may appear in more than one block, see \citet{Lutkepohl:1993} (p.460).} As before, we have:
$$
\mF^l\;=\;\mP^{-1}\mGamma^l \mP
$$
where from \citet{Lutkepohl:1993} (p.460),
$$
\mGamma=blockdiag(\mGamma_1^l,\mGamma_2^l,\ldots,\mGamma_k^l)
$$
and
\begin{equation}
\label{Lamjl}
\mGamma_j^l\;=\;\left[\,\begin{array}{ccccc} \gamma_j^l & { l \choose 1}\gamma_j^{l-1} & { l \choose 2}\gamma_j^{l-2} & \hdots & { l \choose n_j-1}\gamma_j^{l-n_j+1}\\ 0 & \gamma_j^l & { l \choose 1}\gamma_j^{l-1} & \hdots & { l \choose n_j-2}\gamma_j^{l-n_j+2}\\\vdots & \vdots & \vdots &{}& \hdots\\ 0 & 0 & 0 &\hdots & { l \choose 1}\gamma_j^{l-1}\\0 & 0 & 0 & \hdots & \gamma_j^l\end{array}\,\right].
\end{equation}
where ${l \choose c}=0$ if $c>l$.\\[0.1in]

Given the block diagonal structure,
$$
\|\mGamma^l\|\;=\; \left(\sum_{i=1}^k \|\mGamma_i^l\|^2\,\right)^{1/2},
$$
and from (\ref{Lamjl}),
$$
\|\mGamma_i^l\|^2\;=\;\sum_{m=0}^{n_j-1}\sum_{c=0}^m\left\{{l \choose c}\gamma_i^{l-c}\right\}^2.
$$

\vspace*{0.1in}
Set $\overline n=\max_j\{n_j\}$. Then:
$$
\|\mGamma_i^l\|^2\;\leq\;\sum_{m=0}^{\overline n}\sum_{c=0}^m\left\{{l \choose c}\gamma_*^{l-c}\right\}^2,
$$
\black{where $\gamma^*=\max\{\tilde{\gamma}_i; \,i=1,2,\ldots,np\}$, and $\tilde{\gamma}_i=|\gamma_i|$.}
For $0\leq c\leq \overline n$ and $l=1,2,\ldots \overline n$, we have (as $|\gamma_i|<1$)
$$
\left\{{l \choose c}\gamma_*^{l-c}\right\}^2\leq \left\{l^{\overline n}\gamma_*^{l-\overline n}\right\}^2
$$
and so
\begin{equation}
\label{normGAmmal}
\|\mGamma^l\|^2\;\leq \;(\overline n +1)^2(l^{\overline n }\gamma_*^{l-\overline n})^2.
\end{equation}

Therefore, for $l\geq \overline n$, we have
$$
\|\mGamma^l\|\;\leq \;K_1l^{\overline n}\gamma_*^{l},
$$
for some finite positive constant $K_1$.\\[0.1in]
%
We have
$$
\sum_{l=1}^{(\Delta\tau_sT)-1}\|\mF^l\|^2\,l C_\Sigma\;=\;\sum_{l=1}^{\overline n}\|\mF^l\|^2\,l C_\Sigma\,+\,\sum_{l=\overline n +1}^{{(\Delta\tau_sT)-1}}\| \mF^l\|^2\,l C_\Sigma,
$$
where the first term on the right-hand side is evidently $O(1)$ as $\overline n$ is finite. So now consider the other term. By similar arguments to Case (i),
$$
\sum_{l=\overline n+1}^{{(\Delta\tau_sT)-1}}\|\mF^l\|^2\,l C_\Sigma\;\leq\; K_3 \sum_{l=\overline n+1}^{{(\Delta\tau_sT)-1}}l\|\mGamma^l\|\;\leq\;K \sum_{l=\overline n+1}^{{(\Delta\tau_sT)-1}}l^{\overline n +1}\gamma_*^{l}\;=\;O(1)
$$
(The last quality holds by the following reasoning. We have $\sum_{j=0}^\infty x^j=(1-x)^{-1}$. Differentiating both sides with respect to $x$ gives $\sum_{j=1}^\infty jx^{j-1}=(1-x)^{-2}$ and so $\sum_{j=1}^\infty jx^{j}=x(1-x)^{-2}$. The second derivative gives $\sum_{j=2}^\infty j(j-1)x^{j-2}=2(1-x)^{-3}$ and so $\sum_{j=2}^\infty j(j-1)x^{j}=2x^2(1-x)^{-3}$ and $\sum_{j=2}^\infty j^2x^{j}=2x^2(1-x)^{-3}-\{x(1-x)^{-2}-x\}=O(1)$. Repeated differentiation and this logic yields: $\sum_{j=\overline n}^\infty j^{\overline n}x^{j}=O(1)$.)\\[0.1in]

\noindent
It then follows by similar arguments to Case (ii) that $\mathcal{C}_{2}=o(1)$. The proof that $\mathcal{B}_6=o_p(1)$  and $\mathcal{B}_{1,2}\to\mathbb{B}(\tau_{s-1},\tau_s)$ follow via similar arguments and so is omitted for brevity.\\[0.1in]

\noindent
{\it Case 3: distinct complex roots}\\[0.1in]
Without loss of generality, we consider the case where the first two eigenvalues of $F$ are complex and the remainder are all real. Let $\gamma_1=\mu+i\nu$ and $\gamma_2=\mu-i\nu$. Given that the VAR is stationary, we have $r=\sqrt{\mu^2+\nu^2}<1$. In this case, we have:
$$
\mF^l=\mP^{-1} \overline \mGamma^l \mP
$$
where
$$
\overline \mGamma\;=\;\left[\,\begin{array}{ccccccc} \mu & -\nu & 0 & 0 &\hdots &0\\\nu & \mu & 0 & 0 & \hdots & 0\\0 & 0 & \gamma_3 & 0 &\hdots &0\\\vdots &\vdots & \vdots &\ddots &\hdots \\ 0 & 0 & 0 & 0 &\hdots & \gamma_{np}\end{array}\,\right].
$$
Using polar coordinate, we write $\mu=rCos(\theta)$ and $\nu=rSin(\theta)$. It then follows that
$$
\overline \mGamma^l\;=\;\left[\,\begin{array}{ccccccc} r^lCos(l\theta) & -r^lSin(l\theta) & 0 & 0 &\hdots &0\\ r^lSin(l\theta) & r^lCos(l\theta) & 0 & 0 & \hdots & 0\\0 & 0 & \gamma_3^l & 0 &\hdots &0\\\vdots &\vdots & \vdots &\ddots &\hdots \\ 0 & 0 & 0 & 0 &\hdots & \gamma_{np}^l\end{array}\,\right].
$$
Thus we have
$$
\| \overline \mGamma^l\|^2\;=\;2r^{2l}\,+\,\gamma_3^{2l}\,+\,\ldots\,+\,\gamma_{np}^{2l}.
$$
We can then apply the same argument as in Case (i) to show $\mathcal{C}_{2}=o(1)$ provided we replace $\gamma_{max,\mGamma^2}$ by $\max\{r, |\gamma_3|,\ldots,|\gamma_{np}|\}$. The proof that $\mathcal{B}_6=o_p(1)$ and $\mathcal{B}_{1,2}\to\mathbb{B}_2(\tau_{s-1},\tau_s)$ follow via similar arguments and so is omitted for brevity.\\[0.1in]

\noindent
{\it Case 4: repeated complex roots}\\[0.1in]
Without loss of generality, we consider the case where the first four eigenvalues of $\mF$ are complex and repeated, and the remainder are all real. Let $\gamma_1=\gamma_3=\mu+i\nu$ and $\gamma_2=\gamma_4=\mu-i\nu$. Given that the VAR is stationary, we have $r=\sqrt{\mu^2+\nu^2}<1$. In this case, we have:
$$
\mF^l=\mP^{-1} \overline \mGamma^l \mP
$$
where
$$
\overline \mGamma\;=\;\left[\,\begin{array}{cccccccc} \mu & -\nu & 1 & 0 &0&0&\hdots &0\\\nu & \mu & 0 & 1 &0& 0&\hdots & 0\\
0&0&\mu & -\nu  &0&0&\hdots &0\\0&0&\nu & \mu & 0&0&\hdots & 0\\
0&0&0 & 0 & \gamma_5 & 0 &\hdots &0\\\vdots &\vdots &\vdots &\vdots & \vdots &\vdots &\hdots &\vdots \\ 0&0&0 & 0 & 0 & 0 &\hdots & \gamma_{np}\end{array}\,\right],
$$
and
$$
\overline \mGamma^l\;=\;\left[\,\begin{array}{cccccccc} r^la(l) & -r^lb(l) & lr^{l-1}a(l-1) &-lr^{l-1}b(l-1) &0&\hdots &0\\r^lb(l) & r^la(l) & lr^{l-1}b(l-1) & lr^{l-1}a(l-1) &0& 0&\hdots & 0\\
0&0&r^la(l) & -r^lb(l)  &0&0&\hdots &0\\0&0&r^lb(l) & r^la(l) & 0&0&\hdots & 0\\
0&0&0 & 0 & \gamma_5^l & 0 &\hdots &0\\\vdots &\vdots &\vdots &\vdots & \vdots &\vdots &\hdots &\vdots \\ 0&0&0 & 0 & 0 & 0 &\hdots & \gamma_{np}^l\end{array}\,\right],
$$
where $a(l)=Cos(l\theta)$ and $b(l)=Sin(l\theta)$. Therefore, we have\footnote{Using $Cos^2(\theta)+Sin^2(\theta)=1$.}
$$
\|\overline \gamma \|^2\;=\;4 r^{2l}\,+\,2l^2r^{2(l-1)}\,\gamma_5^{2l}\,+\,\ldots\,+\,\gamma_{np}^{2l}
$$
and
$$
\|\overline \mGamma \|\;=\;Kl^2\gamma_*^{(l-1)}
$$
where $\gamma_*=\max\{r, |\gamma_5|,\ldots,|\gamma_{np}|\}$. By a similar argument to Case (ii), it follows that $\mathcal{C}_{2}=o(1)$. The proof that $\mathcal{B}_6=o_p(1)$ and $\mathcal{B}_{1,2}\to\mathbb{B}_2(\tau_{s-1},\tau_s)$ follow via similar arguments and so is omitted for brevity.

Inspecting the derivations, it also becomes clear that we showed that $\sum_{l=0}^{\infty} \| \mF_s^l \| < \infty$.

\section{Proof of \eqref{mathcalR}} \label{supsec3}
We  have
$$
\sum_{\ell=0}^\infty \|\hat \mF_s^\ell\,-\,\mF_s^\ell\|\;=\;\sum_{\ell=1}^\infty \|\hat \mF_s^\ell\,-\,\mF_s^\ell\|.
$$
Using
$$
\hat \mF_s^\ell\,-\,\mF_s^\ell\;=\;\hat \mF_s^{\ell-1}(\hat \mF_s\,-\,\mF_s)\,+\,(\hat \mF_s^{\ell-1}\,-\,\mF_s^{\ell-1})\mF_s,
$$
we have, via repeated back substitution,
$$
\hat \mF_s^\ell\,-\,\mF_s^\ell\;=\;\sum_{j=1}^\ell\hat \mF_s^{j-1}\,(\hat \mF_s\,-\,\mF_s)\mF_s^{\ell-j}.
$$
Therefore, it follows that
$$
\|\hat \mF_s^\ell\,-\,\mF_s^\ell\|\;\leq\;\sum_{j=1}^\ell\|\hat \mF_s^{j-1}\|\,\|\hat \mF_s\,-\,\mF_s\|\,\|\mF_s^{\ell-j}\|\;=\;\|\hat \mF_s\,-\,\mF_s\|\,\sum_{j=1}^\ell\|\hat \mF_s^{j-1}\|\,\|\mF_s^{\ell-j}\|,
$$
and
$$
\sum_{\ell=0}^\infty \|\hat \mF_s^\ell\,-\,\mF_s^\ell\|\;=\sum_{\ell=1}^\infty \|\hat \mF_s^\ell\,-\,\mF_s^\ell\|\;\leq\;\|\hat \mF_s\,-\,\mF_s\|\sum_{\ell=1}^\infty\sum_{j=1}^\ell\|\hat \mF_s^{j-1}\|\,\|\mF_s^{\ell-j}\|.
$$

We now show that $\sum_{\ell=1}^\infty\sum_{j=1}^\ell\|\hat \mF_s^{j-1}\|\,\|\mF_s^{\ell-j}\| =\op(1)$ for the cases in Section \ref{supsec2}.\\[0.1in]

\noindent
{\it Case (i): real and distinct roots}\\
From \eqref{normFl}, it follows that (using $K$ repeatedly to denote any finite constant)
\begin{eqnarray*}
\|\hat \mF_s^{j-1}\|&\leq&K\hat \gamma_*^{j-1},\\
\|\mF_s^{\ell-j}\|&\leq&K \gamma_*^{\ell-j}.
\end{eqnarray*}
where $\gamma_*$ denotes $\gamma_{max,\mGamma}^2$ from Section \ref{supsec2}, and $\hat \gamma_*$ is its analogue based on $\hat \mF_s$. Let $\gamma_\bullet=max\{\gamma_*,\hat \gamma_*\}$. Then, we have
$$
\sum_{j=1}^\ell\|\hat \mF_s^{j-1}\|\,\|\mF_s^{\ell-j}\|\;\leq\;K\ell \gamma_\bullet^{\ell-1},
$$
and so
$$
\sum_{\ell=1}^\infty\sum_{j=1}^\ell\|\hat \mF_s^{j-1}\|\,\|\mF_s^{\ell-j}\|\;\leq\;K\sum_{\ell=1}^\infty\ell \gamma_\bullet^{\ell-1}\;=\;O_p(1),
$$
from Assumption \ref{a7}.\footnote{Note that $\hat \mF_s\stackrel{p}{\to}\mF_s$ implies $\hat\gamma_*\stackrel{p}{\to}\gamma_*$ and $0<\gamma_*<1$ from Assumption \ref{a7}.}

\vspace*{0.1in}
\noindent
{\it Case (ii): roots are real but repeated}\\
Using \eqref{normGAmmal}, we have for $\ell\geq \bar{n}$ (and $\bar n$ is defined as in Section \ref{supsec2} Case (ii))
\begin{eqnarray*}
\|\hat \mF_s^{j-1}\|&\leq&K(j-1)^{\bar n}\hat \gamma_*^{j-1},\\
\|\mF_s^{\ell-j}\|&\leq&K (\ell-j)^{\bar n}\gamma_*^{\ell-j}.
\end{eqnarray*}
where $\gamma_*$ is defined as in Section \ref{supsec2} Case (ii) and $\hat \gamma_*$ is defined analogously only replacing $\mF_s$ with $\hat \mF_s$. For ease of presentation, assume the same multiplicity of eigenvalues for $\mF_s$ and $\mF_s$; it is straightforward to modify the argument to allow for different multiplicities. Using $j\leq \ell$, we have
$$
\|\hat \mF_s^{j-1}\|\,\|\mF_s^{\ell-j}\|\;\leq\;K \ell^{2\bar n}\gamma_\bullet^{\ell-1},
$$
and
$$
\sum_{j=1}^\ell\|\hat \mF_s^{j-1}\|\,\|\mF_s^{\ell-j}\|\;\leq\;K\ell^{2\bar n+1}\gamma_\bullet^{\ell-1}.
$$
This gives
$$
\sum_{\ell=1}^\infty\sum_{j=1}^\ell\|\hat \mF_s^{j-1}\|\,\|\mF_s^{\ell-j}\|\;\leq\;K\sum_{\ell=1}^\infty\ell^{2\bar n+1} \gamma_\bullet^{\ell-1}\;=\;O_p(1),
$$
using the results noted toward the end of Section \ref{supsec2}, Case (ii).\footnote{If roots of $\hat{\mF}_s$ and $\mF_s$ have different multiplicities then define $\bar{n}$ to be the max over the two. Recall that $\hat\gamma_*\stackrel{p}{\to}\gamma_*$ and hence the multiplicities match with probability one as $T\to\infty$.} The proof for the rest of the cases follow by similar reasoning.

\section{Proof of \eqref{mathcaloi}} \label{supsec4}
We need to analyze the following cases:
\begin{itemize}
\item[2)] $t-\kappa=t^*-\kappa^*, t-l=t^*-l^*, t-\kappa\neq t-l$
\item[3)] $t-l=t^*-\kappa^*,t-\kappa=t^*-l^*, t-\kappa\neq t-l$
\item[4)] $t-l=t-\kappa=t^*-l^*=t^*-\kappa^*$
\item[5)] $t-\kappa\neq t-l\neq t^*-\kappa^*\neq t^*-l^*$
\item[6)] $t-\kappa=t-l=t^*-\kappa^*, t-\kappa\neq t^*-l^*$
\item[7)] $t-\kappa=t-l=t^*-l^*, t-\kappa\neq t^*-\kappa^*$
\item[8)] $t-\kappa=t^*-l^*=t^*-\kappa^*, t-\kappa\neq t-l$
\item[9)] $t-l=t^*-l^*=t^*-\kappa^*, t-\kappa\neq t-l$
\end{itemize}
For cases 2)-9) we show that$\mathcal{O}_i=\op(1)$, $i=2,\ldots,9$. Recall that the estimation error in $\hat\vg_{t}$ is asymptotically negligible and $\hat\mF_s^{l}-\mF_s^l= \mathcal{R}_{s,l}$, where $\sum_{l=0}^{\infty} \|\mathcal{R}_{s,l}\|=\| \hat \mF_s-\mF_s\| \ \Op(1)=\op(1)$, and $\hat\mA_{s,\#}^{-1}=\mA_{s,\#}^{-1}+O_p(T^{-1})$. These arguments will be used to replace the estimated values with the true values in all terms $\mathcal{O}_i$, $i=2,\ldots,9$, because the limits will be the same.
Consider case 2).
\begin{eqnarray*}
        \mathcal{O}_2&=& T^{-2}\sum_{t,t^*\in \tilde I_s}\sum_{l,\kappa=0,l\neq \kappa}^{\tilde{t}-1}((\hat\mF_s^\kappa\hat\mA_{s,\#}^{-1})\otimes(\hat\mF_s^l\hat\mA_{s,\#}^{-1}))\mathcal{G}((\hat\mF_s^{t^*-t+\kappa}\hat\mA_{s,\#}^{-1})'\otimes(\hat\mF_s^{t^*-t+l}\hat\mA_{s,\#}^{-1})')\\
        &\times&1_{t^*-t+\kappa>0,t^*-t+l>0,\kappa\neq l}
\end{eqnarray*}
where
\begin{eqnarray*}
        \mathcal{G}&=& \E^b(\vg^b_{t-\kappa}\vg_{t-\kappa}^{b\prime})\otimes\E^b(\vg^b_{t-l}\vg^{b\prime}_{t-l})=((\hat\vg_{t-k}\hat\vg_{t-k}')\otimes(\hat\vg_{t-l}\hat\vg_{t-l}'))\odot(\E^b(\vnu_{t-\kappa}\vnu_{t-\kappa}')\otimes\E^b(\vnu_{t-l}\vnu_{t-l}'))\\
        &=& ((\hat\vg_{t-k}\hat\vg_{t-k}')\otimes(\hat\vg_{t-l}\hat\vg_{t-l}'))\odot(\mathcal{J}\otimes\mathcal{J})= ((\hat\vg_{t-k}\hat\vg_{t-k}')\odot\mathcal{J})\otimes((\hat\vg_{t-l}\hat\vg_{t-l}')\odot\mathcal{J}).
\end{eqnarray*}
Note that
\begin{eqnarray*}
        \left\| \mathcal{O}_2\right\|&\leq& \left\|T^{-2} \sum_{t\in \tilde I_s}\sum_{l,\kappa=0,l\neq \kappa}^{\tilde{t}-1}((\mF_s^\kappa\mA_{s,\#}^{-1})\otimes(\mF_s^l\mA_{s,\#}^{-1}))\left((\vg_{t-\kappa}\vg_{t-\kappa}')\odot\mathcal{J})\otimes((\vg_{t-l}\vg_{t-l}')\odot\mathcal{J})\right)\right\|\\
        &\times&\sup_{t-l,t-\kappa}\left\| \sum_{t^*\in \tilde I_s}((\mF_s^{t^*-t+\kappa}\mA_{s,\#}^{-1})'\otimes(\mF_s^{t^*-t+l}\mA_{s,\#}^{-1})')\right\| +\op(1)\\
        &=&\left\|T^{-2} \sum_{t\in \tilde I_s}\sum_{l,\kappa=0,l\neq \kappa}^{\tilde{t}-1}((\mF_s^\kappa\mA_{s,\#}^{-1})\otimes(\mF_s^l\mA_{s,\#}^{-1}))\left((\vg_{t-\kappa}\vg_{t-\kappa}')\odot\mathcal{J})\otimes((\vg_{t-l}\vg_{t-l}')\odot\mathcal{J})\right)\right\|O(1)\\
        &\leq&\sum_{l,\kappa=0}^{\infty}\left\|\mF_s^{\kappa}\right\| \left\|\mF_s^{l}\right\|\left\|\mA_{s,\#}^{-1}\right\|^2T^{-1}\left(\sup_{t}{\left\|\vepsi_t\right\|^2}\right)^2O(1)=O_p(T^{-1})=\op(1),
\end{eqnarray*}
because $(\vg_{t-\kappa}\vg_{t-\kappa}')\odot\mathcal{J} $ is the upper left block $\vepsi_{t-\kappa}\vepsi_{t-\kappa}'$, and $\|(\vepsi_{t-\kappa}\vepsi_{t-\kappa}') \otimes ( \vepsi_{t-l}\vepsi_{t-l}')\|
\leq \sup_{t-\kappa} \| \vepsi_{t-\kappa}\|^2 \sup_{t-l}\| \vepsi_{t-l}\|^2
= (\sup_{t} \| \vepsi_t \|^{2})^2=O(1)$ by Assumption \ref{a8}.

For case 3) we have
\begin{eqnarray*}
        \mathcal{O}_3&=& T^{-2}\sum_{t,t^*\in \tilde I_s}\sum_{l,\kappa=0,l\neq \kappa}^{\tilde{t}-1}((\hat\mF_s^\kappa\hat\mA_{s,\#}^{-1})\otimes(\hat\mF_s^l\hat\mA_{s,\#}^{-1}))\mathcal{G}((\hat\mF_s^{t^*-t+l}\hat\mA_{s,\#}^{-1})'\otimes(\hat\mF_s^{t^*-t+\kappa}\hat\mA_{s,\#}^{-1})')\\
        &\times&1_{t^*-t+l>0,t^*-t+\kappa>0,\kappa\neq l}
\end{eqnarray*}
where $\mathcal{G}= \E^b(\vg^b_{t-\kappa}\vg_{t-l}^{b\prime})\otimes\E^b(\vg^b_{t-l}\vg^{b\prime}_{t-\kappa})$ which is just a permutation of $\mathcal{G}$ from case 2). Hence, by similar arguments as for $\mathcal{O}_2$ and no additional assumptions, we have $\mathcal{O}_3=\op(1)$.

For case 4) we have
\begin{eqnarray*}
        \mathcal{O}_4&=& T^{-2}\sum_{t,t^*\in \tilde I_s}\sum_{\kappa=0}^{\tilde{t}-1}((\hat\mF_s^\kappa\hat\mA_{s,\#}^{-1})\otimes(\hat\mF_s^\kappa\hat\mA_{s,\#}^{-1}))\mathcal{G}((\hat\mF_s^{t^*-t+\kappa}\hat\mA_{s,\#}^{-1})'\otimes(\hat\mF_s^{t^*-t+\kappa}\hat\mA_{s,\#}^{-1})')\\
        &\times&1_{t^*-t+\kappa>0}
\end{eqnarray*}
where
\begin{eqnarray*}
        \mathcal{G}&=& \E^b((\vg^b_{t-\kappa}\vg_{t-\kappa}^{b\prime})\otimes(\vg^b_{t-\kappa}\vg^{b\prime}_{t-\kappa}))=((\hat\vg_{t-\kappa}\hat\vg_{t-\kappa}')\otimes(\hat\vg_{t-\kappa}\hat\vg_{t-\kappa}'))\odot(\E^b(\vnu_{t-\kappa}\vnu_{t-\kappa}')\otimes(\vnu_{t-\kappa}\vnu_{t-\kappa}')).
\end{eqnarray*}
For a generic scalar $N$ and generic matrices $\mA$, $\mA_{j}$, $j=1,\ldots,N$, $\vec_{N}(\mA)=\vec_{j=1:N}(\mA)$; $\mat_{N,N}(\mA)$ a $N\times N$ block matrix with typical block given by $\mA$. Note that $\vec_N(\mA)=\mat_{N,1}(\mA)$ and $\mat_{N,N}(a)$ is a $N\times N$ matrix with typical scalar element $a$. Then:\begin{align*}
\E^b((\vnu_{t-\kappa}\vnu_{t-\kappa}')\otimes(\vnu_{t-\kappa}\vnu_{t-\kappa}'))&=\left[\begin{array}{ccc}\mat_{1+p_1,1+p_1}(\mathcal{Y}_1)&\mat_{1+p_1,p_{2}}(\mathcal{Y}_2)&\mat_{1+p_1,n(p-1)}(\vzeros_{np\times np})\\
(\mat_{1+p_1,p_2}(\mathcal{Y}_2))'&\mat_{p_2,p_2}(\mathcal{J})&\mat_{p_2,n(p-1)}(\vzeros_{np\times np})\\
(\mat_{1+p_1,n(p-1)}(\vzeros_{np\times np}))'&\mat_{p_2,n(p-1)}(\vzeros_{np\times np})&\mat_{n(p-1),n(p-1)}(\vzeros_{np\times np})
\end{array}   \right]\\
\mathcal{Y}_1&=\left[\begin{array}{ccc}\mat_{1+p_1,1+p_1}(\E^b(\nu_t^4))&\mat_{1+p_1,p_2}(\E^b(\nu_t^3))&\vzeros_{(1+p_1)\times n(p-1)}\\
(\mat_{1+p_1,p_2}(\E^b(\nu_t^3)))'&\mJ_{p_2}&\vzeros_{p_2\times n(p-1)}\\
\vzeros_{(1+p_1)\times n(p-1)}'&\vzeros_{p_2\times n(p-1)}'&\vzeros_{n(p-1)\times n(p-1)}
\end{array} \right]\\
\mathcal{Y}_2&=\left[\begin{array}{ccc}\mat_{1+p_1,1+p_1}(\E^b(\nu_t^3))&\mat_{1+p_1,p_2}(1)&\vzeros_{(1+p_1)\times n(p-1)}\\
(\mat_{1+p_1,p_2}(1))'&\vzeros_{p_2\times p_2}&\vzeros_{p_2\times n(p-1)}\\
\vzeros_{(1+p_1)\times n(p-1)}'&\vzeros_{p_2\times n(p-1)}'&\vzeros_{n(p-1)\times n(p-1)}
\end{array} \right].
\end{align*}
We can see from above that the only nonzero elements of $\E^b(\vnu_{t-\kappa}\vnu_{t-\kappa}')\otimes(\vnu_{t-\kappa}\vnu_{t-\kappa}')$ are $\E^b(\nu_t^4)$, $\E^b(\nu_t^3)$ and 1.  By similar arguments as for case 2) with the additional assumption that $\E^b(\nu_t^4)\leq \bar c$, $\bar c>0$, it follows that $\mathcal{O}_4=\op(1)$.

For case 5) we have
\begin{eqnarray*}
        \mathcal{O}_5&=& T^{-2}\sum_{t,t^*\in \tilde I_s}\sum_{l,\kappa=0}^{\tilde{t}-1}\sum_{l^*,\kappa^*=0}^{\tilde{t}^*-1}((\hat\mF_s^\kappa\hat\mA_{s,\#}^{-1})\otimes(\hat\mF_s^l\hat\mA_{s,\#}^{-1}))\mathcal{G}((\hat\mF_s^{\kappa^*}\hat\mA_{s,\#}^{-1})'\otimes(\hat\mF_s^{l^*}\hat\mA_{s,\#}^{-1})')\times 1_{t-\kappa\neq t-l\neq t^*-\kappa^*\neq t^*-l^*}\\
        \mathcal{G}&=&\E^b\left(\left(\vg^b_{t-\kappa}\otimes\vg^b_{t-l}\right)\left(\vg^b_{t^*-\kappa^*}\otimes\vg^b_{t^*-l^*}\right)^\prime\right)\\
        &=&((\hat\vg_{t-\kappa}\hat\vg_{t^*-\kappa^*}^\prime)\otimes(\hat\vg_{t-l}\hat\vg_{t^*-l^*}^\prime))\odot\E^b((\vnu_{t-\kappa}\vnu_{t^*-\kappa^*}^\prime)\otimes(\vnu_{t-l}\vnu_{t^*-l^*}^{\prime})),
\end{eqnarray*}
where\begin{align*}
\E^b((\vnu_{t-\kappa}\vnu_{t^*-\kappa^*}^\prime)\otimes(\vnu_{t-l}\vnu_{t^*-l^*}^{\prime}))&=\left[\begin{array}{ccc}\mat_{1+p_1,1+p_1}(\vzeros_{np\times np})&\mat_{1+p_1,p_{2}}(\vzeros_{np\times np})&\mat_{1+p_1,n(p-1)}(\vzeros_{np\times np})\\
(\mat_{1+p_1,p_2}(\vzeros_{np\times np}))'&\mat_{p_2,p_2}(\mathcal{J})&\mat_{p_2,n(p-1)}(\vzeros_{np\times np})\\
(\mat_{1+p_1,n(p-1)}(\vzeros_{np\times np}))'&\mat_{p_2,n(p-1)}(\vzeros_{np\times np})&\mat_{n(p-1),n(p-1)}(\vzeros_{np\times np})
\end{array}   \right]
\end{align*}
which only selects the cross products $(\vzeta_{t-\kappa}\vzeta_{t-l}')\otimes (\vzeta_{t^*-\kappa^*}\vzeta_{t^{*}-l^{*}}^\prime)$. We have
\begin{eqnarray*}
        \left\| \mathcal{O}_5\right\|&\leq& \left\|T^{-2} \sum_{t,t^*\in \tilde I_s}\sum_{l,\kappa=0,l\neq \kappa}^{\tilde{t}-1}\sum_{l,\kappa=0}^{\tilde{t}-1}\sum_{l^*,\kappa^*=0}^{\tilde{t}^*-1}((\mF_s^\kappa\mA_{s,\#}^{-1})\otimes(\mF_s^l\mA_{s,\#}^{-1}))\left((\vg_{t-\kappa}\vg^\prime_{t^*-\kappa^*})\otimes(\vg_{t-l}\vg^\prime_{t^*-l^*}))\odot\mathcal{J}_2)\right)\right.\\
        &\times&\left. ((\mF_s^{\kappa^*}\mA_{s,\#}^{-1})'\otimes(\mF_s^{l^*}\mA_{s,\#}^{-1})')\right\|+\op(1)\\
        &\leq&\sum_{l,\kappa,l^*,\kappa^*=0}^{\infty}\left\|\mF_s^{\kappa}\right\| \left\|\mF_s^{l}\right\|\left\|\mF_s^{\kappa^*}\right\| \left\|\mF_s^{l^*}\right\|\left\|\mA_{s,\#}^{-1}\right\|^4T^{-1}\left(\sup_{t}{\left\|\vzeta_t\right\|^2}\right)^2 +\op(1)=\op(1).
\end{eqnarray*}

Consider now case 6). We have:
\begin{eqnarray*}
        \mathcal{O}_6&=&T^{-2}\sum_{t,t^*\in \tilde I_s}\sum_{\kappa=0}^{\tilde{t}-1}\sum_{l^*=0}^{\tilde{t}^*-1}((\hat\mF_s^\kappa\hat\mA_{s,\#}^{-1})\otimes(\hat\mF_s^\kappa\hat\mA_{s,\#}^{-1}))\mathcal{G}((\hat\mF_s^{\kappa^*}\hat\mA_{s,\#}^{-1})'\otimes(\hat\mF_s^{l^*}\hat\mA_{s,\#}^{-1})')
\end{eqnarray*}
where
\begin{eqnarray*}
        \mathcal{G}&=&\E^b(\left(\vg^b_{t-\kappa}\otimes\vg^b_{t-\kappa}\right)\left(\vg^b_{t-\kappa}\otimes\vg^b_{t^*-l^*}\right)^\prime)\\
        &=&((\hat\vg_{t-\kappa}\hat\vg_{t-\kappa}^\prime)\otimes(\hat\vg_{t-\kappa}\hat\vg_{t^*-l^*}^\prime))\odot\E^b((\vnu_{t-\kappa}\vnu_{t-\kappa}^\prime)\otimes(\vnu_{t-\kappa}\vnu_{t^*-l^*}^{\prime})),
\end{eqnarray*}
with\begin{align*}
\E^b((\vnu_{t-\kappa}\vnu_{t-\kappa}^\prime)\otimes(\vnu_{t-\kappa}\vnu_{t^*-l^*}^{\prime}))&=\left[\begin{array}{ccc}\mat_{1+p_1,1+p_1}(\mathcal{\tilde Y}_1)&\mat_{1+p_1,p_{2}}(\mathcal{\tilde Y}_2)&\mat_{1+p_1,n(p-1)}(\vzeros_{np\times np})\\
(\mat_{1+p_1,p_2}(\mathcal{\tilde Y}_2))'&\mat_{p_2,p_2}(\mathcal{J}_2)&\mat_{p_2,n(p-1)}(\vzeros_{np\times np})\\
(\mat_{1+p_1,n(p-1)}(\vzeros_{np\times np}))'&\mat_{p_2,n(p-1)}(\vzeros_{np\times np})&\mat_{n(p-1),n(p-1)}(\vzeros_{np\times np})
\end{array}   \right]\\
\mathcal{\tilde Y}_1&=\left[\begin{array}{ccc}\vzeros_{(1+p_1)\times (1+p_1)}&\mat_{1+p_1,p_2}(\E^b(\nu_t^3))&\vzeros_{(1+p_1)\times n(p-1)}\\
\vzeros_{p_2\times (1+p_1)}&\vzeros_{p_2\times p_2}&\vzeros_{p_2\times n(p-1)}\\
\vzeros_{(1+p_1)\times n(p-1)}'&\vzeros_{p_2\times n(p-1)}'&\vzeros_{n(p-1)\times n(p-1)}
\end{array} \right]\\
\mathcal{\tilde Y}_2&=\left[\begin{array}{ccc}\vzeros_{(1+p_1)\times (1+p_1)}&\mat_{1+p_1,p_2}(1)&\vzeros_{(1+p_1)\times n(p-1)}\\
\vzeros_{p_2\times (1+p_1)}&\vzeros_{p_2\times p_2}&\vzeros_{p_2\times n(p-1)}\\
\vzeros_{(1+p_1)\times n(p-1)}'&\vzeros_{p_2\times n(p-1)}'&\vzeros_{n(p-1)\times n(p-1)}
\end{array} \right]
\end{align*}
which only has $\E(\nu_t^3)<c^*$ (by Assumption \ref{aboot}) and 1 as non-zero elements. Hence,
\begin{eqnarray*}
        \left\| \mathcal{O}_6\right\|&\leq& \left\|T^{-2} \sum_{t,t^*\in \tilde I_s}\sum_{\kappa=0}^{\tilde{t}-1}\sum_{l^*=0}^{\tilde{t}^*-1}((\mF_s^\kappa\mA_{s,\#}^{-1})\otimes(\mF_s^\kappa\mA_{s,\#}^{-1}))\left(((\vg_{t-\kappa}\vg_{t-\kappa}')\otimes(\vg_{t-\kappa}\vg^\prime_{t^{*}-l^{*}}))\odot\mathcal{J}_3\right)\right.\\
        &\times&\left.((\mF_s^{\kappa}\mA_{s,\#}^{-1})'\otimes(\mF_s^{l^*}\mA_{s,\#}^{-1})')\right\|+\op(1)\\
        &\leq&\sum_{l^*,\kappa=0}^{\infty}\left\|\mF_s^{\kappa}\right\|^3 \left\|\mF_s^{l^*}\right\|\left\|\mA_{s,\#}^{-1}\right\|^4T^{-1}\left(\sup_{t}{\left\|\vepsi_t\right\|^2}\right)^2\left\|\mathcal{J}_3\right\|^2+\op(1)=\op(1),
\end{eqnarray*}
which follows by similar arguments as for case 2) and the fact that $\left\|\mathcal{J}\right\|^2=O(1)$. Similarly, under no additional assumptions, $\mathcal{O}_{i}=\op(1)$, $i=7,8,9$.

\section{Validity of the WF and WR bootstrap for the $\sup$-$F$-statistic}\label{supsec5}

We begin by defining the $\sup$-$F$ statistics in both the sample and the bootstrap. For ease of reference some equations in the main paper are repeated.

\noindent
{\it Case (i): $H_0:\, m=0$ versus $H_1:\,m=k$}\\
\noindent
Under $H_0$, the second stage estimation involves regression via OLS of $y_t$ on $\hat{\vw}_t$ where $\hat{\vw}_t =  (\hat{\vx}_t^\prime,\vz_{1,t}^\prime)^\prime$ using the complete sample. Let $SSR_0$ denote the residual sum of squares from this estimation. Under $H_1$, the second stage estimation involves estimation via OLS of the model in equation \eqref{2sls_step2_mod} in the paper, that is,
$$
y_t\;=\;\hat{\vw}_t^\prime\vbeta_{(i)}\,+\,\text{error},\quad i=1,...,k+1,\quad  t\in I_{i,\vlambda_k},
$$
for all possible $k$-partitions $\vlambda_k$. Let $SSR_k(\vlambda_k,\hat \vbeta_{\vlam_k})$ denote the residual sum of squares associated with this estimation. The $\sup$-$F$ test statistic is defined as:
\begin{align}
\label{supF_def}
\sup\text{-}F_T\;=\;\sup_{\vlambda_k\in\mLambda_{\epsilon,k}}F_T(\vlambda_k).
\end{align}
where
\begin{equation}
\label{F_def}
F_T(\vlambda_k)\;=\;\left(\,\frac{T-(k+1)d_\beta}{kd_\beta}\,\right)\left(\,\frac{SSR_0-SSR_k(\vlambda_k;\,\hat \vbeta_{\vlambda_k})}{SSR_k(\vlambda_k;\,\hat \vbeta_{\vlambda_k})}\,\right),
\end{equation}
and  $\mLambda_{\epsilon,k}=\{\vlambda_k: |\lambda_{i+1}-\lambda_i|\ge\epsilon,\lambda_1\ge\epsilon,\lambda_k\le 1-\epsilon\}$ and $d_\beta=\dim(\vbeta_{(i)})=q_1+p_1$.

The WR bootstrap version of the $\sup$-$F$ statistic is calculated as follows.  Let $\hat \vw_t^b$ be calculated as described for the WR bootstrap in Section \ref{sect2.3}.  For a given $k$-partition $\vlambda_k$, the second stage of the 2SLS in the bootstrap samples involves OLS estimation of \eqref{2sls_step2_mod_boot} that is,
$$
y_t^b\;=\;\hat{\vw}_t^{b\prime}\vbeta_{(i)}\,+\,\text{error},\quad i=1,...,k+1,\quad  t\in I_{i,\vlambda_k};
$$
let $SSR_k(\vlambda_k,\hat \vbeta_{\vlam_k}^b)$ denote the residual sum of squares associated with this estimation. Let $SSR_0^b$ be the residual sum of squares associated with estimation of \eqref{2sls_step2_mod_boot} subject to the restriction that $\vbeta_{(i)}$ takes the same value in each of the $k+1$ regimes. The WR bootstrap version of the $\sup$-$F$ statistic:
\begin{align}
\label{supF_defb}
\sup\text{-}F_T^b\;=\;\sup_{\vlambda_k\in\mLambda_{\epsilon,k}}F_{T\vlambda_k}^b.
\end{align}
where
\begin{equation}
\label{F_defb}
F_{T\vlambda_k}^b\;=\;\left(\,\frac{T-(k+1)d_\beta}{kd_\beta}\,\right)\left(\,\frac{SSR_0^b-SSR_k^b(\vlambda_k;\,\hat \vbeta_{\vlambda_k}^b)}{SSR_k^b(\vlambda_k;\,\hat \vbeta^b_{\vlambda_k})}\,\right).
\end{equation}
The only difference for the WF bootstrap version of the $\sup$-$F$ statistic the only difference is that $\hat \vw_t^b$ is calculated as described for the WF bootstrap in Section \ref{sect2.3} of the paper.\\[0.1in]

\noindent
{\it Case (ii): $H_0:\, m=\ell$ versus $H_1:\,m=\ell+1$}\\
\noindent
As in Section \ref{sect2.2}, let the estimated break fractions for the $\ell$-break model be $\hat{\vlambda}_\ell$ and the associated break points be denoted $\{\hat{T}_i\}_{i=1}^\ell$ where $\hat{T}_i=[T\hat{\lambda}_i]$. Let $\hat{I}_i=I_{i,\hat{\vlambda}_\ell}$, the set of observations in the $i^{th}$ regime of the $\ell$-break model and partition this set as $\hat{I}_i=\hat{I}_i^{(1)}(\varpi_i)\cup\hat{I}_i^{(2)}(\varpi_i)$ where  $\hat{I}_i^{(1)}(\varpi_i)=\{t:\,[\hat{\lambda}_{i-1}T]+1,[\hat{\lambda}_{i-1}T]+2,\ldots,[\varpi_iT]\}$ and  $\hat{I}_i^{(2)}(\varpi_i)=\{t:\,[\varpi_iT]+1,[\varpi_iT]+2,\ldots,[\hat{\lambda}_iT]\}$.  Consider estimation of the model in \eqref{2sls_step2_mod_caseii} that is,
$$
y_t\;=\;\hat{\vw}_t^\prime\vbeta_{(j)}\,+\,\text{error},\quad j=1,2\quad  t\in \hat{I}_{i}^{(j)},
$$
for all possible choices of $\varpi_i$ (where for notational brevity we suppress the dependence of $\vbeta_{(j)}$ on $i$). Let $SSR_i(\varpi_i)$ be the residual sum of squares associated with this estimation, and let $SSR_i$ be the residual sum of squares associated with estimation of the model subject to the restriction that $\vbeta_{(1)}=\vbeta_{(2)}$. The $sup-F$ statistic for the same  test is given by
\begin{equation}
\label{supF_ell}
\sup\text{-}F_T(\ell+1\,|\,\ell)\;=\;\max_{i=1,2,\ldots\ell+1}\,\left\{\,\sup_{\varpi_i\in \mathcal{N}(\hat{\vlambda}_\ell)}\,\left(\,\frac{SSR_i\,-\,SSR_i(\varpi_i)}{SSR_i}\,\right)\,\left(\frac{\hat{T}_i-\hat{T}_{i-1}-d_\beta)}{d_\beta}\right)\,\right\}
\end{equation}
where $\mathcal{N}_i(\hat{\vlambda}_\ell)=[\hat{\lambda}_{i-1}+\epsilon,\hat{\lambda}_i-\epsilon]$.

For each bootstrap the first stage of the 2SLS estimation and the construction of $\hat{\vw}_t$ is the same as described above in Case (i). The second stage of the 2SLS involves estimation via OLS of \eqref{2sls_step2_mod_caseii_boot} that is,
$$
y_t^b\;=\;\hat{\vw}_t^{b\prime}\vbeta_{(j)}\,+\,\text{error},\quad j=1,2\quad  t\in \hat{I}_{i}^{(j)},
$$
for all possible $\varpi$ (where, once again, we suppress the dependence of $\vbeta_{(j)}$ on $i$). Let $SSR_i^b(\varpi_i)$ be the residual sum of squares associated with this estimation, and $SSR_i^b$ be the residual sum of squares associated with estimation of the model subject to the restriction that $\vbeta_{(1)}=\vbeta_{(2)}$. The bootstrap version of $\sup\text{-}F_T^b(\ell+1\,|\,\ell)$ is given by
\begin{equation}
\label{supF_ellb}
\sup\text{-}F_T^b(\ell+1\,|\,\ell)\;=\;\max_{i=1,2,\ldots\ell+1}\,\left\{\,\sup_{\varpi_i\in \mathcal{N}(\hat{\vlambda}_\ell)}\,\left(\,\frac{SSR_i^b\,-\,SSR_i^b(\varpi_i)}{SSR_i^b}\,\right)\,\left(\frac{\hat{T}_i-\hat{T}_{i-1}-d_\beta)}{d_\beta}\right)\,\right\}.
\end{equation}

\vspace*{0.1in}
To establish the asymptotic validity of the bootstrap versions $\sup\text{-}F_T^b$ and $\sup\text{-}F_T^b(\ell+1\,|\,\ell)$, it is most convenient to work with alternative formulae for the $F$-statistics.  To illustrate, consider Case (i). From standard LS theory,\footnote{For example, see \citet{Greene:2012} (p.163). We are grateful to a referee for drawing our attention to this alternative representation of the $F$-statistic.} it follows that:
$$
F_{T\vlambda_k}\;=\;\frac{T}{kd_\beta} \, \hat \vbeta_{\vlam_k}^\prime \,
\mR_k' \,\left(\mR_k \bar \mV_{\vlam_k} \mR_k^\prime\right)^{-1}\mR_k \,\hat \vbeta_{\vlam_k}
$$
where $d_\beta$ is the dimension of $\beta$, $\bar{\mV}_{\vlam_k}=\hat{\sigma}^2(\vlambda_k)\diag_{i=1:k+1}(\hat \mQ_{(i)}^{-1})$ and $\hat{\sigma}^2(\vlambda_k)=SSR_k(\vlambda_k;\,\hat \vbeta_{\vlambda_k}/(T-(k+1)d_\beta)$. The asymptotic validity of the bootstrap version of the $sup$-$F$ statistics can then be established using similar arguments to the proofs of Theorems \ref{theo1boot}-\ref{theo2boot}.

\begin{theoA}
        \label{theof}
        Under Assumption \ref{a1}-\ref{aboot} for the WF bootstrap, and Assumptions \ref{a1}-\ref{aboot} and \ref{A8prime} for the WR bootstrap,
(i) under the null hypothesis $m=0$, $$\sup_{c\in\mathbb{R}}\left| P^{b} \left(\sup\text{-}F^b_T\leq c\right)-P(\sup\text{-}F_T\leq c)\right|\inp 0 \mbox{ as }T\rightarrow\infty;$$
(ii) under the null hypothesis $m=\ell$,
     $$\sup_{c\in\mathbb{R}}\left| P^{b} \left(\sup\text{-}F^b_T(\ell+1|\ell)\leq c\right)-P(\sup\text{-}F_T(\ell+1|\ell)\leq c)\right|\inp 0\mbox{ as }T\rightarrow\infty.$$ \end{theoA}


\begin{proof}[Proof of Theorem B \ref{theof}]\hfill \\
(i) Recall that
$
Wald_{T\vlam_k}\; =\;  T \, \hat \vbeta_{\vlam_k}' \,
        \mR_k' \,\left(\mR_k \hat \mV_{\vlam_k} \mR_k'\right)^{-1}\mR_k \,\hat \vbeta_{\vlam_k},
$
where $\hat{\mV}_{\vlam_k}=\diag_{i=1:k+1}(\hat \mV_{(i)})$, where $\hat \mV_{(i)}\;=\;\hat \mQ_{(i)}^{-1}  \ \hat \mM_{(i)} \ \hat \mQ_{(i)}^{-1}$. It can be recognized that the version of the $F_{T\vlambda_k}$ described above only differs from the $Wald_{T\vlambda_k}$ in terms of the choice of matrix in the center of the quadratic form. Clearly an analogous representation is available for the bootstrap versions of the test.

        By Lemmas \ref{lem2}, \ref{lem5} and \ref{lem6} in the paper, it can be shown that the bootstrap equivalent of $\hat{\sigma}^2(\vlambda_k)$ is such that its difference with $\hat{\sigma}^2(\vlambda_k)$ is $\op^b(1)$. For the rest of the quantities, the analysis is similar to the proof of Theorem \ref{theo1boot} in the paper.

(ii) Recalling the alternative representation of $\sup\text{-} Wald_T(\ell+1|\ell)$ in the proof of Theorem B \ref{theo_ellvsell+1}, the proof follows as for part (i).
\end{proof}
\newpage
\FloatBarrier
\section{More simulation evidence}\label{supsec6}

In this section we present further simulation evidence on the WR and WF bootstrap $\sup$-$Wald$ and $\sup$-$F$ using the same DGPs as in the main paper.

We have considered the behavior of the $\sup$-$F$ test under both the null and the alternative hypotheses. From the first two columns of Tables \ref{tabhm=00F}-\ref{tabhm=11F} it can be seen that the WR $\sup$-$F$ test works better than the WF $\sup$-$F$ who is in general oversized. Comparing the two tests, the $\sup$-$F$ and the $\sup$-$Wald$, their WR versions are similar under the null, while the WF $\sup$-$Wald$ is less size distorted than the WF $\sup$-$F$ in general (see column 2 of Tables \ref{tabhm=00}-\ref{tabhm=11} from the paper for the $\sup$-$Wald$ and column 2 of Tables \ref{tabhm=00F}-\ref{tabhm=11F} below for the $\sup$-$F$). Regarding the power, it can be seen from columns 3-6 of Tables \ref{tabhm=00F}-\ref{tabhm=11F} that the $\sup$-$F$ is more powerful than the $\sup$-$Wald$ for $T=120$, but for $T=240$, $480$, the $\sup$-$Wald$ is as powerful or slightly more powerful than the $\sup$-$F$ after adjusting for size (see for example Table \ref{tabhm=11} of the paper and Table  \ref{tabhm=11F} below).

In Tables \ref{tabhm=10} and \ref{tabhm=11} of the main paper we have  tested sequentially for the presence of max $2$ breaks in the RF for $x_t$ (in \eqref{xtsim3i}-\eqref{xtsim3ii} and \eqref{xtsim5i}-\eqref{xtsim5ii} respectively).  The fraction of times that $0,1,2$ breaks were detected in RF (out of 1,000 replications of the scenarios), is given in Tables \ref{tabhm=10fracRF}-\ref{tabhm=11fracRF}. These tables indicate that for all sample sizes considered, only in about 5 percent of the cases the null $H_0:h=\ell=1$ was not rejected and that the null $H_0:h=\ell=0$ was rejected all the time in general.

In order to assess the impact of the pre-testing in RF, in Tables \ref{tabhm=10true2} and \ref{tabhm=11true2} we have obtained the rejection frequencies of the bootstrap tests when the number of breaks in the RF is held at the true number, $h=1$, and the estimated location is imposed in the estimation of RF and SE and computing the test statistics $\sup$-$Wald$ and $\sup$-$F$ for 2SLS.
Comparing these tables with the first two columns of Tables \ref{tabhm=10} and  \ref{tabhm=11} from the paper, and Tables \ref{tabhm=10F} and \ref{tabhm=11F} below
we can see that the rejection frequencies are similar. Note that the true number of breaks in RF for Tables \ref{tabhm=10true2} and \ref{tabhm=11true2}
is the same, $h=1$, but the DGP for the RF is different since in Table \ref{tabhm=11true2} (corresponding to scenario $(h,m)=(1,1)$)  the break in the SE results in a break in the mean of $y_{t-1}$, a regressor in the RF. Table \ref{tabhm=10true2} corresponds to the scenario when there is no break in SE.

In Tables \ref{tabhm=10small} and \ref{tabhm=11small} we have considered a break in RF of smaller size than the one mentioned after \eqref{xtsim3i}-\eqref{xtsim3ii} by taking $\vdelta_{r,(1)}^0=(1,1,1,1)'$ (and the rest of the parameters' values are as mentioned after \eqref{xtsim3i}-\eqref{xtsim3ii}). Tables \ref{tabhm=10small} and \ref{tabhm=11small} present the rejection frequencies for the WR and WF bootstrap $\sup$-$Wald$ under the null hypothesis when we have sequentially   tested for the presence of max $2$ breaks in the RF for $x_t$ (in \eqref{xtsim3i}-\eqref{xtsim3ii} and \eqref{xtsim5i}-\eqref{xtsim5ii} respectively) using the WR/WR $\sup$-$Wald$ for OLS, and the resulting number of RF breaks was imposed in each simulation prior to estimating the RF and SE and computing the test statistics for 2SLS (see the first two columns of Tables \ref{tabhm=10small} and \ref{tabhm=11small}). Columns 3-4 of Tables \ref{tabhm=10small} and \ref{tabhm=11small} present the rejection frequencies when the number of breaks in the RF is held at the true number, $h=1$  (and the estimated location is taken into account in the estimation of SE). The last two columns report fraction of times  that $0,1,2$ breaks were detected in RF out of 1,000 replications of the scenarios. We notice that for $T=120$, in $5$-$9$ percent of the cases no break was detected in RF, but nevertheless the rejection frequencies of the bootstrap $\sup$-$Wald$ (columns 1-2 of Tables \ref{tabhm=10small} and \ref{tabhm=11small}) remain close to the case when no pre-testing in RF took place (columns 3-4 of Tables \ref{tabhm=10small} and \ref{tabhm=11small}) with the WR bootstrap performing again better than the WF bootstrap. Note that the true number of breaks in RF for Tables \ref{tabhm=10small} and \ref{tabhm=11small}
is the same, $h=1$, but the DGP for the RF is different because in Table \ref{tabhm=11small} (which corresponds to scenario $(h,m)=(1,1)$)  the break in the SE results in a break in the mean of $y_{t-1}$, a regressor in the RF. In Table \ref{tabhm=10small}, there is no break in SE.

\begin{table}[h!]
\caption{\label{tabhm=00F}{\it Scenario:(h,m)=(0,0)} - rejection probabilities from testing $H_0:\,m=0$ vs. $H_1:\,m=1$ with bootstrap $\sup$-$F$ test.}
\begin{center}
\addtolength{\tabcolsep}{-2pt}
\begin{tabular}{@{}c ccc c ccc c ccc c ccc c ccc c ccc@{}}
\toprule
\midrule
&\multicolumn{3}{c}{WR bootstrap}&&\multicolumn{3}{c}{WF bootstrap}&&\multicolumn{3}{c}{WR bootstrap}&&
\multicolumn{3}{c}{WF bootstrap}&&\multicolumn{3}{c}{WR bootstrap}&&\multicolumn{3}{c}{WF bootstrap}\\
&\multicolumn{3}{c}{Size }&&\multicolumn{3}{c}{Size}&&\multicolumn{3}{c}{Power }&&
\multicolumn{3}{c}{Power}&&\multicolumn{3}{c}{Power}&&\multicolumn{3}{c}{Power }\\
&\multicolumn{3}{c}{$g$=0}&&\multicolumn{3}{c}{$g$=0}&&\multicolumn{3}{c}{$g=-0.007$}&&
\multicolumn{3}{c}{$g=-0.007$}&&\multicolumn{3}{c}{$g=-0.009$}&&\multicolumn{3}{c}{$g=-0.009$}
\\
\cmidrule{2-4}\cmidrule{6-8} \cmidrule{10-12} \cmidrule{14-16} \cmidrule{18-20} \cmidrule{22-24}\\
T&10\%&5\%&1\%&&10\%&5\%&1\%&&10\%&5\%&1\%&&10\%&5\%&1\%&&10\%&5\%&1\%&&10\%&5\%&1\%\\
\cmidrule{2-4}\cmidrule{6-8} \cmidrule{10-12} \cmidrule{14-16} \cmidrule{18-20} \cmidrule{22-24}
\multicolumn{24}{c}{{\it Case A}}\\
\midrule
120&9.8&        5.1&    1.4&&   9.5&    5.6     &1.7&&  99.9&   99.9&   99.7&&  99.9&   99.8&   99.7&&  100&    100&    100&&   100     &100&   100\\
240&9.4&        4.6&    0.5&&   8.8&    4.5&    0.7&&   99.9&   99.9&   97&&    100&    100&    100&&   100&    100&    100&&   100&    100&    100\\
480&10.9&       5.9&    1.5&&   9.4&    5.4&    1.6&&   100&    100&    100&&   100&    100&    100&&   100&    100&    100&&   100&    100&    100\\

\midrule
\multicolumn{24}{c}{{\it Case B}}\\
\midrule
120&11.3&       6.1&    1.7&&   11.5&   6.3&    2.2&&   99.4&   99.4&   98.6&&  99.6&   99.4&   99&&    99.8&   99.7&   99.6&&  99.9&   99.9&   99.7\\
240&10.6&       5.9&    1       &&9.7   &4.9&   1&&     100&    100&    100&&   100&    100&    100&&   100&    100&    100&&   100&    100&    100\\
480&10.3        &6.2&   1.9&&   9.4&    5.4&    1.3&&   100&    100     &100&&  100&    100&    100&&   100&    100&    100&&   100&    100&    100\\
\midrule
\multicolumn{24}{c}{{\it Case C}}\\

\midrule
120&9.1&        5       &1      &&9.8&  5.4&    1&&     99.7&   99.5&   99.4&&  99.7&   99.6&   99.4&&  99.9&   99.9&   99.9&&  100&    100&    100\\
240&11.1&       5.7&    0.9&&   11.2&   5.5&    0.8&&   100&    100&    100&&   100&    100&    100&&   100&    100&    100&&   100&    100&    100\\
480&9.6&        4.7&    1.1&&   8.7&    4.9&    0.7&&   100&    100&    100&&   100&    100&    100&&   100&    100&    100&&   100&    100&    100\\

\midrule
\multicolumn{24}{c}{{\it Case D}}\\
\midrule
120&10.4&       5.8&    0.9&&   10.6&   4.9&    0.9&&   100&    100&    99.7&&  99.9&   99.9&   99.7&&  100&    100&    100&&   100&    100&    99.9\\
240&10  &4.5&   1       &&9.6&  5       &1.1&&  100&    100&    100&&   100&    100&    100&&   100&    100&    100&&   100&    100&    100\\
480&10.5&       3.6&    0.9&&   8.8&    3.6&    0.9&&   100&    100&100&&       100&    100&    100&&   100&    100&    100&&   100&    100&    100\\
\midrule
\bottomrule
\multicolumn{24}{l}{\small{Notes. The first two columns refer to the case when $H_0: m=0$ is true ($g$=0 in equation \eqref{power}). The next columns refer }}\\ \multicolumn{24}{l}{\small{  to the case when we test for $H_0: m=0$, but $H_1:m=1$ is true ($g=-0.007, -0.009$ in equation \eqref{power}). Under the null }}\\\multicolumn{24}{l}{\small{ and the alternative hypotheses we impose $h=0$ in the RF.}}
\end{tabular}%
\end{center}
\end{table}

\begin{table}[h!]
\caption{\label{tabhm=01F}{\it Scenario:(h,m)=(0,1)} - rejection probabilities from testing $H_0:\,m=1$ vs. $H_1:\,m=2$ with bootstrap $\sup$-$F$ test.}
\begin{center}
\addtolength{\tabcolsep}{-2pt}
\begin{tabular}{@{}c ccc c ccc c ccc c ccc c ccc c ccc@{}}
\toprule
\midrule
&\multicolumn{3}{c}{WR bootstrap}&&\multicolumn{3}{c}{WF bootstrap}&&\multicolumn{3}{c}{WR bootstrap}&&
\multicolumn{3}{c}{WF bootstrap}&&\multicolumn{3}{c}{WR bootstrap}&&\multicolumn{3}{c}{WF bootstrap}\\
&\multicolumn{3}{c}{Size }&&\multicolumn{3}{c}{Size}&&\multicolumn{3}{c}{Power }&&
\multicolumn{3}{c}{Power}&&\multicolumn{3}{c}{Power}&&\multicolumn{3}{c}{Power }\\
&\multicolumn{3}{c}{$g$=0}&&\multicolumn{3}{c}{$g$=0}&&\multicolumn{3}{c}{$g$=0.3}&&
\multicolumn{3}{c}{$g$=0.3}&&\multicolumn{3}{c}{$g$=0.4}&&\multicolumn{3}{c}{$g$=0.4}
\\
\cmidrule{2-4}\cmidrule{6-8} \cmidrule{10-12} \cmidrule{14-16} \cmidrule{18-20} \cmidrule{22-24}\\
T&10\%&5\%&1\%&&10\%&5\%&1\%&&10\%&5\%&1\%&&10\%&5\%&1\%&&10\%&5\%&1\%&&10\%&5\%&1\%\\
\cmidrule{2-4}\cmidrule{6-8} \cmidrule{10-12} \cmidrule{14-16} \cmidrule{18-20} \cmidrule{22-24}
\multicolumn{24}{c}{{\it Case A}}\\
\midrule
120&10.8&       5.6&    1.2&&   14.2&   7.7&    1.7&&   100&    100&    100&&   100&    100&    100&&   100&    100&    100&&   100&    100     &100\\
240&9.3&        4.7&    1.1&&   11.4&   5.5&    1.4&&   100&    100&    100&&   100&    100&    100&&   100&    100&    100&&   100&    100&    100\\
480&8.7 &4.7&   1.5&&   9.8&    5.2&    1.7&&   100&    100&    100&&   100&    100&    100&&   100&    100&    100&&   100&    100&    100\\

\midrule
\multicolumn{24}{c}{{\it Case B}}\\
\midrule
120&13.1&       6.9&    1.9&&   15.3&   8.7&    2.7&&   98.6&   97.5&   95.1&&  98.4&   97.6&   95.9&&  99.4&   99.3&   98.8&&  99.6&   99.3&   98.6\\
240&11.1&       4.9&    0.9&&   13&     6.2&    1.2&&   99.9&   99.7&   98.9&&  99.9&   99.6&   98.9&&  100&    100&    99.8&&  100     &99.9&  99.8\\
480&10& 5.1&    1.3&&   10.6&   5.8&    1.2&&   100&    99.9&   99.8&&  99.9&   99.8&   99.8&&  100&    100&    100&&   100&    100&    100\\

\midrule
\multicolumn{24}{c}{{\it Case C}}\\

\midrule
120&12.1        &7.2&   1.1&&   17.6&   10.8&   2.4&&   96.8&   93.7&   82.5&&  98.5&   96.3&   87.1&&  99.9&   99.8&   99.1&&  100     &100    &99.4\\
240&10.7&       4.9&    0.7&&   15.5&   8.5&    2       &&100   &100&   100&&   100&    100&    100&&   100     &100&   100&&   100&    100&    100\\
480&11.3&       6.4&    1       &&14.1& 8.4&    1.9&&   100&    100&    100&&   100&    100&    100&&   100&    100&    100&&   100&    100&    100\\

\midrule
\multicolumn{24}{c}{{\it Case D}}\\
\midrule
120&11.9&       7.1&    1.4&&   16.3&   10.2&   2.9&&   99.8&   99.8&   99.2&&  99.9&   99.8&   99.8&&  100     &100&   100&&   100&    100&    100\\
240&11.5&       5.4&    1.4&&   17&     8.4&    2.2&&   100&    100&    100&&   100&    100&    100&&   100&    100&    100&&   100&    100&    100\\
480&11.6&       6       &1.1    &&14.2& 7.5&    1.7&&   100&    100&    100&&   100&    100&    100&&   100&    100&    100&&   100&    100&    100\\

\midrule
\bottomrule
\multicolumn{24}{l}{\small{Notes. The first two columns refer to the case when $H_0: m=1$ is true ($g$=0 in equation \eqref{power}). The next columns refer }}\\ \multicolumn{24}{l}{\small{  to the case when we test for $H_0: m=1$, but $H_1:m=2$ is true ($g=0.3, 0.4$ in equation \eqref{power}). Under the null and }}\\\multicolumn{24}{l}{\small{the alternative hypotheses we impose $h=0$ in the RF.}}
\end{tabular}%
\end{center}
\end{table}

\begin{table}
\caption{\label{tabhm=10F}{\it Scenario:(h,m)=(1,0)} - rejection probabilities from testing $H_0:\,m=0$ vs. $H_1:\,m=1$ with bootstrap $\sup$-$F$ test; number of breaks in the RF was estimated and imposed in each simulation using a sequential strategy based on the WR/WF $\sup$-$F$ for OLS}
\begin{center}
\addtolength{\tabcolsep}{-2pt}
\begin{tabular}{@{}c ccc c ccc c ccc c ccc c ccc c ccc@{}}
\toprule
\midrule
&\multicolumn{3}{c}{WR bootstrap}&&\multicolumn{3}{c}{WF bootstrap}&&\multicolumn{3}{c}{WR bootstrap}&&
\multicolumn{3}{c}{WF bootstrap}&&\multicolumn{3}{c}{WR bootstrap}&&\multicolumn{3}{c}{WF bootstrap}\\
&\multicolumn{3}{c}{Size }&&\multicolumn{3}{c}{Size }&&\multicolumn{3}{c}{Power }&&
\multicolumn{3}{c}{Power }&&\multicolumn{3}{c}{Power }&&\multicolumn{3}{c}{Power }\\
&\multicolumn{3}{c}{$g=0$}&&\multicolumn{3}{c}{$g=0$}&&\multicolumn{3}{c}{$g=-0.05$}&&
\multicolumn{3}{c}{$g=-0.05$}&&\multicolumn{3}{c}{$g=-0.07$}&&\multicolumn{3}{c}{$g=-0.07$}\\
\cmidrule{2-4}\cmidrule{6-8} \cmidrule{10-12} \cmidrule{14-16} \cmidrule{18-20} \cmidrule{22-24}\\
T&10\%&5\%&1\%&&10\%&5\%&1\%&&10\%&5\%&1\%&&10\%&5\%&1\%&&10\%&5\%&1\%&&10\%&5\%&1\%\\
\cmidrule{2-4}\cmidrule{6-8} \cmidrule{10-12} \cmidrule{14-16} \cmidrule{18-20} \cmidrule{22-24}
\multicolumn{24}{c}{{\it Case A}}\\
\midrule
120&7.6 &3.6&   0.7&&   11.5&   5.8&    0.9&&   72.6&   67&     56.1&&  76.4&   70.7&   60&&    86.9&   81.6&   71.5&&  88.4&   84.4&   74.2\\
240&7.7&        3.2&    0.2&&   11.3&   5       &0.8&&  97.5&   96.9&   93.8&&  99.7&   97.4&   94.5&&  99.5&   99.1&   98&&    99.7&   99.3&   98.2\\
480&9.7&        5       &0.9&&  12.4&   5.6&    1.3&&   100&    99.9&   99.9&&  100&    100&    99.9&&  100&    100&    100&&   100&    100&    100\\

\midrule
\multicolumn{24}{c}{{\it Case B}}\\
\midrule
120&10& 5.8&    1.1&&   12.7&   6.4&    1.3&&   71.8&   64.9&   53.4&&  76.4&   68.9&   56.3&&  85.3&   80.1&   69.6&&  87.6&   83.1&   73.8\\
240&8.8&        3.9&    0.8&&   10.3&   4.7&    0.9&&   96.7&   95.3&   91.1&&  97.6&   96.1&   92.7&&  99.1&   98.6&   96.8&&  99.3&   98.9&   97.6\\
480&9.1&        4.7&    1.2&&   10.6&   6       &1.2&&  100&    100&    99.6&&  100&    99.9&   99.7&&  100&    100&    100&&   100&    100&    99.9\\

\midrule
\multicolumn{24}{c}{{\it Case C}}\\

\midrule

120&8.9&        4.4&    0.6&&   14.4&   7.7&    1.8&&   61.9&   54.4&   40.9&&  67.2&   60.6&   47.2&&  75.6&   68.8&   58.2&&  80.5&   74.1&   62.7\\
240&10.1&       5.1&    0.8&&   14.4&   8       &1.5&&  93.5&   91.5&   86.1&&  94.2&   93.1&   89.1&&  97.2&   96.3&   93.7&&  99.7&   97.1&   95\\
480&8.4&        4       &0.6&&  10.9&   5.6&    0.9&&   99.8&   99.6&   99.4&&  99.9&   99.8&   99.5&&  100&    100&    100&&   100&    100&    100\\

\midrule
\multicolumn{24}{c}{{\it Case D}}\\
\midrule
120&9.3&        4.4&    0.6&&   13.1&   6.1&    1.6&&   63.9&   56.1&   42.5&&  68.2&   60.6&   45.8&&  78.1&   72.8&   61.1&&  80.6&   75.2&   64.8\\
240&9.3&        4.2&    0.5&&   12.5&   6.1&    1.5&&   94&     92.5&   88.7&&  94.7&   93&     89.8&&  97.9&   96.9&   94.3&&  98.2&   97.4&   94.7\\
480&8.5&        4&      0.8&&   11.6&   5.4&    0.8&&   99.8&   99.5&   99.3&&  99.8&   99.5&   99.4&&  100&    100&    100&&   100&    100&    99.9\\
\midrule
\bottomrule
\multicolumn{24}{l}{\small{Notes. The first two columns refer to the case when $H_0: m=0$ is true ($g$=0 in equation \eqref{power}). The next columns refer }}\\ \multicolumn{24}{l}{\small{  to the case when we test for $H_0: m=0$, but $H_1:m=1$ is true ($g= -0.05, -0.07$ in equation \eqref{power}). Prior to testing    }}\\ \multicolumn{24}{l}{\small{$H_0:m=0$ vs $H_1:m=1$ (for all columns above), we tested sequentially for the presence of maximum two breaks in the   }}\\ \multicolumn{24}{l}{\small{ RF (we used the WR/WF bootstrap $\sup$-$F$ for OLS\ to test $H_0:h=\ell$ vs. $H_1:\ell+1$, $\ell=0,1$). If breaks are detected in }}\\ \multicolumn{24}{l}{\small{RF, the number of breaks and the estimated location are imposed when estimating the SE.}}
\end{tabular}%
\end{center}
\end{table}

\begin{table}
\caption{\label{tabhm=11F}{\it Scenario:(h,m)=(1,1)} - rejection probabilities from testing $H_0:\,m=1$ vs. $H_1:\,m=2$ with bootstrap $\sup$-$F$ test;  number of breaks in the RF was estimated and imposed in each simulation using a sequential strategy based on the WR/WF $\sup$-$F$ for OLS}
\begin{center}
\addtolength{\tabcolsep}{-2pt}
\begin{tabular}{@{}c ccc c ccc c ccc c ccc c ccc c ccc@{}}
\toprule
\midrule
&\multicolumn{3}{c}{WR bootstrap}&&\multicolumn{3}{c}{WF bootstrap}&&\multicolumn{3}{c}{WR bootstrap}&&
\multicolumn{3}{c}{WF bootstrap}&&\multicolumn{3}{c}{WR bootstrap}&&\multicolumn{3}{c}{WF bootstrap}\\
&\multicolumn{3}{c}{Size }&&\multicolumn{3}{c}{Size}&&\multicolumn{3}{c}{Power }&&
\multicolumn{3}{c}{Power}&&\multicolumn{3}{c}{Power}&&\multicolumn{3}{c}{Power }\\
&\multicolumn{3}{c}{$g$=0}&&\multicolumn{3}{c}{$g$=0}&&\multicolumn{3}{c}{$g=0.5$}&&
\multicolumn{3}{c}{$g=0.5$}&&\multicolumn{3}{c}{$g=-0.5$}&&\multicolumn{3}{c}{$g=-0.05$}
\\
\cmidrule{2-4}\cmidrule{6-8} \cmidrule{10-12} \cmidrule{14-16} \cmidrule{18-20} \cmidrule{22-24}\\
T&10\%&5\%&1\%&&10\%&5\%&1\%&&10\%&5\%&1\%&&10\%&5\%&1\%&&10\%&5\%&1\%&&10\%&5\%&1\%\\
\cmidrule{2-4}\cmidrule{6-8} \cmidrule{10-12} \cmidrule{14-16} \cmidrule{18-20} \cmidrule{22-24}
\multicolumn{24}{c}{{\it Case A}}\\
\midrule
120&11.4&       6       &0.7&&  14.3&   8.6&    1.8&&   98.2&   98.2&   97.7&&  98.4&   98.2&   98&&    98.9&   98.9&   98.6&&  99&     98.8&   98.5\\
240&9.8&        4.8&    0.7&&   12.2&   6.1&    1.5&&   98.7&   98.7&   98.6&&  98.8&   98.8&   98.7&&  98.9&   98.9&   98.6&&  99&     98.9&   98.6\\
480&8.2&        4.5&    0.9&&   9.5&    5.4&    1.1&&   100&    100&    99.9&&  100&    100&    99.8&&  99.9&   99.9&   99.9&&  99.9&   99.9&   99.9\\

\midrule
\multicolumn{24}{c}{{\it Case B}}\\
\midrule
120&11.4&       5.4&    1.4&&   14.4&   7.2&    1.9&&   96.7&   95.8&   93.6&&  96.9&   96.4&   94.8&&  97.4&   96.8&   94.9&&  97.5&   97&     95.9\\
240&9.7&        5.1&    1.1&&   11.3&   6.1&    1.2&&   98.2&   98.1&   97.2&&  98&     97.9&   97.2&&  98.4&   98.3&   97.3&&  98.3&   98&     96.9\\
480&9.7&        5       &1.1&&  10.8&   5.6&    0.9&&   98.7&   98.7&   98.5&&  98.6&   98.6&   98.3&&  99.2&   99.2&   99.2&&  98.8&   98.7&   98.7\\

\midrule
\multicolumn{24}{c}{{\it Case C}}\\

\midrule
120&9.8&        3.7&    0.4&&   14.5&   7.6&    1       &&96.3& 94.1&   86&&    97.2&   95.7&   89.8&&  93.1&   90.4&   78.3&&  94.3&   91.8&   80.2\\
240&10.6&       5.2&    0.13&&  14.2&   8.4&    1.6&&   97.4&   97.4&   97.3&&  97.5&   97.5&   97.4&&  97.1&   97.1&   97&&    96.7&   96.7&   96.6\\
480&10.1&       5&      0.8&&   12.4&   7.5&    1.2&&   98.5&   98.5&   97.9&&  98.6&   98.6&   98.1&&  97.1&   97.1&   97&&    97.1&   97&     96.8\\

\midrule
\multicolumn{24}{c}{{\it Case D}}\\
\midrule
120&8.6&        4.6&    0.8&&   13&     7.5&    1.3&&   98.3&   97.8&   96.4&&  98.2&   97.9&   96.8&&  98.5&   98.1&   96.6&&  99      &98.4&  96.9\\
240&10.2&       4.6&    1.2&&   14.1&   7.1&    1.4&&   99&     98.9&   98.6&&  99&     98.9&   98.5&&  98.5&   98.5&   98.3&&  98.6&   98.6&   98.4\\
480&9.9&        5&      0.9&&   12.4&   6.9&    1.6&&   99.7&   99.6&   99.4&&  99.7&   99.6&   99.5&&  99.6&   99.1&   98.8&&  99.5&   99.2&   98.7\\

\midrule
\bottomrule
\multicolumn{24}{l}{\small{Notes. The first two columns refer to the case when $H_0: m=1$ is true ($g$=0 in equation \eqref{power}). The next columns refer }}\\ \multicolumn{24}{l}{\small{  to the case when we test for $H_0: m=1$, but $H_1:m=2$ is true ($g= -0.5, 0.5$ in equation \eqref{power}). Prior to testing    }}\\ \multicolumn{24}{l}{\small{$H_0:m=1$ vs $H_1:m=2$ (for all columns above), we tested sequentially for the presence of maximum two breaks in the   }}\\ \multicolumn{24}{l}{\small{ RF (we used the WR/WF bootstrap $\sup$-$F$ for OLS\ to test $H_0:h=\ell$ vs. $H_1:\ell+1$, $\ell=0,1$). If breaks are detected in }}\\ \multicolumn{24}{l}{\small{RF, the number of breaks and the estimated location are imposed when estimating the SE.}}
\end{tabular}%
\end{center}
\end{table}

\begin{table}
\caption{\label{tabhm=10true2}{\it Scenario:(h,m)=(1,0)} - rejection probabilities from testing $H_0:\,m=0$ vs. $H_1:\,m=1$ with bootstrap $\sup$-$Wald$ and $\sup$-$F$ tests; the number of breaks in the RF is held at the true number ($h=1$); $H_0:m=0$ is true.}
\begin{center}
\addtolength{\tabcolsep}{-2pt}
\begin{tabular}{@{}c ccc c ccc c ccc c ccc c ccc c ccc@{}}
\toprule
\midrule
&\multicolumn{3}{c}{WR bootstrap}&&\multicolumn{3}{c}{WF bootstrap}&&\multicolumn{3}{c}{WR bootstrap}&&
\multicolumn{3}{c}{WF bootstrap}\\
&\multicolumn{3}{c}{$\sup$-$Wald$}&&\multicolumn{3}{c}{$\sup$-$Wald$}&&\multicolumn{3}{c}{$\sup$-$F$}&&
\multicolumn{3}{c}{$\sup$-$F$}\\
\cmidrule{2-4}\cmidrule{6-8} \cmidrule{10-12} \cmidrule{14-16} \\
T&10\%&5\%&1\%&&10\%&5\%&1\%&&10\%&5\%&1\%&&10\%&5\%&1\\
\cmidrule{2-4}\cmidrule{6-8} \cmidrule{10-12} \cmidrule{14-16}
\multicolumn{16}{c}{{\it Case A}}\\
\midrule
120&11& 4.3&    0.9&&   14.8&   7&      1.3&&   8       &4.2&   0.9&&   12.1&   5.9&    1.1\\
240&10.3&       5.5&    1       &&13.1& 5.8&    0.9&&   8.6&    4.1&    0.4&&   11.6&   5.9&    0.9\\
480&10.5&       4.8&    0.5&&   12.3&   6.2&    1       &&10.1& 5.3&    0.9&&   12.5&   6.2&    1.2\\

\midrule
\multicolumn{16}{c}{{\it Case B}}\\
\midrule
120&9.4&        5.2&    0.8&&   13.2&   7       &1.6&&  11.4&   5.5&    1.3&&   11.4&   5.5&    1.3\\
240&9.6&        4.7&    1.3&&   12.2&   6.3&    1.4&&   9.8&    4.7&    1.1&&   9.8&    4.7&    1.1\\
480&11.5&       4.9&    0.4&&   12.9&   6.4&    1.1&&   9.9&    5.3&    1.5&&   9.9&    5.3&    1.5\\

\midrule
\multicolumn{16}{c}{{\it Case C}}\\
\midrule
120&10.2&       4.4&    1&     &14.8&  7&      1.8&&   10.2&   5&      0.8&&   12&     5&      0.8\\
240&10.8&       5&      0.4     &&15.4& 7.4&    1.6&&   10.8&   5.3&    0.6&&   10.8&   5.3&    0.6\\
480&9.7&        5.6&    1&&     11.8&   6.1&    1.3&&   9.1&    4.9&    1       &&9.1&  4.9&    1       \\

\midrule
\multicolumn{16}{c}{{\it Case D}}\\
\midrule
120&10.6        &5&     1.2&&   14.1&   7&      1.5&&   9.6&    4.8&    1       &&9.6&  4.8&    1\\
240&10.3&       5.9&    1.1&&   14.7&   7.5&    2&&     9.7&    4.7&    0.5&&   9.7&    4.7&    0.5\\
480&10.3&       5.6&    0.8&&   13&     6.3&    1.5&&   9.8&    4.4&    0.7&&   9.8&    4.4&    0.7\\
\midrule
\bottomrule
\multicolumn{16}{l}{\small{Notes. For both $\sup$-$F$ and $\sup$-$Wald$ bootstrap tests the number of breaks in}}\\\multicolumn{16}{l}{\small{the RF is held at the true number ($h=1$), we estimated the  location of the
}}\\\multicolumn{16}{l}{\small{RF break and imposed it when the SE was estimated.}}

\end{tabular}\end{center}
\end{table}

\begin{table}
\caption{\label{tabhm=11true2}{\it Scenario:(h,m)=(1,1)} - rejection probabilities from testing $H_0:\,m=1$ vs. $H_1:\,m=2$ with bootstrap $\sup$-$Wald$ and $\sup$-$F$ tests; the number of breaks in the RF is held at the true number ($h=1$) $H_0:m=1$ is true.}
\begin{center}
\addtolength{\tabcolsep}{-2pt}
\begin{tabular}{@{}c ccc c ccc c ccc c ccc c ccc c ccc@{}}
\toprule
\midrule
&\multicolumn{3}{c}{WR bootstrap}&&\multicolumn{3}{c}{WF bootstrap}&&\multicolumn{3}{c}{WR bootstrap}&&
\multicolumn{3}{c}{WF bootstrap}\\
&\multicolumn{3}{c}{$\sup$-$Wald$}&&\multicolumn{3}{c}{$\sup$-$Wald$}&&\multicolumn{3}{c}{$\sup$-$F$}&&
\multicolumn{3}{c}{$\sup$-$F$}\\
\cmidrule{2-4}\cmidrule{6-8} \cmidrule{10-12} \cmidrule{14-16} \\
T&10\%&5\%&1\%&&10\%&5\%&1\%&&10\%&5\%&1\%&&10\%&5\%&1\\
\cmidrule{2-4}\cmidrule{6-8} \cmidrule{10-12} \cmidrule{14-16}
\multicolumn{16}{c}{{\it Case A}}\\
\midrule

120&9.3&        5.2&    0.7&&   8.8&    4.5&    0.8&&   10.9&   5.8&    1.1&&   14.8&   8.8&    2.1\\
240&10.4&       5.6&    0.7&&   10.2&   5.2&    0.9&&   10.5&   5.6&    1.1&&   12.7&   7.6&    1.5\\
480&9.6&        4.3&    0.7&&   10&     4.6&    0.8&&   8.7&    4.8&    1.1&&   10.3&   5.6&    1.4\\

\midrule
\multicolumn{16}{c}{{\it Case B}}\\
\midrule
120&9.4&        4.1&    0.9&&   8.6&    3.2&    0.8&&   12.1&   5.7&    1.5&&   15.1&   8&      2.4\\
240&10.4&       4.6&    0.8&&   10.3&   5.2&    0.9&&   10.6&   6&      1.3&&   12.7&   6.9     &2\\
480&10.3&       4.2&    0.8&&   11&     5.3&    0.7&&   10.4&   5.4&    1.2&&   11.6&   6&      1.2\\

\midrule
\multicolumn{16}{c}{{\it Case C}}\\
\midrule
120&9.8&        3.9&    1.1&&   9.4&    4&      0.3&&   10&     4.7&    0.6&&   15.1&   8&      1.6\\
240&10& 5.1&    1.2&&   10.2&   5       &1&&    10.5&   5       &0.8&&  15.1&   8.6&    2.7\\
480&11.2&       4.7&    1       &&12.1& 5.3&    0.6&&   10.3&   5.9&    1&&     12.6&   8&      1.5\\
\midrule
\multicolumn{16}{c}{{\it Case D}}\\
\midrule
120&10.1&       4.4&    1.6&&   8.5&    3.8&    0.6&&   9.7&    4.8&    1.1&&   13.5&   7.5     &1.7\\
240&11& 5.1&    0.9&&   11.8&   5.1&    0.6&&   11.2&   4.5&    1.3&&   14.8&   8.2&    1.6\\
480&9.7&        5.1&    1.4&&   10.7&   5.2&    1.2&&   11.4&   5.7&    1.2&&   12.9&   7.3&    1.8\\

\midrule
\bottomrule
\multicolumn{16}{l}{\small{Notes. For both $\sup$-$F$ and $\sup$-$Wald$ bootstrap tests the number of breaks in}}\\\multicolumn{16}{l}{\small{the RF is held at the true number ($h=1$), we estimated the  location of the
}}\\\multicolumn{16}{l}{\small{RF break and imposed it when the SE was estimated.}}

\end{tabular}\end{center}
\end{table}

\begin{table}
\caption{\label{tabhm=10fracRF}Percentage of times (out of 1,000 replications) when 0, 1 and 2 breaks in RF were detected with $\sup$-$Wald$ and $\sup$-$F$ bootstrap  tests before testing the true null hypothesis $H_0:m=0$ vs. $H_1:m=1$ in {\it Scenario:(h,m)=(1,0)}.}
\begin{center}
\addtolength{\tabcolsep}{-2pt}
\begin{tabular}{@{}c ccc c ccc c ccc c ccc @{}}
\toprule
\midrule
&\multicolumn{3}{c}{WR $\sup$-$Wald$}&&\multicolumn{3}{c}{WF $\sup$-$Wald$}&&\multicolumn{3}{c}{WR $\sup$-$F$}&&
\multicolumn{3}{c}{WF $\sup$-$F$}\\\cmidrule{2-4}\cmidrule{6-8} \cmidrule{10-12} \cmidrule{14-16}
&\multicolumn{3}{c}{\% of RF  breaks}&&\multicolumn{3}{c}{\% of RF  breaks}&&\multicolumn{3}{c}{\% of RF  breaks}&&
\multicolumn{3}{c}{\% of RF  breaks}\\
\cmidrule{2-4}\cmidrule{6-8} \cmidrule{10-12} \cmidrule{14-16}
T&0&1&2 &&0&1&2&&0&1&2&&0&1&2\\
\cmidrule{2-4}\cmidrule{6-8} \cmidrule{10-12} \cmidrule{14-16}

\multicolumn{16}{c}{{\it Case A}}\\
\midrule
120&0&  94.6&   5.4&&   0&      98.7&   1.3&&   0&      93.3&   6.7&&   0&      92.6&   7.4\\
240&0&  95.8&   4.2&&   0&      95.6&   4.4&&   0&      93.6&   6.4&&   0&      93.3&   6.7\\
480&0&  94.4&   5.6&&   0&      93.6&   6.4&&   0&      95.1&   4.9&&   0&      94.7&   5.3\\

\midrule
\multicolumn{16}{c}{{\it Case B}}\\
\midrule
120&0.3&        93.7&   6&&     0.5&    98.1&   1.4&&   0&      91.1&   8.9&&   0&      91.1&   8.9\\
240&0&  95.9&   4.1&&   0&      95.5&   4.5&&   0&      90.9&   9.1&&   0&      90.3&   9.7\\
480&0&  94.8&   5.2&&   0&      94.8&   5.2&&   0&      93.3&   6.7&&   0&      93.3&   6.7\\

\midrule
\multicolumn{16}{c}{{\it Case C}}\\
\midrule
120&0&  94&     6&&     0.2&    98.4&   1.4&&   0&      93.5&   6.5&&   0&      93&     7\\
240&0&  93.8&   6.2&&   0&      95.9&   4.1&&   0&      93.5&   6.5&&   0&      92.7&   7.3\\
480&0&  93.6&   6.4&&   0&      93.2&   6.8&&   0       &93.9&  6.1&&   0&      93.5&   6.5\\

\midrule
\multicolumn{16}{c}{{\it Case D}}\\
\midrule
120&0&  94.6&   5.4&&   0.2&    97.4&   2.4&&   0&      94.3&    5.7&&     0&      94.1&   5.9\\
240&0   &95.8&  4.2&&   0&      95.8&   4.2&&   0&      94.3&   5.7&&   0&      94.1&   5.9\\
480&0   &93.5&  6.5&&   0&      91.5&   8.5&&   0&      94.2&   5.8&&   0&      94.1&   5.9\\

\midrule\bottomrule
\multicolumn{16}{l}{\small{Notes. In RF we tested $H_0:h=\ell$ vs $H_1:h=\ell+1$, $\ell=0,1$ with the WR   }}\\\multicolumn{16}{l}{\small{and WF bootstrap $\sup$-$Wald$ and $\sup$-$F$ tests. A(n) (additional) break was}}\\\multicolumn{16}{l}{\small{  detected if the WR or WF bootstrap $p$-value was smaller than $5\%$.}}
\end{tabular}%
\end{center}
\end{table}

\begin{table}
\caption{\label{tabhm=11fracRF}Percentage of times (out of 1,000 replications) when 0, 1 and 2 breaks in RF were detected with $\sup$-$Wald$ and $\sup$-$F$ bootstrap  tests before testing the true null hypothesis $H_0:m=1$ vs. $H_1:m=2$ in {\it Scenario:(h,m)=(1,1)}.}
\begin{center}
\addtolength{\tabcolsep}{-2pt}
\begin{tabular}{@{}c ccc c ccc c ccc c ccc @{}}
\toprule
\midrule
&\multicolumn{3}{c}{WR $\sup$-$Wald$}&&\multicolumn{3}{c}{WF $\sup$-$Wald$}&&\multicolumn{3}{c}{WR $\sup$-$F$}&&
\multicolumn{3}{c}{WF $\sup$-$F$}\\\cmidrule{2-4}\cmidrule{6-8} \cmidrule{10-12} \cmidrule{14-16}
&\multicolumn{3}{c}{\% of RF  breaks}&&\multicolumn{3}{c}{\% of RF  breaks}&&\multicolumn{3}{c}{\% of RF  breaks}&&
\multicolumn{3}{c}{\% of RF  breaks}\\
\cmidrule{2-4}\cmidrule{6-8} \cmidrule{10-12} \cmidrule{14-16}
T&0&1&2 &&0&1&2&&0&1&2&&0&1&2\\
\cmidrule{2-4}\cmidrule{6-8} \cmidrule{10-12} \cmidrule{14-16}

\multicolumn{16}{c}{{\it Case A}}\\
\midrule
120&0   &94.8&  5.2&&   0&      98.6&   1.4&&   0&      94.5&   5.5     &&0&    93.9&   6.1\\
240&0&  95.9&   4.1&&   0&      95.8&   4.2&&   0&      94.4&   5.6&&   0&      93.9&   6.1\\
480&0   &96.3&  3.7&&   0&      96.2&   3.8&&   0&      94.7&   5.3&&   0&      94.6&   5.4\\

\midrule
\multicolumn{16}{c}{{\it Case B}}\\
\midrule
120&0.5&        93.6&   5.9&&   0.8&    98&     1.2&&   0&      92.8&   7.2&&   0.01&   92.89&  7.1\\
240&0.2&        95.6&   4.2&&   0.2&    95      &4.8    &&0     &92     &8&&     0       &91.9&  8.1\\
480&0   &96.2&  3.8&&   0&      96.1&   3.9&&   0&      93.4&   6.6&&   0&      93.2&   6.8\\

\midrule
\multicolumn{16}{c}{{\it Case C}}\\
\midrule
120&0.1&        94.2&   5.7&&   0.4&    98&     1.6&&   0&      93.4&   6.6&&   0&      93.6&   6.4\\
240&0&  93.7&   6.3&&   0&      95.8&   4.2&&   0&      94.2&   5.8&&   0&      93.8&   6.2\\
480&0&  94.1&   5.9&&   0       &94.6&  5.4&&   0&      94.4&   5.6&&   0&      94.2&   5.8\\

\midrule
\multicolumn{16}{c}{{\it Case D}}\\
\midrule
120&0&  94.4&   5.6&&   0.2&    97.8&   2&&     0       &94.1&  5.9&&   0       &94.2&  5.8\\
240&0&  95.7&   4.3&&   0       &96&    4       &&0     &94.7&  5.3&&   0&      94.9&   5.1\\
480&0   &95&    5       &&0&    93.6&   6.4&&   0&      93.8&   6.2&&   0&      94.1&   5.9\\

\midrule\bottomrule
\multicolumn{16}{l}{\small{Notes. In RF we tested $H_0:h=\ell$ vs $H_1:h=\ell+1$, $\ell=0,1$ with the WR   }}\\\multicolumn{16}{l}{\small{and WF bootstrap $\sup$-$Wald$ and $\sup$-$F$ tests. A(n) (additional) break was }}\\\multicolumn{16}{l}{\small{  detected if the  WR or WF bootstrap $p$-value was smaller than $5\%$.}}
\end{tabular}%
\end{center}
\end{table}

\begin{table}
\caption{\label{tabhm=10small}{\it Scenario:(h,m)=(1,0)} - rejection probabilities from testing $H_0:\,m=0$ vs. $H_1:\,m=1$ with bootstrap $\sup$-$Wald$ test (size of break in RF is smaller than in Table \ref{tabhm=10} of the paper)}
\begin{center}
\addtolength{\tabcolsep}{-2pt}
\begin{tabular}{@{}c ccc c ccc c ccc c ccc c ccc c ccc@{}}
\toprule
\midrule
&\multicolumn{3}{c}{WR bootstrap}&&\multicolumn{3}{c}{WF bootstrap}&&\multicolumn{3}{c}{WR bootstrap}&&
\multicolumn{3}{c}{WF bootstrap}&&\multicolumn{3}{c}{WR bootstrap}&&\multicolumn{3}{c}{WF bootstrap}\\
&\multicolumn{3}{c}{RF tested}&&\multicolumn{3}{c}{RF tested}&&\multicolumn{3}{c}{RF not tested}&&
\multicolumn{3}{c}{RF not tested}&&\multicolumn{3}{c}{breaks detected}&&\multicolumn{3}{c}{breaks detected}\\
\cmidrule{2-4}\cmidrule{6-8} \cmidrule{10-12} \cmidrule{14-16} \cmidrule{18-20} \cmidrule{22-24}\\
T&10\%&5\%&1\%&&10\%&5\%&1\%&&10\%&5\%&1\%&&10\%&5\%&1\%&&0&1&2&&0&1&2\\
\cmidrule{2-4}\cmidrule{6-8} \cmidrule{10-12} \cmidrule{14-16} \cmidrule{18-20} \cmidrule{22-24}
\multicolumn{24}{c}{{\it Case A}}\\
\midrule
120&9.9&        4.9&    1.1&&   15.5&   8.8&    2.2&&   10.3&   5.7&    1       &&15.6& 7.6&    2.3&&   5&      87.2&   7.8&&   5.5&    92.4&   2.1\\
240&9.5&        4.6&    0.8&&   13.4&   6.5&    1.2&&   9.6&    4.5&    0.8&&   12.8&   7.5&    1.2&&   0       &93.2&  6.8&&   0&      96.3&   3.7\\
480&9.3&        4.5&    0.8&&   11.6&   5.7&    1       &&9.5&  4.7&    0.8&&   11.5&   5.6&    1.1&&   0&      95.4&   4.6&&   0       &95.7&  4.3\\

\midrule
\multicolumn{24}{c}{{\it Case B}}\\
\midrule
120&10.6&       5.3&    1&&     15&     8.1&    2       &&11.6& 5.9&    1.1&&   14.6&   8.3&    1.9&&   6.4&    85.3&   8.3&&   6.8&    91.5&   1.7\\
240&9.7&        4.8&    0.9&&   11.9&   6.2&    1.3&&   9.5&    4.6&    1&&     12&     6.1&    1.4&&   0.1&    93&     6.9&&   0.1&    95.6&   4.3\\
480&9.8&        5.2&    0.9&&   11.3&   6.4&    1.3&&   9.9&    4.9&    0.5&&   11.3&   5.9&    1.4&&   0&      95&     5&&     0&      94.8&   5.2\\

\midrule
\multicolumn{24}{c}{{\it Case C}}\\
\midrule
120&9.7&        4.8&    1&&     15.6&   7.8&    1.7&&   9.5&    4.7&    1       &&14.2& 7.9&    1.4&&   8       &84.9&  7.1&&   9.1&    88.9&   2\\
240&9.9&        5.4&    0.7&&   15.5&   10&     2.5&&   10.7&   5.5&    1.2&&   15&     9.7&    1.9&&   0&      94.8&   5.2&&   0&      96.3&   3.7\\
480&9.4&        4.9&    1.5&&   11.9&   6.5&    2&&     9.8&    4.4&    1.2&&   12.5&   6.4&    2&&     0&      92.8&   7.2&&   0&      92.3&   7.7\\

\midrule
\multicolumn{24}{c}{{\it Case D}}\\
\midrule
120&7.6&        3.6&    0.7&&   11.3&   6.2&    1.4&&   7.6&    3.7&    0.4&&   11.6&   5&      1.3&&   8.1&    86.3&   5.6&&   8.9&    89.1&   2\\
240&8.9&        4.1&    0.8&&   13.2&   7.5&    1.9&&   9.5&    4.4&    0.8&&   12.9    &7.7&   1.3&&   0&      93.9&   6.1&&   0       &95.6&  4.4\\
480&10.3&       5.5&    1.1&&   13.2&   7.4&    1.2&&   10.1&   5.7&    1.1&&   13.4&   7.1&    1.7&&   0&      92.9&   7.1&&   0&      93.3&  6.7\\

\midrule
\bottomrule

\multicolumn{24}{l}{\small{Notes. All the columns refer to the case when $H_0: m=0$ is true. The first two columns correspond to the case when }}\\ \multicolumn{24}{l}{\small{we tested sequentially for a maximum of 2 breaks (using the WR/WF bootstrap $\sup$-$Wald$ for OLS to test }}\\ \multicolumn{24}{l}{\small{ $H_0:h=\ell$ vs $H_1:h=\ell+1$, $\ell=0,1$). The resulting number of RF breaks and their estimated location was imposed  }} \\ \multicolumn{24}{l}{\small{in each simulation prior to estimating RF and SE and testing $H_0:\,m=0$. The next two columns refer  to the case }} \\ \multicolumn{24}{l}{\small{when under the null and the alternative hypotheses the number of breaks in RF is held at the true number ( $h=1$) }}\\ \multicolumn{24}{l}{\small{and the estimated location of the  RF break in imposed when the SE was estimated. The last two columns give the }}\\ \multicolumn{24}{l}{\small{percentage of times that the bootstrap tests detected and    imposed 0, 1 or 2 breaks when SE was estimated.}}
\end{tabular}%
\end{center}
\end{table}

\begin{table}
\caption{\label{tabhm=11small}{\it Scenario:(h,m)=(1,1)} - rejection probabilities from testing $H_0:\,m=1$ vs. $H_1:\,m=2$ with bootstrap $\sup$-$Wald$ test (size of break in RF is smaller than in Table \ref{tabhm=11} of the paper)}
\begin{center}
\addtolength{\tabcolsep}{-2pt}
\begin{tabular}{@{}c ccc c ccc c ccc c ccc c ccc c ccc@{}}
\toprule
\midrule
&\multicolumn{3}{c}{WR bootstrap}&&\multicolumn{3}{c}{WF bootstrap}&&\multicolumn{3}{c}{WR bootstrap}&&
\multicolumn{3}{c}{WF bootstrap}&&\multicolumn{3}{c}{WR bootstrap}&&\multicolumn{3}{c}{WF bootstrap}\\
&\multicolumn{3}{c}{RF tested}&&\multicolumn{3}{c}{RF tested}&&\multicolumn{3}{c}{RF not tested}&&
\multicolumn{3}{c}{RF not tested}&&\multicolumn{3}{c}{breaks detected}&&\multicolumn{3}{c}{breaks detected}\\
\cmidrule{2-4}\cmidrule{6-8} \cmidrule{10-12} \cmidrule{14-16} \cmidrule{18-20} \cmidrule{22-24}\\
T&10\%&5\%&1\%&&10\%&5\%&1\%&&10\%&5\%&1\%&&10\%&5\%&1\%&&0&1&2&&0&1&2\\
\cmidrule{2-4}\cmidrule{6-8} \cmidrule{10-12} \cmidrule{14-16} \cmidrule{18-20} \cmidrule{22-24}
\multicolumn{24}{c}{{\it Case A}}\\
\midrule
120&9.4&        5       &0.8&&  9.8&    5.2&    1       &&9.7&  4.9&    0.8&&   9.3&    4.4&    0.7&&   7.5&    85&    7.5&&   8.2&    89.5&   2.3\\
240&9.1&        5.3&    0.8&&   10.5&   4.7&    0.6&&   9.7&    5.5&    0.8&&   10.3&   4.6&    0.5&&   0.5&    92.6&   6.9&&   0.7&    95.9&   3.4\\
480&9.7&        5.5&    1.1&&   10.8&   5.6&    1       &&9.6&  5.3&    1&&     10.9&   5.4&    0.9&&   0.5&    95.7&   3.8&&   0.5&    96&     3.5\\

\midrule
\multicolumn{24}{c}{{\it Case B}}\\
\midrule
120&8.3 &3.7&   0.6&&   8.1&    3.6&    0.9&&   8.4&    4       &0.7&&  7.7&    3.2&    0.6&&   7.7&    84&     8.3&&   7.8&    90&     2.2\\
240&10.1&       5       &0.9&&  10.6&   5.2&    1.1&&   10.1&   5.2&    1&&     10.5&   5.3&    1&&     1.3&    91.5&   7.2&&   1.5&    94.4&  4.1\\
480&9.5&        5.2&    1.1&&   10.6&   5.7&    1.5&&   9.1&    4.9&    0.9&&   10.3&   5.4&    1.2&&   0.8&    94.5&   4.7&&   1&      94.8    &4.2\\

\midrule
\multicolumn{24}{c}{{\it Case C}}\\
\midrule
120&8.7&        4.3&    0.8&&   9.3&    4.4&    0.4&&   9.3&    4.3&    0.9&&   9&      4.1&    0.2&&   8.9&    84.8&   6.3&&   10.6&   87.5&   1.9\\
240&10.2&       5.3&    1.1&&   10.9&   4.9&    0.9&&   10.3&   5.5&    1.2&&   10.7&   5&      0.9&&   1.2&    93.1&   5.7&&   1.5&    94.9&   3.6\\
480&10.6        &5.4&   1&&     12.2&   6       &1&&    10.4&   5.2&    1.1&&   11.9&   5.6&    0.8&&   0.4     &93.6&  6       &&0.5&  93.9&   5.6\\

\midrule
\multicolumn{24}{c}{{\it Case D}}\\
\midrule
120&9.8 &4.4&   0.5&&   9.3&    4.2&    0.5&&   10&     4.2&    0.5&&   8.6&    3.5&    0.1&&   9&      84.2&   6.8&&   10.3&   87.4&   2.3\\
240&9.8&        5.4     &0.8&&  11.1&   4.5&    0.4&&   9.9&    5.5     &0.8&&  11.1&   4.5&    0.4&&   0.4&    93.7&   5.9&&   0.5&    95.1&   4.4\\
480&9.7&        4.9&    1.5     &&10.6& 5.2&    1.4&&   9.8&    4.7&    1.5&&   10.6&   5.1&    1.3&&   0.1&    94      &5.9&&  0.2&    94.2&   5.6\\

\midrule
\bottomrule

\multicolumn{24}{l}{\small{Notes. All the columns refer to the case when $H_0: m=1$ is true. The first two columns correspond to the case when }}\\ \multicolumn{24}{l}{\small{we tested sequentially for a maximum of 2 breaks (using the WR/WF bootstrap $\sup$-$Wald$ for OLS to test }}\\ \multicolumn{24}{l}{\small{ $H_0:h=\ell$ vs $H_1:h=\ell+1$, $\ell=0,1$). The resulting number of RF breaks and their estimated location was imposed  }} \\ \multicolumn{24}{l}{\small{in each simulation prior to estimating RF and SE and testing $H_0:\,m=1$. The next two columns refer  to the case }} \\ \multicolumn{24}{l}{\small{when under the null and the alternative hypotheses  the number of breaks in RF is held at the true number ( $h=1$) }}\\ \multicolumn{24}{l}{\small{and the estimated location of the  RF break in imposed when the SE was estimated. The last two columns give the }}\\ \multicolumn{24}{l}{\small{percentage of times that the bootstrap tests detected and    imposed 0, 1 or 2 breaks when SE was estimated.}}
\end{tabular}%
\end{center}
\end{table}
\clearpage
\newpage

\end{document}